%% file: main.tex
\documentclass{article}

\usepackage[letterpaper,top=2cm,bottom=2cm,left=3cm,right=3cm,marginparwidth=1.75cm]{geometry}

\PassOptionsToPackage{dvipsnames}{xcolor}

\usepackage{amsmath, amsthm, amsfonts, amssymb}
\usepackage{mathtools}
\numberwithin{equation}{section}
\usepackage{tikz}
\usetikzlibrary{positioning, arrows.meta, calc, decorations.pathreplacing, fit, backgrounds, patterns}
\usepackage{graphicx}
\usepackage{booktabs}
\usepackage{float}

\usepackage[section]{placeins}

\AtBeginDocument{\captionsetup{skip=4pt,font=small}}
\definecolor{paulix}{HTML}{D55E00}
\definecolor{mcmcsample}{HTML}{6B8E3D}
\definecolor{panelcolor}{HTML}{5B7896}
\definecolor{pauliy}{HTML}{009E73}
\definecolor{pauliz}{HTML}{0072B2}
\definecolor{routecol}{HTML}{5C7299}
\definecolor{surfhero}{HTML}{7C74D6}
\definecolor{routeraddition}{HTML}{7A1FA2}
\definecolor{changegreen}{HTML}{008000}
\usepackage{graphicx}
\theoremstyle{plain}
\newtheorem{theorem}{Theorem}[section]
\newtheorem{proposition}[theorem]{Proposition}
\newtheorem{lemma}[theorem]{Lemma}

\theoremstyle{definition}

\theoremstyle{remark}

\usepackage[colorlinks=true, allcolors=blue]{hyperref}
\usepackage{soul}

\makeatletter
\newcommand*\l@smhead[2]{%
  \addpenalty\@secpenalty
  \addvspace{1.0em \@plus\p@}%
  \begingroup
    \large\bfseries
    \parindent \z@
    \leavevmode
    #1\par
    \nobreak
  \endgroup}
\makeatother

\usepackage[numbers,square,sort&compress]{natbib}%
\usepackage[capitalize,noabbrev]{cleveref}
\usepackage{caption}
\usepackage{subcaption}

\usepackage[textsize=tiny]{todonotes}

\input{dataset-instance.tex}

\newif\ifshowdiff \showdifftrue
\newcommand{\newcitep}[1]{\ifshowdiff\begingroup\hypersetup{citecolor=BrickRed}\citep{#1}\endgroup\else\citep{#1}\fi}
\newcommand{\newcitet}[1]{\ifshowdiff\begingroup\hypersetup{citecolor=BrickRed}\citet{#1}\endgroup\else\citet{#1}\fi}

\title{\vspace{-3em}
Hamilton-Zero: A Neural Tensor-Network Foundation Model for Ground States of Arbitrary Quadratic Qubit Hamiltonians
}

\author{%
Timothy Heightman$^{1,2,\ast}$ \quad
Elena Orlova$^{2}$ \quad
Philip Mantrov$^{2}$ \quad
Aleksei Ustimenko$^{2,\ast}$
\\[0.35em]
{\small $^{1}$ICFO -- Institut de Ci\`{e}ncies Fot\`{o}niques, The Barcelona Institute of Science and Technology,}\\
{\small 08860 Castelldefels, Barcelona, Spain}\\
{\small $^{2}$Simulacra Research Inc., London, UK \& Chicago, USA}\\[0.15em]
{\small \texttt{timothyheightman@simulacra-ai.com} \quad \texttt{aleksei@simulacra-ai.com}}%
}

\date{}

\begin{document}

\maketitle
\renewcommand{\thefootnote}{\fnsymbol{footnote}}
\footnotetext[1]{These two authors contributed equally. Code: \href{https://github.com/simulacra-research/HamiltonZero}{https://github.com/simulacra-research/HamiltonZero}}

\renewcommand{\thefootnote}{\arabic{footnote}}
\setcounter{footnote}{0}

\vspace{-1.2em}
\begin{abstract}
A central promise of useful quantum advantage is the ability to compute ground states of Hamiltonian systems beyond the reach of classical simulation methods. Here we demonstrate that this problem can be effectively amortized across an arbitrary and universal set of Hamiltonians by a foundation model with $\sim0.5$B variational parameters, trained with contemporary techniques from large language models and deep reinforcement learning. To do this, we formulate $\text{spin-}1/2$ quantum ground-state learning as manifold variational optimisation over centrally odd scalar functions on $\mathrm{SU}(2)^N$. This replaces explicit Hilbert-space vector amplitudes with manifold functions on which the Hamiltonian acts through Lie derivatives, evaluated by custom automatic differentiation primitives. We prove that the resulting variational principle on this manifold preserves the $\text{spin-}1/2$ sector's ground-state upper bound using the Peter-Weyl theorem and justify the choice of such a representation with a no-go theorem for pure state foundation NQS. We then pre-train our foundation model on a dataset of hundreds of thousands of different Hamiltonian systems, varying the connection topology, system size, interaction types and strengths, bringing together a century of many-body literature. Using a novel $\mathrm{SU}(2)$ replica-exchange Langevin sampler and sharded natural-gradient optimisation, we train our model with our own extension of the Kronecker-Factored Approximate Curvature (KFAC) optimiser on system sizes up to 64 qubits. On a held-out generalisation dataset, we fine-tune our model on system sizes of up to 1024 qubits, and evaluate on systems up to 8100 qubits.
\end{abstract}

\setcounter{tocdepth}{1}
\tableofcontents
\newpage

\section{Introduction}

A central problem in quantum many-body systems is computing the ground state of a given spin-$1/2$ Hamiltonian. Through a fermion-to-qubit mapping such as Jordan-Wigner, it encodes the electronic-structure Hamiltonians of quantum chemistry, and with them this ground state problem is relevant from materials design through to drug discovery \cite{bauer2020quantum,mcardle2020quantum,chien2026simulating}. Ground states across a Hamiltonian family also trace out its phase diagram, as well as encode the solutions of combinatorial optimisation problems relevant to supply-chain \cite{singh2022combinatorial}, portfolio optimisation \cite{he2014two}, electric vehical charging \cite{touati2012combinatorial,james2019online} or flight routing \cite{garcia2019combinatorial} to name some examples \cite{yu2013industrial, du2013handbook}. 

For a large many-body system, however, computing such ground states has proven famously difficult due to the curse of dimensionality. This curse refers to the fact that a spin-$\tfrac12$ Hamiltonian on $N$ sites has a ground state in $2^N$ complex dimensions. Decades of research have uncovered many complementary methods to solve this problem as generally as possible. Of note are tensor-network methods~\citep{shi2006classical,murg2010simulating,nakatani2013efficient}, variational quantum algorithms (VQAs)~\citep{peruzzo2014variational,cerezo2021variational}, and neural quantum states (NQS)~\citep{carleo2017solving, lange2024architectures}. The first is the natural tool in one dimension, where the ground states of gapped local Hamiltonians obey an area law and a fixed bond dimension captures them. In two dimensions, that same area law turns against the method, since the entanglement across a region's boundary, and with it the bond dimension required, grows with that boundary, confining efficient tensor networks to a low-entanglement corner of Hilbert space. The second, VQAs, has proved initially successful for smaller many-body systems, but has trainability and expressivity tradeoffs due to the so-called barren plateau phenomenon \cite{larocca2025barren, ragone2024lie, mcclean2018barren}. As such, the third method, NQS, has taken off in popularity over the last years, solving for variational upper-bounds to ground states on large systems  , with the largest systems approaching $1500-2000$ qubits \cite{viteritti2026approaching}. This accuracy on single systems alone however has not, until recently, proven transferable across \emph{different} Hamiltonians \cite{rende2025foundationnqs}.

Indeed, recent works aiming for this transferrability constitute early foundation neural quantum states (FNQS), in which a single neural quantum state is trained on many different hamiltonians \cite{rende2025foundationnqs}. This amortises the cost of training a new network from scratch at the per-system level. Thus far however, the community has gone in a direction where they fix the Hamiltonian's interaction graph (i.e. the topology of the connections between spin sites), varying only the coefficients of a given Hamiltonian structure, and keeping both said structure fixed as well as the number of spins \cite{zaklama2025attention, rende2025foundationnqs}, see Figure~\ref{fig:generalisation-surface}.

This lineage is short but fast moving and highly active area of investigation. The transformer quantum state of Zhang and Di Ventra \citep{zhang2023transformer} was the first to carry a single network across many instances of one model class. Rende and collaborators \citep{rende2024finetuning} then pretrained at one point of a phase diagram and fine-tuned outward from it. The same group made the idea systematic in their foundation neural quantum states \citep{rende2025fnqs}, a single network conditioned on the couplings of one lattice and swept across its phase diagram and disorder. The same line later settled a two-dimensional spin glass \citep{viteritti2025spinglass}, and the construction carried over to lattice fermions \citep{zaklama2025attention}. The cost these models amortise is that of deep variational Monte-Carlo. This per-system paradigm exemplified in the electronic setting by FermiNet, PauliNet and DeepErwin \citep{pfau2020ferminet,hermann2020pauli,gerard2022goldstandard}, each of which solves a single system from scratch to high accuracy, at a price paid anew for every Hamiltonian. Several of the second-order optimisation primitives we build on below originate there. Within the spin-FNQS lineage reviewed above however, prior models condition primarily on the coefficients of a fixed interaction graph, with the only reconfigurable such variational models being the variational architectures on quantum computers, which are themselves also trained from scratch for every new instance of a problem.

\begin{figure}[!htbp]
  \centering
  \input{generalisation-surface.tikz}
  \caption{The generalisation surface of foundation wavefunction
    models. Each bar is one pretrained-model family. Its length is how the model varies its coupling across its training and evaluation corpus, with the axes ordered by structural depth. Solid bars mark demonstrated generalisation, hatched bars partial or model-dependent reach, with each reviewed in the main text.}
  \label{fig:generalisation-surface}
\end{figure}
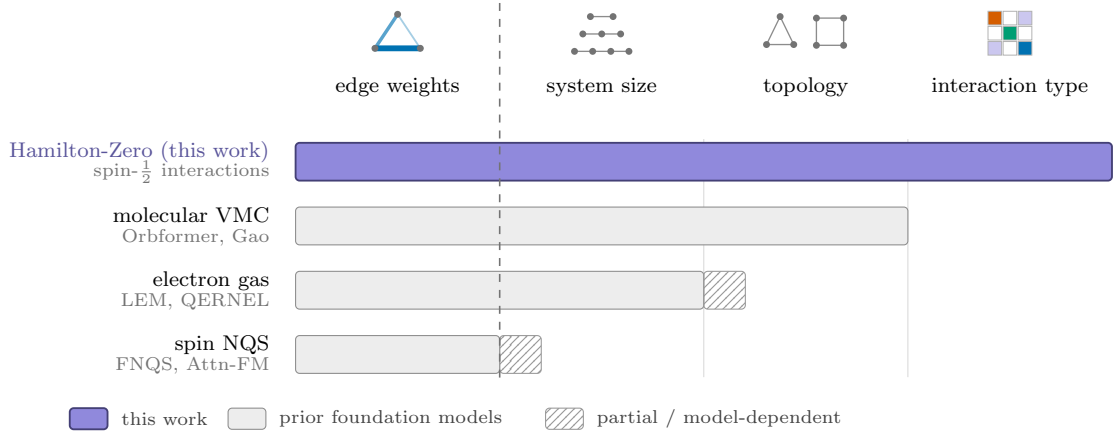

Two neighbouring programmes do vary the structure, though outside the qubit paradigm considered here. The Large Electron Model~\cite{zaklama2026lem} and its successor QERNEL~\cite{nazaryan2026qernel} are foundation models for interacting electrons which generalise across particle number. QERNEL further conditions on the depth of the moir\'e superlattice potential, the single parameter that tunes a semiconductor heterobilayer between its liquid and crystalline electronic states. The molecular wavefunction programme reaches furthest of all on structure. It runs from weight sharing across geometries~\cite{scherbela2022weightsharing,scherbela2023foundation} to the joint-molecule ansatz of Gao and G\"unnemann~\cite{gao2023generalizing}, and on to Orbformer~\cite{foster2025orbformer}, pretrained across $22{,}000$ molecular structures. These genuinely change the structure they condition on. Yet the type interaction remains fixed even as the geometric structure varies owing to the nature of the systems modelled. That is, in these models, we see only the Coulomb for the electrons, and a single symmetry class for the spins. Indeed, this obstruction is mechanical. A Hamiltonian read in as a fixed-length vector of couplings has fixed its interaction graph in the length and the meanings of the encoding of that vector. It cannot then be conditioned on a different type of interaction, nor a different number of spins. To cross from the coefficients to the structure, the Hamiltonian must enter as a variable and typed interaction graph. To-date, no such model has proven successful with such an architectural choice. Two final threads sit just off this map, neither a transfer model in itself. The first is the per-system transformer wavefunction of Roca-Jerat and collaborators~\cite{rocajerat2024transformer}. The second is the neural scaling laws of Jiang et al.~\cite{jiang2025scaling}, our empirical reason to expect pretraining at the scale of hundreds of thousands of different Hamiltonians to amortize, also more recently reported in \cite{rende2026scaling}.

In this contribution, we introduce Hamilton-Zero, a novel foundation model that generalises across both multiple system sizes, different interaction topologies, and different interaction types, trained on a corpus of Hamiltonian systems spanning a century of quantum many-body literature. This pretraining corpus has two levels of structure. At the top level sit thousands of distinct quadratic interaction graphs. Within each graph family, we then vary both the interaction parameters, so that a single family's data is comparable in scope to the entire corpus of an existing FNQS model, and the number of spins, up to $N = 64$. Combined with a hot-swapped perturbation scheme in pretraining, these expansions yield the hundreds of thousands of distinct training Hamiltonian topologies, system sizes and interaction types.

For larger systems up to $8000$ qubits, in many instances our results offer energy bounds that are not accessible by tensor network methods, or DMRG due to the nature of the non-local interaction topologies made accessible by this model. Indeed, the only other architecture that could currently find these bounds is neural quantum states, trained from scratch on each instance of Hamiltonian interaction topology. However, this is amortised by the open-source release of this model's weights.

To do so, we work in a different configuration space for many-body spin systems, akin to the many-body phase space formalism of \cite{heightman2025foundations} with two notable improvements. First, we represent each factor $\mathrm{SU}(2)\cong S^3$ by four global embedding coordinates subject to one unit-norm constraint. Thus the configuration manifold has intrinsic dimension $3N$ and is supplied to the network through $4N$ constrained coordinates. This representation results in a smooth, site-local differential structure, on which the Hamiltonian acts through Lie derivatives. These, in turn, enable the use of automatic differentiation primitives, scalable to 8000+ qubits, well beyond the early proposals of \cite{heightman2025foundations}. Second, our model genuinely respects the variational principle in the Hilbert space of spin-$1/2$ systems, something which was acknowledged to be an open issue of phase-space-based quantum machine learning \cite{heightman2025foundations}.

Our model's wavefunction hence constitutes a centrally odd scalar function on $\mathrm{SU}(2)^N$, a function class properly containing all possible neural quantum states, as we prove in Theorem~\ref{thm:expressivity-ladder}. To do so, when optimising, we constrain the model to the physical sector identified by the Peter-Weyl theorem. At $547{,}521{,}152$ variational parameters, Hamilton-Zero is, to our knowledge, the largest wavefunction ansatz yet trained for spin systems. At this scale, pretraining practices developed for large language models and reinforcement learning become directly relevant. On the fixed pretraining workload defined in SM~\S~\ref{sec:part4}, an off-the-shelf JAX ecosystem \cite{jax2018github} with \texttt{folx} \cite{folx} required more than one year of wall time. The forked versions we release with this work reduce that estimate to approximately four days.

The rest of this work is structured as follows. In Sec.~\ref{sec:mt-autograd} we describe briefly the key mathematical steps that underpin our automatic differentiation primitives and our choice of manifold function. In Sec.~\ref{sec:mt-ansatz} we briefly describe Hamilton-Zero's structure and similarities to Large Language Models, including a brief overview of our deep reinforcement learning pipeline for inferring entanglement structure. Next in Sec.~\ref{sec:mt-datasets} we describe the key aspects of pretraining and data augmentation. Sec.~\ref{sec:mt-results} then presents the results, in which we show (i) a single shared optimisation trajectory improves the energy of a heterogeneous distribution of hundreds of thousands of Hamiltonians, (ii) that it can solve for energies of unseen systems when evaluated on frozen weights, and fine-tune to within $1\%$ of ED energies on such systems with relatively little compute, (iii) that Hamilton-Zero can be taken far outside the system-size distribution on which it was trained on three case-studies, (iv) that it can witness a phase transition from a single checkpoint of its weights, (v) that it becomes approximately equivariant with respect to gauge groups in spinor-physics and (vi) that the entanglement structure of both pretraining and larger out-of-distribution systems are effectively learned by the router. Of note here is our scaling laws, which empirically suggest our model has $5\times$ better scaling than today's LLMs.

The supplementary material develops each component of Hamilton-Zero in full. In SM~\S~\ref{sec:part1}, we derive
our variational principle for the function class $SU(2)^N \rightarrow \mathbb{C}$,
describing the necessary theory to facilitate Hamiltonian (and other
observable) action via automatic differentiation. Here we also resolve an open problem around harmonic leakage for phase space quantum machine learning that enables future generations of the model to gain even more expressiveness. In SM~\S~\ref{sec:part2} we describe the foundation ansatz, which contains a featuriser for the Hamiltonian dataset's differing interaction topologies, weights and system sizes, an LLM style trunk of attention and FFN layers, and a readout mechanism trained with deep reinforcement learning to enable generalisation across unseen system topologies and system
sizes. In SM~\S~\ref{sec:datasets}, we lay out the training routines we employed for pretraining in full, as well as detail for the training, validation and evaluation datasets we created for pre-training, fine-tuning and zero-shot evaluation. SM~\S~\ref{sec:part4} details our research into optimisation and engineering that made this model's training feasible. Of note here is our novel sampling scheme for MCMC on Riemannian manifolds which facilitates high-quality, decorrelated samples from our model's quaternionic density at unprecedented speeds. Here we also detail an extension of the KFAC optimiser of Martens and Grosse~\cite{martens2015optimizing}, which allows it to work with higher-order Fisher-information tensors than its rank-two counterpart. 

Everything above can be found in the open-source \href{https://github.com/simulacra-research/HamiltonZero}{github repository} for this paper, along with usage instructions for installing and using the model. It is also available on \href{https://huggingface.co/simulacra-research/HamiltonZero}{hugging-face}.

\section{Theory}
\label{sec:mt-autograd}

We consider the ground-state problem for the universal \cite{nielsen2000quantum} spin-$1/2$ family of Hamiltonians,
\begin{equation} \hat{H} = \frac14 \sum_{i<j}\sum_{a,b} J^{ab}_{ij}\,\hat{\sigma}^a_i \hat{\sigma}^b_j + \frac12 \sum_i \sum_a h^a_i\,\hat{\sigma}^a_i ,
\label{eq:hamilton_zero_hamiltoinian}
\end{equation} where $J^{ab}_{ij}$ encodes arbitrary quadratic Pauli interactions and $h^a_i$ encodes local fields.

It is useful to think of $J$ as the data of a weighted graph on $N$ sites, where each unordered pair $(i, j)$ with $i \neq j$ carries a $3 \times 3$ Pauli-pair coupling matrix $J_{ij} \in \mathbb{R}^{3 \times 3}$, which we read as nine \emph{flavours} of edge, one per Pauli-pair $(a, b) \in \{x, y, z\}^{2}$ giving the operator $\sigma^a_i \sigma^b_j$, see Fig.~\ref{fig:bonds-examples}.

\begin{figure}[!htbp]
  \centering
  \begin{subfigure}[b]{0.46\linewidth}
    \centering
    \input{graph-heisenberg.tikz}
    \label{fig:bonds-heisenberg}
    \caption{(a) Heisenberg XXX model on a 1D spin-chain. Each nearest neighbour pair of sites carries XX, YY and ZZ channels.}
  \end{subfigure}
  \hfill
  \begin{subfigure}[b]{0.46\linewidth}
    \centering
    \input{graph-kitaev.tikz}
    \caption{Cross-channel coupling (Kitaev-$\Gamma$) on a 1D chain. Each bond carries a distinct off-diagonal Pauli pair, 1--2 is YZ, 2--3 is ZX, 3--4 is XY. Source and destination Pauli differ, so each edge changes colour at its midpoint.}
    \label{fig:bonds-kitaev}
  \end{subfigure}

  \vspace{1em}

  \begin{subfigure}[b]{0.46\linewidth}
    \centering
    \input{graph-compass.tikz}
    \caption{Compass model on a 2D $2 \times 3$ grid, where horizontal bonds carry XX; vertical bonds carry YY. Bond orientation determines the Pauli channel.}
    \label{fig:bonds-compass}
  \end{subfigure}
  \hfill
  \begin{subfigure}[b]{0.46\linewidth}
    \centering
    \input{graph-longrange.tikz}
    \caption{A J1J1 Heisenberg chain with nearest-neighbour ($J_1$, straight) and next-nearest-neighbour ($J_2$, arcs) couplings, both ZZ.}
    \label{fig:bonds-longrange}
  \end{subfigure}

  \caption{Four interaction-graph examples drawn from the training data. The colour map
    \textcolor{paulix}{\rule{6pt}{6pt}}~X,
    \textcolor{pauliy}{\rule{6pt}{6pt}}~Y,
    \textcolor{pauliz}{\rule{6pt}{6pt}}~Z shows different types of edges.}
  \label{fig:bonds-examples}
\end{figure}
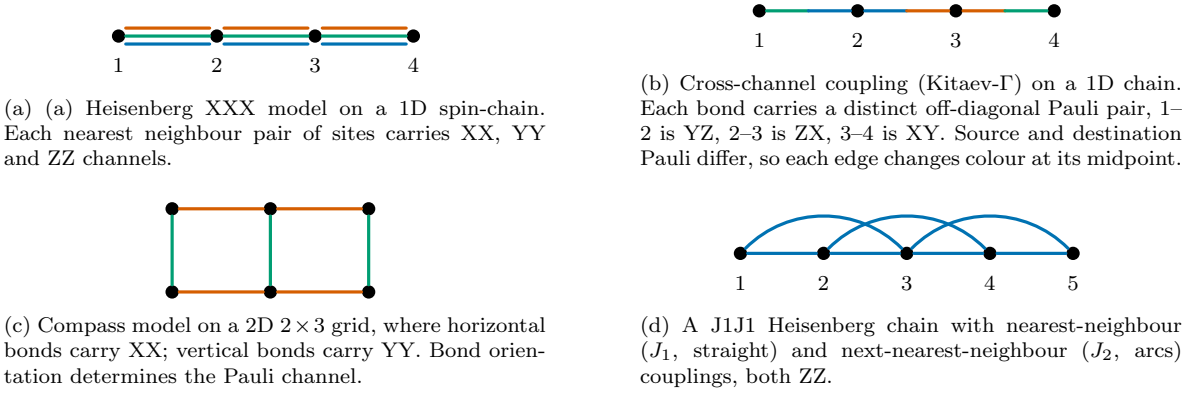

For example, the diagonal entries $a = b$ (XX, YY, ZZ) are the Heisenberg exchanges and the off-diagonal entries (XY, XZ, YX, YZ, ZX, ZY) encode Kitaev-$\Gamma$ \cite{sorensen2021heart,zhang2021variational}, Dzyaloshinskii--Moriya \cite{yang2023first}, compass, and other cross-channel couplings. We will informally refer to the $N \times N$ edges as \emph{bonds}. Meanwhile the diagonal $(i, i)$ self-bonds do not correspond to any quadratic interaction in this graph.

Indeed, standard Neural Quantum States (NQS) \cite{carleo2017solving,zhang2023transformer,rocajerat2024transformer} learn amplitudes for these interaction graphs on the discrete spin basis $\{\uparrow,\downarrow\}^N$. This makes each model tied to one system size and one Hamiltonian instance, so transfer across Hamiltonian structure is difficult \cite{rende2024finetuning,rende2025fnqs,viteritti2025spinglass}. Only recently researchers have tried to start generalising within a given Hamiltonian interaction topology \cite{zhang2023transformer,rende2024finetuning,rende2025fnqs,viteritti2025spinglass,zaklama2025attention}. Hamilton-Zero instead represents the state as a smooth function on the product Lie group $\psi_\theta : \mathrm{SU}(2)^N \to \mathbb{C}$. Each spin is parametrised by a unit quaternion $q_i \in SU(2) \cong S^3 \hookrightarrow\mathbb{R}^4$, so the network consumes $4N$ continuous coordinates rather than the discrete basis of NQS \cite{carleo2017solving}. Thus spin operators can act on this manifold through left-invariant vector fields $L_i^a$, $\hat{\sigma}_i^a \longleftrightarrow -i L_i^a$ for first order and for second, $\hat{\sigma}_i^a \hat{\sigma}_j^b \longleftrightarrow -L_i^a L_j^b $. The Hamiltonian becomes a second-order differential operator on $\mathrm{SU}(2)^N$, which we can compute with automatic differentiation of the model with respect to its input coordinates,
\begin{equation} \hat{H} \longleftrightarrow -\frac14 \sum_{i<j,a,b} J^{ab}_{ij} L_i^a L_j^b - \frac{i}{2}\sum_{i,a} h_i^a L_i^a,
\end{equation}
extending recent works on moment generating functions \cite{heightman2025foundations} to the manifold function directly via custom JAX-primitives of these Lie derivatives. Hence, the Hamiltonian's interaction data is simply a bare coupling tensor $(J,h) \in \mathbb{R}^{N \times N \times 9} \times \mathbb{R}^{3N}$, while its action on states is evaluated by autograd of the manifold wavefunction's inputs. The quaternionic frame realising each $L_i^a$ and the log-derivative identities behind the local energy are Supplimentary Material (SM) \S\ref{sec:part1}; the kernels evaluating these operators at scale are SM~\S\ref{ssec:energy-kernels}.

\label{sec:mt-vp}

Indeed, by moving from amplitudes to functions on $\mathrm{SU}(2)^N$, we have a space $L^2(\mathrm{SU}(2)^N)$ which is larger than the physical spin-$1/2$ Hilbert space. However, the Peter--Weyl theorem \cite{hall2013lie, peter1927vollstandigkeit} gives us a spectrum for this function space,
\begin{equation}
L^2(\mathrm{SU}(2)) = \bigoplus_{j=0,\frac12,1,\ldots} V_j^* \otimes V_j^ ,
\end{equation}
decomposing it site-wise into irreducible spin sectors, $V_j$. The eigenfunctions in the physical subspace satisfy $-\Delta_i \psi = \frac34 \psi$ for every site $i$, where $\Delta_i$ is the Laplace--Beltrami operator on each local set of $SU(2)$ coordinates and $-\Delta_i$ realises the per-site Casimir $\hat S_i^2$ \cite{folland2016course, knapp2001representation}. All physical spin-$1/2$ functions on this manifold satisfy this Casimir relation, however an unconstrained neural function on $\mathrm{SU}(2)^N$ may leak into higher $j$ sectors giving unphysical energies below the ground state.

Our ansatz avoids the leak by construction, as we impose per-site oddness $\psi(\ldots,-q_i,\ldots) = -\psi(\ldots,q_i,\ldots)$, and we keep the wavefunction linear in each site's quaternion. This is because oddness removes all integer-$j$ sectors, and linear functions of $q_i$ are exactly the span of the four spin-$1/2$ Wigner $D$-matrix elements $D^{1/2}_{mn}(q_i)$, so a centrally odd wavefunction that is linear in every quaternion lies analytically in the spin-$1/2$ sector,

\begin{equation} \psi_\theta(q_1, \ldots, q_N) = \sum_{\boldsymbol{m}, \boldsymbol{n}} T_{\boldsymbol{m}, \boldsymbol{n}}(\theta) \prod_{k=1}^{N} D^{1/2}_{m_k n_k}(q_k),
\end{equation}
 On this sector, the Lie-derivative operator is unitarily equivalent to the physical Hamiltonian, so we may therefore define a variational principle,
\begin{equation}
  \frac{\langle\psi_\theta|\hat H|\psi_\theta\rangle}
       {\langle\psi_\theta|\psi_\theta\rangle} \;\ge\; E_0 ,
\end{equation}
and every energy our optimiser reports is a rigorous upper bound on the physical ground-state energy. See SM~\S~\ref{sec:part1} for a full exposition and derivation of this variational principle. We emphasise that despite this linearity, nothing prevents a non-linear conditioning on the Hamiltonian's interaction data, which is precisely where Hamilton-Zero's expressivity lies. We defer details of its architecture to Sec.~\ref{sec:mt-ansatz} and SM~\S~\ref{sec:part2}.

Had the ansatz been nonlinear in the quaternions however, oddness alone would leave the higher half-integer sectors open, meaning leakage could occur in principle. However, we can resolve this open problem of leakage to higher-spin sectors recorded in \cite{heightman2025foundations} with a regulariser that pushes the model towards the physical sector, and an energy penalty that analytically restores the variational principle. For a full discussion on these cases, we refer to Sup. Mat.~\ref{sec:part1} and Appendix~\ref{app:casimir-reg} for the proofs. For this generation of model however, the regular variational principle is sufficient when we have centrally odd functions that are multi-linear in the manifolds quaternionic coordinates.

Our wavefunction representation belongs to a theoretical lineage nearly a century old. Casimir first treated angular momentum as invariant differential operators on the rotation group \cite{casimir1931rotation}, and the Bopp--Haag model built spin-$\tfrac12$ from wavefunctions on its double cover \cite{bopp1950spin}. Dankel then recovered the Pauli equation from Nelson diffusions on that same manifold \cite{nelson1966derivation,dankel1970mechanics}, with variational and analytic treatments following from Guerra--Morato and Wallstrom respectively \cite{guerra1983quantization,wallstrom1990pauli}. $\mathrm{SU}(2)$-aware NQS are recent by comparison, working in coupled angular-momentum bases \cite{vieijra2020rbm}, through Clebsch--Gordan contractions \cite{vieijra2021manybody}, or with gauge-equivariant wavefunctions on continuous link variables \cite{spriggs2026su2}, and all enforce exact symmetry through representation labels or equivariant layers. Hamilton-Zero assumes no symmetry. It takes the group manifold as configuration space for generic spin-$\tfrac12$ Hamiltonians, pins the sector by per-site oddness and degree-one quaternion dependence, and thus descends from the older line of works.

Next let us compare the expressivity of our model's ambient function class in comparison to NQS. In brief, NQS span a variational sub-manifold of their full Hilbert space unless they have control over $\mathcal{O}(2^N)$ independent amplitudes through variational amplitudes, which is dimensionally cursed (and the reason why practical NQS architectures span a smaller variational sub-manifold). Hamilton-Zero does not inherit this
restriction. Since the model is centrally odd and multilinear in the quaternion coordinates, it has the capacity to act in
\begin{equation}
  \mathcal F_{\mathrm{lin}}
  \;\cong\;
  \bigl(V_{1/2}^{*}\otimes V_{1/2}\bigr)^{\otimes N},
  \qquad
  \dim\mathcal F_{\mathrm{lin}}=4^N,
\end{equation}
which strictly contains the canonical lift of the entire NQS Hilbert space, and therefore also contains every practical NQS variational manifold. The Hamiltonian acts on the one Peter–Weyl index and leaves the additional index untouched, so this larger function class still respects the physical variational principle. If we allow the model to depend non-linearly on each quaternion (but remain centrally odd), its capacity grows again to include all higher half-integer sectors. This however requires an energy penalty to remain a variational principle on the original Hilbert space, as well as regularisation. We discuss these aspects further in SM~\S~\ref{sec:part1} and prove an expressivity hierarchy in Theorem~\ref{thm:expressivity-ladder}. The theorem compares ambient function classes; the expressivity of a finite parameterisation remains architecture-dependent.

Finally, we prove a topological no-go theorem for pure state foundation NQS, which implies that no matter how expressive such model can get, on non-trivial obstructed gapped Hamiltonian families, it's worst case energy error will always somewhere be greater or equal to one full gap and the neural wavefunction will be orthogonal to the ground state itself in the family of Hamiltonians that such foundation NQS attempts to represent.

\section{Hamilton-Zero's Architecture}
\label{sec:mt-ansatz}

In this section, we give a brief account of the architecture of our model, deferring its technicalities to SM~\S~\ref{sec:part2}.

The architecture has two types of inputs. First is a sampled spin configuration $q \in \mathrm{SU}(2)^N$, and second is a Hamiltonian's interaction tensor $(J,h)$. The two are routed through separate pathways and meet at a single, controlled point, so that the
central oddness and per-site linearity of Sec.~\ref{sec:mt-vp} are
preserved exactly by construction (see Fig.~\ref{fig:architecture-main}). We give a brief account of each component here, deferring technicalities to SM~\S~\ref{sec:part2}.

\begin{figure}[t]
\centering
\input{architecture-short.tikz}
\vspace{-0.5em}
\caption{\small Hamilton-Zero information flow. The Hamiltonian $(J,h)$ is spectrally normalised, featurised in polar form, and propagated through the trunk; a route policy conditioned on the trunk's outputs permutes sites and padding slots before the leaf contextualiser; the spin configuration $q$ enters only at the odd leaf builder, enforcing per-site oddness before the shared rank-four merge tree returns $\log\psi_\theta(q\,|\,p)$. The rail beneath the pipeline is the global stream $g$, refreshed at every stage on the unit-RMS sphere. Component diagrams are SM~\S\ref{sec:part2}.}
\label{fig:architecture-main}
\vspace{-0.75em}
\end{figure}
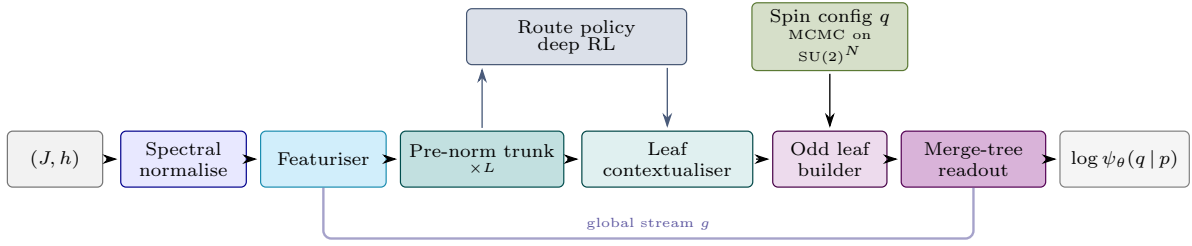

The Hamiltonian data enter the model through a fixed cache followed by a learnable featuriser. The cache bond-symmetrises $J$ and forms a temporary Hermitian matrix from $J$ and $h$ solely to determine one shared spectral scale. Once that scale has been computed, the diagonal field block is discarded and the normalised bond tensor and local field are supplied to the model as separate channels. This places every system in a common numerical regime without needing the model to relearn an overall energy scale.

Next, a featuriser then reads the spectrally normalised Hamiltonian data and produces embeddings for it for the downstream blocks of our model. The featuriser normalises in learned groups by root-mean-square (RMS) maps, and aggregates by column and row attention to produce a per-bond embedding, a per-site embedding, and a global context vector. Bonds and fields absent from a system substitute learned absent tokens rather than zeros. This is so that the absence of a connection, which changes the system's topology, is a representable direction of the feature basis. Meanwhile padding, which carries no structural meaning to the system's interaction topology, is done with masks.

Once the Hamiltonian's data is embedded by the featuriser, it is passed through $L$ pre-norm residual blocks. Each block updates the bonds first through an edge map, biases the per-site attention logits with the fresh bonds so heads can attend along the interaction graph, and closes with a feed-forward update. Every sublayer normalised by the same root-mean-square rule and is damped by a $1/\sqrt L$ residual gain. A global data-stream runs alongside this so that downstream components can also act directly on the feature basis. This global stream lives on the sphere of unit root-mean-square norm, and is refreshed once per block by descriptor attention over the states the block has just written. It moves by the normalised-interpolation update of nGPT, so near-identical readings from symmetric systems saturate instead of accumulating with depth.

We emphasise that the spin configuration never enters this trunk, and that everything up to the leaves is a function of $(J,h)$ alone. This means the bulk of the model's parameters do not need any symmetry constraints that would limit its expressivity. The symmetry constraints are enforced at the single point where quaternionic coordinates enter.


For the coordinate dependence, a bias-free map first lifts each $q_i$ into a linear odd data-stream.  Hamiltonian context supplies only the coefficients of the map applied to that stream.  The resulting raw leaf carrier is therefore linear in $q_i$.  We normalise the carrier, $u$, for numerical stability, record the divided-out leaf magnitude in a log-scale, $s$, and restore it at the root together with all later merge scales.
The normalised carrier alone is odd but nonlinear, meaning only the represented pair $(u,s)$ reconstructs the linearity property needed from Sec.~\ref{sec:mt-autograd} carrier exactly.  Bilinear tree merges then make the final amplitude multilinear in the site coordinates, while the Hamiltonian-only stream retains the model's
nonlinear capacity.

Finally, a readout mechanism contracts the $N$ carriers through a balanced binary tree of merges, in the spirit of tree tensor networks \cite{shi2006classical,murg2010simulating,nakatani2013efficient}, with leaves in a bit-reversed canonical frame and one merge shared across every level and position. The shared merge is a rank-four tensor contracting blockwise, bilinear in its two children so a sign flip at any leaf passes to the root and thus preserves our symmetry constraints from Sec.~\ref{sec:mt-autograd}. The leaves and every internal node renormalise their carriers to unit root-mean-square while storing the divided-out magnitudes in an additive
log-scale, all of which is again restored at the root readout. A Hamiltonian context stream is computed up the tree beside the carriers by the same nGPT-style normalised interpolation as the global context data-stream. This context stream reads the coarse-grained bond field between the two subtrees being merged, and a level attention with a tree-distance bias lets subtrees exchange information laterally, deleting the depth penalty on long-range correlations. Thus our merge tree is conditioned by a coarse-graining of the even data-stream, which we think of as a neural augmentation of tree tensor networks. The closest neural precedent for such a hierarchy is the symmetric tree ansatz of \cite{vieijra2021manybody}, which established that hierarchically composing spin carriers through pairwise Clebsch--Gordan projections and learned reductions yields state-of-the-art variational states with exact non-Abelian symmetry. Our merge tree shares that pairwise fusion structure, whilst its contractions act on continuous manifold carriers through Hamiltonian-conditioned bilinear maps under a learned routing, rather than resolving fixed Clebsch--Gordan channels of discrete configurations. Note that padding slots carry learned contexts of their own, so one tree serves many system sizes.

A learned site-to-slot map adapts the contraction tree to different interaction topologies and system sizes. To learn entanglement structures, we make a map from physical sites to the leaf slots of the merge contraction itself learnable. This map is a discrete latent variable between physical site indices and leaf slots, which we learn with deep reinforcement learning \cite{sutton1998reinforcement, vincent2018introduction}. A pointer-network policy conditioned on the trunk's outputs decodes a permutation slot by slot. At every row it composes every remaining candidate afresh from its site state, the fixed Hamiltonian global, the decode clock, bidirectional prefix and suffix bond summaries, an order-aware prefix, and hole-placement statistics. The selected row states become the leaves of a partial dyadic merge tree. Causal same-level refinement builds its complete subtrees, their mixed-scale dyadic cover produces the next query, and the remaining candidates attend the cover, earlier queries, and one another before the pointer chooses. A conditional quotient then removes indistinguishable choices before sampling. The sampled route relabels sites, bonds, and walkers before the merge mechanism and its context construction runs, so that the corresponding context around each site and its bonds is carried through the learned reordering.

From a physical perspective, the model represents the incoherent mixture of states over sampled routes, a relaxation that costs nothing because the energy is linear in the density matrix and is minimised at a pure state. The policy trains by a score-function gradient with a beam approximation to the policy mode and a self-normalised importance-sampling baseline evaluated on the routed walkers. The walker channel is centred at its route-conditional mean. Here $\pi$ is the routing policy, $m$ is the slot mask, and $\operatorname{Aut}(J,h,m)$ is the group of slot permutations preserving the bond tensor $J$, site fields $h$, and mask $m$. For a route $p$, its orbit and stabiliser are
\begin{equation}
      \mathcal O(p):=\{g\!\cdot\!p:g\in\operatorname{Aut}(J,h,m)\},
  \qquad
  \operatorname{Stab}(p)
  :=\{g\in\operatorname{Aut}(J,h,m):g\!\cdot\!p=p\}.
\end{equation}
For every $g\in\operatorname{Aut}(J,h,m)$, a route $p$ and
$g\!\cdot\!p$ have equal energy. A symmetrised policy is uniform within each
occupied orbit. Writing
$H[\pi]:=-\sum_p\pi(p)\log\pi(p)$ for the policy's Shannon entropy, a
policy supported on one orbit $\mathcal O(p)$ has
\begin{equation}
H[\pi]=\log|\mathcal O(p)|
=\log\frac{|\operatorname{Aut}(J,h,m)|}
{|\operatorname{Stab}(p)|}.
\end{equation}

Thus the measured entropy of a
learned policy and its analytic correspondence to automorphisms of the interaction graph provide us
with a heuristic tool to measure whether the router’s training has converge to the true optimum, see
SM~\S~\S\ref{sssec:entropy-floor}.

Indeed in SM~\S\ref{sec:part2} we give a full exposition of this component along with all the others, including a proof of optimality when the router learns automorphisms. We now proceed to Hamilton-Zero's training data and pretraining.

\section[Pretraining]{\texorpdfstring{Pretraining}{Pretraining}}
\label{sec:mt-datasets}

To train Hamilton-Zero, we collated a curated dataset of $5000$ different interaction topologies, spanning a century of quantum many-body and quantum computing literature \cite{anderson1958absence,lieb1961two,hubbard1963electron,majumdar1969next,ising1925beitrag,su1979solitons,rice1982elementary,kitaev2001unpaired,chen2011classification,sachdev1999quantum,sandvik2007evidence,senthil2004deconfined,monroe2021programmable,browaeys2020many,hensgens2017quantum,heras2014digital,kitaev2006anyons,anderson1973rvb,balents2010spinliquids,savary2016qsl,basko2006mbl,nandkishore2015mbl,abanin2019mbl,page1993average,lipkin1965validity,greenberger1989ghz,wilson1974confinement,kogut1975hamiltonian,kogut1979lattice,jordan1928pauli,bravyi2002fermionic,tranter2018comparison,goldstone1961field,nambu1960quasiparticles,auerbach1994interacting,turner2018scars,bernier2017rydberg,aklt1988vbs,yang1962odlro,sachdev1993gapless,kitaevsuh2018soft,maldacena2016remarks,dzyaloshinsky1958weak,moriya1960anisotropic,kugel1982jahn,nussinov2015compass}.
Because the corpus also includes Hamiltonians native to superconducting, trapped-ion, Rydberg, spin-qubit, and molecular platforms, Hamilton-Zero can be evaluated as a classical foundation surrogate for hardware-relevant ground-state problems. The dataset is stratified into physics cells and coverage tiers, from canonical anchors through frustrated, disordered, gauge and volume-law regimes, each system built deterministically, see SM~\S\ref{sec:datasets} for a full account of the contents of the training dataset. By the standards of large language models, however, such a dataset is small, and trained on it as-is, the model would memorise couplings, fields, and frames alike.

As such, here we briefly discuss the important aspects of the augmentations of our training datasets that prevent overfitting to specific edge-values and field-values, or overfitting to the topologies seen in pre-training. To that end, our training routine involves a series of perturbations of the training data. The first tier perturbs the physics itself, new bonds, new fields, new reference energies to compute; the second applies exact symmetry transformations, which generate unlimited new coupling tensors whose physics, and whose references,
carry over unchanged. Together the two tiers expand the dataset into the hundreds of thousands of distinct training Hamiltonians of our corpus, which is again a relatively small dataset by today's standard for large-language models.

During training, every new epoch receives a perturbed version of the pretraining data, which is made in parallel to the training loop's live progress. This perturbed data is then hot-swapped at every epoch end. To keep the training loop running continuously, this perturbation routine stays \dsLookahead{} rounds ahead of the loop, precomputing exact-diagonalisation references for every perturbed system to $N \le \dsEDmax$.

Each structured system receives one of five changes at equal probability. The changes can be one of the following: (i) leave the system untouched, (ii) add one new bond, (iii) add noise to one existing bond, (iv) remove one bond, or (v) adds one random orbit-symmetric field, with sites, bonds, and orbits sampled stratified over the Weisfeiler--Leman orbit partition of $(J,h)$. Note that for (v) the magnitudes are set relative to the system's own coupling scale. The orbit-symmetric field draws one random three-vector of strength $\eta \sim \mathcal{U}(0.25, 1)$ times the coupling scale and adds it to every site of one orbit. For small random systems $N \le \dsDensifyN$ , we additionally densify towards all-to-all coupling, with a field-free such system gaining a substantial random dense field half the time, and occasional bond drops.

On top of each round's perturbation the worker applies an exact symmetry
change. It draws one Haar-random $u\in\mathrm{SU}(2)$ per
Weisfeiler--Leman orbit of the interaction graph, shared by the sites of
that orbit. Write $u_i$ for the draw assigned to the orbit containing site
$i$. Let $R_{\mathrm{ker}}(u_i)$ denote the adjoint rotation
expressed in the kernel's ordered
$(x,y,z)=(\hat k,\hat\jmath,\hat\imath)$ frame. The exact gauge
augmentation is
\begin{equation}
  J_{ij}\mapsto R_{\mathrm{ker}}(u_i)J_{ij}
                    R_{\mathrm{ker}}(u_j)^\top,
  \qquad
  h_i\mapsto R_{\mathrm{ker}}(u_i)h_i,
  \qquad
  q_i\mapsto q_i\bar u_i,
  \label{eq:mt-gauge-augmentation}
\end{equation}
where $\bar u_i=u_i^{-1}$. Under this right multiplication the
left-invariant fields transform by the adjoint action, so the simultaneous
rewrite of $(J,h,q)$ leaves the local energy and spectrum unchanged.
Sharing $u$ within a class preserves the interaction graph's symmetry
structure and leaves the routing algorithm unchanged.

Physically, this means the model cannot overfit to a choice of frame, and the physics of the interaction graph is identical under this transformation. Hence,
training across gauged rounds penalises whatever gauge violation the model carries. What this encourages is approximate $\mathrm{SU}(2)$ covariance resolved per orbit. Single-site fields transform by one rotation per orbit; a bond between two different orbits sees independent left and right factors, hence approximate $\mathrm{SU}(2) \times \mathrm{SU}(2)$ covariance across inter-orbit bonds, and in the random small-$N$ band the rotations are fully per-site, so every bond there trains under independent factors. The walkers transform on $\mathrm{SU}(2)$ itself, the double cover of the $\mathrm{SO}(3)$ rotations acting on the couplings, and the prediction delta across a gauge transformation is a direct probe of learned covariance (see SM~\S\ref{ssec:augmented}).

\section{Results}
\label{sec:mt-results}

In this section, we show that one shared optimisation trajectory improves the energy of a heterogeneous distribution of Hamiltonians in a quantitatively regular way in \ref{ssec:scaling_laws}. We opt for the V-score \citep{wu2024variational} as our measure of quality here, as it provides a fair comparison across system types and sizes. In \ref{ssec:zero_shot_finetune}, we show the results for zero-shot inference and fine-tuning. Here, zero-shot refers to taking Hamilton-Zero's frozen model weights and evaluating them on the three datasets above, and hence zero training steps were taken for those shots. We applied both to three different datasets detailed in SM~\S~\ref{ssec:v4-eval}. In brief, these datasets cover (A) combinatorial optimisation, (B) unseen interaction and channel topologies, and (C) systems we observed to be hardest when initially developing Hamilton-Zero that we subsequently held out for the model weights released in this work.

We next show how the same pretrained wavefunction can be taken far outside the system-size distribution on which it was trained in \ref{ssec:large_systems}. We study three rather different problems: weighted MaxCut instances \cite{glover2018tutorial,glover2022quantum}, the Pariser--Parr--Pople (PPP) model of a conjugated carbon chain \cite{pariser1953semi1, pariser1953semi2,pople1953electron}, and frustrated Heisenberg antiferromagnets on square and triangular lattices \cite{becca2017quantum}. Together, these cases test non-local weighted graphs, long-range fermionic interactions, and geometrically local frustration. This allows us to see the maximum system sizes Hamilton-Zero can compile and sample from, and whether the resulting zero-shot and fine-tuned energies remain close to reference values when available at such scales.

Next in \ref{ssec:phase_transitions} we show that Hamilton-Zero can witness a phase transitions via the fidelity susceptibility method \cite{gu2010fidelity}. We will apply this to the Ising model and the J1J2 Heisenberg model. We then show that Hamilton-Zero recovers approximate equivariance with respect to the gauge transformations in the perturbation scheme of its pretraining data in \ref{ssec:equivariance_results}.

Finally, in \ref{ssec:router_results} we show how the router learned the entanglement structure of our pretraining dataset, how its commitment to a certain entangling contraction path is indicative of the quality of the energy bounds, and that it generalises to unseen interaction topologies and interaction types on the scale of the large-system case studies. Here we will also show the overall model has learned approximate equivariance with respect to gauge groups actions on the interaction graph $(J,h)$.

\subsection{Pretrain and scaling laws}
\label{ssec:scaling_laws}

After the initial transient, Fig.~\ref{fig:compute-scaling-laws} shows the V-score falls approximately as $C^{-0.53}$, while the ED-relative gap falls as $C^{-0.61}$. Their similar compute dependence shows that local-energy fluctuations remain informative about energy accuracy across the training distribution, including where ED is unavailable. Moreover, the convergence of $V^{1/2}$ across the system-size bands shows that its normalisation removes most of the extensive dependence on $N$. These exponents describe this pretraining run and its evolving Hamiltonian mixture. We note that $C$ is measured in cumulative eight-H200 hours rather than exact FLOPs,
so this result is hardware dependent. By the standards of neural scaling laws in other domains like LLMs, this exponent shows highly favourable scaling with compute, sitting around $5\times$ better than the values reported for LLMs \cite{kaplan2020scaling}.

\begin{figure}[htbp]
    \centering
    \includegraphics[width=0.75\textwidth]{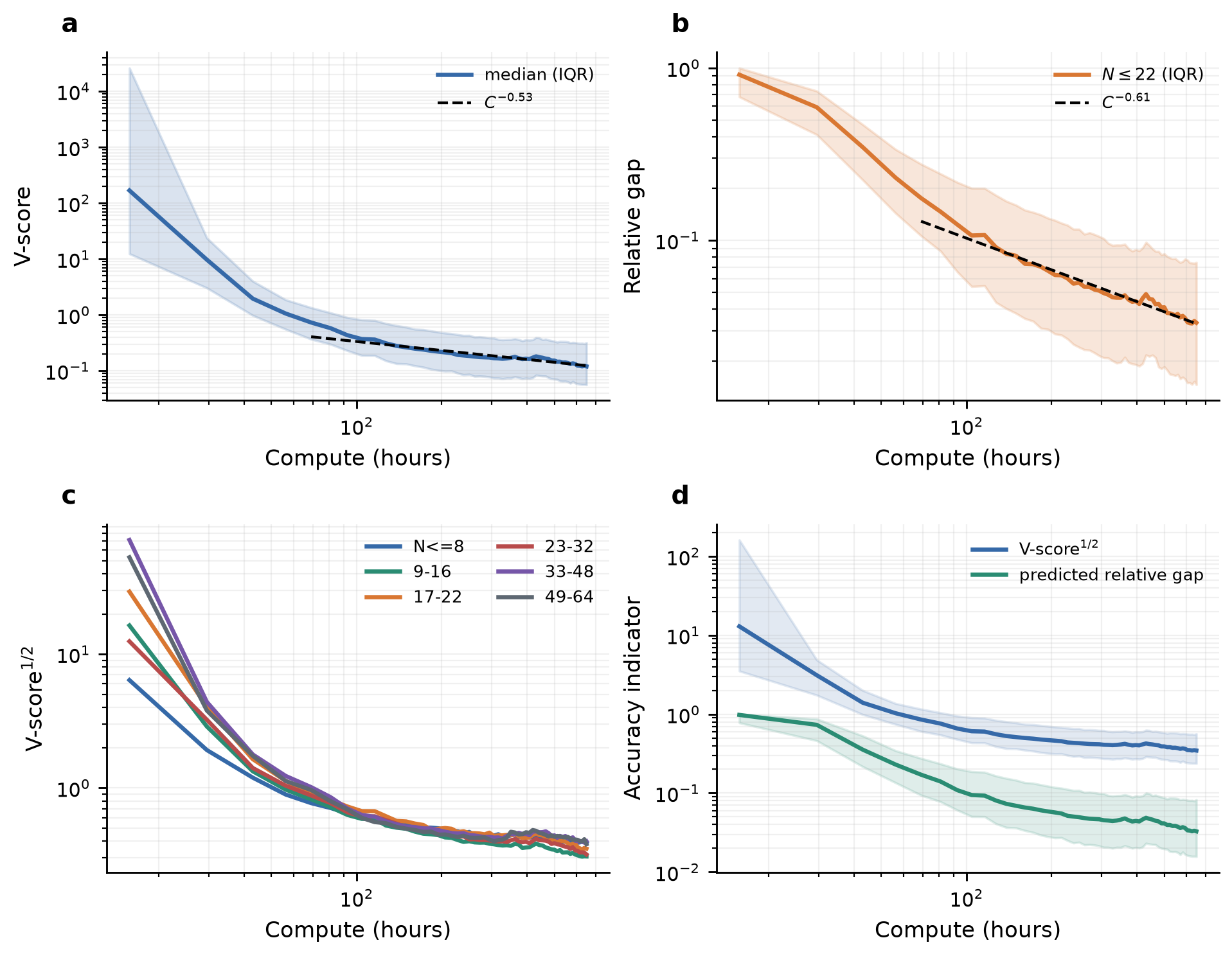}
    \caption{Compute scaling of variance and variational energy error through step 280\,000.
    \textbf{(a)} V-score decreases approximately as $V\propto C^{-0.53}$ after the initial transient.
    \textbf{(b)} The observed relative gap for ED-covered systems decreases approximately as $C^{-0.61}$.
    \textbf{(c)} $V^{1/2}$ approaches a common late-time scale across the six system-size bands, showing that the normalization removes most of the
    extensive size dependence.
    \textbf{(d)} $V^{1/2}$ and the predicted relative gap decrease together over the full 5,000-system dataset, connecting variance reduction to improved inferred accuracy.
    Compute is cumulative execution time across eight NVIDIA H200 GPUs. The V-score is $V=N\,\mathrm{Var}(E_L)/(E-E_\infty)^2$, with $E_\infty=0$ here, and is therefore a normalized variational-accuracy metric rather than raw variance.
    The relative gap denotes $r=(E-E_{\mathrm{ED}})/|E_{\mathrm{ED}}|$.
    Solid curves show medians and shaded regions show interquartile ranges.
    }
    \label{fig:compute-scaling-laws}
\end{figure}

Next we turn our attention to the zero-variance principle \cite{becca2017quantum}. As shown in Fig.~\ref{fig:zero-variance-calibration}, the bounded zero-variance model follows the calibration data without a systematic residual drift with $N$, and a calibration restricted to $N\leq18$ predicts the unseen $N=19$--22 systems with relative ease.  This is supporting evidence for transfer across size \textit{within} the ED-accessible regime.

\begin{figure}[htbp]
    \centering
    \includegraphics[width=\textwidth]{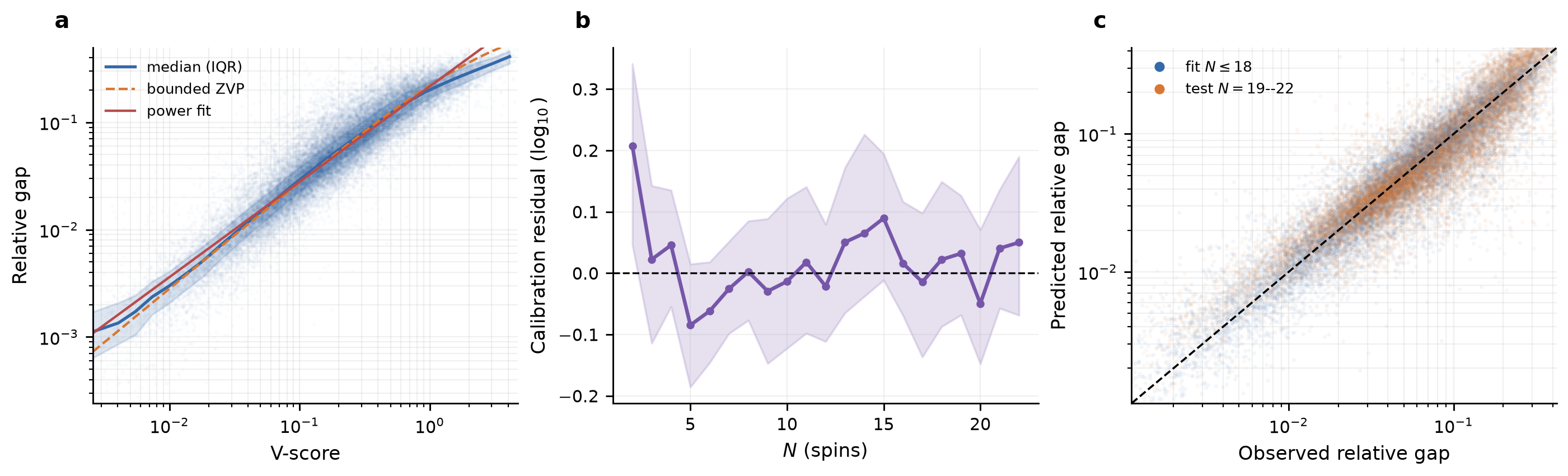}
    \caption{ Zero-variance calibration and held-out size validation.  \textbf{(a)} Relative gap against V-score, both defined as in Fig.~\ref{fig:compute-scaling-laws}. Here we show the median and IQR as a blue line and shaded blue region, with a bounded zero-variance principle in dashed orange, obtained from $q=(E-E_{\mathrm{ED}})/|E|=\kappa V$ with $\kappa=0.2825$ and $r=q/(1+q)$; a free robust fit gives $r\propto V^{0.89}$.
    \textbf{(b)} The median calibration residual shows no monotonic drift with
    $N$, supporting transfer across system size.
    \textbf{(c)} Fit of gap prediction
    showing sizes $N\geq 18$ follow the same relationship as $19 \leq N \leq 22$. Blue points belong to the fitted $N\leq18$ population and orange points to the unseen $N=19$--22 population. Their overlap around the
    one-to-one line demonstrates held-out size generalization within the range where ED is still applicable. See also Fig.~\ref{fig:gap-extrapolation-by-size} for further size extrapolation.
    }
    \label{fig:zero-variance-calibration}
\end{figure}

Subject to that qualification within the ED regime, Fig.~\ref{fig:gap-extrapolation-by-size} shows continuity at the $N=22$ ED boundary (beyond which we no longer computed ED reference values). For $N>22$, the predicted median relative gap is $3.45\%$, with an interquartile range of $1.86\%$--$10.21\%$. The median therefore remains at the few-percent scale, while the widening distribution shows that larger systems are increasingly heterogeneous rather than uniformly harder. Finally, the exchange discrepancy tracks the total-energy error, whereas the field discrepancy is smaller. Since the component energies are not separately variational and may cancel, this identifies where the remaining energy discrepancy lies.

\begin{figure}[!htbp]
    \centering
    \includegraphics[width=\textwidth]{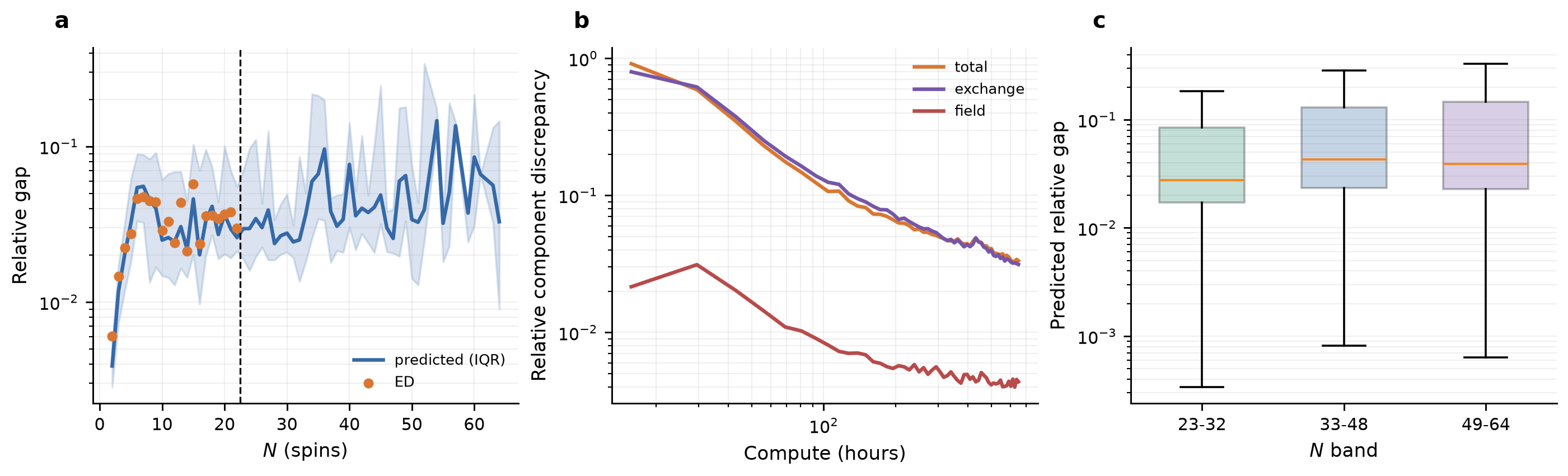}
    \caption{Energy gap extrapolation across full
    \textbf{(a)} Orange points show observed median relative gaps where ED is available; the blue curve and band show the median and interquartile range of the zero-variance predictions of Fig.~\ref{fig:zero-variance-calibration}. The dashed line marks the $N=22$ ED boundary. The extrapolation remains continuous across this boundary indicating the pretraining run's relative gap is consistently non-monotonic with $N$, rather than some cuttoff past the ED regime in accuracy.
    \textbf{(b)} Total, exchange, and field discrepancies are normalized by $|E_{\mathrm{ED}}|$ of the total Hamiltonian. Only total energy obeys the variational bound; component discrepancies are displayed in absolute value. Exchange tracks the total error, while the field contribution is smaller.
    \textbf{(c)} Predicted relative-gap distributions for $N>22$ show modest median variation but broader system-to-system tails with size. Across all $N>22$ systems, the median prediction is 3.45\%, with interquartile range 1.86\%--10.21\%.
    }
    \label{fig:gap-extrapolation-by-size}
\end{figure}

\FloatBarrier
\subsection{Zero-shot Inference and Fine-tuning}
\label{ssec:zero_shot_finetune}

In order to evaluate the model, we let the pretrained router select the tree ordering, and compile the wavefunction on that ordering, and sample the fixed model for 512 measurement steps with 256 walkers across eight tempered replicas and 24 adaptive Langevin moves per step (see SM~\S~\ref{sec:part4} for full details on our sampling scheme), retaining $512\times256\times8\simeq1.05\times10^{6}$ samples per system. Energies are reported as $E=E_{\mathrm{kernel}}/2$ with the pooled standard deviation of per-walker means as the uncertainty.


Zero shot, the pretrained model's energies against ED values (where available) are reported in the top row of Fig.~\ref{fig:axis-abc-ckpt280k-loglog}, with the cumulative frequency of systems below a given percentage threshold on the bottom row. We see from the top row an overall agreement of the systems energies against ED values, with median signed gaps of $4.02\%$, $4.21\%$ and $12.63\%$ on set (A) , (B) and (C). There are $28\%$, $21\%$ and $3\%$ of systems within $1\%$ of ED for these three sets, and the degradation across them is structured. Indeed, the gap's rank correlation with system size grows from $0.14$ (A) through $0.25$ (B) to $0.47$ (C). Within each evaluation set, the hardest systems are on a distinct family: the maximum-independent-set instances on (a) (gaps up to $31\%$), the $t$-fractals on B ($35$--$39\%$), and the Schwinger chains at nonzero $\theta$ on C ($40$--$50\%$), consistent with Axis C sitting farthest from the pretraining distribution (see also SM~\S~\ref{sec:datasets}).

\begin{figure}[!htbp]
    \centering
    \includegraphics[width=\linewidth]{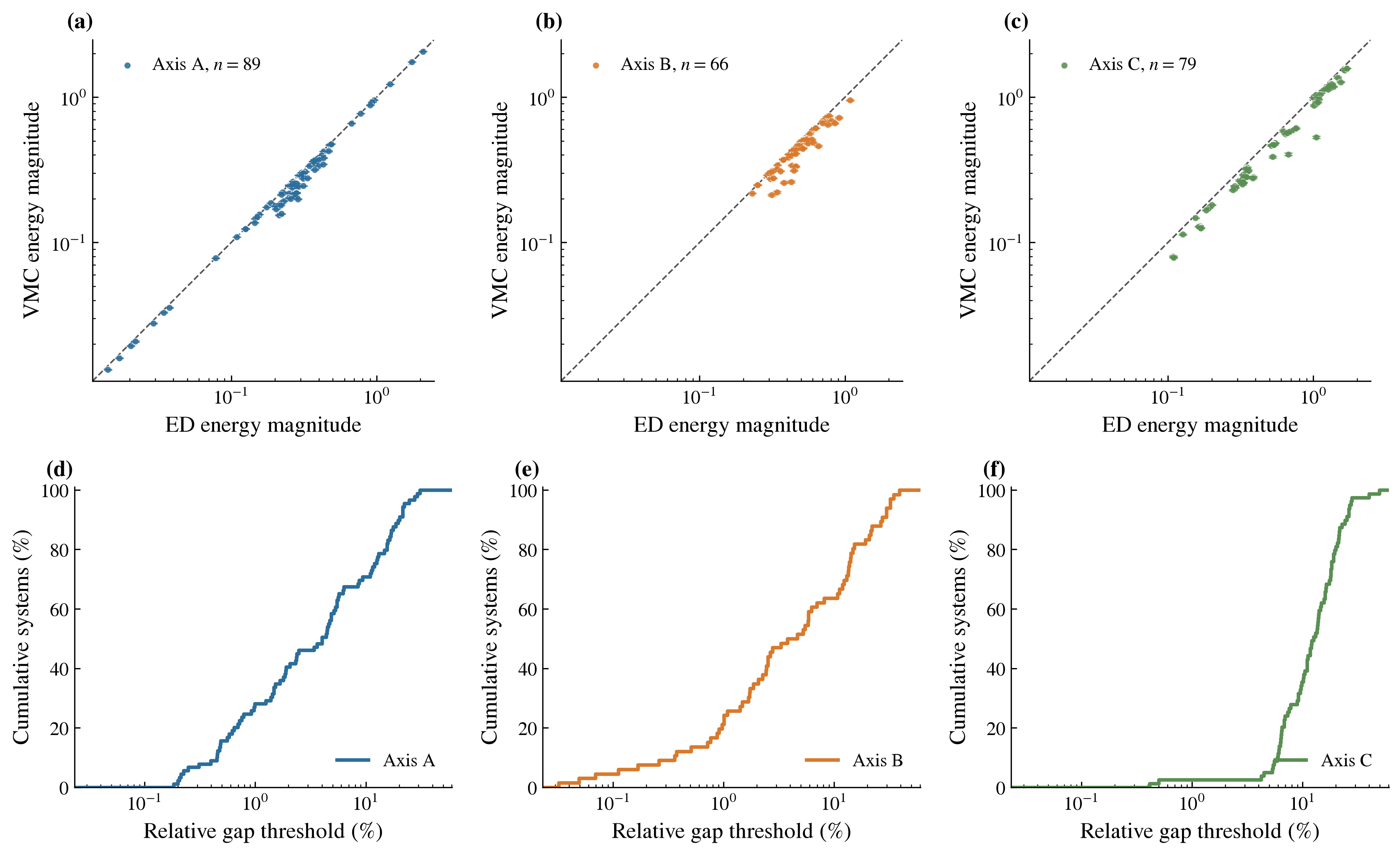}
    \caption{Held-out energy accuracy for the 512-step evaluation. Panels (a)--(c) show the magnitudes of the variational Monte Carlo and exact-diagonalization energies per site for Axes A, B, and C, respectively, on logarithmic axes. Energy magnitudes are plotted because the reported energies are negative. Every point carries its corresponding walker-mean tail standard-deviation bar; where no bar extends beyond its marker, the uncertainty is smaller than the marker diameter. Panels (d)--(f) show the empirical cumulative percentage of systems whose relative energy gap is at most the threshold for Axes A, B, and C, respectively, with a logarithmic gap-threshold axis. Colors identify the same axis in both rows. All energies are additionally divided by $N$ for per-site values. The empirical-CDF denominators are the systems with exact-diagonalization references: 89 for Axis A, 66 for Axis B, and 79 for Axis C.}
    \label{fig:axis-abc-ckpt280k-loglog}
\end{figure}

Next, we turn our attention to fine-tuning. For fine-tuning, we restart from the same checkpoint as evaluation, let the pretrained router pick the tree ordering, and compile the wavefunction on that ordering. We then only train the merge-tree on that fixed ordering, leaving a compiled ansatz with $\sim$4.7M parameters, against 547M for the full model as the object being trained. Our rationale for this is that the trunk, leaf contextualiser and other upstream parts of the Hamiltonian context need not be retrained on a \textit{single} system since they are amortizing over many different Hamiltonian families, so it is counter-productive to train them to fit the key features of a single system. This is also justified by the fact that once compiled, we can then fine-tune on a single A100 GPU at a considerably faster step-weight owing to the fact that the number of variational parameters is $<1\%$ of the overall model. Each system is optimised independently on a single A100 for $10^{4}$ KFAC steps with 256 walkers, four MCMC moves per step and a 128-step burn-in, with learning rate $0.01/(1+t/10^{4})$ (ending at $\approx0.005$), damping $10^{-3}$, and curvature EMA $0.9$. After training, every system is re-evaluated with its new tree weights. The final checkpoint is resumed at an exact zero learning rate for an additional 512 measurement steps with the sampler state carried over. ED references are never supplied to evaluation or fine-tuning and enter only in the analysis below.

Fig.~\ref{fig:axis-abc-finetuned-parity} shows the results for fine-tuning against ED values (where available). We see the remaining errors from zero-shot decrease by two to three orders of magnitude on every set of the (A)-(C) evaluation datasets. The median signed gap falls to $3.1\times10^{-4}\%$, $1.2\times10^{-3}\%$ and $6.0\times10^{-3}\%$ respectively for (A)-(C), with $90\%$, $94\%$ and $100\%$ of systems landing within $1\%$ of ED, and $75\%$, $49\%$ and $35\%$ already at the $10^{-3}\%$ error. Because every fine-tuned value comes from the frozen-weight re-evaluation rather than from the training tail, and the two agree within combined uncertainties on all 256 systems ($|z|\le2.7$), the comparison against zero shot is like for like, and we see the failure modes do not persist through fine-tuning. Hamilton-Zero was thus able to find these ground states too.



\begin{figure}[htbp]
\centering\includegraphics[width=\textwidth]{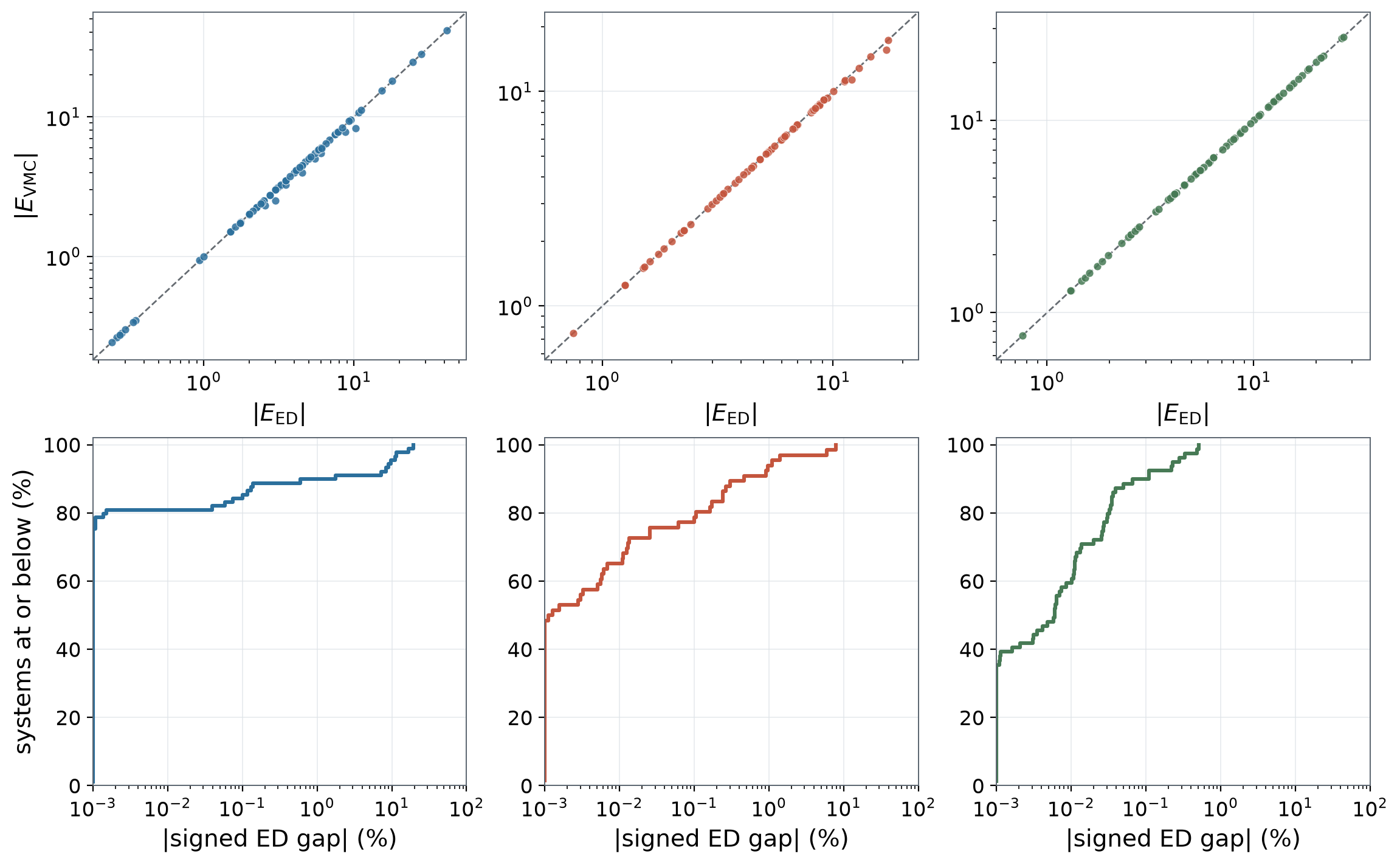}
\caption{Held-out energy accuracy for the 512-step evaluation after finetuning. Panels (a)--(c) show the magnitudes of the variational Monte Carlo and exact-diagonalization energies per site for Axes A, B, and C, respectively, on logarithmic axes. Energy magnitudes are plotted because the reported energies are negative. Every point carries its corresponding walker-mean tail standard-deviation bar; where no bar extends beyond its marker, the uncertainty is smaller than the marker diameter. Panels (d)--(f) show the empirical cumulative percentage of systems whose relative energy gap is at most the threshold for Axes A, B, and C, respectively, with a logarithmic gap-threshold axis. Colors identify the same axis in both rows. All energies are additionally divided by $N$ for per-site values. The empirical-CDF denominators are the systems with exact-diagonalization references: 89 for Axis A, 66 for Axis B, and 79 for Axis C.}
\label{fig:axis-abc-finetuned-parity}
\end{figure}

\begin{figure}[!htbp]
    \centering
    \includegraphics[width=0.85\linewidth]{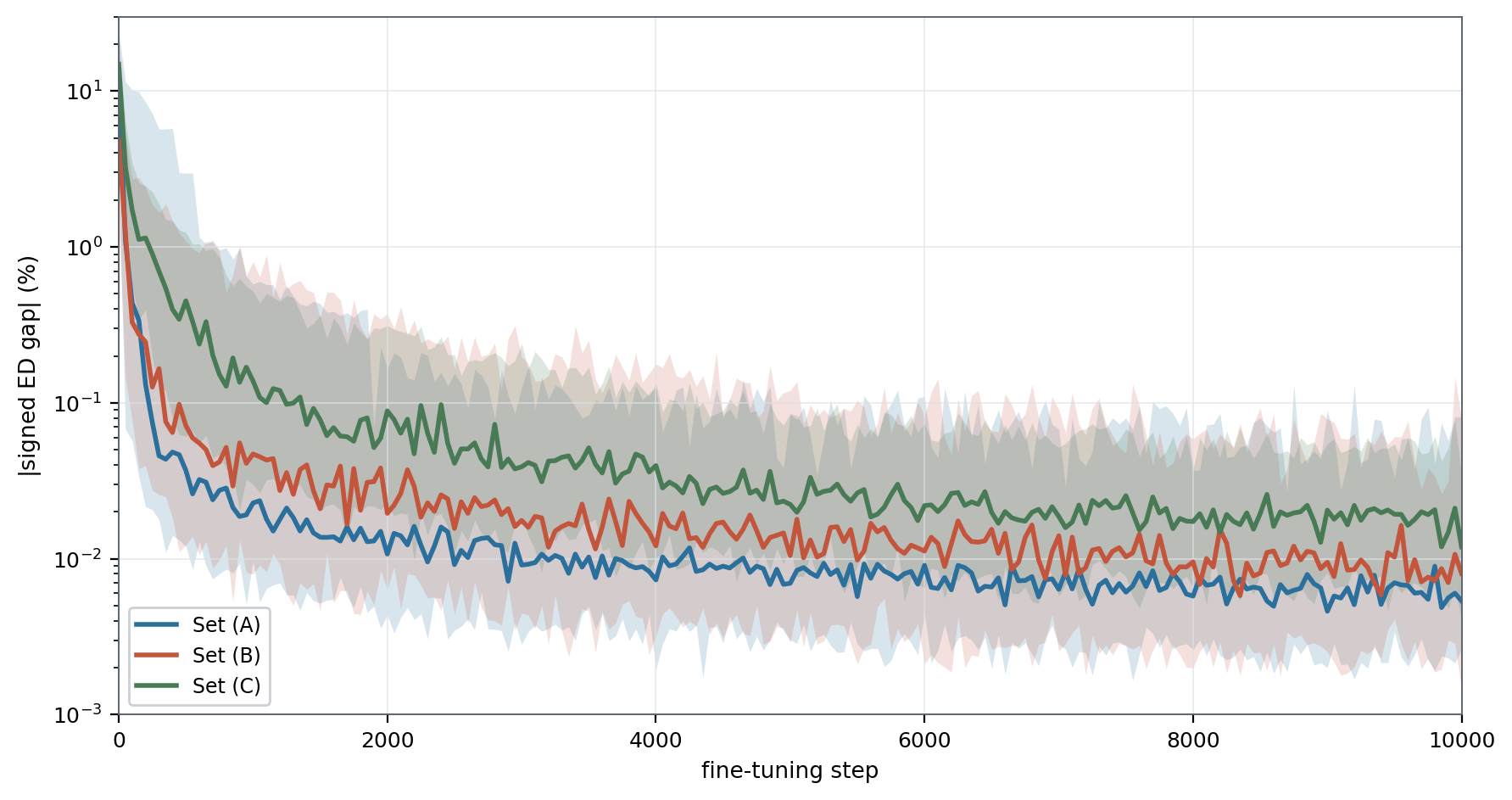}
    \caption{Convergence of the fine-tuning runs on the ED-referenced systems of datasets A, B, and C, respectively. Lines show the per-axis median of the absolute signed relative energy gap to exact diagonalization as a function of fine-tuning step, and shaded bands span the interquartile range across systems. Gaps are computed from the per-step walker-mean energies recorded every 50 steps, so the late-time level reflects single-step Monte Carlo noise of order $10^{-2}\%$. The converged medians quoted in the text are lower because the frozen-weight evaluation averages 512 such steps. }
    \label{fig:traj-gap-median-linx}
\end{figure}

The cost of this accuracy is modest, with convergence heavily front-loaded in the number of training steps as shown in Fig.~\ref{fig:traj-gap-median-linx}. The mean gap drops below $1\%$ within a few hundred steps on B and C, and Axis A plateaus at $3.1\%$ by roughly $2000$ steps. In terms of compute cost, we took $234$ GPU\,h to fine-tune all 256 systems of sets (A)-(C), plus $24$ GPU\,h to re-evaluate them, against $47$ GPU\,h for the zero-shot evaluation alone, roughly $5.5\times$ the zero-shot cost for the two to three orders of magnitude in accuracy. Wall time is set by the padded tree width rather than by physical $N$, so cost is determined by the power of two a system falls into, with the largest instances ($N=40$--$54$) taking about $3.5$\,h each.

\FloatBarrier
\subsection{Large System Case Studies}
\label{ssec:large_systems}

We begin by recalling some important properties of the three types of systems we study here: MaxCut, a PPP carbon chain, and frustrated 2D magnets. We then show the results for zero-shot and fine-tuning for each system.

For a weighted graph $G=(V,\mathcal E,w)$, MaxCut asks for a bipartition of the vertices which maximises the total weight of edges crossing the partition. If $z_i\in\{-1,+1\}$ labels the side containing vertex $i$, its objective is
\begin{equation}
 C(z)=\frac{1}{2}\sum_{(i,j)\in\mathcal E}w_{ij}(1-z_i z_j).
 \label{eq:large-maxcut-objective}
\end{equation}
We drop the additive constant and give Hamilton-Zero the diagonal Ising
Hamiltonian
\begin{equation}
 H_{}=\frac{1}{2}\sum_{(i,j)\in\mathcal E}w_{ij}Z_iZ_j,
 \qquad
 \langle C\rangle=\frac{1}{2}\sum_{(i,j)\in\mathcal E}w_{ij}
    -\langle H_{}\rangle .
 \label{eq:large-maxcut-hamiltonian}
\end{equation}
The ground state is consequently a computational-basis state encoding an optimal cut. The variational energy gives an expected cut weight. Below we compare this expectation with archived feasible cuts. Because the present evaluation does not include a discrete basis-sampling stage, the inferred, generally noninteger cut values are not themselves feasible cuts or optimality certificates.

Next, the PPP Hamiltonian is an interacting $\pi$-electron model for conjugated carbon systems.  For a chain of $L$ carbon sites, at half filling, we use the particle--hole-centred form
\begin{align}
 H_{\mathrm{PPP}}={}&-
 \sum_{\langle i,j\rangle,\sigma}t_{ij}
 \left(c^{\dagger}_{i\sigma}c_{j\sigma}+c^{\dagger}_{j\sigma}c_{i\sigma}\right)
 +U\sum_i\left(n_{i\uparrow}-\frac{1}{2}\right)
          \left(n_{i\downarrow}-\frac{1}{2}\right)
 \nonumber\\
 &+\sum_{i<j}V_{ij}(n_i-1)(n_j-1),
 \qquad
 V_{ij}=\frac{U}{\sqrt{1+\left(Ur_{ij}/e^2\right)^2}},
 \qquad e^2=14.397\,\mathrm{eV\,\mathring A}.
 \label{eq:large-ppp-hamiltonian}
\end{align}
Here $t_{ij}$ is the nearest-neighbour hopping, $U$ the on-site repulsion and $V_{ij}$ the Ohno interpolation for the long-range Coulomb interaction \citep{ohno1964ppp}.  Its Jordan--Wigner image contains $N=2L$ qubits, one for each spin orbital.  We report the centred energy per carbon, $E_{\mathrm{cent}}/N_{\mathrm C}$; any comparison must use this convention. In particular, the infinite-chain value quoted below is a literature scale for the corresponding PPP convention rather than a rigorous finite-chain bound, since Coulomb geometry and end corrections differ between calculations
\citep{soos1984ppp}.

Finally, the square and triangular spin systems can be written without committing to a drawing of the lattice. Let $\mathcal E_1$ and $\mathcal E_2$
denote the sets of nearest- and next-nearest-neighbour bonds. We study
\begin{equation}
 H_{J_1J_2}=J_1\sum_{(i,j)\in\mathcal E_1}\mathbf S_i\!\cdot\!\mathbf S_j
            +J_2\sum_{(i,j)\in\mathcal E_2}\mathbf S_i\!\cdot\!\mathbf S_j,
 \qquad \mathbf S_i=\frac{1}{2}(X_i,Y_i,Z_i).
 \label{eq:large-j1j2-hamiltonian}
\end{equation}
On the square lattice, $\mathcal E_1$ contains the horizontal and vertical bonds and $\mathcal E_2$ the plaquette diagonals. The case studied here has
$J_2/J_1=1/2$, near the maximally frustrated region of the model \citep{gong2014j1j2}. The triangular calculation instead sets $J_2=0$ and places the $J_1$ bonds on a triangular lattice, for which every bulk spin has six neighbours and the antiferromagnetic interactions cannot all be minimised simultaneously.  The four constructions are shown in
Fig.~\ref{fig:large-system-schematics}.

\begin{figure}[!htbp]
    \centering
    \includegraphics[width=0.75\linewidth]{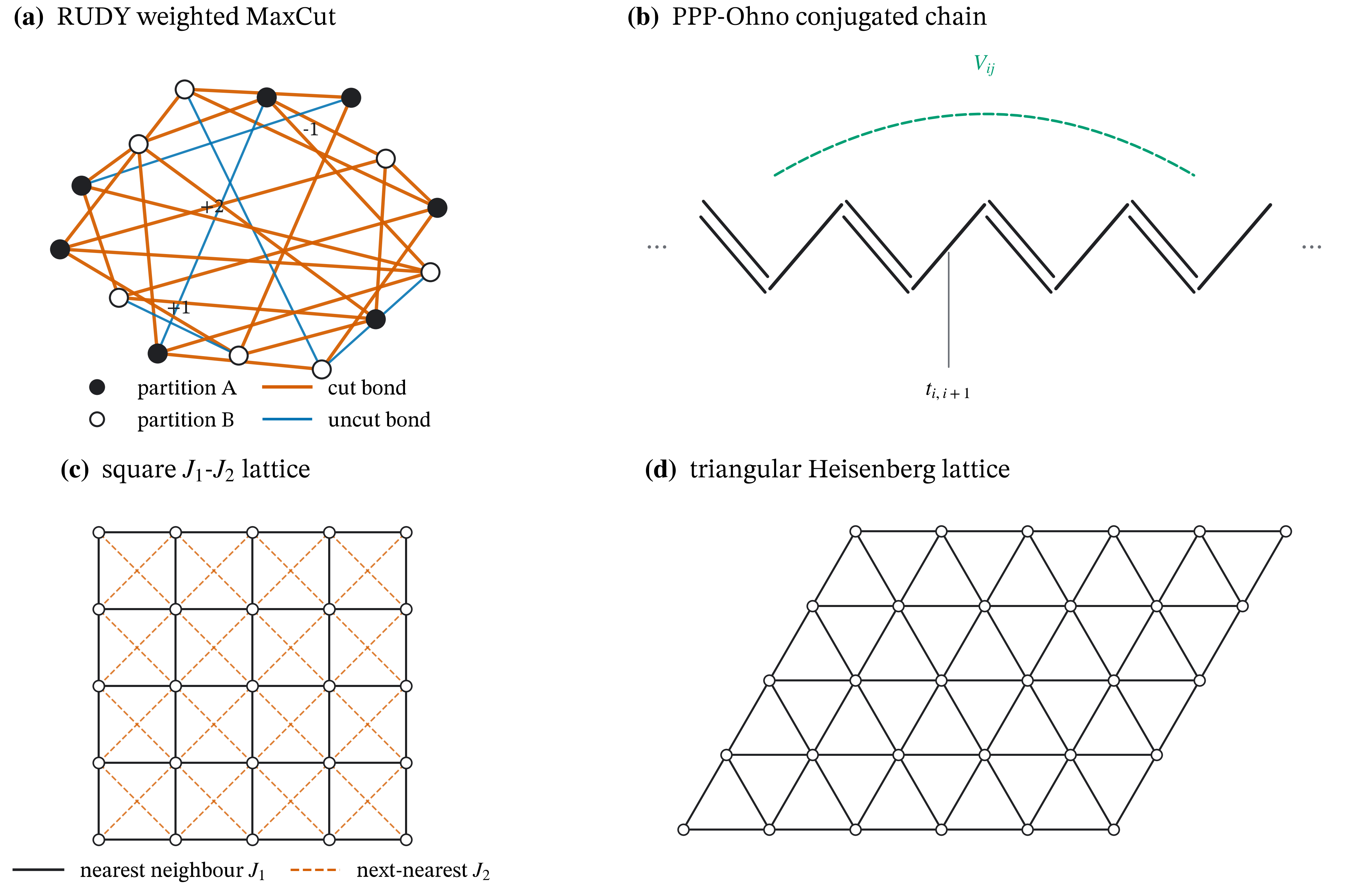}
    \caption{Large-system case studies.
    (a) A weighted graph for MaxCut.
    Blue and white vertices show a candidate bipartition; orange edges cross
    the cut, while the labels illustrate that the instances are weighted.
    (b) The PPP model on a conjugate carbon backbone, drawn in standard skeletal notation. Alternating bonds support the $\pi$ system, with nearest-neighbour hopping $t_{i,i+1}$ and the long-range Ohno interaction $V_{ij}$; two spin orbitals per carbon map to $N=2L$ qubits.
    (c) The square $J_1$--$J_2$ model, with solid nearest-neighbour and dashed diagonal bonds. (d) The
    nearest-neighbour triangular Heisenberg antiferromagnet. The drawings illustrate the interaction graphs, and the evaluated systems are larger than those drawn.}
    \label{fig:large-system-schematics}
\end{figure}

On these systems, Tab.~\ref{tab:large-system-zero-shot} shows how Hamilton-Zero's performance scales from the intended 280,000-step foundation checkpoint. For RUDY MaxCut, the zero-shot expected-cut recovery increases from $90.7\%$ at $N=256$ to $94.3\%$ at $N=1024$. The PPP chains retain $80.1$--$85.9\%$ of the magnitude of the infinite-chain comparison scale through $N=1024$. The separate lattice stress tests show relative agreement up to roughly 2000 physical spins, followed by degradation at 4096 and failure at 8100, locating the current zero-shot size-extrapolation limit for those frustrated lattices.

\begin{table}[!htbp]
\centering
\caption{Zero-shot and fine-tuned large-system evaluation from the 280,000-step foundation checkpoint. Energies are per qubit in textbook coupling units, except the PPP rows, which are particle--hole-centred energies in eV per carbon, with $N_{\mathrm C}=N/2$. Zero-shot uncertainties are Monte Carlo standard errors divided by the same site or carbon count. For MaxCut and PPP, fine-tuned entries are means over the final 50 logged optimisation steps. The square lattices use open boundary conditions and the triangular lattice periodic boundary conditions. MaxCut comparisons are $E_{\mathrm{feas}}/N=(W/2-C_{\mathrm{ref}})/N$ from archived feasible cuts, not certified optima. PPP and lattice comparisons are thermodynamic-limit literature scales rather than exact finite-size energies \cite{pariser1953semi1,pariser1953semi2,pople1953electron}. Recovery is $\langle C\rangle/C_{\mathrm{ref}}$ for MaxCut and the magnitude ratio $|e/e_{\mathrm{ref}}|$ otherwise; it uses the fine-tuned value when available and the zero-shot value for rows marked $^*$.}
\label{tab:large-system-zero-shot}
\scriptsize
\renewcommand{\arraystretch}{1.12}
\setlength{\tabcolsep}{4pt}
\begin{tabular}{@{}l r l l l r@{}}
\toprule
System & $N$ & Zero-shot $E/N$ & Fine-tuned $E/N$ & Comparison & Recovery \\
\midrule
MaxCut--12 & 256  & $-1.052265\pm0.000217$ & $-1.505829$ (20k) & $-1.943359$ $(C_{\mathrm{ref}}=2456)$  & $95.4\%$ \\
MaxCut--12 & 512  & $-1.483534\pm0.000706$ & $-2.041354$ (20k) & $-2.785156$ $(C_{\mathrm{ref}}=9275)$  & $95.9\%$ \\
MaxCut--12 & 1024 & $-1.971880\pm0.001024$ & $-2.621615$ (20k) & $-3.947754$ $(C_{\mathrm{ref}}=35469)$ & $96.2\%$ \\
\midrule
PPP--Ohno  & 256  & $-4.307745\pm0.002533$ & $-4.843845$ (10k) & $-5.015(20)$ & $96.6\%$ \\
PPP--Ohno  & 512  & $-4.087032\pm0.000897$ & $-4.822364$ (10k) & $-5.015(20)$ & $96.2\%$ \\
PPP--Ohno  & 1024 & $-4.018025\pm0.001008$ & $-4.817144$ (10k) & $-5.015(20)$ & $96.1\%$ \\
\midrule
Square $J_1$--$J_2$ & 484  & $-0.446661$  & $-0.462856$ (100k) & $\simeq-0.4968$ & $93.2\%$ \\
Square $J_1$--$J_2$ & 1024 & $-0.445674$  & $-0.451456$ (30k)  & $\simeq-0.4968$ & $90.9\%$ \\
Square $J_1$--$J_2$ & 2025 & $-0.421763\pm0.000046$ & N/A                & $\simeq-0.4968$ & $84.9\%^*$ \\
Square $J_1$--$J_2$ & 4096 & $-0.350009\pm0.000048$ & N/A                & $\simeq-0.4968$ & $70.5\%^*$ \\
Square $J_1$--$J_2$ & 8100 & $+0.128196\pm0.000101$ & N/A                & $\simeq-0.4968$ & $-25.8\%^*$ \\
Triangular $J_1$    & 1764 & $-0.398369\pm0.000050$ & N/A                & $-0.5503(8)$    & $72.4\%^*$ \\
\bottomrule
\end{tabular}
\end{table}

We next restart from the same 280,000-step foundation checkpoint, retain the route chosen by the pretrained router, and train only the compiled merge tree. This reduces the optimised object from the 547M-parameter foundation model to approximately 8.8M trainable parameters. The $N=256$ and $512$ MaxCut and PPP runs use four H200 GPUs, while the $N=1024$ runs use eight. They use 256 walkers, eight replicas, two MCMC moves per KFAC step, and a burn-in of $N$ steps (see SM~\S~\ref{sec:part4}). Their learning rate is $0.002/(1+t/10^4)$, with damping $10^{-3}$ and curvature exponential moving average $0.99$. The $22\!\times\!22$ and $32\!\times\!32$ square-lattice runs use the same walker, replica and update counts, a 128-step burn-in and $0.01/(1+t/10^4)$ with curvature moving average $0.9$.

Figure~\ref{fig:large-system-finetuning} shows the energy intensiveness in the site count of each problem: qubits for MaxCut, carbons for PPP and physical spins for $J_1$--$J_2$. The optimisation traces are strongly front-loaded, as we saw before on smaller systems. The final-50-step mean MaxCut energies for $N=256$, $512$, and $1024$ correspond to expected cut weights $2344.0$, $8894.2$, and $34111.0$, or $95.4\%$, $95.9\%$, and $96.2\%$ of the archived feasible references. The corresponding PPP chains reach $-4.8438$, $-4.8224$, and $-4.8171$~eV/C after $10^4$ steps, within $3.4\%$, $3.8\%$, and $3.9\%$ of the infinite-chain comparison scale under the centred convention. For square-lattice $J_1$--$J_2$, the $22\!\times\!22$ case has 484 physical spins in the padded $N=512$ routing bucket and improves from $E/N=-0.44666$ at initialisation to $-0.46316$ after $10^5$ steps. The $32\times32$ case has $N=1024$ physical spins and reaches $E/N=-0.451456$ after $3\times 10^4$ steps, essentially its initial value of $-0.44567$ after an early transient.

\begin{figure}[!htbp]
    \centering
    \includegraphics[width=\linewidth]{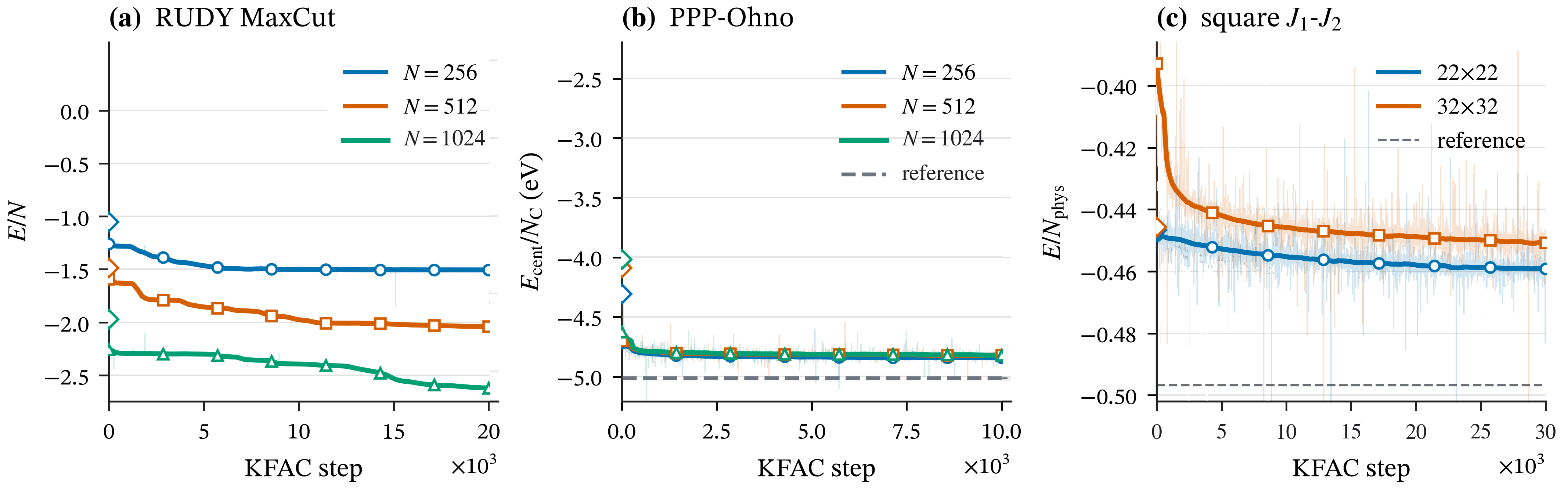}
    \caption{Fine-tuning on the large-system case studies. The faint curves are per-step training estimates and the solid curves are centred rolling means (501 KFAC steps for MaxCut and PPP; 1001 for $J_1$--$J_2$, with shorter windows at the boundaries). Open markers show representative points on the rolling means. Diamonds at step zero are the independent step-280000 zero-shot evaluations in (a,b) and the logged training initialisations in (c). The dashed grey lines in (b,c) are literature scales, not optimisation targets. (a) MaxCut energy per qubit; more negative energy corresponds to a larger expected cut through Eq.~\eqref{eq:large-maxcut-hamiltonian}. (b) Particle--hole-centred PPP energy per carbon. (c) Textbook $J_1$--$J_2$ energy per physical spin: $22\!\times\!22$ means 484 physical spins in the padded $N=512$ routing bucket, while $32\!\times\!32$ means $N=1024$ physical spins.}
    \label{fig:large-system-finetuning}
\end{figure}

\FloatBarrier
\subsection{Witnessing a Phase Transition with Zero-Shot Weights}
\label{ssec:phase_transitions}

Consider the periodic transverse-field Ising Hamiltonian
\begin{equation}
  H(g)
  =
  -J\sum_{i=1}^{N}\sigma_i^z\sigma_{i+1}^z
  -gJ\sum_{i=1}^{N}\sigma_i^x,
  \qquad \sigma_{N+1}\equiv\sigma_1 .
\end{equation}
Here, $g$ is the dimensionless transverse-field strength relative to the
nearest-neighbour Ising coupling $J$. Varying this parameter drives the
model through its thermodynamic quantum critical point at $g_c=1$.

For the normalized ground state $\lvert\psi(g)\rangle$ of $H(g)$, the
fidelity between two nearby Hamiltonians is
\begin{equation}
  F(g,g+\delta g)
  =
  \left|
    \langle\psi(g)\mid\psi(g+\delta g)\rangle
  \right|.
\end{equation}
The leading response defines the fidelity susceptibility,
\begin{equation}
  \chi_F\!\left(g+\frac{\delta g}{2}\right)
  \simeq
  -\frac{2\ln F(g,g+\delta g)}{(\delta g)^2}.
\end{equation}
Near the transition, the ground-state wavefunction changes rapidly with $g$, decreasing the fidelity and producing a peak in $\chi_F$. The finite-size quantity $\chi_F/N$ therefore provides a practical witness of the transition.

In Fig.~\ref{fig:tfim-zero-shot-fidelity-susceptibility}, we see Hamilton-Zero witnesses this phase-transition with zero-shot inference. This size-consistent maximum near $g_c=1$ provides a proof of principle that foundation models could in future serve as search engines for phase transitions, see Sec.~\ref{sec:mt-outlook} for further discussion.

\begin{figure}[!htbp]
  \centering
  \includegraphics[width=0.75\linewidth]{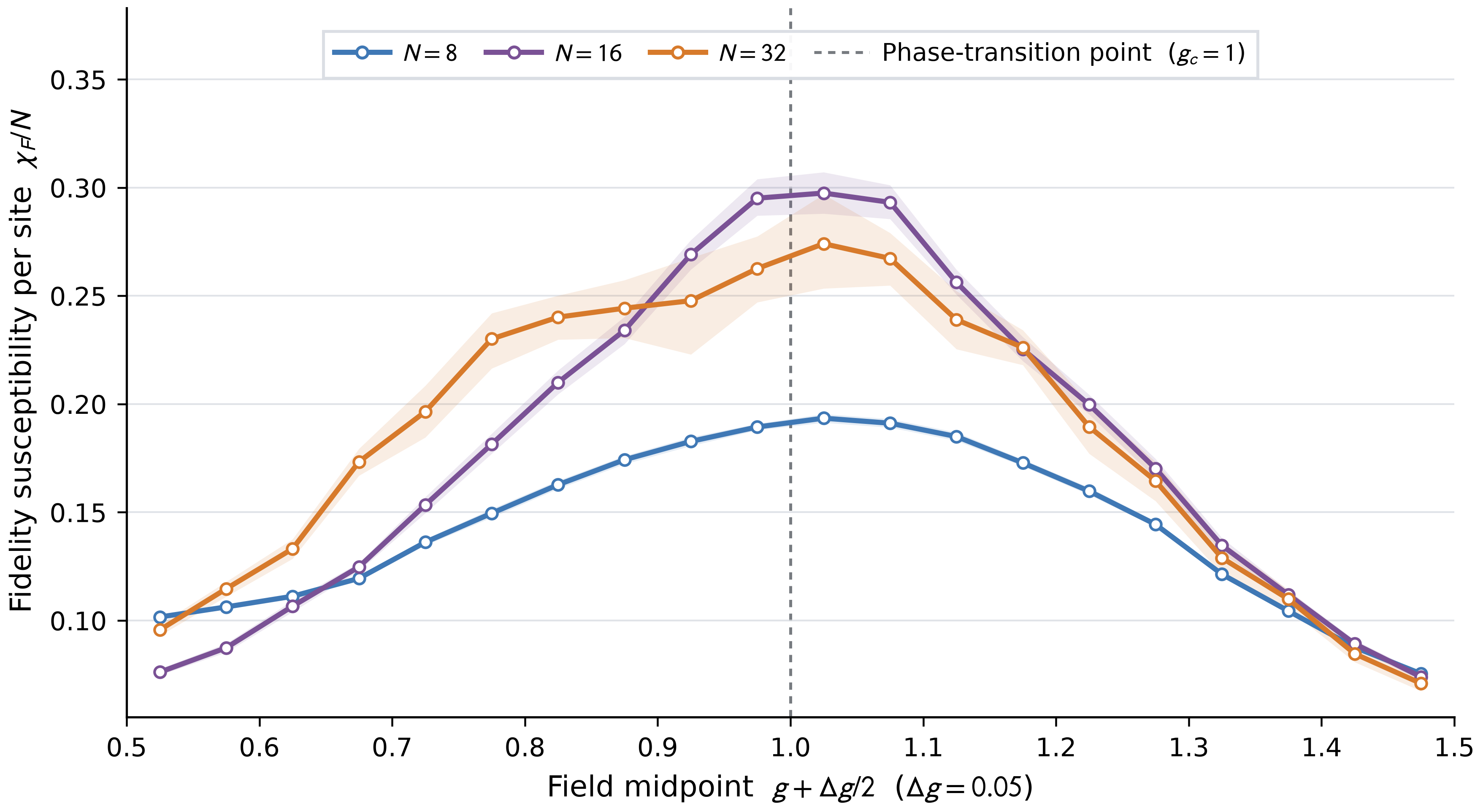}
  \caption{
  Fidelity susceptibility per site, $\chi_F/N$, of the periodic
  transverse-field Ising chain is shown for $N=8$, $16$, and $32$ spins.
  Every curve is evaluated directly from the same pretrained Hamilton-Zero
  checkpoint without fine-tuning; markers denote adjacent-field overlap
  estimates at midpoint $g+\Delta g/2$ with $\Delta g=0.05$, and shaded
  regions give 95\% confidence intervals. The vertical dashed line marks the thermodynamic phase-transition point, $g_c=1$. The zero-shot response develops a finite-size maximum close to $g_c$ for
  all three sizes, providing a phase-transition witness from the fixed
  pretrained wavefunction.}
  \label{fig:tfim-zero-shot-fidelity-susceptibility}
\end{figure}

\FloatBarrier
\subsection{Applications to Fermionic Systems}

We can also apply Hamilton-Zero to ground-state fermionic problems via so-called fermion--qubit mappings. Here the input is the second-quantized active-space Hamiltonian
\begin{equation}\label{eq:res-Hf} H_{\rm f}=E_0I+\sum_{p,q=0}^{n-1}h_{pq}a_p^\dagger a_q+\sum_{p,q,r,s=0}^{n-1}v_{pqrs}a_p^\dagger a_q^\dagger a_ra_s, \end{equation}
where $a_p,a_p^\dagger$ ($0\le p<n$) satisfy $\{a_p,a_q\}=0$ and $\{a_p,a_q^\dagger\}=\delta_{pq}I$, $n=2N_o$ is the number of spin-orbitals, and CAS($N_e$e,$N_o$o) restricts to the eigenspace $\hat N=\sum_pa_p^\dagger a_p=N_e$.

HamiltonZero accepts Hamiltonians with at most two-local terms, per our definition in Eq.~\ref{eq:hamilton_zero_hamiltoinian}. As such, there are two steps we need to take to make Hamilton-Zero apply to fermionic systems. First, we need to apply a ferim-to-qubit mapping, and second we need to encode the result of that mapping into a second-order interaction graph using ancilla qubits. This is similar to the steps taken when compiling fermionic systems onto digital quantum gates in quantum computing. Here we opt the JW-order-CF frame, which keeps the standard Jordan--Wigner mapping but optimises the order in which the modes are placed. The Jordan--Wigner mapping~\cite{jw} assigns one qubit per fermionic mode and represents the ladder operators as
\begin{equation}\label{eq:jw-map} a_p^\dagger \mapsto \tfrac12\left(X_p - iY_p\right)Z_{p-1}\cdots Z_0 , \end{equation}
where the string of $Z$ operators below mode $p$ carries the anticommutation relations. Every term of Eq.~\ref{eq:res-Hf} then becomes a real combination of Pauli words, tensor products over $\{I,X,Y,Z\}$, and we call the number of non-identity factors in a word its weight. Since the $Z$ strings run between the modes of a term, the weight of a term depends on where its modes sit along the chain. In the JW-order-CF frame we place mode $p$ at chain position $\pi(p)$ and choose the permutation $\pi$ by simulated annealing, so that strongly interacting modes sit close together and the qubit count after the second-order encoding below is minimised~\cite{chiew}. The spectrum is identical for every $\pi$, so the ordering only changes the size of the compiled register.

Generically, the JW-order CF frame, as well as most fermi-qubit mappings available in the literature~\cite{mcardle,jw,bk,jkmn,chiew,hatt,treespilation,rsd,clifford}, do not result in a Pauli word of weight two, so these mappings do are not de-facto compatible with Hamilton-Zero. We therefore use ancilla qubits to make an effective second-order interaction. Since our model is classical and already can scale to thousands of qubits, this is a minor penalty in the capacity needed for fermionic systems. However, it is still in our interest to minimise the size of the ancilla as it results in faster compute.

To reduce a Pauli word $P$ of weight $k>2$ to second order we use a private Jordan--Farhi gadget~\cite{jordanfarhi,kkr}. We introduce $k$ fresh ancilla qubits for that word alone and give them a strong ferromagnetic penalty $H^{\rm anc}=\Delta\sum_{i<j}(I-Z_{r_i}Z_{r_j})/2$, whose ground space contains only the two aligned ancilla states. Each of the $k$ Pauli factors of $P$ is then coupled to its own ancilla by a two-local term, and the coupling strengths are fixed so that at order $k$ in perturbation theory the ancillas mediate exactly the product interaction $\alpha P$. The error of this construction is of order $\lVert V\rVert^{k+1}/\Delta^{k}$, where $V$ collects the gadget couplings, so it is controlled by the size of the penalty $\Delta$. As such, every interaction in the compiled graph acts on at most two qubits, and the register grows from $n$ to $n+\sum_\ell w_\ell$ over the gadgetised words. The gadget makes no use of symmetry and its register grows only linearly with the number of modes. This means it can be generically applied to fermionic systems. We defer the construction and its proof to Supp. Mat.~\ref{app:fermionic-systems}.

Thus a two-local compilation is always available, and Hamilton-Zero can be applied to fermionic systems, with this gadget forming a compilation chain
\begin{equation}\label{eq:gadget-chain}
\left(h_{pq},\,v_{pqrs}\right)
\;\longrightarrow\;
H_{\rm f}
\;\xrightarrow{\ \text{JW-order-CF}\ }\;
\textstyle\sum_\ell\alpha_\ell P_\ell
\;\xrightarrow{\ \text{gadgets}\ }\;
\left(J_{ij}^{ab},\,h_i^a\right),
\end{equation}
where the mapping arrow is exact and isospectral for every mode ordering, and the gadget arrow is the only step that introduces a perturbative approximation. For a large penalty $\Delta$ the ancillas are frozen into their aligned ground space, and treating the couplings $V$ as a perturbation reproduces the target word $\alpha P$ at order $k$, with a remainder of order $\lVert V\rVert^{k+1}/\Delta^{k}$, see Supp. Mat.~\ref{app:fermionic-systems}.

Next we will see that small systems admit an exact alternative which is commonplace in quantum chemistry ~\cite{roos,knowleshandy,helgaker}. Instead of mapping every mode to a qubit, we can use the symmetries of Eq.~\ref{eq:res-Hf} to shrink the problem before it reaches the model. Since $H_{\rm f}$ conserves particle number, it is block diagonal over the charge sectors, and the ground state lies in one known block. Restricting to that block is a change of basis and carries no approximation; it is the same restriction that classical active-space methods diagonalise directly. We then encode the restricted block itself as a two-local system, with a register whose size is set by the dimension of the block rather than by the number of modes. For spin-resolved modes $p=(o,\sigma)$ with $[H_{\rm f},\hat N_\sigma]=0$, the $N_e$-electron ground state at spin projection $m$ lies in an invariant subspace $\mathcal V$ of dimension
\begin{equation}
\label{eq:res-dim} 
d(N_e,N_o;m)=\binom{N_o}{N_e/2+m}\binom{N_o}{N_e/2-m}. 
\end{equation}

The full compilation is then the chain
\begin{equation}\label{eq:res-chain}
\left(h_{pq},\,v_{pqrs}\right)
\;\longrightarrow\;
H_{\rm f}
\;\xrightarrow{\ \text{restrict to }\mathcal V\ }\;
\widetilde M
\;\xrightarrow{\ \text{encode}\ }\;
\left(J_{ij}^{ab},\,h_i^a\right).
\end{equation}
Here $\widetilde M$ is the $d\times d$ matrix of $H_{\rm f}$ on the invariant subspace $\mathcal V$, shifted to be traceless, and the encoding step writes $\widetilde M$ in at most two-local Pauli operators, densely on two qubits for $d\le4$ and in a one-hot register with a commuting penalty for larger $d$. The resulting couplings $(J,h)$ are exactly of the Hamilton-Zero form of Eq.~\ref{eq:hamilton_zero_hamiltoinian}. This route carries no compilation error at any penalty strength, in contrast to the perturbative gadget above, but it pays for exactness in register size: the qubit count is set by the sector dimension of Eq.~\ref{eq:res-dim}, which grows exponentially with the active space, whereas the gadget register grows only linearly with the number of modes. The two routes are therefore complementary. The gadget is the generic choice for large systems, and the exact compilation is the preferred choice whenever the sector fits on the register, as it does for all of the systems we study below.

We demonstrate this proof-of-principle for seven small molecular systems so that we can compare with CASCI energies and evaluate the accuracy on our model. Figure~\ref{fig:qchem-symmetry-panels} shows the fine-tuning error of the  compiled systems, with the public checkpoint's model weights and the same hyperparameters as Sec.~\ref{ssec:zero_shot_finetune}. We see that with finetuning all seven small molecular systems fall to within the chemical accuracy threshold of $1.6$mHa. None of these systems' interaction graphs were present in the pretraining datasets. We refer to Supp. Mat.~\ref{app:fermionic-systems} for a full account of the systems, their fermion--qubit mappings and their symmetry-induced compilation.

\begin{figure}[!htbp]
\centering
\includegraphics[width=\linewidth]{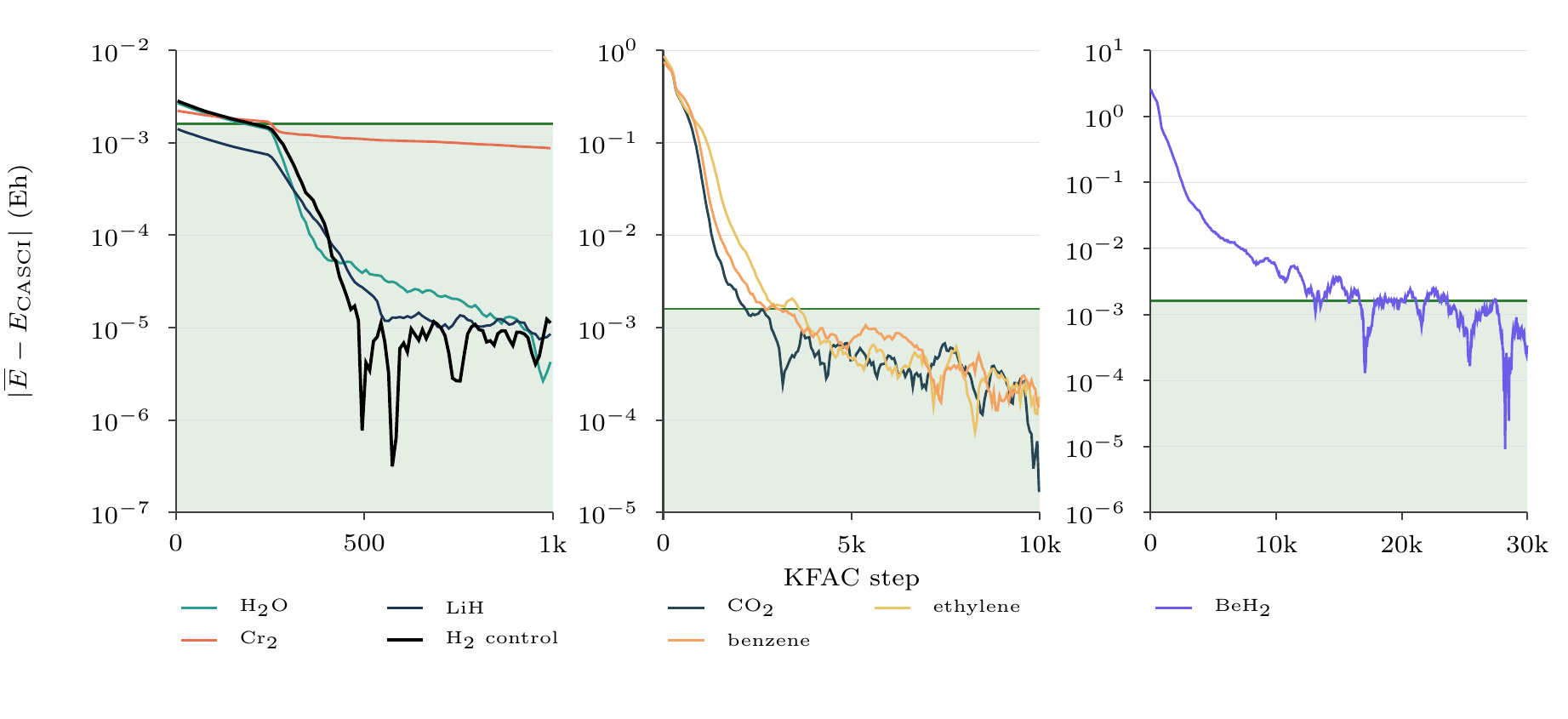}
\caption{Convergence classes of the exact symmetry-sector arm (500-step centred rolling means; colors and labels shared across panels; $x$- and $y$-ranges are per panel). The green line marks chemical accuracy ($1.6$\,mHa) and the shading beneath it the target band.}
\label{fig:qchem-symmetry-panels}
\end{figure}

\FloatBarrier
\subsection{Approximate Equivariance} 
\label{ssec:equivariance_results}

We tested the invariance of the predicted mean energy under a group action on the complete physical triplet $(J,h,q)$ corresponding to the interaction graph $(J,h)$ and quaternionic coordinates $q$. Recall for site-dependent rotations $u_i\in\mathrm{SU}(2)$, the action is
\begin{equation}
    J_{ij}\mapsto R(u_i)J_{ij}R(u_j)^{\mathsf T},
    \qquad
    h_i\mapsto R(u_i)h_i,
    \qquad
    q_i\mapsto q_i\bar u_i,
    \label{eq:joint-su2-action}
\end{equation}
where $R(u_i)\in\mathrm{SO}(3)$ is the adjoint rotation and $\bar u_i=u_i^{-1}$. The global test is the diagonal subgroup obtained by setting $u_i=u$ at every site. The site-local test instead draws the $u_i$ independently, inducing independent left and right $\mathrm{SU}(2)$ actions at the two ends of each
bond.

For each transformed system, we measured the relative mean-energy residual
\begin{equation}
    r_E =
    \frac{\left|\bar E_g-\bar E\right|}{\left|\bar E\right|},
    \label{eq:equivariance-residual}
\end{equation}
using paired baseline and group-transformed samples. To distinguish symmetry violation from Monte Carlo noise, we also estimated a paired-sample standard
error. For walker $w$, let
\begin{equation}
    \overline{\Delta E}_w =
    \frac{1}{T}\sum_{t=1}^{T}
    \left(E^{(g)}_{w,t}-E_{w,t}\right).
\end{equation}
The relative standard error plotted in
Fig.~\ref{fig:joint-triplet-equivariance-learning} is
\begin{equation}
    \mathrm{SE}_{\mathrm{rel}} =\frac{\operatorname{sd}_{w}\!\left(\overline{\Delta E}_w\right)
    }{\sqrt{W}\,|\bar E|
    },
    \label{eq:paired-equivariance-se}
\end{equation}
with  $W=256$ walkers and $T=128$ steps. This quantity estimates the Monte Carlo uncertainty of each paired residual.

At initialization, the median residual was on the order of unity for both actions: $1.30$ for the global action and $1.39$ for the site-local action. By step $70\,000$, these values had fallen to $1.41\times10^{-2}$ and $3.28\times10^{-2}$, respectively. Thus, most of the equivariant behaviour was acquired during the first quarter of training. Subsequent checkpoints show a non-monotonic plateau rather than a steady improvement: over steps $70\,000$--$280\,000$, the global median lies between $1.14\times10^{-2}$ and $2.01\times10^{-2}$, whereas the site-local median lies between $2.57\times10^{-2}$ and $3.95\times10^{-2}$.

At step $280\,000$, the median residuals are $1.14\times10^{-2}$ for the global action and $2.91\times10^{-2}$ for the site-local action. The corresponding median paired-sample standard errors are $2.12\times10^{-3}$ and $6.53\times10^{-3}$. The reduction from order-one residuals at initialization to percent-level residuals after training demonstrates that the model has learned approximate equivariance. The consistently larger site-local residual shows that the full $\mathrm{SU}(2)^N$ action is learned less accurately than its global diagonal
subgroup.

\begin{figure}[!htbp]
    \centering
    \includegraphics[width=0.75\linewidth]{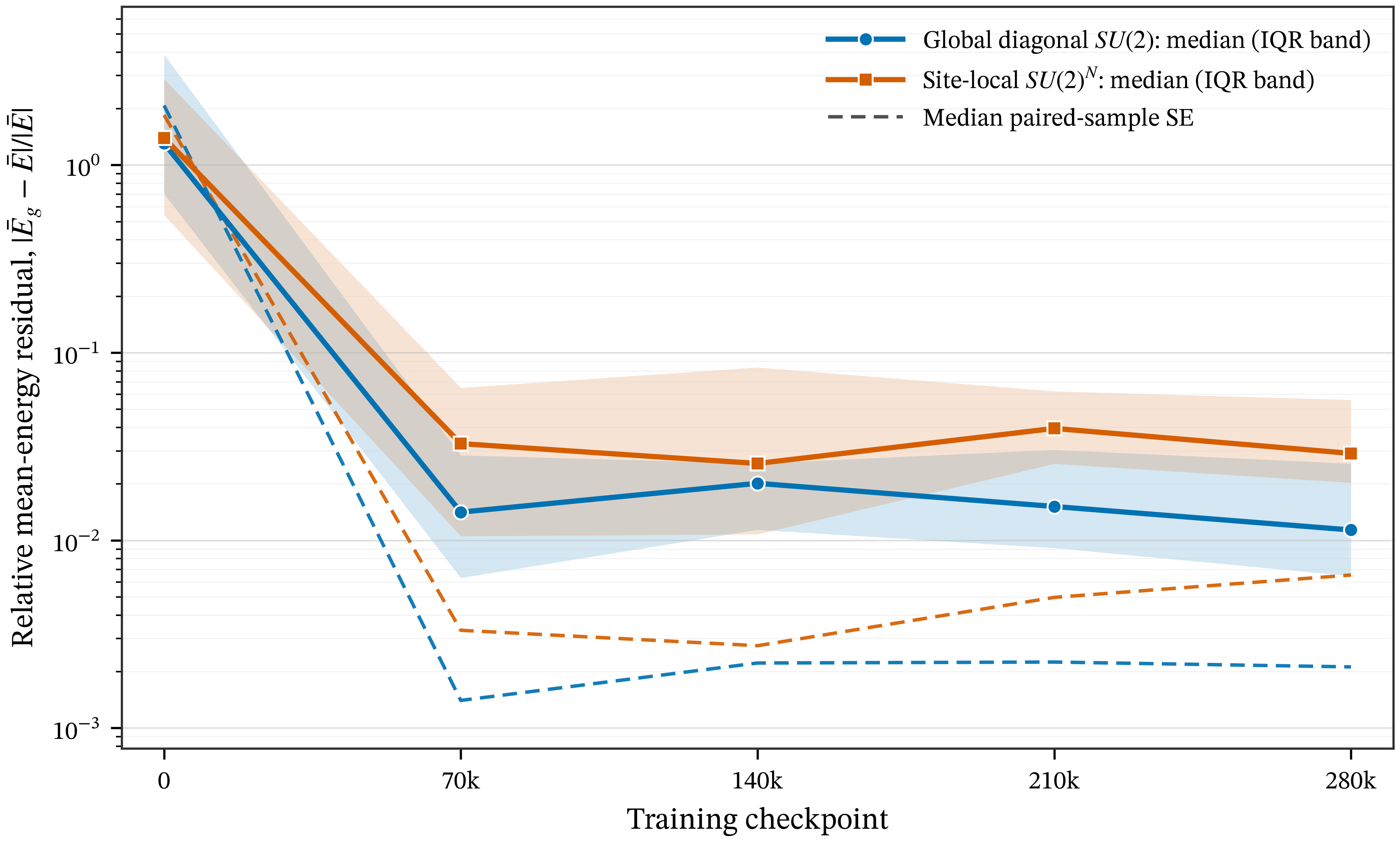}
    \caption{
        Learning of mean-energy equivariance under joint transformations of $(J,h,q)$ on 20 held-out systems from the (B) evaluation set (see SM~\S~\ref{sec:datasets}). Blue denotes the global diagonal $\mathrm{SU}(2)$ action and orange the site-local $\mathrm{SU}(2)^N$ action. Solid curves and shaded bands show the
        median and interquartile range of the relative energy residual, respectively; dashed curves show the median paired-sample Monte Carlo
        standard error estimated from 256 walkerwise averages of 128 steps.
    }
    \label{fig:joint-triplet-equivariance-learning}
\end{figure}

\FloatBarrier
\subsection{Router Commitment and Generalisation}
\label{ssec:router_results}

Recall from Sec.~\ref{sec:mt-ansatz} (see also SM~\S~\ref{sec:part2}) that the router converts the Hamiltonian-conditioned trunk embeddings into a permutation of the physical sites and padded leaves, thereby selecting a binary contraction tree used by the wavefunction. In this section, we visualise the map that the router picks, showing that it recovers the entanglement structure of the case studies of \ref{ssec:large_systems}. For visualisation, we map the saved route from its padded, bit-reversed compiler labels back to canonical physical coordinates. The resulting diagrams therefore show the contraction actually selected by the model.

At zero-shot evaluation, the router decodes eight candidate trees. Each candidate is subjected to a short burn-in and measurement contest, after which the winning ordering is compiled and used for the full evaluation. Owing to the fact that our algorithm selects from a contest, we can quantify the concentration of the route policy by its \textit{commitment},
\begin{equation}
  C_{\mathrm{route}}
  :=1-\frac{H_{\mathrm{policy}}}{H_{\max}},
\end{equation}
where $C_{\mathrm{route}}\simeq0$ denotes a diffuse policy and $C_{\mathrm{route}}\simeq1$ denotes concentration on a small set of preferred, potentially symmetry-related routes. Figure~\ref{fig:router-commitment-280k}(a) shows the median commitment across the 13 sufficiently populated system categories of the pretraining dataset (SM~\S~\ref{sec:datasets}), with the grey band giving the interquartile range, while Fig.~\ref{fig:router-commitment-280k}(b) resolves the individual category trajectories. Commitment rises from approximately $0.5$ to at least $0.96$ during the first eight augmentation rounds and remains close to saturation thereafter. The agreement across categories is indicative of the router learning a stable structural preference early in pretraining, rather than acquiring a different arbitrary ordering for each family late in the run. It is thus able to train effectively over a heterogeneous set of Hamiltonians in the pretraining set.

\begin{figure}[!htbp]
  \centering
  \includegraphics[width=\textwidth]{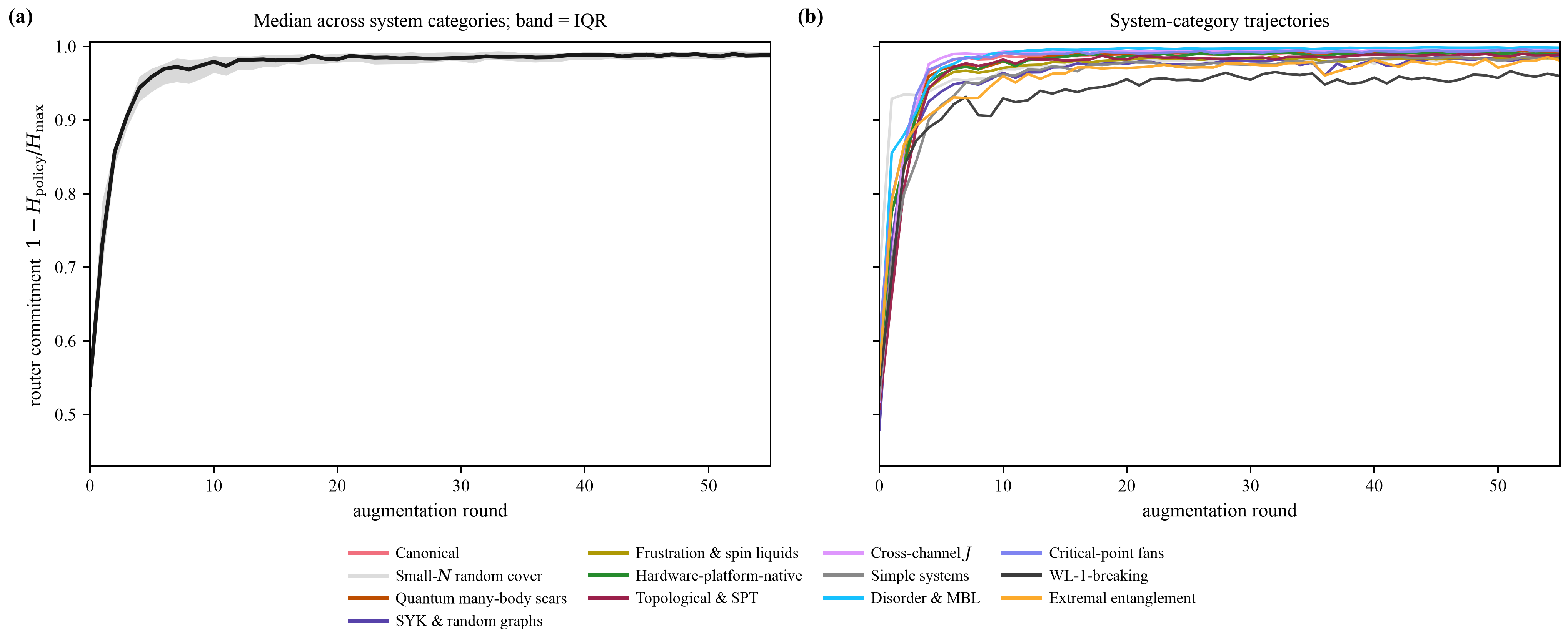}
  \caption{Entropy-realized commitment,
  $1-H_{\mathrm{policy}}/H_{\max}$, by augmentation round for the 13 system categories with $S\geq80$.
  (a) Round-wise median across categories, with the grey band denoting
  the interquartile range.
  (b) Stratification by system-category, see Fig.~\ref{fig:panel-v9-size} of SM~\S~\ref{sec:datasets}. All categories commit rapidly during the first ${\sim}8$ rounds and subsequently remain near saturation.}
  \label{fig:router-commitment-280k}
\end{figure}

Indeed, figure~\ref{fig:router-phase-portrait-280k} shows there is a direct correlation between commitment and the energy quality of Hamilton-Zero's predictions; a higher level of commitment generally signals a lower gap to ED. Panel (a) follows the round-wise median in the plane of router commitment and relative gap to exact diagonalisation for systems in the pretraining set where ED data is available. Panel (b) shows the corresponding trajectories for the individual system categories, showing this trend is consistent between different types of system. Commitment is therefore a correlational diagnostic rather than an energy certificate: a concentrated policy may still select a poor ordering.

\begin{figure}[!htbp]
  \centering
  \includegraphics[width=\textwidth]{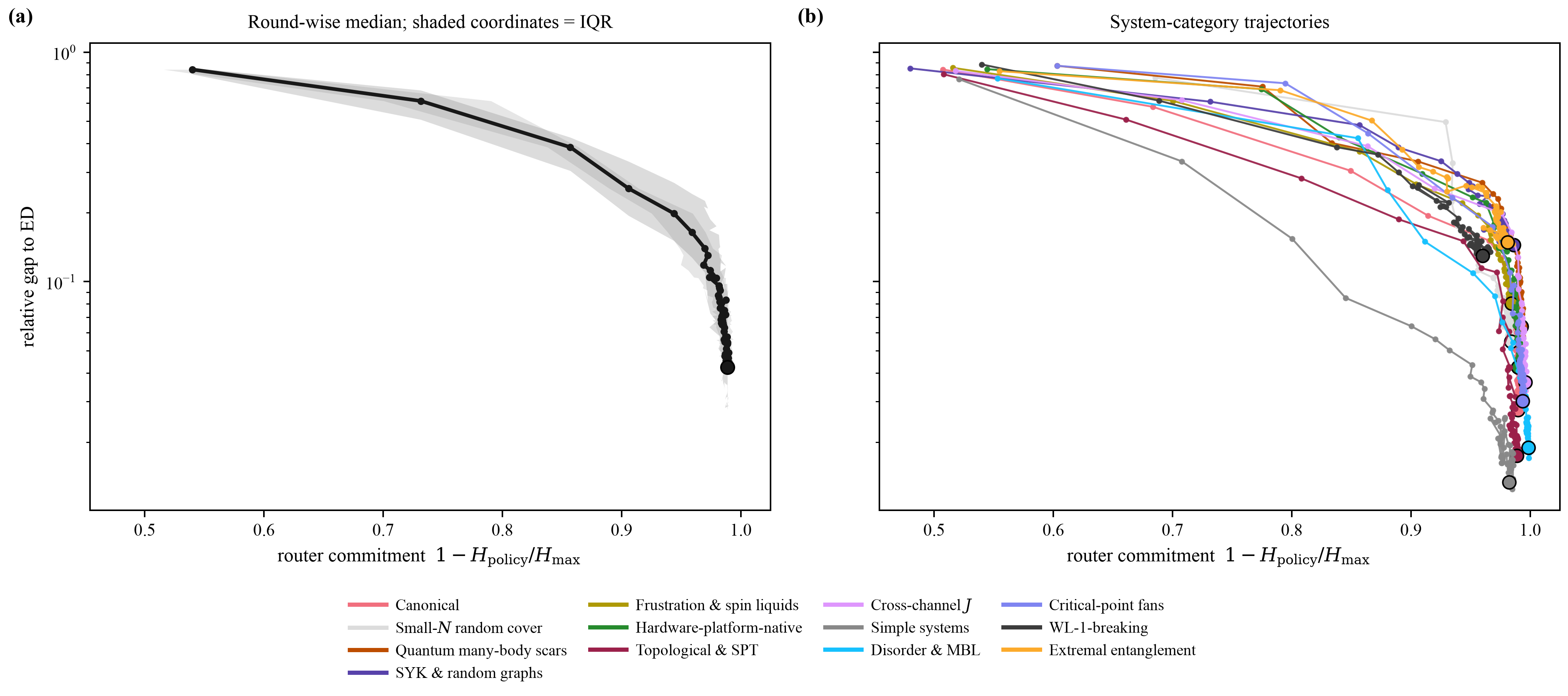}
  \caption{Router commitment versus relative gap to ED over augmentation rounds $0\rightarrow55$ for the 13 system categories of our pretraining dataset (SM~\S~\ref{sec:datasets}).
  (a) The round-wise median trajectory; grey shading denotes the
  interquartile range in the commitment and relative-gap coordinates.
  (b) System-category trajectories, with the enlarged marker indicating
  the final round. }
  \label{fig:router-phase-portrait-280k}
\end{figure}

The router's selected routes also generalise to system sizes and topologies beyond those seen in training. On the $10\times10$ $J_1$--$J_2$ torus (Fig.~\ref{fig:router-coarsegraining-torus}), the first merge pairs nearest-neighbour spins and the second produces exact $2\times2$ plaquettes for every fully occupied four-site cell. Larger cells then remain spatially local as they are combined towards the root. No lattice geometry is supplied to the contraction mechanism, this plaquette hierarchy is inferred from the Hamiltonian-conditioned representations and selected because it supports a lower variational energy. This is supporting evidence that the router is able to recover the locally entangling geometry of the J1J2 model, at a scale above its pretraining.

More strikingly, the same behaviour survives far beyond the pretraining sizes. The $45\times45$ square lattice in Fig.~\ref{fig:router-coarsegraining-large-square} contains 2,025 physical sites, more than thirty times the largest physical size used in pretraining. Nevertheless, the zero-shot route first forms local cells and subsequently organises them into extended diagonal bands consistent with the frustrated $J_1$--$J_2$ interaction structure. At Step~6, $75.0\%$ of nearest-neighbour bonds lie within the same 64-leaf cell, compared with $3.1\%$ under random ordering; at the final root split, $91.4\%$ lie within one branch, compared with the random baseline of $50.0\%$. The router has therefore preserved interaction locality across eleven merge levels without having encountered comparable system sizes during training.

\begin{figure}[!htbp]
  \centering \includegraphics[width=\textwidth]{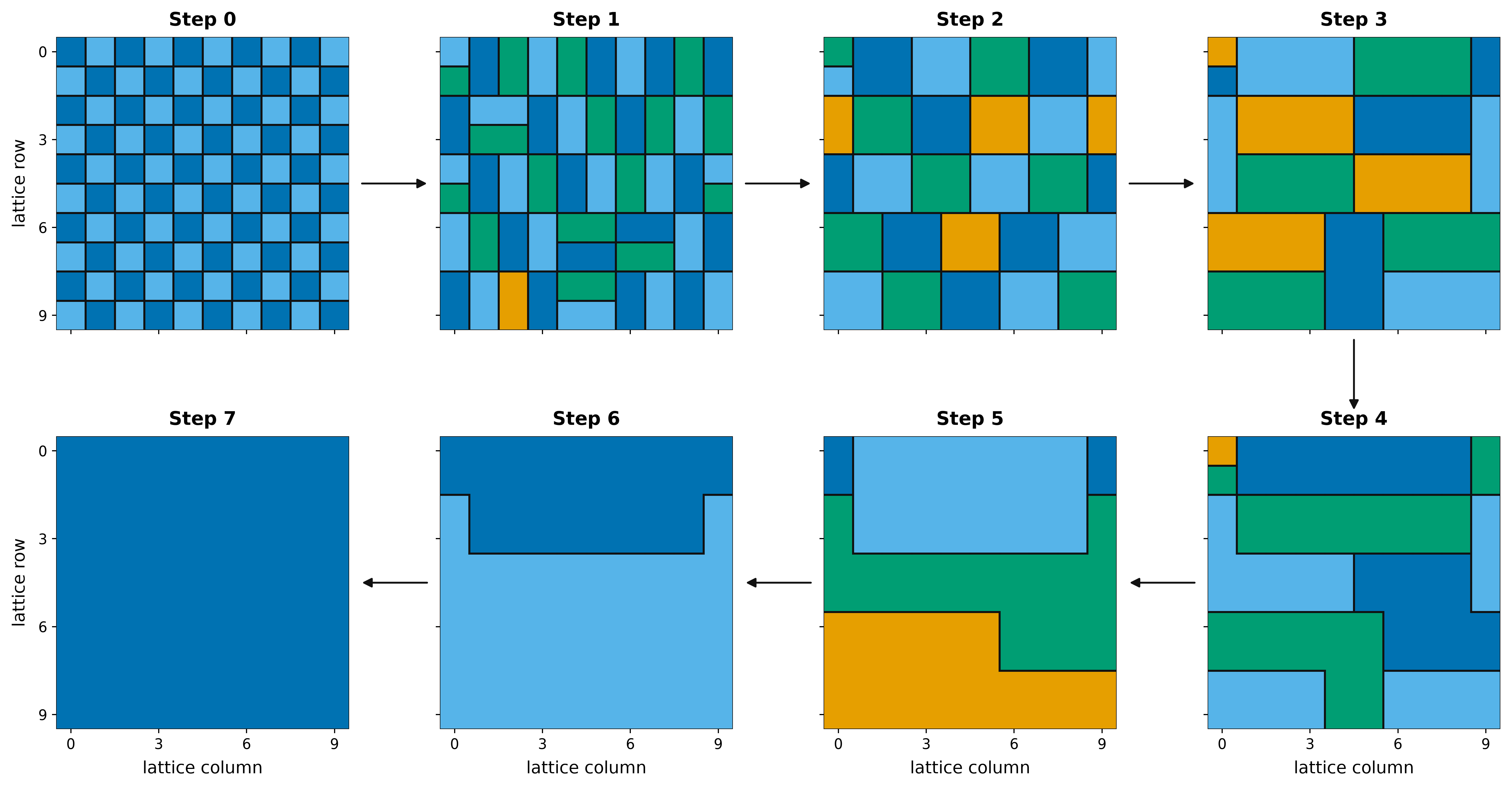}
  \caption{Learned coarse-graining of the $10\times10$ $J_1$--$J_2$ torus. The panels are read along the arrows, with successive rows traversed in opposite directions. Step~0 shows the 100 physical lattice sites. At Step~$k$, the black boundaries enclose the effective cells remaining after $k$ levels of binary merging; a boundary disappears when the two cells on either side are merged. Colours are reassigned independently at each step solely to distinguish adjacent cells and therefore do not track cell identity between panels. For $J_2/J_1=0.3$, the first merge pairs only nearest-neighbour sites. At Step~2, all 24 fully occupied four-site cells are exact $2\times2$ plaquettes. The few incompletely occupied cells arise because the 100 physical sites are embedded in a 128-leaf binary tree, with 28 virtual leaves. Subsequent levels combine the plaquettes into progressively larger spatially local regions, culminating in two root branches at Step~6 and a single scalar at Step~7 of the merge-tree.}
  \label{fig:router-coarsegraining-torus}
\end{figure}

\begin{figure}[!htbp]
  \centering
  \includegraphics[width=\textwidth]{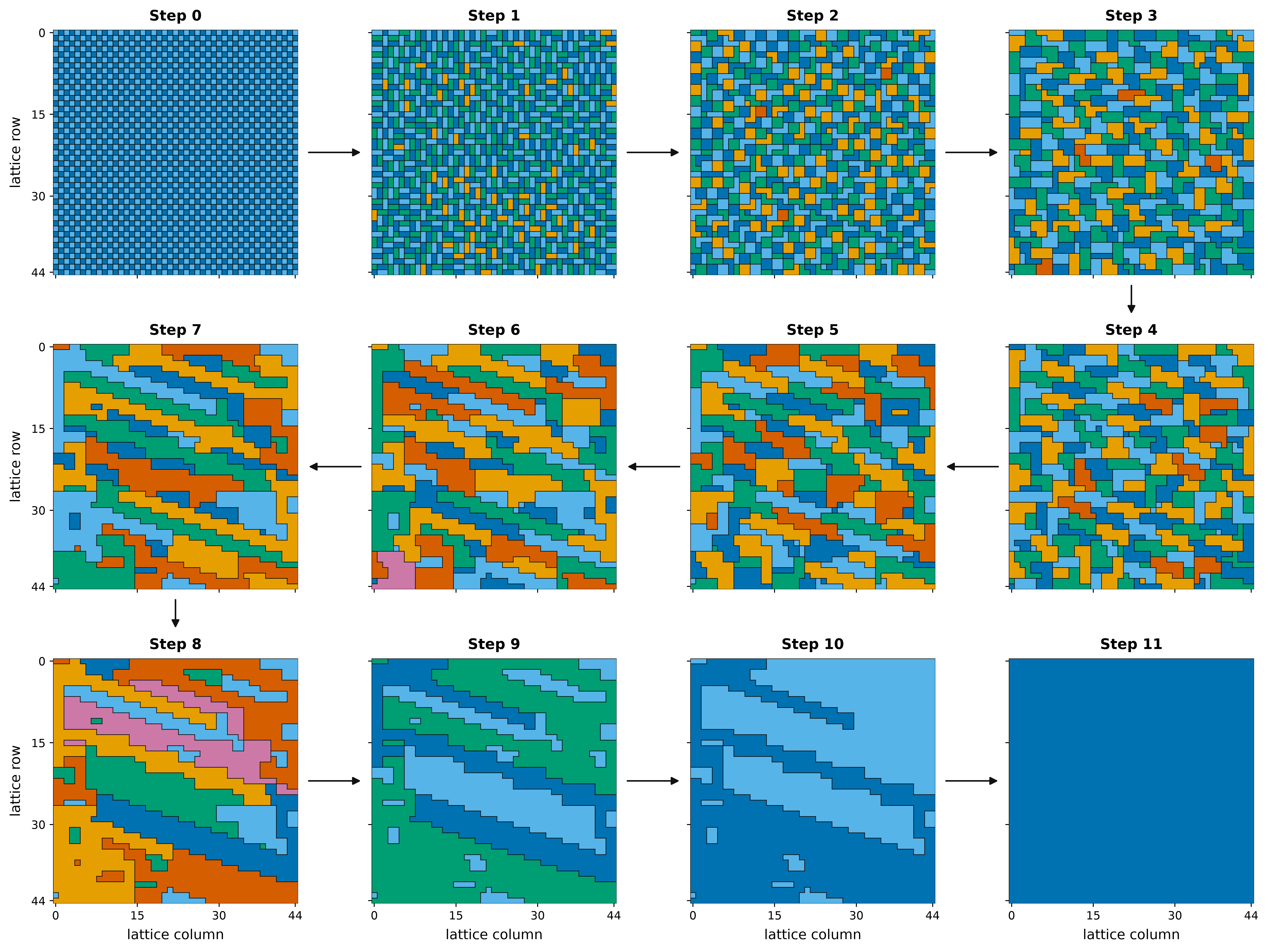}
  \caption{Zero-shot coarse-graining of the $45\times45$ square $J_1$--$J_2$ lattice. The panels follow the arrows from the 2,025 physical sites at Step~0 to the scalar output at Step~11. At each step, black boundaries delimit the current effective cells and disappear when those cells are merged. Colours are reassigned independently within each panel only to separate adjacent cells visually. At $J_2/J_1=0.5$, the early merges are predominantly local. The larger effective cells subsequently organize into extended diagonal bands, showing that the learned hierarchy preserves spatial structure far beyond the training sizes. At Step~6, $75.0\%$ of nearest-neighbour bonds lie within the same 64-leaf cell, compared with a random-order mean of $3.1\%$. At the root split in Step~10, $91.4\%$ lie within one branch, compared with $50.0\%$ for a random ordering. Some cells contain fewer than $2^k$ physical sites because the lattice is embedded in a 2,048-leaf tree containing 23 virtual leaves.}
  \label{fig:router-coarsegraining-large-square}
\end{figure}

The PPP--Ohno case of \ref{ssec:large_systems} provides a contrasting route. Fig.~\ref{fig:router-coarsegraining-ppp} shows the router first pairs the two spin orbitals belonging to each carbon, then combines neighbouring carbons into four-orbital cells, and thereafter doubles the length of the contiguous carbon blocks at each level. It recovers the chemically natural nested factorisation of a long-range interacting fermionic chain despite operating on qubit-labelled Hamiltonian data \cite{pariser1953semi1, pariser1953semi2,pople1953electron}.  This shows the router demonstrates transfer across interaction type as well as scale.

Together, these examples show that the router has not merely memorised a site ordering or learned a generic preference for balanced trees. It extracts a topology-dependent, multiscale contraction hierarchy from the Hamiltonian representations and transfers that hierarchy across periodic and open lattices, long-range chains, interaction families, and system sizes. The result is evidence that pretraining has produced a reusable structural prior for organising entanglement, one that can be converted directly into the compiled contraction path used for inference.

\begin{figure}[!htbp]
  \centering
  \includegraphics[width=\textwidth]{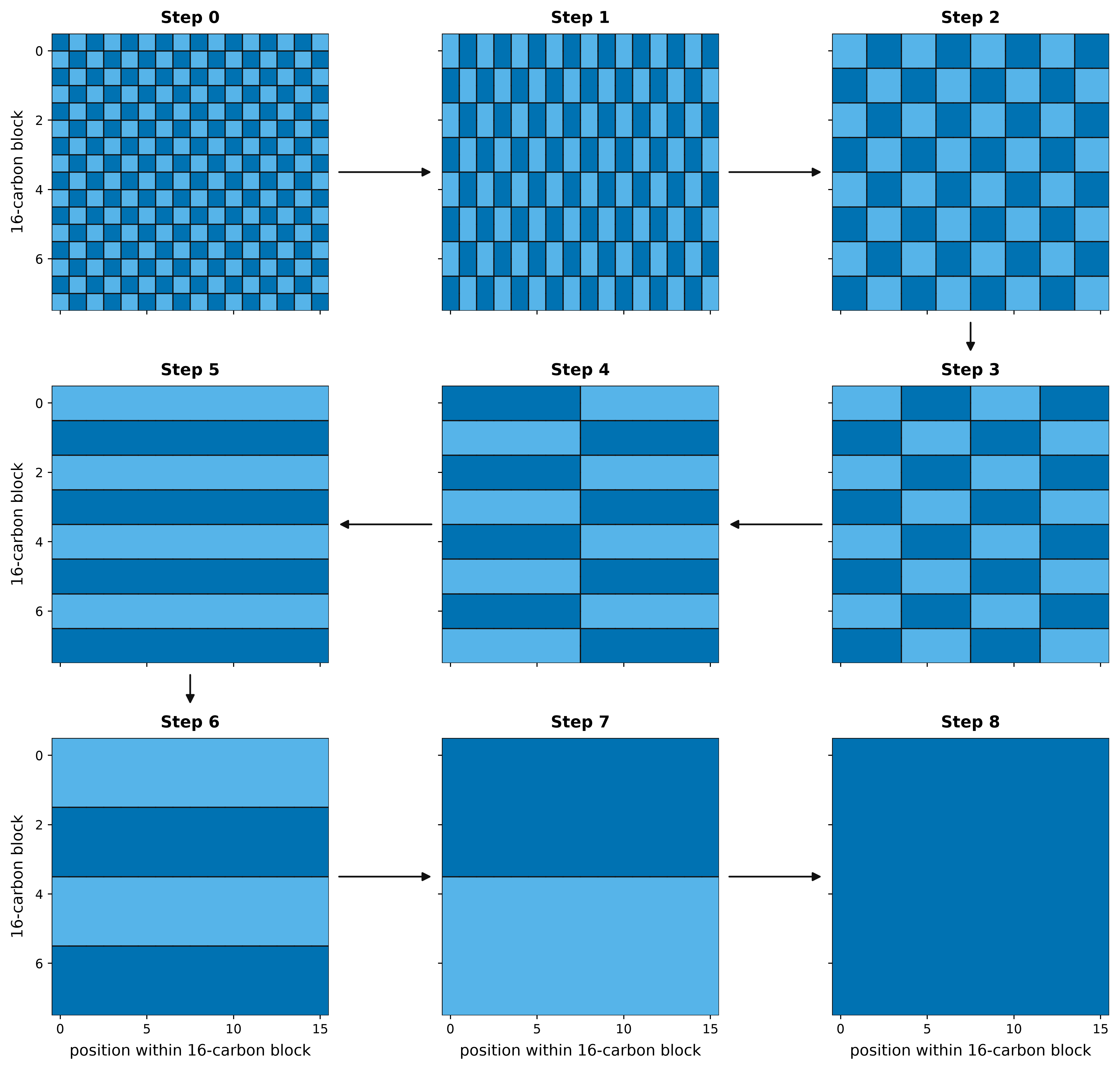}
  \caption{Learned coarse-graining of the 128-carbon PPP--Ohno chain. The 256 spin orbitals are folded into a serpentine arrangement for visualization only, so that consecutive carbons remain adjacent when the chain passes between rows. The two spin orbitals belonging to each carbon are placed vertically. Panels are read along the arrows. Black boundaries enclose the effective cells present at each merge level and disappear when cells are combined; colours are reassigned independently at every step and serve only to distinguish neighbouring cells. At Step~1, every effective cell contains the two spin orbitals of one carbon. Step~2 combines adjacent carbons into four-orbital cells, and every subsequent level doubles the length of the contiguous carbon block. The hierarchy therefore proceeds through blocks of $1,2,4,8,\ldots,64$ carbons before the final scalar merge at Step~8. This exact nested structure provides direct evidence that the router has learned
  the interaction-local factorisation of the PPP Hamiltonian.}
  \label{fig:router-coarsegraining-ppp}
\end{figure}

\clearpage
\section{Discussion and outlook}
\label{sec:mt-outlook}

Quantum computing has attracted great interest in both academic and industrial settings on the premise that classical cost will continue rising until useful quantum many-body calculations have to move onto quantum hardware. Hamilton-Zero takes the other side of this premise for ground state problems specifically. As a foundation model, pay once for learning ground states across a distribution of Hamiltonians and store what the model learns in its weights. As a result, we have a classical model that improves with pretraining, transfers across interaction graphs and system sizes, and returns an explicit, differentiable wavefunction on hardware that exists today.

We believe this changes the economics of quantum ground-state computation. Most classical ground-state solvers restart, or substantially re-optimise, for each new Hamiltonian. Their cost scales with the number of problems solved. Hamilton-Zero replaces repeated training from scratch with one pretraining run whose cost is shared across every subsequent use of the checkpoint. If we let $M$ denote the number of target Hamiltonians evaluated using the same pretrained checkpoint, and let $R$ denote the number of additional fine-tuning runs performed across those targets. So $R=0$ for a completely zero-shot workload, $R=M$ if every target is fine-tuned separately, and $0<R<M$ if one fine-tuned checkpoint is reused across several related Hamiltonians. Then we denote $C_{\mathrm{pre}}$ for the dollar cost of producing the pretrained checkpoint, $C_{\mathrm{eval}}^{(m)}$ for the dollar cost of applying a fixed checkpoint to target Hamiltonian $m$, including inference, MCMC sampling and estimation of the reported observables, and $C_{\mathrm{ft}}^{(r)}$ for the dollar cost of fine-tuning run $r$. The average cost per target Hamiltonian is then
\begin{equation}
  \frac{C_{\mathrm{HZ}}(M)}{M}
  =
  \frac{C_{\mathrm{pre}}}{M}
  +
  \frac{1}{M}\sum_{m=1}^{M}C_{\mathrm{eval}}^{(m)}
  +
  \frac{1}{M}\sum_{r=1}^{R}C_{\mathrm{ft}}^{(r)}.
  \label{eq:hz-amortisation}
\end{equation}
The first term falls as $1/M$, while the evaluation term remains the marginal cost of using the model. Meanwhile the fine-tuning term depends on how often fine-tuning is required and how widely each fine-tuned checkpoint can subsequently be reused. In the zero-shot setting the final sum vanishes.

At the measured throughput of our inference stack, drawing one sample from Hamilton-Zero costs below $10^{-6}$ USD. Public quantum hardware is already priced in the same unit, with services like Amazon Braket listing on-demand QPU prices from $4.25\times10^{-4}$ to $8.0\times10^{-2}$ USD per shot in addition to a $0.30$ USD charge for every submitted job \cite{amazon2026braketpricing}. Even if we discard the job charge entirely, one quantum shot therefore costs between $425$ and $80{,}000$ times one Hamilton-Zero sample at current public prices. Amazon's own error-mitigation example prices an IonQ task of $2{,}500$ shots at $200.30$ USD \cite{amazon2026braketpricing}; the same number of Hamilton-Zero samples costs approximately $0.0025$ USD. This is still the comparison most favourable to quantum hardware, one shot returns one measurement outcome, whereas estimating a ground-state energy requires repeated state preparation and measurements across many operator groups.

Indeed, when we compare this more broadly with the larger research and engineering costs or the continued development and improvement of quantum hardware, the difference is also stark. Improving Hamilton-Zero requires more pretraining compute, which distributes immediately with the release of new model weights, whereas hardware expenditure depreciates. Furthermore, our scaling laws show Hamilton-Zero's compute costs scale around 5 times more favourably than today's LLMs \cite{kaplan2020scaling}, making this a relatively resource-efficient compute demand by today's standards. For reference, the checkpoint we release in this work, including its previous iterations, cost around $\$30k$ to pretrain from scratch. Evaluation and fine-tuning cost another $\$30k$ using on-demand resources. Since a foundation model can improve, absorb new problem classes, and serve more users from the same weights, the economic question of useful quantum advantage is now whether a quantum device can produce a ground-state estimate at lower total cost than a pretrained classical model whose marginal cost will continue to fall with every subsequent release of such models.

Hence, with Hamilton-Zero, we believe the baseline of useful quantum advantage in ground state computation has moved. Comparisons against a classical model trained from scratch establish only an advantage over cold-start optimization, but should, in fact, compare against amortized computation now that a model like Hamilton-Zero is known to be workable. This is especially crucial in the knowledge that to date, no quantum computer has solved quantum chemistry problems beyond a few dozens of qubits at a time for example, yet Hamilton-Zero is able to compile wave-functions to the 8000 qubits scale. From this point, a ground-state quantum-advantage claim should report the same Hamiltonians, the same observables, the same target error and the same confidence level, while counting state preparation, sampling, mitigation, classical outer-loop optimisation and verification. It should then compare total wall time, dollar-cost and energy-cost against both zero-shot Hamilton-Zero and Hamilton-Zero fine-tuned at the same budget if such a claim is to hold up against this model. Until a quantum device wins that comparison, it beats only a now-obsolete classical baseline.

We also believe in future Hamilton-Zero could change how we can search for new physics through phase diagrams. The fact that the fixed weights could witness a phase transition in the Ising model is promising evidence that this is a realistic goal for such foundation models. Indeed, a conventional phase-diagram study solves a sparse collection of coupling points in the Hamiltonian’s parameter-space and spends most of its compute rediscovering nearby ground states. Whereas Hamilton-Zero swept the coupling space of the Ising model in zero-shot inference. This avenue of identifying sharp changes in fidelity, correlations, or order parameters, can let us concentrate fine-tuning where the phase structure changes in future. It could also be used to supply synthetic samples for moment-matrices that underpin recent advances in semi-definite programming techniques for mapping phase diagrams \cite{jansen2026mapping}. As such in this role the model nominates regions for exact methods or experiment.

Hamilton-Zero is also not without limitations however. First, the cost of compute is $\mathrm{poly}(2^{\mathrm{ceil}(log_2 N)})$ due to padding that aims to construct a perfectly balanced tree. So, a 129-qubit system costs the same as a 256-qubit one, as they share the next power of 2 available. Second, the automatic differentiation primitives to evaluate energies have an order of differentiation that scales with the weight of the Pauli string present in the Hamiltonian. As such, we restricted ourselves to quadratic Hamiltonians. However, we do note that such a set is universal in the sense of quantum computation \cite{nielsen2000quantum,aharonov1999quantum} as it corresponds to two-qubit gates. Hence, higher-order terms can be compiled onto a system of physical and ancilla qubits at quadratic order. Finally, even though QMC is variational by design, finite-sample Monte-Carlo estimate may violate the variational inequality by either of two mechanisms: estimator noise and/or mixing bias, which is dependent on the spectral gap \cite{ledoux2015stein} and how long the burn-in is \cite{becca2017quantum} , which can always be alleviated with larger effective sample size and longer burn-in. We report confidence intervals, which also do not guarantee that its upper bound would be above the ground state energy, as confidence intervals only have a guarantee to contain true value within repeated sampling. To empirically minimize bias effects, we conducted 3 different stationarity tests on energy values to detect distribution drifts during estimation, and chose burn-in duration empirically based on this; see SM~\S~\ref{sec:part4}.

Three extensions follow naturally from this work. 
First, we can create versions of Hamilton-zero that were fine-tuned on larger datasets to boost its performance in sector-specific problem families such as quantum chemistry or combinatorial optimisation. This would likely further amortize subsequent fine-tuning of individual systems and lead to even stronger zero-shot inference. Second, we can extend Hamilton-Zero's capabilities from static ground state problems to dynamics via a time encoding. For a pretraining dataset of time-dependent Hamiltonians $H(t)$, we could train against the Schr\"odinger residual $i\partial_t\psi_\theta-H(t)\psi_\theta$ together with an initial-condition loss, as was done in deep stochastic mechanics of bosonic systems \citep{orlova2023deep}. This is doubly motivated by the fact that it would allow our foundation models to compile fermion--qubit mappings into a time-like causal order for arbitrary systems, rather than the perturbative gadget (see Supp. Mat~\ref{app:fermionic-systems}) that uses ancilla qubits, and introduces a perturbative error. Finally, a per-site-odd but nonlinear successor could use higher half-integer Peter--Weyl sectors. Unlike the multilinear architecture used here, it would require the certified Casimir penalty of Appendix~\ref{app:casimir-reg} to retain a physical upper bound. We leave both constructions and their validation to future work.

\subsection*{Acknowledgements}

We thank the NVIDIA Inception Program for hardware support and technical engagement, the Google for Startups Program for cloud infrastructure support, and Nebius for providing the GPU compute that supported this work's training runs. We would also like to thank Jose Ramon Martinez for his help proof-reading the manuscript.

\clearpage
\pdfbookmark[1]{Supplementary Material}{smdivider}
\addtocontents{toc}{\protect\begingroup\protect\small}
\addtocontents{toc}{\protect\contentsline{smhead}{Supplementary Material}{\thepage}{}}
\begin{center}
  {\LARGE Supplementary Material}
\end{center}
\setcounter{section}{0}
\setcounter{figure}{0}
\setcounter{table}{0}
\setcounter{equation}{0}
\renewcommand{\thesection}{S\arabic{section}}
\renewcommand{\thefigure}{S\arabic{figure}}
\renewcommand{\thetable}{S\arabic{table}}
\renewcommand{\theHsection}{SM.\arabic{section}}
\renewcommand{\theHfigure}{SM.\arabic{figure}}
\renewcommand{\theHtable}{SM.\arabic{table}}
\renewcommand{\theHequation}{SM.\theequation}
\let\SMoldappendix\appendix
\renewcommand{\appendix}{\SMoldappendix\renewcommand{\theHsection}{SMapp.\Alph{section}}}

In the supplementary material, we describe the theory, architecture, datasets, optimisation, and engineering that underpin Hamilton-Zero. S\ref{sec:part1} covers variational principles on spinor manifolds, with a full exposition on the theory behind our automatic-differentiation primitives of Lie derivatives, the Peter-Weyl theorem that shows why our model has support on a larger manifold than NQS, as well as how our variational principle manages the case of non-linear quaternionic functions through a Casimir-operator-based regulariser and energy penalty that preserves the spin-$1/2$ variational principle. Next S\ref{sec:part2} describes Hamilton-Zero's architecture, drawing parallels and design choices from large-scale pretraining runs of large-language models. Here we also describe the router we use to learn entanglement structure through deep reinforcement learning, characterise its automorphism-orbit entropy, as well as the tree tensor-network-style contraction mechanism that is conditioned on the model's embeddings of interaction data. Next, we describe in full detail the pretraining routines and datasets we used in S\ref{sec:datasets}. We then describe in S\ref{sec:part4} our state-of-the-art sampling scheme for MCMC on the configuration manifold $SU(2)^N$, our energy kernels, and our extension of the KFAC optimiser, all of which we made shard-mappable and GEMM-optimised for training on Nvidia H200 GPUs.

\section{Variational principles on Spinor Manifolds}
\label{sec:part1}

\subsection{From basis amplitudes to wavefunctions on $\mathrm{SU}(2)^{N}$}
\label{ssec:two-pictures}

The standard spin-$\tfrac12$ wavefunction is a complex-valued function on $\{+,-\}^N$, equivalently a vector $\Psi \in \mathbb{C}^{2^N}$ with one amplitude per computational-basis configuration. Two features of this picture constrain every method built on it. The dimension is exponential in $N$, and the domain is discrete. Thus a Hamiltonian acts through explicit sums over connected configurations, enumerated anew for each interaction structure. Tensor networks and NQS address the first constraint by replacing the $2^N$-dimensional space with a lower-dimensional variational sub-manifold. The second constraint they inherit unchanged however, and it is what ties each trained model to one Hamiltonian's connectivity. The central objective of this section is to remove both at once, by replacing the discrete domain with the continuous group manifold $\mathrm{SU}(2)^N$, on which Hamiltonians act by differentiation.

Tensor networks and neural quantum states replace the exponentially large
Hilbert space by a lower-dimensional variational family \cite{becca2017quantum, carleo2017nqs}, but the resulting
optimisation can still be limited by stability and trainability. On the
quantum-circuit side, barren plateaus provide a distinct failure mode in
which gradient variance can vanish exponentially with system
size~\citep{mcclean2018barren}.

The central objective of this section is to expose the methods
that allow us to sidestep both the curse of dimensionality and the curse of linearity, by
replacing the discrete domain of the spin-$1/2$ basis vectors with
a continuous Lie group domain. The wavefunctions we employ,
 \begin{equation}
          \psi : \bigl(\mathrm{SU}(2)\bigr)^{N} \to \mathbb{C},
\end{equation}
are smooth functions on the Lie group manifold. Here, each spin is parametrised by a unit quaternion $q_i \in S^3$ via the standard identification $\mathrm{SU}(2) \cong S^3$ embedded in $\mathbb{R}^4$. Hence any neural network model in our framework thus consumes $4N$ continuous inputs for $N$ bodies which is linear in $N$.

As usual with representation theory, there is no free lunch. In the following, we will see the signature of the Hilbert-space curse of dimensionality in this picture, and explain how to side-step it by relaxing into a larger domain, $L^{2}(\mathrm{SU}(2)^N)$, than the spin-$\tfrac{1}{2}$ sector, while maintaining the variational principle, so that we can find ground-state energies of many-body systems.

First, though, let us build some intuition for the quaternionic representation, so that examples
later in this part can be written directly in $(q_0, q_1, q_2, q_3)$. Recall a unit quaternion $q = (q_0, q_1, q_2, q_3) \in S^3$ ($q_0^2 + q_1^2 + q_2^2 + q_3^2 = 1$) is the standard coordinate on $\mathrm{SU}(2)$ via
\begin{equation}
  g(q) \;=\; q_0\, I + q_a\, (i \sigma^a)
        \;=\;
        \begin{pmatrix}
          q_0 + i\, q_3 & q_2 + i\, q_1 \\[4pt]
          i\, q_1 - q_2 & q_0 - i\, q_3
        \end{pmatrix},
  \qquad g(q) \in \mathrm{SU}(2),
\end{equation}
where $\det g(q) = q_0^2 + q_1^2 + q_2^2 + q_3^2 = 1$, and $g(q)\, g(q)^{\dagger} = I$.

The four entries of $g(q)$ are the spin-$\tfrac{1}{2}$ \emph{Wigner D-matrix} elements. That is, for any spin $j$, $D^{j}_{mn}(q)$ is the matrix element of the rotation representation on the spin-$j$ irrep. For $j = \tfrac{1}{2}$ with row/column indices $m, n \in \{0, 1\}$,
\begin{equation}
  D^{1/2}_{mn}(q) \;=\; g_{mn}(q),
  \qquad
  \begin{array}{l}
    D^{1/2}_{00}(q) = q_0 + i\, q_3, \\[2pt]
    D^{1/2}_{01}(q) = q_2 + i\, q_1, \\[2pt]
    D^{1/2}_{10}(q) = i\, q_1 - q_2, \\[2pt]
    D^{1/2}_{11}(q) = q_0 - i\, q_3.
  \end{array}
\end{equation}

To identify physical single-site spin states with
$\mathrm{SU}(2)$-functions, we fix the row $m=0$, because the column index
carries the physical spin. Thus
\begin{equation}
  |\uparrow\rangle \longleftrightarrow
  D^{1/2}_{00}(q)=q_0+i q_3,
  \qquad
  |\downarrow\rangle \longleftrightarrow
  D^{1/2}_{01}(q)=q_2+i q_1.
\end{equation}
A general single-site state
$\alpha|\!\uparrow\rangle+\beta|\!\downarrow\rangle$ becomes
\begin{equation}
  \psi(q)=\alpha(q_0+i q_3)+\beta(q_2+i q_1).
\end{equation}
For $N$ sites, let $s_k\in\{0,1\}$ with
$0\equiv\uparrow$ and $1\equiv\downarrow$. Then
\begin{equation}
  |s_1\cdots s_N\rangle
  \longleftrightarrow
  \prod_{k=1}^{N}D^{1/2}_{0,s_k}(q_k),
\end{equation}
and a general spin-$\tfrac12$ state is a linear combination of these
products; see SM~\S~\ref{ssec:peter-weyl}.



Recall a variational ansatz $\Psi_{\theta}$ over the
physical Hilbert space $(\mathbb{C}^{2})^{\otimes N}$ gives an upper
bound
\begin{equation}
  \frac{\langle \Psi_{\theta} |\, H\, |\Psi_{\theta} \rangle}
       {\langle \Psi_{\theta} |\, \Psi_{\theta} \rangle}
  \;\;\ge\;\; E_0
\end{equation}
for \emph{any} normalised $\Psi_{\theta}$. Minimising this over $\theta$
converges to (or above) $E_{0}$. This is the
variational principle in its standard form, the rest of SM~\S~\ref{sec:part1} is
about preserving this inequality under our non-standard quaternionic
parameterisation.

We instantiate the manifold amplitude as a neural network returning
\begin{equation}
  \log\psi_\theta:\mathrm{SU}(2)^N\longrightarrow\mathbb C.
\end{equation}
For the multilinear architecture, its Peter--Weyl coefficients carry a
row-index multiplicity and a column-index physical-spin factor. Fixing a row multiplicity vector recovers the state-vector lift described above; a general model state is interpreted through the decomposition of Sec.~\ref{ssec:peter-weyl}.

We defer the details of the architecture to SM~\S~\ref{sec:part2}, but the signature is a smooth non-linear function on the manifold, to be optimised by gradient descent on a variational loss. As is standard practise in Quantum Variational Monte Carlo, we use $\log \psi_{\theta}$ to prevent underflow, and to allow for stable natural gradient flow \cite{becca2017quantum}.

Having now fixed our model to be a smooth $4N$-input, $1$-output complex-valued network, predicting $\log \psi_{\theta}$ on the Lie-group manifold, we can now turn  to the question of how the Hamiltonian acts on this object, and how to project this object into the physical spin-$\tfrac{1}{2}$ sector of this function family.



\subsection{Operators with automatic differentiation}
\label{sec:ops_and_autodiff}

In this subsection we explain how the shift away from Hilbert space
lets us realise the action of spin operators as Lie derivatives,
executable directly through automatic differentiation. We restrict
attention to quadratic, that is two-body, Hamiltonians throughout, since
two-body interactions are universal for quantum \cite{cubitt2018universal}, and
higher-order interactions can be built from them \cite{oliveira2005complexity, kempe2006complexity}.

Recall from the standard theory that
at each site $i$, the spin is carried by three Hermitian operators
$\hat S_i^x, \hat S_i^y, \hat S_i^z$, closing the algebra
$[\hat S_i^a, \hat S_i^b] = i\, \varepsilon_{abc}\, \hat S_i^c$. Their
squared magnitude is the Casimir
$\hat S_i^2 := (\hat S_i^x)^2 + (\hat S_i^y)^2 + (\hat S_i^z)^2$,
which on a single spin-$\tfrac{1}{2}$ is the scalar $\tfrac{3}{4}$,
and we write $\hat\sigma_i^a := 2\, \hat S_i^a$ for the associated
Pauli matrices. The most general quadratic Hamiltonian we target
couples these Paulis pairwise, together with an on-site field,
\begin{equation}
  \hat H
  \;=\; \frac{1}{4}\!\sum_{i < j} \sum_{a, b}
          J_{ij}^{ab}\, \hat\sigma_i^a\, \hat\sigma_j^b
        \;+\; \frac{1}{2}\!\sum_{i, a} h_i^a\, \hat\sigma_i^a,
  \label{eq:hilbert-hamiltonian}
\end{equation}
with exchange tensor $J \in \mathbb{R}^{N \times N \times 3 \times 3}$
and field $h \in \mathbb{R}^{N \times 3}$. As stated in the main text, it is useful to think of $J$ as the data of a weighted interaction graph on $N$ sites, where each unordered pair $(i,j)$ with $i \neq j$ carries a $3 \times 3$ Pauli-pair coupling matrix $J_{ij} \in \mathbb{R}^{3\times3}$.

To move this operator onto the manifold, we use on the $i$th copy of
$\mathrm{SU}(2)$ the left-invariant vector fields
\begin{equation}
  (L_i^a f)(q_i)
  :=\left.\frac{\mathrm d}{\mathrm dt}
      f\!\left(q_i\exp(tX_a)\right)\right|_{t=0},
\end{equation}
the anti-Hermitian infinitesimal generators of right translations. For
$D^j_{mn}(q)=\langle j,m|D^j(q)|j,n\rangle$, they act on the column
index:
\begin{equation}
  L_i^aD^j_{mn}(q_i)
  =\sum_r D^j_{mr}(q_i)\,
    \bigl(\mathrm dD^j(X_a)\bigr)_{rn}.
\end{equation}
We therefore use the column index $n$ as the physical spin index and the
row index $m$ as the Peter--Weyl multiplicity. The fields close
$[L_i^a,L_i^b]=\kappa\,\varepsilon_{abc}L_i^c$, with
$\kappa=-2$ in the unit-length (half-Lie) normalisation we adopt. Setting
$\hat S_i^a=cL_i^a$ and matching the two algebras forces
$c\kappa=i$, hence $c=i/\kappa=-\tfrac{i}{2}$ and
$\hat\sigma_i^a=2cL_i^a=-iL_i^a$.

Under this identification
each Pauli pair becomes
$\hat\sigma_i^a \hat\sigma_j^b = -L_i^a L_j^b$ and each field term
becomes $\hat\sigma_i^a = -i\, L_i^a$, so the Hamiltonian
\eqref{eq:hilbert-hamiltonian} turns into the second-order
differential operator
\begin{equation}
  \hat H \;\leftrightarrow\;
    -\,\frac{1}{4}\!\sum_{i<j, a, b} J_{ij}^{ab}\, L_i^a L_j^b
    \;-\; \frac{i}{2}\!\sum_{i, a} h_i^a\, L_i^a
  \label{eq:lie-hamiltonian}
\end{equation}
on $L^{2}(\mathrm{SU}(2)^N)$. The same identification sends the spin
magnitude to a Laplacian:
$\hat S_i^2 = c^2 \sum_a (L_i^a)^2 = -\Delta_i$, where
$\Delta_i := \tfrac{1}{4} \sum_a (L_i^a)^2$ is the Laplace--Beltrami
operator on the $i$-th sphere.

What makes \eqref{eq:lie-hamiltonian} executable is that each $L_i^a$
is a known linear combination of the four partial derivatives
$\partial / \partial q_{i, \mu}$ at site $i$,
\begin{equation}
  L_i^a \;=\; \sum_{\mu = 0}^{3} A^{a}_{\mu}(q_i)\;
              \frac{\partial}{\partial q_{i, \mu}},
  \label{eq:L-frame}
\end{equation}
with coefficients $A^{a}_{\mu}$ the left-invariant frame on
$S^{3} \cong \mathrm{SU}(2)$, obtained by right-multiplying $q_i$ by
the three imaginary quaternion units $\hat\imath, \hat\jmath, \hat k$,
\begin{equation}
  \begin{aligned}
  \bigl[A^{x}_\mu(q)\bigr]_{\mu = 0, 1, 2, 3}
    &= (-q_3, \; q_2, \; -q_1, \; q_0), \\
  \bigl[A^{y}_\mu(q)\bigr]
    &= (-q_2, \; -q_3, \; q_0, \; q_1), \\
  \bigl[A^{z}_\mu(q)\bigr]
    &= (-q_1, \; q_0, \; q_3, \; -q_2).
  \end{aligned}
  \label{eq:vielbein}
\end{equation}
A direct expansion using $|q|^2 = 1$ gives the identity we will lean
on,
\begin{equation}
  \sum_{a \in \{x, y, z\}} A^{a}_{\mu}(q)\, A^{a}_{\nu}(q)
    \;=\; \delta_{\mu\nu} - q_\mu q_\nu
    \;=\; \bigl[P_{T_q S^3}\bigr]_{\mu\nu},
  \label{eq:frame-projector}
\end{equation}
so the three tangents $A^{a}(q)$ are orthonormal and their outer
product is the ambient projector onto $T_q S^{3}$. This is what makes
$\{L^a\}$ a genuine orthonormal frame of the round metric, and hence
$\Delta_i$ the Laplace--Beltrami operator named above. When two Lie
derivatives act on the \emph{same} site the channel sum
\eqref{eq:frame-projector} applies; across two distinct sites the two
frames carry different arguments and are simply transported along.

Equation~\eqref{eq:L-frame} makes the differential action computationally tractable. The custom forward-Laplacian propagates the value, the required first derivatives, and their contracted second derivatives together through one structured program traversal, never needing to materialise a dense Hessian. We can also avoid running an independent second-order pass for every pair. For a Hamiltonian with edge set $\mathcal E$, the pair contraction has leading combinatorial cost $\mathcal{O}(|\mathcal E|)$, which is $\mathcal{O}(N^2)$ in the dense case. Here $\mathcal E$ is the Hamiltonian's interacting edge set. In the balanced readout, whose padded width is $\hat N:=2^{\lceil\log_2N\rceil}$, the cross-derivative work summed over levels is likewise $\mathcal{O}(\hat N^2)$ with model-width factors suppressed. .

Applying the frame \eqref{eq:L-frame} to the Lie form
\eqref{eq:lie-hamiltonian} now gives the action of $\hat H$ on $\psi$
in quaternionic coordinates. In a pair term the outer derivative
$L_i^a$ meets the frame $A^{b}(q_j)$ of the other site as a constant,
so it contributes a pure second derivative with no first-order
remainder, while the field term stays first order. Hence, for any
$J_{ij}^{ab}$,
\begin{equation}
  \hat H\, \psi(q)
  \;=\; -\,\frac{1}{4}\!\sum_{i<j, a, b}\!
          J_{ij}^{a b}\!\sum_{\mu, \nu = 0}^{3}\!
          A^{a}_{\mu}(q_i)\, A^{b}_{\nu}(q_j)\,
          \frac{\partial^{2} \psi}{\partial q_{i,\mu}\, \partial q_{j,\nu}}(q)
        \;-\; \frac{i}{2}\!\sum_{i, a}\!
          h^{a}_{i}\!\sum_{\mu}\!
          A^{a}_{\mu}(q_i)\, \frac{\partial \psi}{\partial q_{i, \mu}}(q).
  \label{eq:H-psi-general}
\end{equation}
One step remains to make this usable. The network returns
$\log \psi_{\theta}$, not $\psi_{\theta}$, whereas the local energy
needs ratios $L_i^a L_j^b\, \psi_{\theta} / \psi_{\theta}$. Writing
$\psi = e^{\log \psi}$ and applying the two first-order operators in
turn bridges the two: for any twice-differentiable $\psi \neq 0$, we have
that,
\begin{equation}
  \frac{L_i^a L_j^b\, \psi}{\psi}
  \;=\;
  L_i^a L_j^b\, \log \psi
  \;+\;
  \bigl(L_i^a\, \log \psi\bigr)\,
  \bigl(L_j^b\, \log \psi\bigr).
  \label{eq:log-form}
\end{equation}
On a single site, summing the diagonal channels means the spin
magnitude $\hat S_l^2 = -\Delta_l$ reads,
\begin{equation}
  \frac{-\Delta_l\, \psi}{\psi}
  \;=\; -\,\Delta_l\, \log \psi
        \;-\; \tfrac{1}{4}\sum_{a}\bigl(L_l^a\, \log \psi\bigr)^{2}.
  \label{eq:casimir-diagnostic}
\end{equation}
Both \eqref{eq:log-form} and \eqref{eq:casimir-diagnostic} depend only on $\log\psi_\theta$ and its first two derivatives at the sampled point. These same quantities determine local estimators for observables containing at most two spin operators, including all two-point correlation functions. In \eqref{eq:H-psi-general}, the Hamiltonian enters through the coefficients $J_{ij}^{ab}$ and $h_i^a$ that contract these derivatives. In particular, the nonzero pattern of $J$ carries the interaction graph. The graph is therefore supplied as data to the same differential expression, rather than encoded in a lattice-specific rule. Changing the topology or system size changes the tensor entries and index ranges, and the form of the calculation remains identical aside from that.

\subsubsection{Three examples from the many-body literature}

To make \eqref{eq:H-psi-general} concrete, we read off three couplings that recur across the spin-model literature. The isotropic Heisenberg model $\hat H_{XXX} = J \sum_{\langle ij \rangle} \sum_a \hat S_i^a \hat S_j^a$ has $J_{ij}^{ab} = J\, \delta^{ab}$ on each bond and no field; its Lie form is $-\tfrac{J}{4} \sum_{\langle ij \rangle, a} L_i^a L_j^a$, since $\hat S_i^a \hat S_j^a = c^2 L_i^a L_j^a = -\tfrac14 L_i^a L_j^a$, and \eqref{eq:H-psi-general} collapses to a diagonal channel sum,
\begin{equation}
  \hat H_{XXX}\, \psi(q)
  \;=\; -\,\tfrac{J}{4}\!
        \sum_{\langle i j \rangle}\!
        \sum_{a \in \{x, y, z\}}
        \sum_{\mu, \nu = 0}^{3}
        A^{a}_{\mu}(q_i)\, A^{a}_{\nu}(q_j)\,
        \frac{\partial^{2} \psi}{\partial q_{i, \mu}\, \partial q_{j, \nu}}(q),
  \label{eq:HXXX-psi}
\end{equation}
three mixed-site Hessian sandwiches per bond, one per Pauli channel. The Kitaev-$\Gamma$ coupling $\hat H_{\Gamma} = \sum_{\langle ij \rangle} \Gamma_{ij}\, (\hat S_i^y \hat S_j^z + \hat S_i^z \hat S_j^y)$ is symmetric off-diagonal, $J_{ij}^{yz} = J_{ij}^{zy} = \Gamma_{ij}$, and pairs the two cross-channels,
\begin{equation}
  \hat H_{\Gamma}\, \psi(q)
  \;=\; -\,\tfrac{1}{4} \sum_{\langle i j \rangle} \Gamma_{ij}
          \sum_{\mu, \nu = 0}^{3}
          \bigl[\, A^{y}_{\mu}(q_i)\, A^{z}_{\nu}(q_j)
                  + A^{z}_{\mu}(q_i)\, A^{y}_{\nu}(q_j) \,\bigr]
          \frac{\partial^{2} \psi}{\partial q_{i, \mu}\, \partial q_{j, \nu}}(q).
  \label{eq:HKG-psi}
\end{equation}
The Dzyaloshinskii--Moriya coupling $\hat H_{\mathrm{DM}} = \sum_{\langle ij \rangle} \vec D_{ij} \cdot (\hat{\vec S}_i \times \hat{\vec S}_j)$ is antisymmetric, $J_{ij}^{bc} = \sum_a D_{ij}^a\, \varepsilon^{abc}$, and collecting its three channels into a cross product (well defined because $L_i^a$ and $L_j^b$ commute for $i \neq j$) gives
\begin{equation}
  \hat H_{\mathrm{DM}}\, \psi(q)
  \;=\; -\,\tfrac{1}{4}\!\sum_{\langle ij \rangle}\!
        \sum_{\mu, \nu = 0}^{3}
        \vec D_{ij} \cdot
        \bigl[\, \vec A_{\mu}(q_i) \times \vec A_{\nu}(q_j) \,\bigr]\,
        \frac{\partial^{2} \psi}{\partial q_{i, \mu}\, \partial q_{j, \nu}}(q),
  \label{eq:HDM-psi}
\end{equation}
with $\vec A_{\mu}(q) := (A^{x}_{\mu}, A^{y}_{\mu}, A^{z}_{\mu})(q)$ the frame sliced at fixed $\mu$. Here the bracket is antisymmetric under the joint swap $(i, \mu) \leftrightarrow (j, \nu)$ while the mixed Hessian is symmetric, so a Dzyaloshinskii--Moriya term contributes nothing when both legs sit on one site: there is no on-site DMI, exactly as the antisymmetry $J^{bc} = -J^{cb}$ demands.

\subsection{Peter--Weyl decomposition and the spin-1/2 sector}
\label{ssec:peter-weyl}

Here we use the Peter--Weyl decomposition to identify the physical spin-$\tfrac{1}{2}$ sector inside $L^{2}(\mathrm{SU}(2)^N)$. We will show that the multilinear manifold ansatz occupies a strictly larger \emph{ambient function class} than a fixed lift of any NQS family, while the physical Hilbert space itself remains $2^N$-dimensional. The comparison concerns ambient function classes; finite-parameter expressivity depends on the chosen architecture. We then establish the variational principle on the target spin-$\tfrac{1}{2}$ sector.

The Peter--Weyl theorem decomposes the single-site space into the matrix coefficients of the irreducible representations,
\begin{equation}
  L^2\bigl(\mathrm{SU}(2)\bigr)
  \;\cong\;
  \bigoplus_{j=0,\tfrac12,1,\tfrac32,\dots}
  V_j^{*}\otimes V_j.
\end{equation}
Here $D^j_{mn}(q)=\langle j,m|D^j(q)|j,n\rangle$ carries the row index $m$ in the first, multiplicity factor $V_j^*$ and the column index $n$ in the second, physical factor $V_j$. The coefficients $\{D^j_{mn}\}_{m,n=-j,\dots,j}$ form a complete orthogonal basis under the Haar measure, with $\|D^j_{mn}\|_{L^2}^{2}=1/(2j+1)$. The spin-$\tfrac12$ coefficients $D^{1/2}_{mn}$ are the four entries of $g(q)$ from \S\,\ref{ssec:two-pictures}; the higher-$j$ coefficients are the degree-$2j$ monomials in those entries.

On $N$ sites the basis is the product
\begin{equation}
  \Phi_{\boldsymbol J,\boldsymbol m,\boldsymbol n}(q_1,\dots,q_N)
  =\prod_{k=1}^{N}D^{j_k}_{m_k,n_k}(q_k),
\end{equation}
labelled by $\boldsymbol J=(j_1,\dots,j_N)$, and the full space splits as
\begin{equation}
  L^{2}(\mathrm{SU}(2)^N)
  \;\cong\;
  \bigoplus_{\boldsymbol J}
  \bigotimes_{k=1}^{N}\bigl(V_{j_k}^{*}\otimes V_{j_k}\bigr).
\end{equation}
with $J_k \in \{0, \tfrac{1}{2}, 1, \dots\}$, and any $\psi$ has a unique expansion,
\begin{equation}
  \psi = \sum_{\boldsymbol{J}, \boldsymbol{m}, \boldsymbol{n}}
  c_{\boldsymbol{J}, \boldsymbol{m}, \boldsymbol{n}}\,
  \Phi_{\boldsymbol{J}, \boldsymbol{m}, \boldsymbol{n}}.
\end{equation}

The per-site Casimir of \S\,\ref{sec:ops_and_autodiff} acts diagonally on this basis. Each $D^{j}_{mn}$ restricts to a degree-$2j$ harmonic polynomial in the quaternion coordinates, on which $-\Delta_i$ carries the eigenvalue $j(j+1)$,
\begin{equation}
  -\,\Delta_i\, D^{j_i}_{m_i, n_i}(q_i)
  \;=\; j_i(j_i + 1)\, D^{j_i}_{m_i, n_i}(q_i),
  \label{eq:casimir-eigen}
\end{equation}
and in particular $-\Delta_i\, D^{1/2}_{mn} = \tfrac{3}{4}\, D^{1/2}_{mn}$,
the spin-$\tfrac{1}{2}$ value $\hat S_i^2 = \tfrac{3}{4}$ recorded in
\S\,\ref{sec:ops_and_autodiff}.

The consequence for the network is immediate. An unconstrained $\psi_{\theta} \in L^{2}(\mathrm{SU}(2)^N)$ generically has $c_{\boldsymbol{J}, \boldsymbol{m}, \boldsymbol{n}} \neq 0$ for every multi-index $\boldsymbol{J}$ in the expansion, integer and half-integer alike. It is not a spin-$\tfrac{1}{2}$ wavefunction, nor even a single-spin-$j$ wavefunction; its image lives across the full direct sum.

Indeed, the physical content occupies a single slice of that sum. A wavefunction $\psi$ describes $N$ genuine spin-$\tfrac{1}{2}$ degrees of freedom, one per site, precisely when it is an eigenfunction of every per-site Casimir at the spin-$\tfrac{1}{2}$ value,
\begingroup
\renewcommand{\theHequation}{target-sector}
\begin{equation}
  -\,\Delta_k\, \psi \;=\; \tfrac{3}{4}\, \psi
  \qquad \text{for every site } k = 1, \dots, N.
  \tag{$\ast$}
  \label{eq:target-sector}
\end{equation}
\endgroup
Condition \eqref{eq:target-sector} selects the Peter--Weyl sector $\boldsymbol{J} = (\tfrac{1}{2}, \dots, \tfrac{1}{2})$, and projecting a generic manifold function onto it is what turns it into a physical state. This leads us to the following proposition.

\begin{proposition}[The spin-$\tfrac{1}{2}$ sector]
\label{prop:halfsector}
A function $\psi \in L^{2}(\mathrm{SU}(2)^N)$ satisfies
\eqref{eq:target-sector} if and only if
\begin{equation}
  \psi(q_1, \dots, q_N)
  \;=\;
  \sum_{\boldsymbol{m}, \boldsymbol{n}}
    T_{\boldsymbol{m}, \boldsymbol{n}}\,
    \prod_{k=1}^{N} D^{1/2}_{m_k, n_k}(q_k),
  \qquad
  T_{\boldsymbol{m}, \boldsymbol{n}} \in \mathbb{C}.
\end{equation}
\end{proposition}

The coefficient tensor $T\in\mathbb C^{4^N}$ parameterises the entire sector. Fixing the fiducial row $m_k=0$ at every site fixes the Peter--Weyl multiplicity and leaves the $2^N$ coefficients indexed by $\boldsymbol n$; these are the amplitudes of $\Psi\in(\mathbb C^2)^{\otimes N}$ under the identification of \S\,\ref{ssec:two-pictures}.

The proof, a term-by-term application of the Casimir eigen-equation \eqref{eq:casimir-eigen}, is given in Appendix~\ref{app:prop}.

To see how Prop. \ref{prop:halfsector} works, it is instructive to write down some entangled states of interest to the quantum information community directly in the quaternionic coordinates of \S\,\ref{ssec:two-pictures}.

The Bell state $|\Phi^{+}\rangle=\tfrac{1}{\sqrt2} (|\!\uparrow\uparrow\rangle+|\!\downarrow\downarrow\rangle)$ and singlet $|\Psi^{-}\rangle=\tfrac{1}{\sqrt2} (|\!\uparrow\downarrow\rangle-|\!\downarrow\uparrow\rangle)$ read
\begin{align*}
  \psi_{\Phi^{+}}(q_1,q_2)
  &=\tfrac{1}{\sqrt2}\Bigl[
    (q_{1,0}+i q_{1,3})(q_{2,0}+i q_{2,3})
    +(q_{1,2}+i q_{1,1})(q_{2,2}+i q_{2,1})
  \Bigr],\\
  \psi_{\Psi^{-}}(q_1,q_2)
  &=\tfrac{1}{\sqrt2}\Bigl[
    (q_{1,0}+i q_{1,3})(q_{2,2}+i q_{2,1})
    -(q_{1,2}+i q_{1,1})(q_{2,0}+i q_{2,3})
  \Bigr].
\end{align*}
The three-site GHZ and W states, which read,
\begin{align*}
  |\mathrm{GHZ}\rangle
  &=\tfrac{1}{\sqrt2}
    (|\!\uparrow\uparrow\uparrow\rangle
     +|\!\downarrow\downarrow\downarrow\rangle),\\
  |\mathrm{W}\rangle
  &=\tfrac{1}{\sqrt3}
    (|\!\uparrow\uparrow\downarrow\rangle
     +|\!\uparrow\downarrow\uparrow\rangle
     +|\!\downarrow\uparrow\uparrow\rangle),
\end{align*}
have manifold representatives,
\begin{align*}
  \psi_{\mathrm{GHZ}}(q_1,q_2,q_3)
  &=\tfrac{1}{\sqrt2}\Bigl[
    \textstyle\prod_{k=1}^{3}(q_{k,0}+i q_{k,3})
    +\prod_{k=1}^{3}(q_{k,2}+i q_{k,1})
  \Bigr],\\
  \psi_{\mathrm{W}}(q_1,q_2,q_3)
  &=\tfrac{1}{\sqrt3}\Bigl[
    (q_{1,0}+i q_{1,3})(q_{2,0}+i q_{2,3})(q_{3,2}+i q_{3,1})
    +(q_{1,0}+i q_{1,3})(q_{2,2}+i q_{2,1})(q_{3,0}+i q_{3,3})\\
  &\qquad\quad
    +(q_{1,2}+i q_{1,1})(q_{2,0}+i q_{2,3})(q_{3,0}+i q_{3,3})
  \Bigr].
\end{align*}
Each state is a linear combination of products with the single fiducial row $m_k=0$, exactly the physical lift described above and the form allowed by Proposition~\ref{prop:halfsector}.

We can now state the comparison with neural quantum states precisely. Let $G=\mathrm{SU}(2)$ and define
\begin{equation}
  \mathcal K_N:=(V_{1/2}^{*})^{\otimes N},
  \qquad
  \mathcal H_N:=V_{1/2}^{\otimes N},
\end{equation}
where $\mathcal K_N$ carries the row/multiplicity indices $\boldsymbol m$ and $\mathcal H_N$ carries the column/physical indices $\boldsymbol n$. The multilinear and per-site-odd function spaces are
\begin{align}
  \mathcal F_{\mathrm{lin}}
  &:=\operatorname{span}\!\left\{
      \prod_{i=1}^{N}D^{1/2}_{m_i n_i}(q_i)
      \right\}
    \cong\mathcal K_N\otimes\mathcal H_N,
  \label{eq:F-linear}\\
  \mathcal F_{\mathrm{odd}}
  &:=\left\{f\in L^2(G^N):
      f(\ldots,-q_i,\ldots)=-f(\ldots,q_i,\ldots)\ \forall i\right\}
  \nonumber\\
  &=\bigoplus_{j_1,\ldots,j_N\in\{\frac12,\frac32,\ldots\}}
      \bigotimes_{i=1}^{N}\bigl(V_{j_i}^{*}\otimes V_{j_i}\bigr).
  \label{eq:F-odd}
\end{align}
Fix any unit vector $\chi\in\mathcal K_N$. Its canonical lift is
\begin{equation}
  \iota_\chi(\Psi)(q_1,\ldots,q_N)
  :=2^{N/2}\!\sum_{\boldsymbol m,\boldsymbol n}
  \chi_{\boldsymbol m}\Psi_{\boldsymbol n}
  \prod_{i=1}^{N}D^{1/2}_{m_i n_i}(q_i).
  \label{eq:canonical-nqs-lift}
\end{equation}

\begin{theorem}[Ambient expressivity hierarchy]
\label{thm:expressivity-ladder}
For $N\geq1$, the lift \eqref{eq:canonical-nqs-lift} is an isometry and
\begin{equation}
  \iota_\chi(\mathcal H_N)
  \;\subsetneq\;
  \mathcal F_{\mathrm{lin}}
  \;\subsetneq\;
  \mathcal F_{\mathrm{odd}}.
  \label{eq:expressivity-ladder}
\end{equation}
Consequently, for every NQS variational family
$\mathcal M_{\mathrm{NQS}}\subseteq\mathcal H_N$,
\begin{equation}
  \iota_\chi(\mathcal M_{\mathrm{NQS}})
  \;\subseteq\;
  \iota_\chi(\mathcal H_N)
  \;\subsetneq\;
  \mathcal F_{\mathrm{lin}}
  \;\subsetneq\;
  \mathcal F_{\mathrm{odd}}.
\end{equation}
Let $\widetilde H$ be the right-regular differential representation of the physical Hamiltonian $H$, generated by the left-invariant fields. On the linear class,
\begin{equation}
  \widetilde H\big|_{\mathcal F_{\mathrm{lin}}}
  \cong I_{\mathcal K_N}\otimes H_{\mathcal H_N},
\end{equation}
where $H_{\mathcal H_N}$ denotes $H$ acting on $\mathcal H_N$. Consequently, every nonzero $f\in\mathcal F_{\mathrm{lin}}$ obeys
\begin{equation}
  \frac{\langle f,\widetilde Hf\rangle}{\langle f,f\rangle}\geq E_0,
  \qquad
  \inf_{f\in\mathcal F_{\mathrm{lin}}\setminus\{0\}}
  \frac{\langle f,\widetilde Hf\rangle}{\langle f,f\rangle}=E_0.
  \label{eq:linear-variational-safety}
\end{equation}
On the nonlinear odd class, define
\begin{equation}
  \widehat C:=\sum_{i=1}^{N}
  \left(\widehat S_i^2-\frac34 I\right).
\end{equation}
Then $\widehat C\succeq0$ on $\mathcal F_{\mathrm{odd}}$ and $\ker\widehat C=\mathcal F_{\mathrm{lin}}$.  For every certified $\lambda\geq\mu_\star(J,h)$ from Theorem~\ref{thm:tight}, proved below,
\begin{equation}
  \frac{\langle f,(\widetilde H+\lambda\widehat C)f\rangle}
       {\langle f,f\rangle}
  \geq E_0
  \quad\text{for all }f\in\mathcal F_{\mathrm{odd}}\setminus\{0\},
  \label{eq:nonlinear-variational-safety}
\end{equation}
and the infimum of the regularised quotient over $\mathcal F_{\mathrm{odd}}\setminus\{0\}$ is exactly $E_0$.
\end{theorem}

The strict containments in Theorem~\ref{thm:expressivity-ladder} concern functions on the group manifold. The additional row Peter--Weyl index is a multiplicity, or purification, degree of freedom on which the physical Hamiltonian acts trivially; the column index carries the physical spin. This does not enlarge the physical $2^N$-dimensional Hilbert space. Likewise, the nonlinear class adds higher-half-integer sectors that become variationally consistent only after the Casimir penalty. The theorem proves the exact ambient hierarchy and the variational principle for both the linear and nonlinear cases.

\subsection{Central oddness and variational principles}
\label{ssec:casimir}

Notice first that the Peter--Weyl expansion splits into two parity classes under the per-site reflection $q_i \to -q_i$. The point $-q$ represents the central element $-I \in \mathrm{SU}(2)$, which acts on the spin-$j$ irrep as the scalar $(-1)^{2j}$, so the Wigner coefficients carry that same sign,
\begin{equation}
  D^{j}_{mn}(-q) \;=\; (-1)^{2j}\, D^{j}_{mn}(q),
  \label{eq:parity-split}
\end{equation}
even ($+1$) on the integer sectors and odd ($-1$) on the half-integer ones. The integer-$j$ sectors are therefore exactly the even subspace of $L^2(\mathrm{SU}(2))$, and the half-integer sectors
the odd subspace.

Hence, per-site oddness,
\begin{equation}
  \psi(\dots, -q_i, \dots) \;=\; -\,\psi(\dots, q_i, \dots)
  \qquad \text{for every site } i,
  \label{eq:oddness-constraint}
\end{equation}
annihilates by \eqref{eq:parity-split} every component with integer $j_i$, and the support collapses onto the half-integer sectors $j_i \in \{\tfrac{1}{2}, \tfrac{3}{2}, \tfrac{5}{2}, \dots\}$, with the physical target $j_i = \tfrac{1}{2}$ now the lowest of them.

The architecture of SM~\S~\ref{sec:part2} enforces \eqref{eq:oddness-constraint}  exactly, so we take it as given here. Although each normalised carrier is a nonlinear function of its quaternion, the leaf and merge log-scales cancel those normalisations exactly in the reconstructed amplitude. Lemma~\ref{lem:restored-multilinearity} therefore shows that the final $\psi_\theta$ is degree one in every $q_i$.  A degree-one function on $\mathrm{SU}(2)$ is a linear combination of the four coefficients $D^{1/2}_{mn}$, so Proposition~\ref{prop:halfsector} places the ansatz in $\mathcal F_{\mathrm{lin}}$ identically.

Two distinct cases follow as a consequence of Theorem~\ref{thm:expressivity-ladder}. First, the foundation architecture we describe next is linear in each quaternion, so lies in $\mathcal F_{\mathrm{lin}}$ and directly obeys the variational principle  $E[\psi_\theta]\geq E_0$. Second, any future generation of model which is per-site odd but \textit{non-linear} in the quaternioncs lies in the larger $\mathcal F_{\mathrm{odd}}$. Such a model's higher-half-integeter sectors invalidate the variational principle as it stands, however it can be restored with an energy penalty restores the upper bound while retaining that
larger ambient function class, see Appendix~\ref{app:casimir-reg}.

\subsection{A topological obstruction to wavefunction-valued foundation models}
\label{ssec:topological-nogo}

In this subsection, we show that continuous pure-state foundation NQS cannot uniformly represent topologically nontrivial ground-state families. Theorem~\ref{thm:expressivity-ladder} already separates the ambient function classes at fixed Hamiltonian. The separation here is of a different kind and appears only at the level of a family of topologically distinct interacting Hamiltonians, where one set of model weights must serve every Hamiltonian of a this family at once. 

Let $X$ be a compact manifold and $H : X \to \mathrm{Herm}(\mathcal H_N)$ a continuous family with a nondegenerate ground state at every $x \in X$. Compactness gives a uniform gap $\Delta := \min_{x \in X}\,[E_1(x) - E_0(x)] > 0$, and the ground-state projector $P_0(x)$, the Riesz integral of the resolvent around $E_0(x)$, is norm-continuous in $x$ \cite{arcozzi1995riesz}. Its range $L := \operatorname{Ran} P_0 \subset X \times \mathcal H_N$ is therefore a complex line bundle over $X$, the Berry bundle of the family \cite{berry1984,simon1983}, classified up to isomorphism by its first Chern class $c_1(L) \in H^2(X;\mathbb Z)$.

A wavefunction-valued foundation model evaluated across the family is a continuous map $\Phi : X \to \mathcal H_N \setminus \{0\}$, where the network reads the interaction data $(J,h)$ at $x$ and returns 
basis amplitudes. Hence, we can create a full state vector of basis amplitudes from this object, which is the form of the FNQS programme \cite{rende2025foundationnqs,zaklama2025attention}. The following theorem shows that a continuous family of Hamiltonians create a situation where a pure-state (wavefunction-valued) FNQS projection onto the ground state eigenspace is necessarily vanishing, creating a topological no-go theorem for such foundation models.

\begin{theorem}[Topological no-go]
\label{thm:topological-nogo}
Let $X$, $H$, $\Delta$, $L$ be as above with $c_1(L) \neq 0$. Then every continuous $\Phi : X \to \mathcal H_N \setminus \{0\}$ satisfies $P_0(x_\star)\,\Phi(x_\star) = 0$ at some $x_\star \in X$. Consequently
\begin{equation}
  \frac{\|P_0(x_\star)\Phi(x_\star)\|^{2}}{\|\Phi(x_\star)\|^{2}} \;=\; 0
  \qquad\text{and}\qquad
  \frac{\langle \Phi(x_\star),\, H(x_\star)\,\Phi(x_\star)\rangle}
       {\langle \Phi(x_\star),\, \Phi(x_\star)\rangle}
  \;\geq\; E_0(x_\star) + \Delta ,
  \label{eq:nogo-floor}
\end{equation}
so the worst-case ground-state fidelity of $\Phi$ over the family is zero and its worst-case energy error is at least one full gap inside $X$.
\end{theorem}

\begin{proof}
$s(x) := P_0(x)\Phi(x)$ is a continuous section of $L$. If $s$ vanished nowhere, $x \mapsto s(x)/\|s(x)\|$ would be a global unit frame and $L$ would be trivial, contradicting $c_1(L) \neq 0$. Hence $s(x_\star) = 0$ at some $x_\star$, i.e.\ $\Phi(x_\star) \in \operatorname{Ran}\,(I - P_0(x_\star))$, on which the Rayleigh quotient of $H(x_\star)$ is at least $E_1(x_\star) \geq E_0(x_\star) + \Delta$.
\end{proof}

This theorem affects the zero-shot regime a foundation model. Per-instance fine-tuning evades it only by replacing the single map $\Phi$ with a different model for each Hamiltonian, which renders null the amortisation the foundation setting exists for. We can see this with the following one-qubit system example.

For $H(n) = -\sum_a n^{a}\hat\sigma^{a}$ with $n \in S^2$, the Berry bundle is the Hopf line bundle and $|c_1(L)| = 1$ \cite{berry1984}. Every continuous $\Phi : S^2 \to \mathbb C^2 \setminus \{0\}$ therefore returns the excited ray at some $n_\star$, with energy error $2$, the full spectral width. The family embeds at any qubit number with fixed spectators,
\begin{equation}
  H_N(n) \;=\; -\sum_a n^{a}\hat\sigma_1^{a} \;-\; \sum_{i=2}^{N} \hat\sigma_i^{z},
  \label{eq:hopf-embedded}
\end{equation}
whose ground state is nondegenerate with uniform gap $2$ and whose Berry bundle is the Hopf line tensored with a trivial spectator line, so $c_1(L) \neq 0$ persists. Single-site fields are part of the interaction data $(J,h)$ a foundation model must accept, so this obstructed two-sphere already lies inside the input class of any FNQS.

Because of this, the fidelity with true ground state will always vanish somewhere in $X$ when the wavefunction goes orthogonal to the true ground state. This can be especially misleading when using foundation models to sweep over phase diagrams, and would a false prediction of a phase transition.

Indeed, the ground-ray map $g : X \to \mathbb P(\mathcal H_N)$, $x \mapsto [\Psi_0(x)]$, is globally continuous. A wavefunction-valued model attempts to recover a continuous lift $\widetilde g : X \to S(\mathcal H_N)$ of $g$ through the Hopf quotient $\pi : S(\mathcal H_N) \to \mathbb P(\mathcal H_N)$ from the unit sphere $S(\mathcal H_N) \subset \mathcal H_N$, and such a lift exists if and only if the pullback of the tautological line bundle, $g^{*}\mathcal O(-1) \cong L$, is trivial \cite{simon1983}. Contrary to that, models like Hamilton-Zero are actually operator-valued and $\mathbb P(\mathcal H_N)$ embed continuously in $\mathcal H_N^{*} \otimes \mathcal H_N=\mathcal K_N \otimes \mathcal H_N$ by $[\Psi] \mapsto |\Psi\rangle\langle\Psi|$, and $x \mapsto P_0(x)$ stays continuous whereas no continuous $\Psi_0$ exists. To elaborate on that, we note that $\mathcal H_N^{*} \otimes \mathcal H_N = \operatorname{End}(\mathcal{H}_N)$ is a space of operators, which contain as a subspace a space of density operators. This is otherwise known as purification \cite{nielsen2010quantum}. We also want to emphasize that this is distinct from router's density mixture, and holds for its fixed permutation setting as well.

On degenerate ground-state manifolds, a $k$-fold degenerate ground-state space defines a rank-$k$ Berry bundle with non-abelian holonomy \cite{wilczekzee1984}. Wavefunction-valued models then require continuous frames, obstructed by its characteristic classes beginning with the Euler class $c_k$. A systematic account of which Hamiltonian families are obstructed, and of what operator-valued ans\"atze can represent uniformly, is deferred to a forthcoming work.

Finally, we note that the bare existence of such cases is a no go for any globally defined foundation NQS in a setting that we operate in: arbitrary quadratic hamiltonian. We also note that this no-go applies should we have chosen a phase space like that of \cite{heightman2025foundations}, since $S^2$ suffers the same lifting problem whenever such a foundation-model is wavefunction-valued, and thus acting on only one component of the Peter-Weyl sector. We also note that doubling the number of binary variables would circumvent such an issue for the NQS community, as it makes an effective density matrix that can act on both components of the Peter-Weyl sector.

\section{Foundation Ansatz and Entanglement Structure via Reinforcement Learning}
\label{sec:part2}

In this section, we describe the foundation model's architecture. The model has two different types of input. First is the Hamiltonian's interaction tensor $(J,h)$, and second is a spin configuration $q\in SU(2)^N$ sampled from the model's own density. In this section, we will assume efficient access to such samples, deferring the details of our novel sampling scheme to SM~\S~\ref{ssec:mcmc}.

These two sources of data naturally divide the architecture into two components that together allow us to simultaneously condition the foundation model wavefunction on the Hamiltonian (allowing for generalisation) as well as preserving the necessary exchange anti-symmetry of our quaternionic representation from SM~\S~\ref{sec:part1}.

The first component, which we call the \textit{trunk}, is a function of the Hamiltonian's interaction tensor. In (\S\,\ref{ssec:hampack}) we describe how a given input $(J,h)$ is cached into two tensors $J''$ and $h''$. In (\S\,\ref{ssec:featurizer}) we then describe our learnable featurizer, which converts the cached inputs $J''$ and $h''$ into embeddings of the Hamiltonian's system context for the trunk. We then describe the architecture of the trunk itself in (\S\,\ref{ssec:trunk-leaf}), which consists of $L$ residual attention blocks nested between feed-forward networks, as well as a per-site leaf builder where the spin configuration $q$ enters (\S\,\ref{ssec:trunk-leaf}).

The second component is a readout contraction mechanism of the trunk that incorporates spin coordinates. This component is akin to a tree-tensor network, augmented by three routines that enable generalisation to arbitrary quadratic Hamiltonians and system sizes. First, we use edge-level attention in a tree to enable it to condition on the Hamiltonian's interaction tensor. Second, we use a learnable routing policy to allow the tree to entangle arbitrary sites in the system.  Finally, we use a contextualizer to allow the tree to condition on Hamiltonian data
while preserving the necessary exchange anti-symmetry of our quaternionic representation.

The contraction itself is detailed in (\S\,\ref{ssec:tree-readout}), where we describe a rank-4 quadrilinear merge with level attention on each contraction layer,  as well as the contextualizer. Section (\S\,\ref{ssec:routes}) then describes the learnable routing policy to optimise over contraction paths with deep reinforcement learning.

Together, the merged contraction with a learnable path enables arbitrary entanglement between sites, with the entanglement conditioned on the Hamiltonian's interaction tensor $(J,h)$ enabling a Hamiltonian-dependent contraction pathway.

Figure~\ref{fig:architecture} shows the information flow of our ansatz over a forward pass. Recall for the spin-configuration input $q \in (\mathrm{SU}(2))^N$, we use the quaternion representation embedded in $\mathbb{R}^{4}$, so $q_{i} = (q_{i,0}, q_{i,1}, q_{i,2}, q_{i,3})$ with $q_{i,0}^{2} + q_{i,1}^{2} + q_{i,2}^{2} + q_{i,3}^{2} = 1$. This makes the input space a product of spheres $\mathrm{SU}(2) \cong S^3 \hookrightarrow \mathbb{R}^4$.

\begin{figure}[!htbp]
  \centering
  \input{architecture.tikz}
  \caption{Information flow through the foundation ansatz. The Hamiltonian $(J, h)$ is cached into the two tensors $J''$ and $h''$ (\S\,\ref{ssec:hampack}) and consumed by the \textsc{SystemFeaturizer} (\S\,\ref{ssec:featurizer}), which emits the per-bond embedding $e^{(0)}$, the per-site local embedding $\ell^{(0)}$ and the seed $g^{(0)}$ of the global stream. The trunk (\S\,\ref{ssec:trunk-leaf}) then propagates all three streams through $L$ residual attention blocks, refreshing $g$ once per block (Figure~\ref{fig:gladder} traces the stream on its own). Notice the spin configuration $q$ does not enter the trunk, it is a function of the Hamiltonian interaction alone. After the trunk, the route policy (\S\,\ref{ssec:routes}) samples a permutation $p \in S_{\hat N}$ from $\pi_\theta(p \mid h, \ell^{(L)}, e^{(L)}, m)$, conditioned by its own fork of the global stream, and that permutation relabels both the trunk outputs $(\ell^{(L)}, e^{(L)})$ and the spin configuration $q$ before downstream consumption. In particular before the contextualizer, whose slot clocks and tree-distance decays encode position and must therefore act in the routed frame. A two-block leaf contextualizer (\S\,\ref{sssec:merge-upgrade}) then refines the routed trunk outputs over the padded slot frame, manufacturing contexts for the holes, and emits the per-slot states $\hat\ell$ together with the rewritten edges that seed the tree's level-$0$ edge field $E^{(0)}$ (the latter not drawn). The per-site leaf builder (\S\,\ref{sssec:leaf-builder}) is the unique entry point for $q$ and consumes $(\hat\ell_i, g, q_i)$ per site, with $g$ the stream's refreshed post-context value (it does not see the per-bond edges). Its output $u_i^{(0)}$ feeds the tree readout (\S\,\ref{ssec:tree-readout}), which contracts the $\hat N$ leaf carriers down to a single root carrier; the per-bond $e^{(L)}$ re-enters the tree internally at the level-edge attention in (\S\,\ref{sssec:level-attn}) (not drawn here). The root readout is the routed log-amplitude $\log \psi_\theta(q \mid p) \in \mathbb{C}$.}
  \label{fig:architecture}
\end{figure}
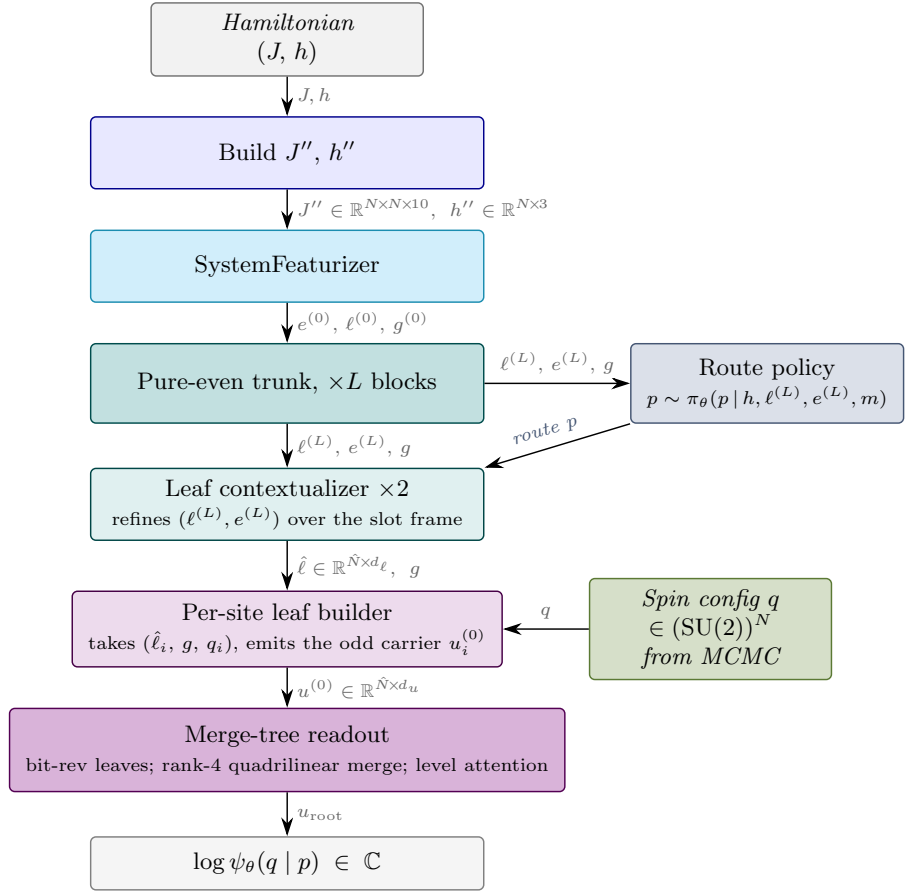

\subsection{Hamiltonian encoding and spectral normalisation}
\label{ssec:hampack}

Recall from \S \ref{sec:ops_and_autodiff} that the Hamiltonian family we target is
\begin{equation}
  \hat H
  \;=\; \frac{1}{4} \sum_{i<j} \sum_{a,k}
        J^{ak}_{ij}\, \sigma^a_i\, \sigma^k_j
  \;+\; \frac{1}{2} \sum_{i} \sum_{a}
        h^a_i\, \sigma^a_i,
  \label{eq:hamiltonian}
\end{equation}
with exchange tensor $J \in \mathbb{R}^{N \times N \times 3 \times 3}$ (cross-channel Pauli-pair couplings on off-diagonal $(i, j)$) and transverse field $h \in \mathbb{R}^{N \times 3}$ (per-site three-vector).

As stated in the main text, it is useful to think of $J$ as the data of a weighted graph on $N$ sites, where each unordered pair $(i, j)$ with $i \neq j$ carries a $3 \times 3$ Pauli-pair coupling matrix $J_{ij} \in \mathbb{R}^{3 \times 3}$, which we read as nine \emph{flavours} of edge, one per Pauli-pair $(a, b) \in \{x, y, z\}^{2}$ giving the operator $\sigma^a_i \sigma^b_j$, see Fig.~\ref{fig:bonds-examples}.

We keep exchange and field as separate input channels for the model. Before feeding them into the model, we normalize them in the following way.

\emph{Step 1 (off-diagonal --- bond-symmetrise).}\;
On Hermitian-conjugate pairs $(i, j) \leftrightarrow (j, i)$, average
to enforce the swap symmetry:
\begin{equation}
  J^{\text{sym}}_{ij, ab}
  \;=\;
  \tfrac{1}{2}\,\bigl(J_{ij, ab} + J^{*}_{ji, ba}\bigr)
  \qquad (i \neq j).
  \label{eq:bond-symmetrise}
\end{equation}
We do this because Pauli interactions are hermitian, so at the level of the interaction tensor and its relation to an interaction graph, it must be bi-directionally symmetric.

\emph{Step 2 (form the common Hermitian scale matrix).}\; For the scale computation only, the cache builder forms
\begin{equation}
  M_{ij,ab}
  \;=\;
  J^{\mathrm{sym}}_{ij,ab}
  +\delta_{ij}\,i\,\varepsilon_{abc}\,h_i^c,
  \label{eq:scale-operator}
\end{equation}
and reshapes it as
\begin{equation}
  \widetilde M
  :=\operatorname{reshape}(M,\,3N\times3N).
\end{equation}
For real $h_i$, the block $A_{i,ab}:=\varepsilon_{abc}h_i^c$ is real and antisymmetric, so
\begin{equation}
  (iA_i)^\dagger=iA_i.
\end{equation}
Together with the Hermitian swap symmetrisation of Step~1, this makes $\widetilde M$ Hermitian. Its eigenvalues are therefore real, and the shared scale is
\begin{equation}
  s
  :=\|\widetilde M\|_2
  =\max_k\bigl|\lambda_k(\widetilde M)\bigr|,
  \label{eq:scale}
\end{equation}
computed once per system with \texttt{jnp.linalg.eigvalsh}. The spectral radius $\max_k|\lambda_k|$ is defined for every square matrix. The role of the factor $i$ is instead to make $\widetilde M$ Hermitian, so that the Hermitian eigensolver is appropriate and the spectral radius equals the spectral norm. The diagonal field block participates only in this scale computation and is discarded afterwards.

\emph{Step 3 (normalise and flatten).}\;
We divide the separate exchange and field channels by the same scale $s$. The exchange keeps its off-diagonal bonds and an identity-marker channel, with the diagonal Pauli block left empty since the field no longer lives there,
\begin{equation}
  J''_{ij, c} \;=\;
  \begin{cases}
    (J^{\text{sym}} / s)_{ij, ab(c)} & c = 0, \dots, 8 \ \text{and}\ i \neq j,\\
    0 & c = 0, \dots, 8 \ \text{and}\ i = j,\\
    \delta_{ij} & c = 9,
  \end{cases}
  \label{eq:normalise-flatten}
\end{equation}
where $c = 0, \dots, 8$ enumerates the nine $(a, b)$ pairs in row-major order, and the field is cached as its own per-site tensor,
\begin{equation}
  h''_i \;=\; h_i / s.
  \label{eq:field-cached}
\end{equation}
Building these inputs costs one call to \texttt{eigvalsh} per system, paid once per system when a panel or an augmentation round is loaded, and once per system at evaluation or fine-tuning. We emphasise here that the size of the interaction tensor, and the graph that follows, is \textit{polynomial} in the number of bodies, since our foundation model handles quadratic interactions which are a small sub-group of the overall Pauli group on $N$ spins.

By construction, this normalisation makes both cached inputs \emph{invariant} under simultaneous rescaling $(J, h) \to \alpha (J, h)$ for $\alpha > 0$, since the spectral radius $s$ scales linearly with $\alpha$. This avoids the attention layers needing to learn the invariance from our datasets. Under $\alpha < 0$ the inputs are \emph{sign-flip-equivariant}: a global sign passes through to the nine Pauli channels of $J''$ and to $h''$, while the identity channel is unaffected. Finally they are \emph{smooth} at $h = 0$, since $h''$ is linear in $h$ and the scale $s$ stays bounded away from zero. These three properties are exactly the ones our featurizer (\S\,\ref{ssec:featurizer}) relies on for equivariance under rescaling and for the $h \to 0$ limit.

This encoding is a bijection between the original $(J, h)$ and the cached pair $(J'', h'')$, and both representations are in the datasets we have
open-sourced.

The diagonal field block in $M$ is solely a common-normalisation device. It is neither a cached model input nor a replacement of the physical first-order field operator. Including both exchange and field in $\widetilde M$ lets one scale normalise the two channels together. Once $s$ has been computed, the diagonal block is discarded, and the model receives the bond tensor $J''$ and the per-site field $h''$ separately, while the local-energy expression retains the field as the first-order term in Eq.~\eqref{eq:H-psi-general}.

The per-system rescaling-invariance is then the standard ML input-normalisation trick applied at the level of the Hamiltonian itself. Dividing through by the spectral radius lands every system in a unit-norm regime regardless of physical units, so the training dataset's overall energy scale stops being a feature the network has to re-learn for every batch. Indeed this has direct analogy to dividing pixel values by $255$ or attention keys by $\sqrt{d}$, removing a known-uninformative scale so the network spends capacity on what actually distinguishes systems. We can simply rescale the output by the renormalisation factor $s$ to get back to physical units, and optimise in the normalised space where the network is more stable and generalises better.

\subsection{Learnable featurizer}
\label{ssec:featurizer}

The featurizer is a single fusion block that consumes the two cached inputs and emits three downstream-ready tensors:
\begin{equation}
  \texttt{SystemFeaturizer}: \;
  \bigl(\, J'' \in \mathbb{R}^{N \times N \times 10},\;
           h'' \in \mathbb{R}^{N \times 3} \,\bigr)
  \;\longmapsto\;
  \bigl(\;
    e \in \mathbb{R}^{N \times N \times 96},\;\;
    \ell \in \mathbb{R}^{N \times 128},\;\;
    g \in \mathbb{R}^{1024}
  \;\bigr).
\end{equation}
These three objects correspond to a per-bond embedding $e$ for the trunk's attention layers, a per-site local embedding $\ell$ for the trunk's attention layers, and a global context vector $g$ that seeds the global stream the trunk maintains (\S\,\ref{ssec:trunk-leaf}). The featurizer is a single-pass architecture with no iteration or recurrence; it is a single function of the Hamiltonian that produces these three tensors in one go. Figure~\ref{fig:featurizer} shows the data-flow and we detail each step below.
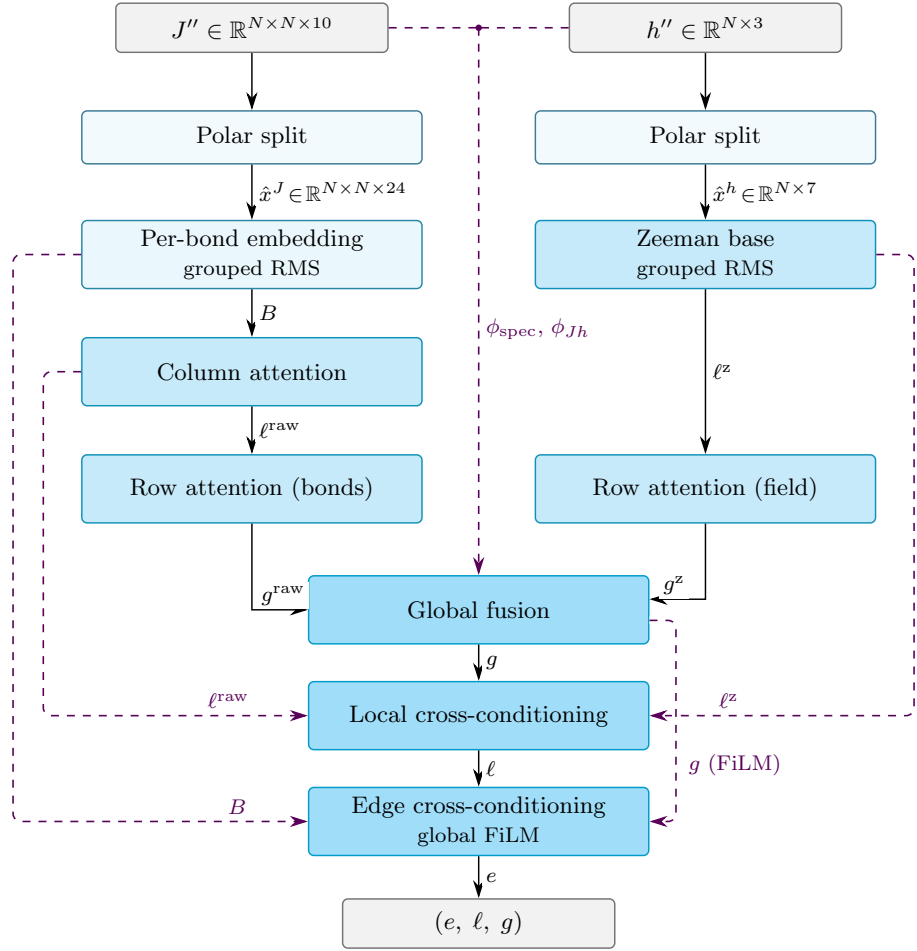
\begin{figure}[!htbp]
  \centering
  \input{architecture-featurizer.tikz}
  \caption{%
    Data flow inside the system featurizer (\S\ref{ssec:featurizer}). The two cached inputs enter on separate spines, each first mapped to its polar form (raw, radial, and angular coordinates): the exchange $J''$ is embedded per bond ($B$) by a grouped-RMS map and aggregated by column attention into the bond descriptor $\ell^{\text{raw}}$, while the field $h''$ is lifted into the Zeeman base $\ell^{\text{z}}$ by the same grouped-RMS pattern. Two row-attention layers pool the two views ($g^{\text{raw}}$, $g^{\text{z}}$), which a fusion map, fed also by the spectrum block $\phi_{\mathrm{spec}}$ and the relative-scale channels, combines into the global vector $g$, the seed of the global stream. Solid arrows are the main pipeline; dashed magenta arrows are skip-connections: the Zeeman base $\ell^{\text{z}}$ is the residual backbone of the local feature $\ell$, conditioned by $\ell^{\text{raw}}$ and $g$, and the edges are conditioned by $B$ and, through a global FiLM, by $g$. The outputs $(e, \ell, g)$ are the three tensors consumed by the trunk.
  }
  \label{fig:featurizer}
\end{figure}

Both cached tensors enter in \emph{polar} form. We read the nine Pauli-pair entries of each bond block as a vector $x^J_{ij} = \mathrm{vec}(J''_{ij,1:9}) \in \mathbb{R}^{9}$ and split it into a radius and a softened direction,
\begin{equation}
  r^J_{ij} = \lVert x^J_{ij} \rVert_2,
  \qquad
  u^J_{ij} = \frac{x^J_{ij}}{\sqrt{\lVert x^J_{ij} \rVert_2^2 + \tau^2}},
\end{equation}
where $\tau = 10^{-3}$ is a fixed feature-scale floor, and $||\cdot||_2$ is the Euclidean norm. The radius $r^J$ shows whether a bond is present, whilst the direction $u^J$ reports its angular type and stays bounded as the bond vanishes. Notice $\tau$ is a fixed feature-scale floor, not a numerical $\varepsilon$. The bond input stacks the raw, radial, angular, and self-bond coordinates, and the field is lifted the same way,
\begin{equation}
  \hat x^J_{ij}
    = \bigl[\, x^J_{ij},\; r^J_{ij},\; u^J_{ij},\; \mathbf 1_{i=j} \,\bigr]
      \in \mathbb{R}^{20},
  \qquad
  \hat x^h_i
    = \Bigl[\, h''_i,\;
              \lVert h''_i \rVert_2,\;
              \tfrac{h''_i}{\sqrt{\lVert h''_i \rVert_2^2 + \tau^2}} \,\Bigr]
      \in \mathbb{R}^{7}.
\end{equation}
They are functions of the cached inputs alone and are computed once per system alongside them (\S\,\ref{ssec:hampack}). We hand the model the coordinate, its magnitude, and its direction separately rather than make a dense layer recover all three: the loci that matter physically, vanishing bonds or fields, isotropic blocks, vanishing skew or symmetric-traceless parts, rank-deficient couplings, are exactly where a plain projection smears presence into direction.

The two embedding paths normalise in groups, and the normalisation is the root-mean-square kind throughout: for a vector $z$ with learned per-channel scale $\gamma$,
\begin{equation}
  \mathrm{RMS}(z)
    \;=\; \gamma \odot \frac{z}{\sqrt{\overline{z^{2}} + \varepsilon}},
  \qquad
  \overline{z^{2}} \;:=\; \tfrac{1}{d}\textstyle\sum_a z_a^{2},
  \label{eq:rms-norm}
\end{equation}
with no mean subtraction and no shift. The mean subtraction of LayerNorm is left out on purpose: it couples every channel through a shared offset and worsens the conditioning of the curvature the optimiser preconditions against, and dropping it is the same trade that moved the large-language-model community from LayerNorm to RMSNorm \citep{zhang2019root}. Every normalisation layer on the even path, here and in every later stage, is of this form. Reshaping a vector into $K$ groups of width $d_{\mathrm{grp}}$, the grouped variant normalises each group on its own with its own scale,
\begin{equation}
  \mathrm{GRMS}(z)_k
    = \gamma_k\, \frac{z_k}{\sqrt{\overline{z_k^{2}} + \varepsilon}},
  \qquad
  \gamma_k = \mathbf 1 \;\text{at init},
\end{equation}
with $\overline{z_k^{2}}$ the mean square within group $k$. Each group is then free to become a learned filter over the raw, radial, and angular coordinates, a data-driven stand-in for the named decomposition $J = tI + Q + A$ into isotropic, symmetric-traceless, and skew parts, without hard-coding it. We use $K = 16$ groups of width $d_{\mathrm{grp}} = 16$ on both paths.

The per-bond embedding is then a two-layer map, a grouped norm, and two more layers,
\begin{equation}
  z^J_{ij}
    = \mathrm{GRMS}\!\bigl( \mathrm{SiLU}(\hat x^J_{ij}\, W^J_1)\, W^J_2 \bigr),
  \qquad
  B_{ij}
    = m_i m_j \cdot \mathrm{SiLU}\!\bigl(
        \mathrm{SiLU}(z^J_{ij}\, W^J_3)\, W^J_4 \bigr),
\end{equation}
with $d_{\text{bond}} = 128$, every bias zero at initialisation, and $m_i m_j$ the site-padding pair mask: it zeroes pairs with a padded endpoint and carries no bond information.

Bonds absent from the system's adjacency pattern are not masked to zero. Wherever a bond feature is consumed, absence substitutes a learned token in its place,
\begin{equation}
  z_{ij} \;\longleftarrow\;
  \begin{cases}
    z_{ij}, & \mathrm{present}_{ij} \\
    t, & \text{otherwise}
  \end{cases}
  \qquad
  \mathrm{present}_{ij} \;=\;
    \bigl[\,\exists\, c \le 8 :\; J''_{ij,c} \neq 0\,\bigr]
    \;\vee\; [\, i = j \,],
  \label{eq:absent-token}
\end{equation}
with one trained token $t$ per consumption site. Six sites carry a token: the bond embedding, the bond attention's key and value input, the Zeeman base, the field's row-attention view, the field branch of the global fusion, and the field's entry in the local cross-conditioning. On the field path a zero field substitutes its tokens the same way, per site everywhere except the fusion branch, which gates per system. The physical reading is that ``no bond'' becomes a representable direction of the basis
rather than the numeric zero, while the substitution still cuts the degenerate-input gradients that a hard zero would poison. The site-padding masks $m_i$ and $M^{\text{key}}_j$ are a different object: they mark sites absent from the padded system, not bonds absent from a present system's graph. Nothing downstream carries a bond mask; bond presence reaches the trunk only through these tokens.

Column attention then attends over the $j$ index for each $(i, \mathsf{h})$ using per-row learnable queries $Q^{\text{col}}_{\mathsf{h}} \in \mathbb{R}^{d_{\mathsf{h}}}$,
\begin{equation}
  K_{ij,\mathsf{h}} = W_K\, \mathrm{RMS}(B_{ij}),
  \quad
  V_{ij,\mathsf{h}} = W_V\, \mathrm{RMS}(B_{ij}),
  \quad
  \alpha^{\text{col}}_{ij,\mathsf{h}}
    = \mathrm{softmax}_j\!\Bigl(
        \tfrac{Q^{\text{col}}_{\mathsf{h}} \cdot K_{ij,\mathsf{h}}}{\sqrt{d_{\mathsf{h}}}}
        + M^{\text{key}}_j
      \Bigr),
\end{equation}
where $\mathsf{h} \in \{1, \dots, n_{\mathsf{h}}\}$ is the attention-head index (we use sans-serif to disinguish from the transverse field $h$), $d_{\mathsf{h}}$ is the per-head width with $n_{\mathsf{h}} \cdot d_{\mathsf{h}} = d_{\text{bond}}$, and $W_K, W_V \in \mathbb{R}^{d_{\text{bond}} \times d_{\text{bond}}}$ are linear projectors whose outputs split into $n_{\mathsf{h}}$ head blocks of width $d_{\mathsf{h}}$ indexed by $\mathsf{h}$. Here $\mathrm{RMS}$ is the map of Eq.~\eqref{eq:rms-norm}, applied along the bond-feature axis, $\mathrm{softmax}_j$ normalises over the second site index, and $M^{\text{key}}_j$ is the key-padding mask which is $0$ when site $j$ is present in the system, and $-\infty$ when it is absent. The kernel computes twice the output head count and gates the heads in bisected pairs,
\begin{equation}
  o_{i, \mathsf{h}}
    = \sum_j \alpha^{\text{col}}_{ij,\mathsf{h}}\, V_{ij,\mathsf{h}},
  \qquad
  \ell^{\text{raw}}_i
    = m_i \cdot
      \Bigl[\, \mathrm{sigmoid}\bigl(o_i^{(1 : n_{\mathsf{h}})}\bigr)
        \odot o_i^{(n_{\mathsf{h}} + 1 : 2 n_{\mathsf{h}})} \,\Bigr],
  \label{eq:col-glu}
\end{equation}
a GLU-style multiplicative gate at the head level \newcitep{dauphin2017language,shazeer2020glu}. Every attention layer of SM~\S~\ref{sec:part2} shares this micro-structure. The descriptor $\ell^{\text{raw}}_i$ carries the exchange structure seen from site $i$.

The field enters on its own per-site channel. Because each site carries its own $h''_i$, the field is a genuinely local quantity, and we lift its polar input $\hat x^h_i$ into a Zeeman base through the same two-layer / grouped-norm / two-layer pattern,
\begin{equation}
  z^h_i
    = \mathrm{GRMS}\!\bigl( \mathrm{SiLU}(\hat x^h_i\, W^h_1)\, W^h_2 \bigr),
  \qquad
  \ell^{\text{z}}_i
    = m_i \cdot \mathrm{SiLU}(z^h_i\, W^h_3)\, W^h_4,
\end{equation}
with $d_\ell = 128$ and, again, every bias zero at initialisation. The field groups learn their own thresholded detectors --- field nearly zero, field isotropic, a single dominant axis --- in the same way the bond groups do.

The global context pools both views of the system. A first row-attention layer with $n_q$ learnable global queries $Q^{\text{row}}_{q,\mathsf{h}} \in \mathbb{R}^{n_q \times n_{\mathsf{h}} \times d_{\mathsf{h}}}$ attends over the bond descriptors $\ell^{\text{raw}}$ to give $g^{\text{raw}}$, and a second, structurally identical layer attends over the Zeeman bases $\ell^{\text{z}}$ to give $g^{\text{z}}$.

Meanwhile the relative weight is already implicit in $(J'',h'')$, and $\phi_{Jh}$ surfaces it as an explicit bounded feature:
\begin{equation}
  \phi_{Jh} \;=\; \bigl[\,\log(1{+}r_J),\; \log(1{+}r_h)\bigr].
  \label{eq:jh-channels}
\end{equation}
The branch is guarded so its gradient signal remains clean on the $h = 0$ majority of the panel. The fusion normalises each pooled view on its own, takes the two derived blocks raw (both are bounded by construction), and closes with an output normalisation,
\begin{equation}
  g \;=\; \mathrm{RMS}\Bigl(
    \mathrm{FFN}_{g}\!\bigl(
      \bigl[\, \mathrm{RMS}(g^{\text{raw}}),\,
        \mathrm{RMS}(g^{\text{z}}),\,
        \phi_{Jh} \,\bigr]
    \bigr)
  \Bigr).
\end{equation}
On a system that carries no field at all, $\mathrm{RMS}(g^{\text{z}})$ is replaced by the fusion's own absent-field token. The per-site local feature then takes the Zeeman base as its backbone and folds in the bond descriptor and the global context through a pre-norm residual map, each input normalised on its own,
\begin{equation}
  \ell_i \;=\; m_i \cdot \mathrm{RMS}\Bigl(\, \ell^{\text{z}}_i
              + \mathrm{FFN}_{\ell}\!\bigl(
                  \bigl[\,\mathrm{RMS}(\ell^{\text{z}}_i),\,
                    \mathrm{RMS}(\ell^{\text{raw}}_i),\,
                    \mathrm{RMS}(g)\,\bigr]
                \bigr) \Bigr),
\end{equation}
and at a zero-field site the first entry is this map's own absent-field token. The field is the residual base because it is the dominant per-site signal, while the exchange reaches a site only by aggregation over its bonds. Edge cross-conditioning then refines each bond with the new local features, and the global vector enters only as a gentle modulation. The core mixes the bond with its two endpoints, while a FiLM scale and shift read off the global context,
\begin{equation}
  c_{ij} = \mathrm{FFN}_{e}\!\bigl(
             \mathrm{RMS}\bigl[\,B_{ij},\, \ell_i,\, \ell_j\,\bigr] \bigr),
  \qquad
  [\gamma, \beta] = W_F\, \mathrm{RMS}(g) + b_F,
  \quad
  \gamma \leftarrow 0.1\tanh\gamma,\;\; \beta \leftarrow 0.1\beta,
\end{equation}
\begin{equation}
  e_{ij} \;=\; m_i m_j \cdot \mathrm{RMS}\!\bigl(
                  P\, B_{ij} + (1 + \gamma)\, c_{ij} + \beta \bigr).
\end{equation}
The FiLM factors start near zero, $P$ is a bias-free residual projector applied when $d_{\text{bond}} \neq d_{\text{edge}}$, and the pair mask is, again, padding alone.

Given the graph structure of the Hamiltonian input, we note here that the featurizer is a GNN in the sense that it performs neighbor aggregation and global pooling. There is, however, no separate message-passing module since attention subsumes the GNN role, framed in a form compatible with our GPU kernel pipeline and natural-gradient optimisers see SM~\S~\ref{sec:part4} for details of our engineering.

\subsection{Trunk and per-site leaf builder}
\label{ssec:trunk-leaf}

Recall that the featurizer of \S\,\ref{ssec:featurizer} emits three tensors: the per-bond embedding $e \in \mathbb{R}^{N \times N \times d_e}$, the per-site local feature $\ell \in \mathbb{R}^{N \times d_\ell}$, and the global context $g \in \mathbb{R}^{d_g}$. Both $e$ and $\ell$ depend on the Hamiltonian $(J, h)$ only through the cached inputs $J''$ and $h''$ of \S\,\ref{ssec:hampack}, and every system-context quantity downstream of this point is, by construction, a function of $(J'', h'')$ and therefore a function of $(J, h)$. The trunk's job is to propagate $(e, \ell, g)$ through $L$ residual blocks: it maintains all three streams, mixing information across sites and bonds and keeping the global summary current as it does so. We emphasise that this is completely independent functions of spin configuration $q$, which enter only once, at a per-site \emph{leaf builder} that sits between the trunk and the merge tree. As such the per-site antisymmetry constraint of \S\,\ref{ssec:casimir},
\begin{equation}
  \psi(\dots, -q_i, \dots) \;=\; -\,\psi(\dots, q_i, \dots)
  \qquad\text{for every site } i,
  \label{eq:per-site-oddness}
\end{equation}
becomes a property of the leaf builder alone, and the expressivity of the trunk is not limited by any symmetry requirement and thus we may use any standard transformer recipe \citep{vaswani2017attention}.

The trunk is structured by blocks, indexed by $\beta \in \{0, 1, \dots, L - 1\}$. Let $\ell^{(\beta)} \in \mathbb{R}^{N \times d_\ell}$ and $e^{(\beta)} \in \mathbb{R}^{N \times N \times d_e}$ be the per-site and per-bond states at block $\beta$, with $\ell^{(0)} = \ell$ and $e^{(0)} = e$ from the featurizer. We write $\ell_i \in \mathbb{R}^{d_\ell}$ for the row of $\ell^{(\beta)}$ at site $i$ and $e_{ij} \in \mathbb{R}^{d_e}$ for the bond slot at $(i, j)$, dropping the $(\beta)$ superscript when the block is unambiguous.

The global context is the trunk's third stream, and it is a maintained one. The natural alternative is to broadcast the featurizer's $g$ unchanged, as a fixed conditioning signal for every block. But a fixed global summarises the couplings as the featurizer saw them, and that summary grows stale exactly as the trunk's refinement makes the site and bond states worth summarising. The trunk therefore refreshes $g$ once per block, from the states the block has just updated. The stream lives on the sphere of unit root-mean-square norm: writing
\begin{equation}
  \mathrm{Norm}(x) \;:=\;
    \frac{x}{\sqrt{\max\!\bigl(\overline{x^{2}},\, \varepsilon\bigr)}},
  \qquad
  \overline{x^{2}} \;:=\; \tfrac{1}{d} \textstyle\sum_a x_a^{2},
  \label{eq:sphere-norm}
\end{equation}
for the parameter-free projection $\mathbb{R}^d \to \mathbb{R}^d$ onto that sphere (the floor $\varepsilon = 10^{-4}$ replaces the usual added constant, so the map is exact whenever $\overline{x^2} \ge \varepsilon$ and sends the origin to itself), the trunk receives the featurizer's global as $g^{(0)} = \mathrm{Norm}(g)$ and returns every update to the same sphere. $\mathrm{Norm}$ carries no parameters; the RMS layers of \S\,\ref{ssec:featurizer} normalise the same way but carry a learned per-channel scale, and the two should not be conflated.

A block applies three residual sublayers in sequence: an edge-update sublayer on the per-bond stream, then an edge-biased multi-head self-attention (MHA) sublayer and a feed-forward network (FFN) sublayer on the per-site stream; a refresh of the global closes the block. The edge stream moves first,
\begin{equation}
  e^{(\beta + 1)}
    \;=\; e^{(\beta)} \;+\;
    \tfrac{1}{\sqrt L}\,
      \mathrm{EdgeUpdate}\!\bigl(
        \mathrm{RMS}(e^{(\beta)}),\; \ell^{(\beta)}
      \bigr),
    \label{eq:trunk-edge-update}
\end{equation}
so that the attention which follows reads \emph{fresh} bonds: the
just-updated $e^{(\beta + 1)}$ enters the attention logits as a
per-head additive bias, defined below. The per-site stream then
takes two residual steps,
\begin{align}
  \ell'
    \;&=\; \ell^{(\beta)} \;+\;
    \tfrac{1}{\sqrt L}\,
      \mathrm{MHA}\!\bigl(\mathrm{RMS}(\ell^{(\beta)});\,
        e^{(\beta + 1)}\bigr),
    \label{eq:trunk-l-update} \\
  \ell^{(\beta + 1)}
    \;&=\; \ell' \;+\;
    \tfrac{1}{\sqrt L}\,
      \mathrm{FFN}\!\bigl(\mathrm{RMS}(\ell') + W_g\, g^{(\beta)}\bigr),
    \label{eq:trunk-ffn-update}
\end{align}
and the global closes the block by reading the sites the block has just written,
\begin{equation}
  g^{(\beta + 1)}
    \;=\; \mathrm{Update}\!\bigl(g^{(\beta)},\;
      \mathrm{pool}\bigl(g^{(\beta)},\, \ell^{(\beta + 1)}\bigr)\bigr),
  \label{eq:trunk-g-refresh}
\end{equation}
with $\mathrm{pool}$ and $\mathrm{Update}$ defined below. In Eq.~\eqref{eq:trunk-l-update} the semicolon indicates that the map is conditioned on the argument after it; $\mathrm{RMS}$ is the normalisation of \S\,\ref{ssec:featurizer}, applied along the feature axis; and $W_g \in \mathbb{R}^{d_\ell \times d_g}$ is a learned projection whose output is broadcast to every site. The global enters the per-site path exactly once, additively, at the FFN's input. We emphasise that the two per-site steps are sequential residuals (the FFN acts on the attention's output $\ell'$, not on $\ell^{(\beta)}$) and that all three sublayers are pre-norm: each reads a normalised copy of its stream and adds its result to the raw stream. The factor $1/\sqrt L$ applies inverse-square-root depth scaling to each residual branch, matching the depth dependence induced by batch normalisation at initialisation \newcitep{de2020batch}. Under the usual independence approximation at initialisation, each additive update contributes variance $\mathcal{O}(1/L)$ to the running stream, so the cumulative variance across $L$ blocks stays $\mathcal{O}(1)$.

The two maps that refresh the global recur at every later stage of the architecture; the trunk's instance is the first of many. The pool is a \emph{descriptor attention}: its query is produced from the current global state and determines the attention weights over the input set. For a set $\{x_i\}_{i=1}^{n}$ of states $x_i \in \mathbb{R}^{d_x}$ and the current global $g$,
\begin{equation}
  q_{\mathsf{h}} = W^Q_{\mathsf{h}}\, g,
  \qquad
  k_{\mathsf{h},i} = W^K_{\mathsf{h}}\, \mathrm{RMS}(x_i),
  \qquad
  v_{\mathsf{h},i} = W^V_{\mathsf{h}}\, \mathrm{RMS}(x_i),
  \label{eq:descriptor-qkv}
\end{equation}
\begin{equation}
  a_{\mathsf{h},i}
    = \mathrm{softmax}_i\!\Bigl(
        \tfrac{q_{\mathsf{h}} \cdot k_{\mathsf{h},i}}{\sqrt{d_k}}
        + M_i \Bigr),
  \qquad
  \mathrm{pool}\bigl(g, \{x_i\}\bigr)
    = \Bigl[\, {\textstyle\sum_i} a_{\mathsf{h},i}\, v_{\mathsf{h},i}
      \,\Bigr]_{\mathsf{h}=1}^{n_{\mathsf{h}}}
    \;\in\; \mathbb{R}^{n_{\mathsf{h}} d_v},
  \label{eq:descriptor-pool}
\end{equation}
where $\mathsf{h}$ is the sans-serif head index of \S\,\ref{ssec:featurizer}, $W^Q_{\mathsf{h}} \in \mathbb{R}^{d_k \times d_g}$, $W^K_{\mathsf{h}} \in \mathbb{R}^{d_k \times d_x}$ and $W^V_{\mathsf{h}} \in \mathbb{R}^{d_v \times d_x}$ are per-head projectors, $[\cdot]_{\mathsf{h}}$ concatenates the head outputs, and $M_i$ is the key-padding mask of \S\,\ref{ssec:featurizer} ($0$ where the element is present, $-\infty$ where it is absent). This is the featurizer's row attention with two changes: the learnable queries are replaced by projections of $g$, and the head outputs concatenate plainly, without the gate of Eq.~\eqref{eq:col-glu}. We use $n_{\mathsf{h}} = 4$ heads of width $d_k = d_v = 64$, so the pooled reading has width $n_{\mathsf{h}} d_v = 256$, a quarter of the stream's own width $d_g = 1024$.

The Hamiltonian's data streams use the normalised interpolation of nGPT \citep{loshchilov2025ngpt}. For a current state $x$ and candidate $\Delta_x$, let $a_x$ be a learned per-channel gate, write $\mathbf 1$ for the all-ones vector and $\odot$ for elementwise multiplication, and define $\operatorname{logit}(r):=\log(r/(1-r))$. Then
\begin{equation}
  \widetilde\alpha_x
    := 0.8\,\operatorname{sigmoid}(a_x),
  \qquad
  a_x\big|_{\mathrm{init}}
    = \operatorname{logit}\!\left(\frac{0.2}{0.8}\right)\mathbf 1,
  \qquad
  \mathcal N_x(x,\Delta_x)
    := \mathrm{Norm}\!\Bigl(
      x + \widetilde\alpha_x\odot
      \bigl(\mathrm{Norm}(\Delta_x)-x\bigr)
    \Bigr).
  \label{eq:ngpt-interpolation}
\end{equation}
Thus every channel starts with interpolation weight $0.2$ and remains in the open interval $(0,0.8)$. Throughout the architecture, $\mathcal N_{\bullet}$ denotes this same normalised-interpolation map with an independently learned gate for the named stream. Given a pooled reading $p$, the global update reads,
\begin{equation}
  \Delta_g
    = \mathrm{FFN}_{g}\!\bigl(
        \mathrm{RMS}[W_{\mathrm{tap}}g,p]
      \bigr),
  \qquad
  \mathrm{Update}(g,p)
    = \mathcal N_g(g,\Delta_g).
  \label{eq:ngpt-update}
\end{equation}
The bottleneck $W_{\mathrm{tap}}$ balances the global and pooled entries before the joint RMS normalisation. Figure~\ref{fig:gladder} traces the stream's stations through the architecture.

\begin{figure}[!htbp]
  \centering
  \input{architecture-gladder.tikz}
  \caption{The global stream. From the featurizer's seed onward, $g$ lives on the unit-RMS sphere, and every station applies the same two maps: a descriptor pool (Eq.~\eqref{eq:descriptor-pool}), then the normalised interpolation of Eq.~\eqref{eq:ngpt-update}. The trunk refreshes $g$ once per block from the updated site states and closes with the factored row-and-column edge reading (Eq.~\eqref{eq:edge-rowcol}). The stream then forks: the physics leg and the router leg each refresh their own copy through a leaf-contextualizer stack and its closing edge reading (\S\,\ref{ssec:tree-readout}, \S\,\ref{ssec:routes}), and the copies never exchange information again. The physics copy is refreshed once per merge level and conditions the root readout; the router copy advances per decode position by reading the dyadic cover, and conditions the pointer.}
  \label{fig:gladder}
\end{figure}
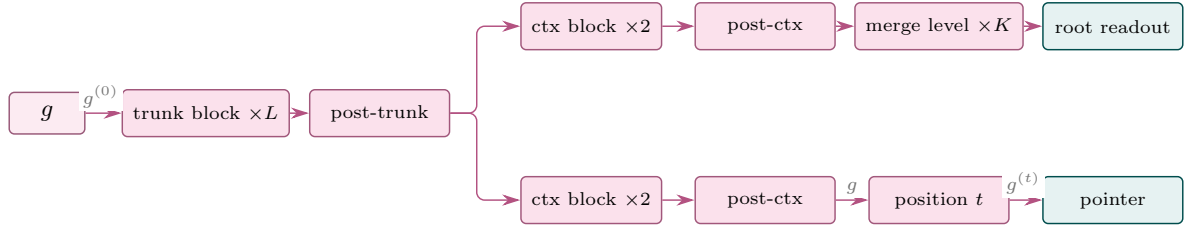

The stream is wide, $d_g = 1024$, and stays wide along the whole ladder; what varies is how the rest of the network reads it. The attention pools and the additive taps ($W_g$ above, and their analogues downstream) are learned projections by construction. The remaining consumers, the leaf gate, the merge context leg and the root readout hypernet, each read the stream through their own RMS-normalised linear projection to a narrower interface of width $256$; where a later display writes $g$ as an input to one of these three, it denotes that consumer's projected copy.

The attention itself is not the vanilla recipe, and both of its deviations recur across the architecture. First, the logits carry the bond stream:
\begin{equation}
  \mathrm{logit}_{\mathsf{h},\, ij}
  \;=\;
  \frac{q_{\mathsf{h}, i} \cdot k_{\mathsf{h}, j}}{\sqrt{d_{\mathsf{h}}}}
  +\;\mathrm{MLP}(\mathrm{RMS}(e^{\beta+1}_{ij}))
  \label{eq:trunk-edge-bias}
\end{equation}
with $\beta_{\mathsf{h}}: \mathbb{R}^{d_e} \to \mathbb{R}$ a small per-head scalar projector, so a head can attend along the interaction graph rather than by content alone. Second, the heads are gated in bisected pairs, the GLU micro-structure the featurizer's attention layers already carry (Eq.~\eqref{eq:col-glu}). The level attention of \S\,\ref{sssec:level-attn} and the contextualizer's slot attention \citep{locatello2020object} reuse it in turn. Figure~\ref{fig:trunk} draws one block.

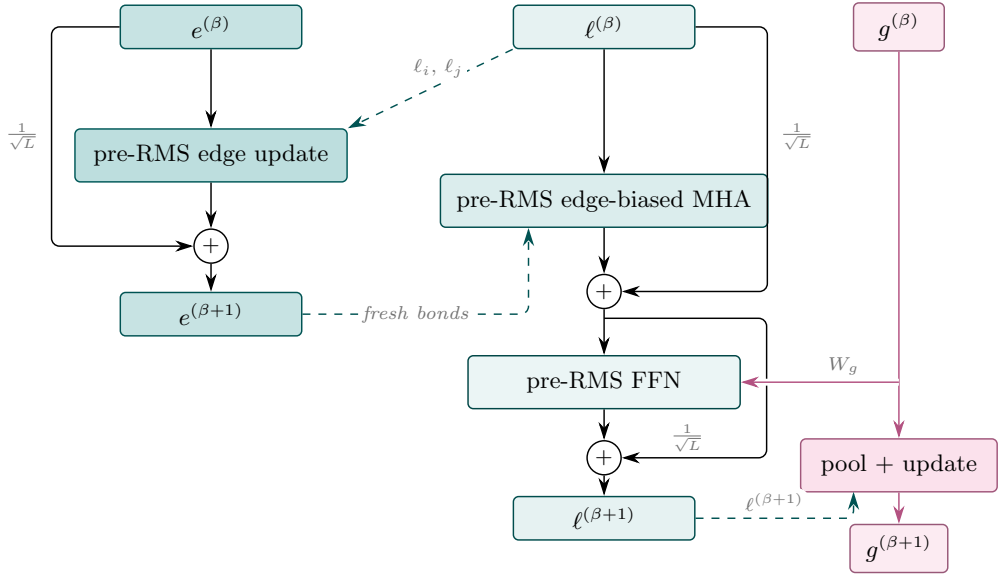
\begin{figure}[!htbp]
  \centering
  \input{architecture-trunk.tikz}
  \caption{One block of the trunk
    (\S\ref{ssec:trunk-leaf}); sublayers run in the order drawn. The per-bond stream $e^{(\beta)} \in \mathbb{R}^{N \times N \times d_e}$ (deeper teal) moves first through the pre-norm edge update of Eqs.~\eqref{eq:edge-update}--\eqref{eq:triangle-aggregate}, consuming the endpoint sites $\ell_i, \ell_j$. The freshly updated bonds $e^{(\beta+1)}$ then bias the per-site attention's logits (Eq.~\eqref{eq:trunk-edge-bias}), whose output is head-gated, and a pre-norm FFN closes the per-site path. Every sublayer normalises with RMS and carries the residual gain $1/\sqrt L$. The global stream $g^{(\beta)}$ (magenta) enters the node FFN through th additive projection $W_g\, g^{(\beta)}$ and is refreshed once per block from the updated site states (Eq.~\eqref{eq:trunk-g-refresh}). $L$ identical blocks are stacked.}
  \label{fig:trunk}
\end{figure}

The edge-update map itself combines a direct-pair contribution with a three-part aggregate that pools over every site $k$ acting as a third vertex. Concretely, its increment at bond $(i, j)$ is
\begin{equation}
  \mathrm{EdgeUpdate}\bigl(\mathrm{RMS}(e),\, \ell\bigr)_{ij}
  \;=\;
    \mathrm{FFN}\!\bigl(\bigl[\mathrm{pair}_{ij},\, P_{ij}\bigr]\bigr),
  \label{eq:edge-update}
\end{equation}
where $\mathrm{pair}_{ij} := [\,e_{ij}, \ell_i, \ell_j\,]$ is the direct-pair feature, i.e. the bond together with its two endpoints, and $P_{ij} \in \mathbb{R}^{d_c}$ is the three-part aggregate,
\begin{equation}
  P_{ij,\, c} \;=\;
  \frac{1}{\sqrt N}
  \sum_{k = 1}^{N}
    \Psi^{\mathrm{L}}_c(\mathrm{pair}_{ik})\,
    \Psi^{\mathrm{R}}_c(\mathrm{pair}_{kj}),
  \qquad c \in \{1, \dots, d_c\}.
  \label{eq:triangle-aggregate}
\end{equation}
The two embedders $\Psi^{\mathrm{L}}, \Psi^{\mathrm{R}}: \mathbb{R}^{d_e + 2 d_\ell} \to \mathbb{R}^{d_c}$ are independent SiLU MLPs that act on the left leg $i \to k$ and the right leg $k \to j$ of the triangle $(i, k, j)$. The $1/\sqrt N$ prefactor keeps the variance of $P_{ij}$ at $\mathcal{O}(1)$ as $N$ grows. Sites $k$ that do not complete a physical triangle with $(i, j)$ enter through legs carrying the absent-bond tokens of Eq.~\eqref{eq:absent-token}, which the embedders learn to discount; absence is a representable input here, not a hard zero. The update costs $\mathcal{O}(N^2 d_e d_c)$ flops per layer, which is the same asymptotic order as the pair-feature
attention already in the block.

This is the triangular update pattern of AlphaFold's Evoformer \citep{jumper2021highly}, adapted from residue-pair representations to coupling graphs. We opt for this three-part update because it is well suited to handle frustrated systems. On a triangular antiferromagnet, a kagome lattice, or any $J_1$--$J_2$ chain whose nearest- and next-nearest-neighbour bonds close into triangles, two bonds can share endpoint features and still sit in physically distinct environments. Thus the bond's effective behaviour is set by the \emph{triangle} of competing couplings it belongs to, not by its endpoints alone. Without $P_{ij}$, Eq.~\eqref{eq:edge-update} reduces to the pair-only form standard in message-passing GNNs, which collapses such bonds into a single representation. With $P_{ij}$, the FFN sees the full multiset $\{(\mathrm{pair}_{ik}, \mathrm{pair}_{kj}) : k\}$ of triangle completions and separates them.

In graph-theoretic language, the pair-only limit of Eq.~\eqref{eq:edge-update} is 1-Weisfeiler--Lehman in distinguishing power on the line graph, which is the level at which most off-the-shelf message-passing GNNs operate. The triangular aggregate $P_{ij}$ lifts the update strictly to the 2-WL refinement on pairs \cite{cai1992optimal, maron2019provably} which distinguishes any two bonds whose triangle structures differ.

After the $L$ blocks the global context stream takes one further update, this time reading the bonds rather than the sites. Placing all $N^2$ bond states under a single softmax would be wasteful, so the reading is factored through row and column descriptors,
\begin{equation}
  r_i = \mathrm{pool}\bigl(g,\, e^{(L)}_{i\,\cdot}\bigr),
  \qquad
  c_j = \mathrm{pool}\bigl(g,\, e^{(L)}_{\cdot\, j}\bigr),
  \qquad
  p = \mathrm{pool}\bigl(g,\, \{r_1, \dots, r_N,\,
      c_1, \dots, c_N\}\bigr),
  \label{eq:edge-rowcol}
\end{equation}
followed by the update $g \mapsto \mathrm{Update}(g, p)$ of Eq.~\eqref{eq:ngpt-update}. The row and column descriptors compress the $N^2$ bonds to $2N$ summaries, and the second-stage pool reads those; every pool here runs under the physical-site mask. The same factored reading recurs whenever the global reads an edge field downstream.

The trunk thus emits the triple $(e^{(L)}, \ell^{(L)}, g)$, where $g$ now denotes the stream's post-trunk value. Each of these is, by construction, a deterministic function of the featurizer's outputs $(e^{(0)}, \ell^{(0)}, g^{(0)})$, which are themselves deterministic functions of the cached inputs $J''$ and $h''$ of the raw Hamiltonian $(J, h)$.

Hence the whole pipeline up to here is a deterministic function of $(J, h)$ alone, and the spin configuration $q$ coordinates have not yet entered the model. The per-site $\ell^{(L)}$ feeds the leaf builder defined immediately below, refined on the way by a leaf contextualizer that the tree readout introduces (\S\,\ref{ssec:tree-readout}); the per-bond $e^{(L)}$ re-enters at the level-edge attention variant of the merge tree (\S\,\ref{ssec:tree-readout}); and the global stream continues through every later stage, refreshed at each in the same pool-and-update pattern.

\subsubsection{The per-site leaf builder}
\label{sssec:leaf-builder}

The leaf builder is the unique entry point for the spin configuration $q$: everything upstream of it is a function of the Hamiltonian alone, and everything downstream of it is a function of $q$ as well. It acts one site at a time. Given the contextualized per-slot state $\hat\ell_i \in \mathbb{R}^{d_\ell}$ --- the trunk's per-site outputs, refined over the readout's slot frame by the leaf contextualizer whose construction we give in \S\,\ref{sssec:merge-upgrade}, the global context $g \in \mathbb{R}^{256}$, the stream's refreshed post-context value read through the leaf's own projection (\S\,\ref{ssec:trunk-leaf}), and the quaternion-embedded spin $q_i \in \mathrm{SU}(2) \hookrightarrow \mathbb{R}^4$ with $\sum_a (q_i^a)^2 = 1$ (\S\,\ref{ssec:two-pictures}), it produces a per-site carrier
\begin{equation*}
  u_i^{(0)} \in \mathbb{R}^{d_u}
\end{equation*}
that encodes everything site $i$ contributes to the wavefunction amplitude. The $N$ carriers $\{u_i^{(0)}\}_{i=1}^N$ are the leaves of the tree readout of \S\,\ref{ssec:tree-readout}, which contracts them into the single scalar $\log \psi_\theta(q) \in \mathbb{C}$; the width $d_u$ is set by that readout.

Recall from the opening of \S\,\ref{ssec:trunk-leaf} that, because $q$ enters nowhere else, the per-site antisymmetry constraint \eqref{eq:per-site-oddness} becomes a property of this one map. The design problem is therefore: build $u_i^{(0)}$ as a function of $(\hat\ell_i, g, q_i)$ that is expressive in all three arguments yet \emph{exactly} odd under $q_i \to -q_i$. Our solution keeps the two parities on separate legs and couples them one way only, which keeps the necessary properties of the sector argument we made in \S\,\ref{ssec:casimir}, namely that the model is odd and linear in each quaternion $q_i$. Indeed, the $q$-leg is linear in $q_i$ because it is bias-free (since a constant would be even), and the other leg is an unconstrained function of $(\hat\ell_i, g)$. The two legs meet in a product, never in a sum, thus preserving the central oddness of the model. Figure~\ref{fig:leaf-builder} draws one site's unit.

\begin{figure}[!htbp]
  \centering
  \input{architecture-leaf-builder.tikz}
  \caption{The per-site leaf builder (\S\,\ref{sssec:leaf-builder}). Two input streams, one output carrier. The Hamiltonian context of per-slot $\hat\ell_i$ and global $g$ (teal) sets the gate $m_i = W_h[\hat\ell_i, g]$ of Eq.~\eqref{eq:even-hypernet}, which is $q$-independent by construction. The spin configuration $q_i$ (from MCMC) enters through the bias-free lift $z_i = W_{q \to z}\, q_i$ (orange) and meets the gate in the rank-$R$ Hadamard product of Eq.~\eqref{eq:leaf-builder} (violet), followed by the bias-free projection $U$ and the parity-odd scale normalisation of Eq.~\eqref{eq:leaf-rms}.  The RMS normalisation is nonlinear in $q_i$, but its divided-out scale is stored as $s_i^{(0)}=\log r_i^{(0)}$ and restored exactly at the root.  Consequently $u_i^{(0)}$ is odd, while the reconstructed carrier $e^{s_i^{(0)}}u_i^{(0)}=B_iq_i$ is exactly linear in $q_i$.}
  \label{fig:leaf-builder}
\end{figure}
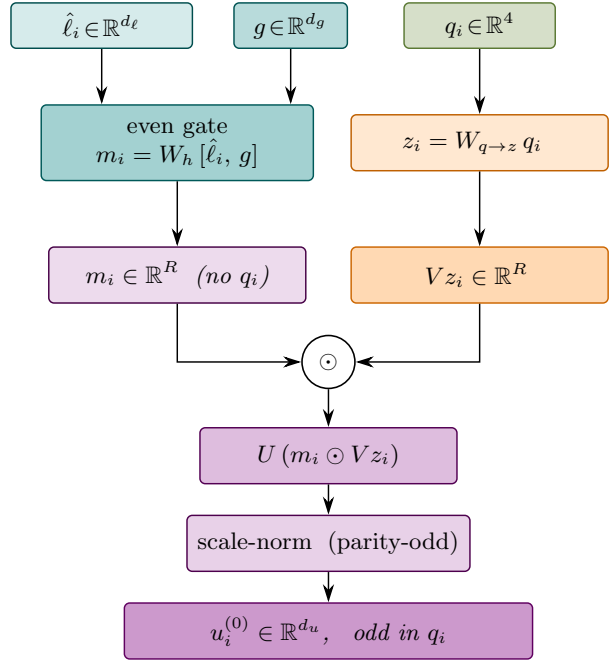

Let $\mathrm{BiasFreeLinear}(d_\mathrm{in}, d_\mathrm{out})$ denote a linear map with no additive bias. The $q$-leg is a single bias-free lift of the quaternion into what we call the \emph{odd stream},
\begin{equation}
  z_i \;:=\; W_{q \to z}\, q_i \;\in\; \mathbb{R}^{d_o},
  \qquad
  W_{q \to z} \in \mathrm{BiasFreeLinear}(4,\, d_o),
  \label{eq:q-to-z}
\end{equation}
with $d_o$ the odd-stream width. The absence of a bias is not a stylistic choice: an additive constant is even under $q_i \to -q_i$, and \eqref{eq:q-to-z} is the first link in a chain whose oddness must be exact --- this is the bias-free constraint of
\S\,\ref{ssec:casimir} in its simplest form.

The even leg maps the per-slot and global context to a gate vector,
\begin{equation}
  m_i \;:=\; W_h\,\bigl[\,\hat\ell_i,\, g\,\bigr]
  \;\in\; \mathbb{R}^{R},
  \label{eq:even-hypernet}
\end{equation}
where here $[\cdot\,,\cdot]$ denotes concatenation, $W_h \in \mathbb{R}^{R \times (d_\ell + 256)}$, and $R$ is a rank hyperparameter we return to in a moment. Since $\hat\ell_i$ and $g$ are functions of the Hamiltonian alone (\S\,\ref{ssec:trunk-leaf}), carrying no $q$-dependence whatsoever, the gate $m_i$ is $q$-independent by construction --- in particular invariant under $q_i \to -q_i$.

The two legs meet in a Hadamard (elementwise) product inside an $R$-dimensional bottleneck, followed by a bias-free output projection,
\begin{equation}
  \;
    u_i^{(0)}
    \;=\;
    U\bigl(\, m_i \odot V z_i \,\bigr)
    \;\in\; \mathbb{R}^{d_u},
  \;
  \qquad
  V \in \mathbb{R}^{R \times d_o},
  \quad
  U \in \mathbb{R}^{d_u \times R},
  \label{eq:leaf-builder}
\end{equation}
with $V$ and $U$ bias-free for the same parity reason as $W_{q \to z}$. Collecting the linear maps shows what
\eqref{eq:leaf-builder} really is:
\begin{equation}
  u_i^{(0)} \;=\; M(\hat\ell_i,\, g)\;\, z_i,
  \qquad
  M(\hat\ell_i,\, g)
  \;:=\;
  U\, \mathrm{diag}(m_i)\, V
  \;\in\; \mathbb{R}^{d_u \times d_o}.
  \label{eq:leaf-matrix-view}
\end{equation}
The even leg emits the \emph{weights} of a linear map of rank at most $R$, and the spin is read out through that map. In ML terms this is a small hypernetwork \newcitep{ha2017hypernetworks}, one network outputs the parameters of the map applied by another, with the Hadamard gate playing the role of FiLM-style multiplicative conditioning \newcitep{perez2018film} in a rank-$R$ bottleneck. In physics terms, the even sector dresses the odd sector without ever mixing into it. The same three-matrix shape $U(\,\cdot\, \odot V \cdot\,)$ reappears as the residual hypernet at every merge node of the tree (\S\,\ref{sssec:merge-upgrade}); the leaf builder is its first and simplest instance.

Finally, we normalise the carrier to unit root-mean-square while banking the divided-out magnitude in the log-scale stream. To that end, we may write
\begin{equation}
  B_i \;:=\; M(\hat\ell_i,g)\,W_{q\to z},
  \qquad
  u_{i,\mathrm{raw}}^{(0)} \;:=\; B_i q_i,
\end{equation}
to define
\begin{equation}
  r_i^{(0)}
  \;:=\;
  \sqrt{\overline{\bigl(u_{i,\mathrm{raw}}^{(0)}\bigr)^2}+\varepsilon},
  \qquad
  u_i^{(0)}
  \;:=\;
  \frac{u_{i,\mathrm{raw}}^{(0)}}{r_i^{(0)}},
  \qquad
  s_i^{(0)} \;:=\; \log r_i^{(0)},
  \label{eq:leaf-rms}
\end{equation}
where,
\begin{equation}
  \bigl(r_i^{(0)}\bigr)^2
  \;=\;
  \frac{1}{d_u}\,q_i^{\top}B_i^{\top}B_iq_i+\varepsilon,
  \qquad
  \overline{x^2}:=\frac{1}{d_u}\sum_a x_a^2.
  \label{eq:leaf-quadratic-scale}
\end{equation}
The normalised carrier $u_i^{(0)}$ is odd but, by itself, nonlinear in $q_i$ because its denominator is the square root of the quadratic form in Eq.~\eqref{eq:leaf-quadratic-scale}.  The pair $(u_i^{(0)},s_i^{(0)})$, however, preserves the raw linear carrier exactly,
\begin{equation}
  e^{s_i^{(0)}}u_i^{(0)}
  \;=\;
  u_{i,\mathrm{raw}}^{(0)}
  \;=\;
  B_iq_i.
  \label{eq:leaf-scale-cancellation}
\end{equation}
Thus no magnitude degree of freedom is deleted at the leaf. The normalisation is a numerical gauge choice whose scale is banked at the beginning of the merge, propagated additively through the tree, and restored at the root readout. This means the non-linearity introduced by normalising cancels with the global scale added through each step of the merge, and the overall model remains multi-linear in each of the manifold's quaternion coordinates.

Next we highlight the the reason for using a product here and not a standard architectural choice of a FFN. A generic FFN on $\mathrm{concat}(\hat\ell_i, g, z_i)$ would mix even and odd inputs nonlinearly, generate terms of even degree in $q_i$, and \eqref{eq:per-site-oddness} would hold only to the extent that training happened to learn it. The gated form \eqref{eq:leaf-matrix-view} is the simplest combiner that makes oddness \emph{structural}. The $q$-dependence stays linear in $z_i$, with coefficients the Hamiltonian context may set freely. Hence, the even side of the model is completely unconstrained when it comes to representational capacity.

Indeed, every map on the $q$-leg is linear and bias-free, so $z_i(-q_i) = -z_i(q_i)$, the matrix $M(\hat\ell_i, g)$ does not move when $q_i$ flips, and the normalisation \eqref{eq:leaf-rms} has an odd numerator and an even denominator. The composition is therefore odd,
\begin{equation}
  u_i^{(0)}(\hat\ell_i,\, g,\, -q_i)
  \;=\;
  -\, u_i^{(0)}(\hat\ell_i,\, g,\, q_i).
  \label{eq:leaf-odd}
\end{equation}
This parity requirement also propagates through the rest of the network. Each merge node of \S\,\ref{ssec:tree-readout} is linear in each of its two children separately, and each scale normalisation along the way flips its output whenever its input flips, so a sign flip at leaf $i$ propagates to a sign flip of the root carrier and changes nothing else. The complex readout at the root converts that flip into $\log \psi_\theta \mapsto \log \psi_\theta + i \pi$. That is, $\psi_\theta \mapsto -\psi_\theta$ with $|\psi_\theta|$ untouched in line with \eqref{eq:per-site-oddness}.

The distinction between the normalised carrier and the represented amplitude is essential.  The normalised carrier $u_i^{(0)}$ contains the even denominator of Eq.~\eqref{eq:leaf-quadratic-scale} and is therefore not a degree-one function of $q_i$.  The tree never discards that denominator however. Instead it carries $s_i^{(0)}=\log r_i^{(0)}$, adds every later merge scale, and restores the complete scale at the root.  Hence the represented leaf amplitude is $e^{s_i^{(0)}}u_i^{(0)}=B_iq_i$, a pure spin-$\tfrac12$ object.  Thus the scale normalisations do not introduce higher per-site Peter--Weyl harmonics into the final wavefunction, since their apparent nonlinear dependence cancels algebraically in the reconstructed amplitude.

\subsection{Neural tensor network readout from leaves}
\label{ssec:tree-readout}

The leaf builder from \S\,\ref{sssec:leaf-builder} produces, at every site $i \in \{1, \dots, N\}$, a carrier $u_i^{(0)} \in \mathbb{R}^{d_u}$ that is linear in its own site's $q_i$ and independent of every other site's. To evaluate a wavefunction, we need to collapse these $N$ carriers into a single complex scalar $\log \psi_\theta(q) \in \mathbb{C}$, and we need the collapse to stay linear in each site separately, as multilinearity is the function class the spin-$\tfrac{1}{2}$ sector demands, see also (\S\,\ref{sssec:leaf-builder}). A natural way to do so with multilinear pieces is a balanced binary tree of merges, which pairs the carriers up, merges each pair into a parent carrier through a bilinear map, then pairs the parents up, merges them in turn, and continue for $K := \log_2 \hat N$ levels until one root carrier survives, where $\hat N$ is the closest power of 2 larger than $N$. A linear readout on the root produces the scalar.

Indeed, the tensor-network community has used this concept in the contraction schedule of a tree tensor network \citep{shi2006classical}. The comparison is worth dwelling on for a moment because it reveals limitations of a binary-tree merge. Indeed, a contraction geometry is an entanglement prior, since the tree structure encodes in advance which sites can correlate cheaply (through a shallow common ancestor) and which only expensively (through many intermediate tensors), so choosing the geometry well classically requires knowing the correlation structure of the ground state one is trying to find. For a single Hamiltonian, we are able to make such a prior contraction geometry by hand, and both the Tensor Network and the Neural Quantum State literatures \citep{carleo2017solving, shi2006classical} have found that a carefully chosen architecture as a training prior can outperform a generic one. However, a foundation model spanning chains, frustrated lattices, disordered ensembles, and long-range couplings cannot commit to one geometry.

In this section, we detail how to circumvent the contraction flexibility issue with a single, shared tree architecture that can be cheaply reconditioned to fit the correlation structure of any Hamiltonian that is quadratic weight in the Pauli group.

One merge module is shared across the entire tree, and it is modulated by Hamiltonian context; the levels exchange information laterally through attention, and the assignment of physical sites to leaves is promoted in \S\,\ref{ssec:routes} to a learned, per-Hamiltonian contraction.

Every node of the tree carries a state of three objects, and a fourth object lives between the nodes of each level:
\begin{equation}
  \underbrace{c \in \mathbb{R}^{d_c}}_{\text{context (even)}},
  \qquad
  \underbrace{u \in \mathbb{R}^{d_u}}_{\text{carrier (odd)}},
  \qquad
  \underbrace{s \in \mathbb{R}}_{\text{log-scale}},
  \qquad
  \underbrace{E^{(\kappa)} \in
    \mathbb{R}^{N_\kappa \times N_\kappa \times d_{\mathrm{edge}}}}_{\text{level edge field}},
  \label{eq:tree-cast}
\end{equation}
where $N_\kappa := \hat N / 2^\kappa$ is the number of nodes at
level $\kappa \in \{0, \dots, K\}$. The context $c$ is the $q$-independent stream, and it tells a merge node what physical Hamiltonian it is contracting through its encoding in the  trunk, see \S\,\ref{sssec:merge-upgrade}.

The normalised carrier $u$ is the odd stream of \S\,\ref{sssec:leaf-builder}, together with the log-scale $s$ it represents the reconstructed carrier
\begin{equation}
  a \;:=\; e^s u.
  \label{eq:reconstructed-carrier}
\end{equation}
At an active leaf, $a_i^{(0)}=B_iq_i$ is exactly linear in its own quaternion.  The scalar $s$ starts at the leaf with $s_i^{(0)}=\log r_i^{(0)}$ and then accumulates every magnitude divided out by the $K$ levels of merge normalisation.  The individual $u$ stream is kept at unit RMS for numerical stability, and multilinearity is an invariant of the reconstructed stream $a=e^su$.

The edge field $E^{(\kappa)}$ is a coarse-grained descendant of the trunk's per-bond data $e^{(L)}$, with one feature vector per \emph{pair} of level-$\kappa$ subtrees.  It plays the role of an effective coupling between blocks, in the spirit of Wilson renormalisation \cite{wilson1983renormalization}.

From these objects in the tree, the readout happens in five steps, one per subsection below, with the routing of \S\ref{ssec:routes} having to its own dedicated section.  We start by defining the balanced frame of slots and padding holes, then state the bilinear merge and identify the two bottlenecks it imposes (\S\,\ref{sssec:balanced-merge}). The first bottleneck, a hidden locality prior on which leaves entangle cheaply, is loosened by the bit-reversed canonical frame (\S\,\ref{sssec:bitrev}) and then we fully remove any locality priors by the learned routes of \S\,\ref{ssec:routes}. The second bottleneck, limited per-merge expressivity, is fixed by widening the merge to a blockwise rank-4 contraction tensor with a context-conditioned residual gate, after which the leaf contextualizer prepares the per-slot contexts and the level-$0$ edge field that the tree consumes at its leaves (\S\,\ref{sssec:merge-upgrade}). Level attention with a tree-aware position bias then deletes the $\log_2 \hat N$ depth penalty on long-range pair correlations (\S\,\ref{sssec:level-attn}), and the root readout closes the construction and states the structural invariant that the energy kernels of SM~\S~\ref{sec:part4} exploit (\S\,\ref{sssec:root-readout}), see Figure~\ref{fig:tree-readout}.

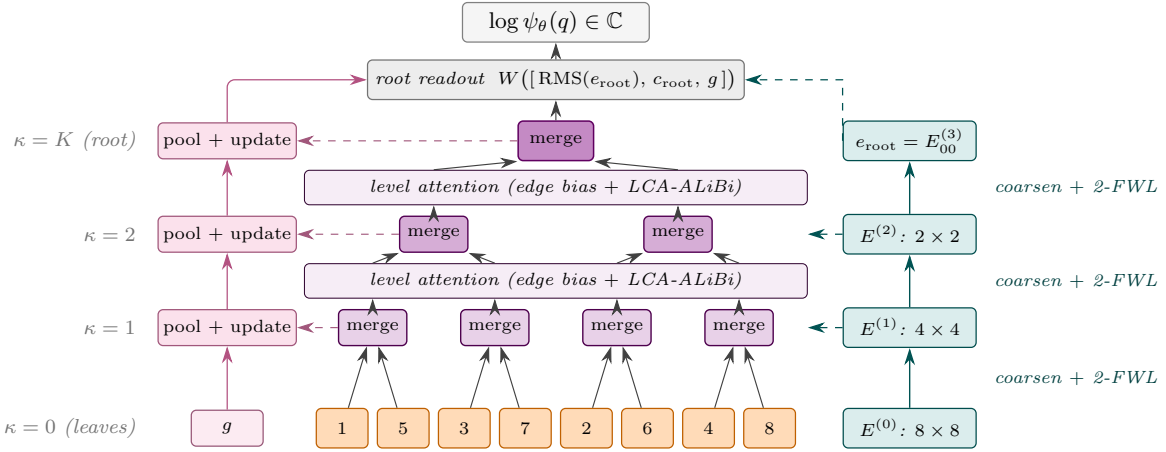
\begin{figure}[!htbp]
  \centering
  \resizebox{\textwidth}{!}{\input{architecture-tree.tikz}}
  \caption{The tree readout at $\hat N = 8$ (\S\,\ref{ssec:tree-readout}). Orange slots carry the per-site carriers $u_i^{(0)} \in \mathbb{R}^{d_u}$ in the bit-reversed canonical frame; per-slot contexts and the level-$0$ edge field come from the leaf contextualizer (\S\,\ref{sssec:merge-upgrade}). Violet merge nodes apply the shared three-stream merge of Figure~\ref{fig:merge-node} to the node states $(c, u, s)$. Lavender bands between merge levels are the level-attention insert of Eq.~\eqref{eq:level-attn-update} with the edge-biased LCA-ALiBi logits of Eq.~\eqref{eq:level-edge-logit}. The teal lane on the right is the edge field, coarsened level by level (Eq.~\eqref{eq:edge-coarsen}, with 2-FWL refinement); it biases the bands and, at the top, conditions the root readout of Eq.~\eqref{eq:root-readout}. The magenta lane on the left is the global stream, refreshed once per level from the freshly merged contexts; its final value enters the root readout's hypernet alongside $e_{\mathrm{root}}$ and $c_{\mathrm{root}}$ (Eq.~\eqref{eq:root-hypernet-input}).}
  \label{fig:tree-readout}
\end{figure}

\subsubsection{The balanced frame: slots, holes, and the bilinear merge}
\label{sssec:balanced-merge}

Let $K$ be the smallest non-negative integer with $2^K \ge N$, and let $\hat N := 2^K$, which we refer to as the $\hat N$ \emph{slots} for leaf positions. The \emph{physical mask} $m \in \{0, 1\}^{\hat N}$ has $m_t = 1$ on the $N$ slots occupied by physical sites and $m_t = 0$ on the $\hat N - N$ padded slots, which we call \emph{holes}. Holes carry zero couplings ($J_{tu} = 0$ whenever $m_t m_u = 0$, $h_t = 0$) and the physical mask gates every carrier:
\begin{equation}
  u_t^{(0)} \;\mapsto\; m_t\, u_t^{(0)},
  \label{eq:carrier-mask}
\end{equation}
so a hole can never inject amplitude into the wavefunction.

The carrier mask does not send a zero hole through an ordinary bilinear contraction however. For a parent node $P$ with child subtrees $A$ and $B$, let $n_X$ be the physical-leaf count of subtree $X$. The odd carrier merge is the masked pass-through
\begin{equation}
  \mathcal M_P(u_A,u_B)
  :=
  \begin{cases}
    M_PT[u_A,u_B], & n_A>0,\ n_B>0,\\
    u_B,           & n_A=0,\ n_B>0,\\
    u_A,           & n_A>0,\ n_B=0,\\
    0,             & n_A=n_B=0,
  \end{cases}
  \label{eq:masked-carrier-merge}
\end{equation}
Here $T$ is the shared bilinear carrier contraction introduced in Eq.~\eqref{eq:trilinear-merge} and implemented blockwise in Eq.~\eqref{eq:quadrilinear-merge}, and $M_P$ is the parent-specific, context-dependent linear refinement defined in Eq.~\eqref{eq:merge-linear-map}. The log-scale is passed through with the live child in the one-live-child cases. Hole contexts and edge features still follow the even merge, so a hole can condition later gates without injecting an odd amplitude. A hole therefore does not annihilate a nonempty sibling subtree.

Holes are \emph{not} erased from the even side either. The tree treats them as structural sites in their own right. A hole slot carries a context like any other slot, with two holes still merging their contexts, and the tree is always the perfect binary tree over $\hat N$ slots. This buys two things. A system's tree is the same regardless of how far the batch later pads it, so evaluation is independent of batch composition, and a single foundation model can therefore serve \emph{many} sizes of physical system with this mechanism. And since the network learns that a subtree contains holes only through the contexts prepared for it, where the holes sit becomes a structural decision the routing of \S\,\ref{ssec:routes} can learn, not a convention fixed in advance. The implementation accordingly keeps two masks: the physical mask above for everything that touches amplitude, and a structural typing (a flag to mark the context features as coming from a hole) for everything that does not. These design choices are what allows the model to condition properly on arbitrary interaction topologies, including any topological defects in an otherwise idealised lattice tensor with block-diagonal $J$ tensor.

Before the trunk's outputs can be loaded into this frame, it is worth noting that the frame already comes equipped with a natural notion of distance between slots. Every piece of machinery in this section measures position with it. For two distinct slots $a, b \in \{0, \dots, \hat N - 1\}$, let $\mathrm{lca}(a, b)$ denote the depth of their lowest common ancestor, counted from the root at depth $0$. Slots $a$ and $b$ lie in the same level-$\kappa$ subtree exactly when their binary expansions agree on all bits above position $\kappa$, since the subtrees of the balanced tree \emph{are} the dyadic blocks of the slot index. The level at which two slots first meet is therefore set by their highest differing bit, which gives the common-ancestor depth a closed form via bitwise XOR,
\begin{equation}
  \mathrm{lca}(a, b)
  \;=\;
  K - 1 - \lfloor \log_2 (a \oplus b) \rfloor,
  \qquad a \neq b.
  \label{eq:lca-bitrev}
\end{equation}
For example, at $K = 3$, slots $0$ and $1$ have $a \oplus b = 1$, so $\mathrm{lca} = 2$ and they are siblings, while slots $3$ and $4$ have $a \oplus b = 7$, so $\mathrm{lca} = 0$, and they only meet at the root, despite being index-adjacent. The example is worth pausing on: the slot metric is \emph{dyadic} rather than translation-invariant. What matters is therefore the largest aligned block two slots share, counter to our intuition of how far apart their indices sit. From the common-ancestor depth we define the \emph{tree distance},
\begin{equation}
  d_{\text{tree}}(a, b)
  \;:=\;
  2\bigl(K - \mathrm{lca}(a, b)\bigr),
  \label{eq:tree-distance}
\end{equation}
which counts edges from $a$ up to the common ancestor and back down to $b$. Computationally, $\mathrm{lca}(a, b)$ is one XOR plus one bit-scan per pair, cheaper than any learned position table and trivially fused into whatever kernel consumes it. And \eqref{eq:lca-bitrev} is a statement about slots, not about sites, so nothing that reassigns sites to slots,  including the learned routes of \S\,\ref{ssec:routes}, touches it.

A merge node at level $\kappa$ takes two children, $(c_A, u_A, s_A)$ and $(c_B, u_B, s_B)$, and produces a parent, on all three per-node objects, $(c_P, u_P, s_P)$, of \eqref{eq:tree-cast}. The carrier half is the part that must stay multilinear, and its simplest form is one bilinear contraction through a learnable merge tensor,
\begin{equation}
  u_P^{\alpha}
  \;=\;
  T\bigl[u_A, u_B\bigr]^{\alpha}
  \;=\;
  \sum_{\beta, \gamma = 1}^{d_u}
    T_{\alpha \beta \gamma}\, u_A^{\beta}\, u_B^{\gamma},
  \qquad
  \alpha \in \{1, \dots, d_u\},
  \label{eq:trilinear-merge}
\end{equation}
which is linear in each child separately. Hence a sign flip can pass through (\S\,\ref{sssec:leaf-builder}), and it is also bilinear in the pair. The
context meanwhile has no parity constraint and hence
we merge bottom-up, computing a candidate update
\begin{equation}
  \Delta_P
  \;=\;
  \mathrm{FFN}_c\!\bigl(\bigl[\,c_A,\, c_B,\,
    E^{(\kappa)}_{AB},\, E^{(\kappa)}_{BA},\,
    \tau_P,\, \Phi_P;\, g\,\bigr]\bigr)
  \label{eq:context-leg}
\end{equation}
and folding it into the averaged children by the normalised interpolation of Eq.~\eqref{eq:ngpt-update},
\begin{equation}
  \bar c
    = \mathrm{Norm}\!\bigl(\tfrac12(c_A+c_B)\bigr),
  \qquad
  c_P
    = \mathcal N_c(\bar c,\Delta_P),
  \qquad
  \widetilde\alpha_c
    = 0.8\,\operatorname{sigmoid}(a_c),
  \quad
  a_c\big|_{\mathrm{init}}
    = \operatorname{logit}\!\left(\frac{0.2}{0.8}\right)\mathbf 1.
  \label{eq:context-interp}
\end{equation}
Here $[\,\cdot\,]$ is concatenation as in \S\,\ref{sssec:leaf-builder}. $E^{(\kappa)}_{AB}$ and $E^{(\kappa)}_{BA}$ are the two directed edge-field entries between exactly the two subtrees being merged via the renormalised coupling of \eqref{eq:tree-cast}, and encode the interaction this contraction is about to absorb, and $\tau_P$ is a fixed sinusoidal encoding \citep{vaswani2017attention} of the node's level and of its dyadic position counted from the root, which we can think of as a \emph{root-centred clock}. Centering the clock at the root rather than the leaves means trees of different depth agree near the top, which is one of the small choices that lets a model trained at one $\hat N$ evaluate at a larger one. The retraction guarantees a context of bounded magnitude at every level, and $\widetilde\alpha_c$ sets the bounded per-channel interpolation from the averaged children toward the update. The same rule, with its own gate, reappears at the router's tree-prefix merge (\S\,\ref{ssec:routes}).

The input $\Phi_P \in \mathbb{R}^{32}$ is a sinusoidal encoding of the merge's counts and of its level. With $n_A$ and $n_B$ the physical-leaf counts of the two children, $n_{\mathrm{rem}}$ the physical sites the subtree under $P$ has not yet absorbed, $\kappa$ the level and $K$ the tree depth,
\begin{equation}
  \begin{aligned}
  \Phi_P = \Bigl[\,&
    \sin(\omega_f\, \lambda),\; \cos(\omega_f\, \lambda)
    \;\Bigm|\; \\
    &
    \lambda \in \bigl\{\log_2(1{+}n_A),\, \log_2(1{+}n_B),\,
      \log_2(1{+}n_{\mathrm{rem}})\bigr\}, \\
    &
    \omega_f = \pi/2^{f},\; f = 0,\dots,3
  \,\Bigr].
  \end{aligned}
  \label{eq:depth-count-features}
\end{equation}
together with the same sine--cosine encoding of the pair $(\kappa,\, K{-}1{-}\kappa)$ at $f = 0, 1$, for $24 + 8 = 32$ channels. The counts are deliberately ordered: $n_A$ enters before $n_B$ because the reduction path breaks the child-swap symmetry. They are also absolute rather than fractional, so identical physical content encodes identically across the sampled tree topologies of \S\,\ref{ssec:routes}. The level pair reports the node's distance from the root and from the leaves.

The global $g$ in Eq.~\eqref{eq:context-leg} is the stream's value at the level being merged, read through the tree's projection (\S\,\ref{ssec:trunk-leaf}): after each level's merges, the tree refreshes the stream once over the contexts they produced, $g \mapsto \mathrm{Update}\bigl(g, \mathrm{pool}(g, \{c_P^{(\kappa)}\})\bigr)$, pooling under the structural mask, so virtual siblings participate, and re-projects it for the next level. The final value feeds the root readout of \S\,\ref{sssec:root-readout}. We defer the log-scale half $s_P$ to \S\,\ref{sssec:merge-upgrade}, where its purpose becomes visible.

Using \eqref{eq:context-leg} allows us to separate top-down conditioning pathway from the trunk of \S \ref{ssec:trunk-leaf} into the tree itself. The contexts are \emph{computed up the tree} alongside the carriers, rooted in per-slot seeds prepared from the trunk's outputs. This is the subject of \S\,\ref{sssec:merge-upgrade}, which we visit before understanding two necessary constraints that a bilinear cascade imposes in terms of expressivity of the model and a locality prior.

Before that, we clarify two limitations of the bare bilinear cascade, since the next three subsections and the design choices that led to them aim to address them. \emph{First}, the slot-to-site map is a hidden locality prior. Any two leaves at slots $a$ and $b$ can only exchange information through their lowest common ancestor in the tree, and the depth of that ancestor depends entirely on the chosen map. To couple two physical sites through a short contraction path, the map must place them in slots with a shallow common ancestor. \emph{Second}, the bilinear contraction \eqref{eq:trilinear-merge} is either too expensive or too narrow in practise. Materialising $T$ at the full carrier width costs $d_u^3$ parameters and flops per merge, giving over $3 \times 10^9$ at our reference width $d_u = 1536$, which is infeasible on every merge node of the tree. On the other hand, shrinking $d_u$ until the merge is cheap enough introduces a bottleneck that limits the model's capacity to pass information. And at any width, the merge at the common ancestor remains the only channel between two leaves, $\mathrm{lca}(a,b)$ levels up. \S\,\ref{sssec:bitrev} loosens the locality prior by re-indexing the frame; \S\,\ref{sssec:merge-upgrade} fixes the width; \S\,\ref{sssec:level-attn} removes the depth penalty.

\subsubsection{The bit-reversed canonical frame}
\label{sssec:bitrev}

The slot metric determines the path lengths induced by the tree, while the site-to-slot assignment determines which physical pairs receive those path lengths. The naive assignment of site $i$ at slot $i - 1$ to start at zero, inherits the dyadic metric of \eqref{eq:lca-bitrev} directly. This is because site pairs meet early when their zero-based indices agree on high bits, so chain neighbours that straddle an aligned block boundary (e.g. sites $4$ and $5$, sitting at the slots $3$ and $4$ of the worked example in \S\,\ref{sssec:balanced-merge}) are maximally separated, while neighbours inside a block are siblings.

The prior that this encodes is therefore biasing towards aligned dyadic blocks of the site index, which is at best a rough proxy for chain locality, and meaningless for frustrated, disordered, or long-range systems.

To get around this, we opt for a different canonical frame, which  instead places site $i$ at slot $\mathrm{brev}_K(i - 1)$, where $\mathrm{brev}_K$ is the $K$-bit reversal. To that end, let $b_r(t) \in \{0, 1\}$ denote the bit of $t$ at position $r \in \{0, \dots, K - 1\}$ (least-significant at $r = 0$), so that $t = \sum_{r = 0}^{K - 1} b_r(t) \cdot 2^r$; the reversal writes those bits in the opposite order:
\begin{equation}
  \mathrm{brev}_K(t)
  \;:=\;
  \sum_{r = 0}^{K - 1} b_r(t) \cdot 2^{K - 1 - r}.
  \label{eq:brev-def}
\end{equation}
At $K = 3$ for example, this gives the map,
\begin{equation}
  0 \mapsto 0,\;\;
  1 \mapsto 4,\;\;
  2 \mapsto 2,\;\;
  3 \mapsto 6,\;\;
  4 \mapsto 1,\;\;
  5 \mapsto 5,\;\;
  6 \mapsto 3,\;\;
  7 \mapsto 7.
\end{equation}
Figure~\ref{fig:bitrev} compares the two placements at $K = 3$ on the same example pair.
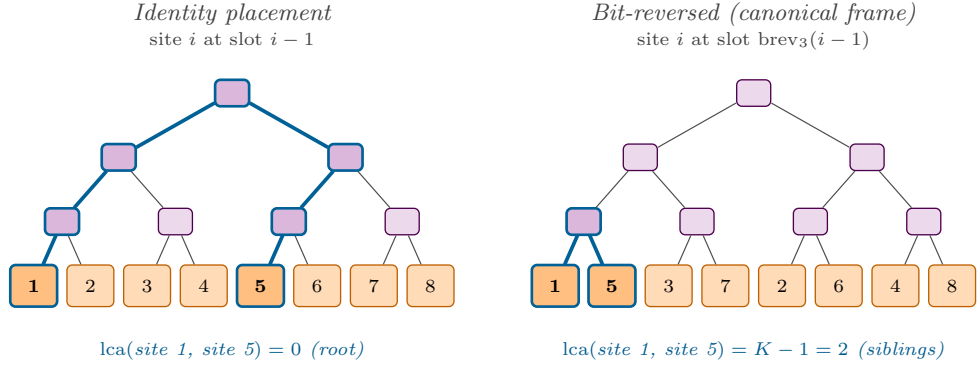
\begin{figure}[!htbp]
  \centering
  \input{architecture-bitrev.tikz}
  \caption{Identity placement (left) versus the bit-reversed canonical frame (right) at $K = 3$ ($\hat N = 8$). Orange squares are slots carrying physical site indices; violet squares are merge nodes. The bold path traces the common-ancestor route for one representative site pair. Under the identity placement, sites $1$ and $5$ (index distance $4$) only meet at the root; in the bit-reversed canonical frame, the same sites are siblings. The slot metric \eqref{eq:lca-bitrev} is identical in both panels; only the occupancy changes.}
  \label{fig:bitrev}
\end{figure}

Because $\mathrm{brev}_K(x) \oplus \mathrm{brev}_K(y) = \mathrm{brev}_K(x \oplus y)$, reversal swaps the roles of high and low bits in \eqref{eq:lca-bitrev}. Hence, under the bit-reversed placement, two sites meet at the level set by the \emph{lowest} differing bit of their indices. Two sites whose indices differ by $1$ (say sites $1$ and $2$) land at slots $0$ and $4$, on opposite halves of the tree; two sites whose indices differ by $\hat N / 2 = 4$ (say sites $1$ and $5$) land at slots $0$ and $1$, which are siblings. The canonical frame is therefore deliberately \emph{anti-local} in the site index, in the sense that long aligned strides merge first, and index neighbours meet last.

The learned routes of \S\,\ref{ssec:routes} condition on $(J,h)$ and can recover a contiguous layout when favoured by the couplings, up to automorphisms of the interaction graph. The second is bookkeeping. The packing is arranged so that a system's $N$ sites plus its $\hat N - N$ holes occupy the contiguous slot block $[0, \hat N)$, bit-reversed \emph{within} that block, however far the batch later pads beyond $\hat N$.

This is what makes a system's tree the same object at every batch shape, so that evaluation across mixed sizes and size extrapolations beyond the pre-training set are well-defined.

\subsubsection{Rank-4 quadrilinear merge and the log-scale chain}
\label{sssec:merge-upgrade}

We now fix the width problem  that comes with a naive bi-linear merge (see \S\,\ref{sssec:balanced-merge}). The full-width merge tensor costs $d_u^3$ per merge, and we seek the $d_u$-wide carrier. The resolution is to split the carrier into \emph{sub-blocks}, contiguous slices of the carrier indexed by a counter $i \in \{1, \dots, B\}$. Let
\begin{equation}
  B \;:=\; d_u \,/\, d_r
\end{equation}
denote the sub-block count (config-time constraint: $d_r$ divides $d_u$), and for each carrier $u \in \mathbb{R}^{d_u}$ define the 2-D reshape
\begin{equation}
  u^{(2\mathrm{D})} \in \mathbb{R}^{B \times d_r},
  \qquad
  u^{(2\mathrm{D})}_{i, k} \;:=\; u_{(i - 1)\, d_r + k},
\end{equation}
with first index $i \in \{1, \dots, B\}$ enumerating sub-blocks and second index $k \in \{1, \dots, d_r\}$ enumerating channels within a sub-block. The merge tensor can then keep a narrow contraction width $d_r$ but increases in expressivity thanks to the sub-block axis, $T \in \mathbb{R}^{B \times d_r \times d_r \times d_r}$. Contracting it with the two reshaped children produces the bilinear output, which we write $\tilde u$; the rank-4 quadrilinear contraction then reads,
\begin{equation}
  \tilde u_{i, j}
  \;=\;
  \sum_{k = 1}^{d_r} \sum_{l = 1}^{d_r}
    T_{i, j, k, l}\,
    u_A^{(2\mathrm{D})}{}_{i, k}\,
    u_B^{(2\mathrm{D})}{}_{i, l},
  \qquad i \in \{1, \dots, B\},\;
  j \in \{1, \dots, d_r\},
  \label{eq:quadrilinear-merge}
\end{equation}
flattened back to $\tilde u \in \mathbb{R}^{d_u}$ for the next contraction level.

Each sub-block $i$ runs an independent rank-3 trilinear sub-merge on its own slice of the children via its own slab $T_{i, :, :, :}$ of the merge tensor, and the parent carrier is the concatenation of the $B$ sub-block outputs. The cost arithmetic is what makes this work. With $B \cdot d_r^3$ parameters instead of $d_u^3$, at the reference configuration ($d_r = 32$, $B = 48$, $d_u = 1536$) that is $\sim\!10^6$ instead of $\sim\!3 \times 10^9$, cheap enough that $T$ is stored as a \emph{learnable tensor} during pre-training.

We emphasise this because it marks a deliberate design choice; the merge tensor itself is \emph{not} produced by a hypernet and does not depend on the context. Instead one $T$, shared across every level and every position of the tree, contracts pairs of leaves to the final pair below the root. Hamiltonian dependence enters the carrier path only through the multiplicative gate we add next. Sharing one merge across the tree is the TTN analogue of weight tying \citep{shi2006classical}, and it is a key step for generalisation. The parameter count is independent of $\hat N$, and a deeper tree at evaluation time reuses the same merge it trained at every level. The sub-block axis $i$ is also not a batch dimension to sum over. \footnote{This is because every sub-block keeps its own bilinear sub-merge with its own output slice, and the natural-gradient geometry of our optimizer (see \S\,\ref{ssec:kfac}) \citep{martens2015optimizing} treats the sub-block axis as part of a genuine Kronecker factorisation of the Fisher on $T$ (a balanced two-factor matricisation over the four axes $(i, j, k, l)$) rather than folding it into the batch.}

A context-conditioned residual gate augments the quadrilinear merge without breaking parity or multilinearity. The bilinear output $\tilde u$ of Eq.~\eqref{eq:quadrilinear-merge} can be refined by a context-conditioned residual hypernet, without projecting out of the carrier width. Let $R \le d_u$ be a hypernet bottleneck, and let $V \in \mathbb{R}^{R \times d_u}$, $U \in \mathbb{R}^{d_u \times R}$, $W_h \in \mathrm{BiasFreeLinear}(d_c, R)$ be three learned linear maps. The hypernet output reads
\begin{equation}
  H(\tilde u, c_P)
  \;:=\;
  U \bigl(\, W_h(c_P) \odot V \tilde u \,\bigr),
  \label{eq:merge-hypernet}
\end{equation}
with $\odot$ the elementwise (Hadamard) product. This is the same three-matrix gate as the leaf builder's Eq.~\eqref{eq:leaf-builder}, meaning it still preserves parity and multilinearity because the even sector modulates the odd sector multiplicatively. The map is dimension-preserving ($\tilde u \mapsto H(\tilde u, c_P) \in \mathbb{R}^{d_u}$) and linear in $\tilde u$ with $q$-independent coefficients, so the parent carrier below stays bilinear in $(u_A, u_B)$. Hence the gate refines the quadrilinear merge with a parity-preserving signal about the context.

For numerical stability, the output is renormalised to unit root-mean-square, with the divided-out magnitude banked in the log-scale stream ,
\begin{equation}
  u_P \;=\; \frac{u_{\mathrm{out}}}{\tilde s},
  \qquad
  \tilde s \;:=\; \sqrt{\overline{u_{\mathrm{out}}^2} + \varepsilon},
  \qquad
  s_P \;=\; s_A + s_B + \log \tilde s,
  \label{eq:merge-scalenorm}
\end{equation}
where $u_{\mathrm{out}} := \tilde u + H(\tilde u, c_P)$, $\overline{x^2} := \tfrac{1}{d_u} \sum_a x_a^2$ as in Eq.~\eqref{eq:leaf-rms}, and $\varepsilon > 0$ is a small numerical-stability constant, see Fig.~\ref{fig:merge-node}.

With the level-zero initialisation of Eq.~\eqref{eq:leaf-rms}, define the reconstructed carrier $a_X:=e^{s_X}u_X$ and the context-dependent linear map
\begin{equation}
  M_P
    := I_{d_u}+U\,\operatorname{diag}\!\bigl(W_h(c_P)\bigr)V.
  \label{eq:merge-linear-map}
\end{equation}
Since $H(\widetilde u,c_P)=U\operatorname{diag}(W_h(c_P))V\widetilde u$, one merge satisfies
\begin{equation}
  a_P=M_P\,T[a_A,a_B].
  \label{eq:reconstructed-merge}
\end{equation}
The coefficients of $M_P$ are independent of $q$, so the map preserves multilinearity while the recorded log-scale cancels the numerical normalisation exactly.

The factor $\tilde s$ and every $s_X$ are even under a single-site sign flip, so the normalised carriers remain odd, and the reconstructed carriers remain multilinear. The accumulated scale is reattached at the root readout (\S\,\ref{sssec:root-readout}). Figure~\ref{fig:merge-node} shows one merge node, with all three streams.

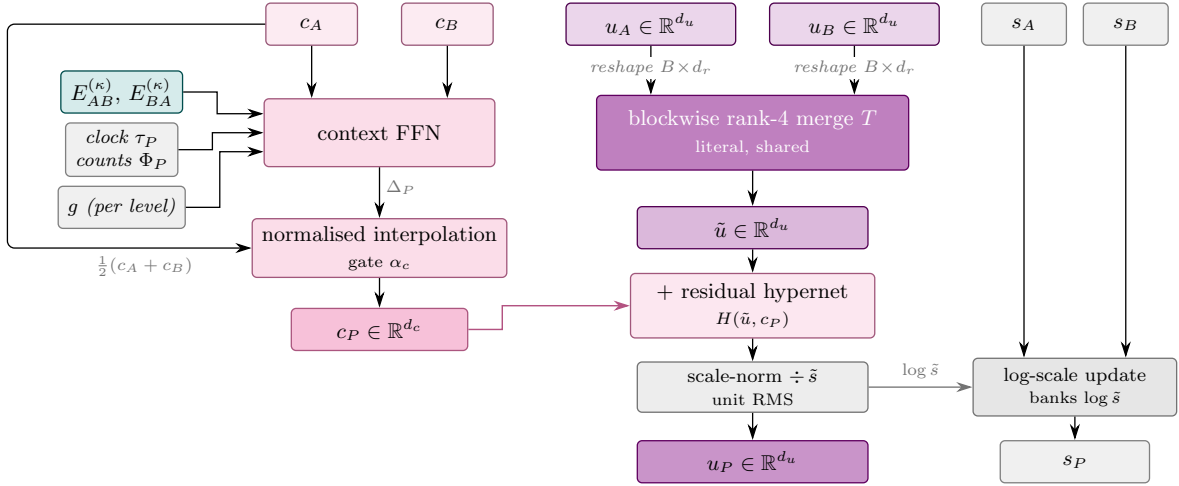
\begin{figure}[!htbp]
  \centering
  \resizebox{\textwidth}{!}{\input{architecture-merge-node.tikz}}
  \caption{One merge node, blown up (\S\,\ref{sssec:merge-upgrade}); the same module is applied at every level and position of the tree. \emph{Left (magenta)}: the even context leg, Eqs.~\eqref{eq:context-leg}--\eqref{eq:context-interp}: an FFN over both child contexts, the two directed sibling entries of the level edge field (teal), the root-centred dyadic clock $\tau_P$, the count features $\Phi_P$ and the per-level-refreshed global $g$, folded into the averaged children by the normalised interpolation with gate $\widetilde\alpha_c$. \emph{Centre (violet)}: the odd carrier leg, children reshaped to $B \times d_r$, contracted by the literal shared merge tensor $T$ (Eq.~\eqref{eq:quadrilinear-merge}), gated by the residual hypernet $H(\tilde u, c_P)$ (Eq.~\eqref{eq:merge-hypernet}; the only place context touches the carrier), and renormalised to unit RMS. \emph{Right (grey)}: the log-scale leg banks $\log \tilde s$, Eq.~\eqref{eq:merge-scalenorm}, so magnitudes accumulate additively instead of multiplicatively.}
  \label{fig:merge-node}
\end{figure}

We now have a mechanism to merge a system of $N$ spin-$1/2$ coordinates that formulates our manifold function's map to a single complex scalar. Our map uses one merge tensor $T$ shared across the tree, and the same three-stream module at every level and position of the binary tree. It also has a conditioning pathway through the contexts, which are computed up the tree and can therefore carry information about the Hamiltonian and the system's structure.

With the merge machinery in hand, what remains is to feed it context from the trunk that carries the $(J,h)$ signal. This corresponds to formulating the bottom row of Figure~\ref{fig:tree-readout} from the trunk's outputs in Figure~\ref{fig:trunk}. Notice the frames of the trunk and readout mechanism are different. The trunk's outputs $(\ell^{(L)}, e^{(L)}, g)$ are indexed by physical sites and know nothing about slots. Meanwhile the holes have no trunk data at all, yet we know from \S\,\ref{sssec:balanced-merge} that holes must carry contexts of their own in order for the model to carry meaningful information across different sized systems. So before the first merge runs, every slot, hole or not, needs a structural description written for it that is a function of the trunk's output data. We construct these descriptions with a leaf-contextualizer module.

As a module, the contextualizer is a single-pass bridge between the two frames,
\begin{equation}
  \texttt{LeafContextualizer}: \;
  \bigl(\,\ell^{(L)},\; e^{(L)},\; g\,\bigr)
  \;\longmapsto\;
  \bigl(\,
    \hat\ell \in \mathbb{R}^{\hat N \times d_\ell},\;\;
    E^{(0)} \in \mathbb{R}^{\hat N \times \hat N \times d_{\mathrm{edge}}},\;\;
    g
  \,\bigr),
  \label{eq:ctxer-contract}
\end{equation}
realised as two blocks in sequence. Two copies of this stack exist, the physics leg described here and the router leg of \S\,\ref{ssec:routes}; each forks its own global from the same post-trunk value, and the two copies never exchange information afterwards. The per-slot output does double duty: $\hat\ell$ is the state that gates the leaf builder of \S\,\ref{sssec:leaf-builder}, and, projected once where $d_c \neq d_\ell$, it seeds the level-$0$ context array $c^{(0)}$ of the merge recursion. Each block makes four updates on the running triple $(c, E, g)$: a context refresh, then a refresh of the global that reads the new contexts, then an edge rewrite that uses both, then a slot attention that reads the rewritten edges. The trunk seeds the state directly in the slot frame, with $c_t = \ell^{(L)}_t$ and $E_{t t'} = e^{(L)}_{t t'}$. A hole inherits whatever the masked trunk left in its rows, and the first update discards exactly this. Two further inputs appear in every step. The physical mask $m_t$ of \S\,\ref{sssec:balanced-merge} decides who keeps their seed. A fixed sinusoidal clock $\tau_t \in \mathbb{R}^{d_\ell}$ on the slot index gives the module its only notion of slot identity, every learned map here is indifferent to slot order, so position enters only through the clock and the tree-distance weights and biases below.

\emph{Step 1 (context refresh --- tree-weighted bond summary).}\;
Every slot first summarises its incident bonds, assigning larger weights to partners that it merges with earlier. Define the normalised tree-decay weights
\begin{equation}
  \omega_{t t'}
  \;=\;
  \frac{w_{t t'}}{\sum_{t'' \neq t} w_{t t''}},
  \qquad
  w_{t t'} \;=\; e^{-\frac{1}{2}\bigl(K - \mathrm{lca}(t, t')\bigr)},
  \label{eq:ctxer-weights}
\end{equation}
recalling that $K - \mathrm{lca}(t, t')$ is the level at which $t$ and $t'$ merge, so siblings ($K - \mathrm{lca} = 1$) dominate while opposite halves of the tree ($K - \mathrm{lca} = K$) barely register. The bond summary and the context refresh are then
\begin{align}
  \bar e_t
  \;&=\;
  \sum_{t' \neq t} \omega_{t t'}\;
    \Psi^{\mathrm{s}}\!\bigl(E_{t t'},\, E_{t' t},\,
      \mathrm{sgn}(t' - t)\bigr),
  \label{eq:ctxer-summary} \\
  \Delta^{\mathrm{ctx}}_t
  \;&=
    \mathrm{FFN}_{\mathrm{ctx}}\!\Bigl(
      \mathrm{RMS}[c_t+\tau_t,\bar e_t,1-m_t]
    \Bigr), \\
  c_t
  \;&\longleftarrow
    \,\mathcal N_{\mathrm{ctx}}(c_t,\Delta^{\mathrm{ctx}}_t)
  \label{eq:ctxer-refresh}
\end{align}
where $\Psi^{\mathrm{s}}: \mathbb{R}^{2 d_{\mathrm{edge}} + 1} \to \mathbb{R}^{d_\ell}$ is a small MLP on the two directed edges of the pair and their index order, and $\mathrm{FFN}_{\mathrm{ctx}}: \mathbb{R}^{2 d_\ell + 1} \to \mathbb{R}^{d_\ell}$ proposes the update.

\emph{Step 2 (global refresh).}\;
The global then takes one update of Eq.~\eqref{eq:ngpt-update}, pooling the just-refreshed contexts over the full slot set, $g \longleftarrow \mathrm{Update}\bigl(g,\, \mathrm{pool}(g, \{c_t\})\bigr)$, holes included: after Step 1 a hole's context carries the tree geometry it was rebuilt from, and that structural information is exactly what the global should read.

\emph{Step 3 (edge rewrite --- synthesise hole-incident bonds).}\;
The refreshed contexts and global then rewrite the edge tensor,
\begin{equation}
  \Delta^E_{tt'}
    = \mathrm{FFN}_{\mathrm{ctx-edge}}\!\Bigl(
      [m_tm_{t'}\mathrm{RMS}(E_{tt'}),
       \widetilde c_t,\widetilde c_{t'},W_Eg]
      \Bigr),
  \qquad
  E_{tt'}\longleftarrow
    \mathcal N_{\mathrm{ctx-edge}}(E_{tt'},\Delta^E_{tt'}).
  \label{eq:ctxer-edge}
\end{equation}
where $\tilde c_t$ is a bias-free projection of the refreshed context's RMS copy, $W_E \in \mathbb{R}^{64 \times d_g}$ reads the refreshed global into 64 further input channels, broadcast to every pair, and $\mathrm{FFN}_{\mathrm{ctx-edge}}$ maps the concatenation back to $\mathbb{R}^{d_{\mathrm{edge}}}$. Unlike the trunk's additive tap of Eq.~\eqref{eq:trunk-ffn-update}, the global enters here as extra channels of the input. The pair mask $m_t m_{t'}$ zeroes the edge input whenever either endpoint is a hole, so a hole-incident edge is synthesised purely from its two endpoint contexts. By the time the tree runs, an edge into a hole is as meaningful as an edge into a site.

\emph{Step 4 (slot attention).}\;
Steps 1 and 3 are local, in the sense that they are fixed pairwise maps with a fixed decay profile, so each block closes with the one operation that moves information across the whole slot set in a content-dependent way: a pre-norm self-attention layer. Queries, keys and values are projected from $\mathrm{RMS}(c_t + \tau_t)$ and split into $n_{\mathsf{h}}$ heads of width $d_{\mathsf{h}}$, with the sans-serif head index of \S\,\ref{ssec:featurizer}, and the logits carry two additive biases,
\begin{equation}
  \mathrm{logit}_{\mathsf{h},\, t t'}
  \;=\;
  \frac{q_{\mathsf{h}, t} \cdot k_{\mathsf{h}, t'}}{\sqrt{d_{\mathsf{h}}}}
  \;-\;
  \alpha_{\mathsf{h}}\, \bigl(K - \mathrm{lca}(t, t')\bigr)
  \;+\;
  \beta_{\mathsf{h}}\bigl(E_{t t'}\bigr),
  \label{eq:ctxer-attn}
\end{equation}
\begin{align}
  \Delta^{\mathrm{slot}}_t
    &= \mathrm{Attn}_{\mathrm{slot}}(\{c_{t'}\};E)_{t},
  & c_t &\longleftarrow
    \mathcal N_{\mathrm{slot-attn}}(c_t,\Delta^{\mathrm{slot}}_t), \\
  \Delta^{\mathrm{slot-ffn}}_t
    &= \mathrm{FFN}_{\mathrm{slot}}(\mathrm{RMS}(c_t)),
  & c_t &\longleftarrow
    \mathcal N_{\mathrm{slot-ffn}}(c_t,\Delta^{\mathrm{slot-ffn}}_t).
\end{align}
where $\alpha_{\mathsf{h}} \ge 0$ is a fixed per-head slope (a scalar per head, not to be confused with the learned per-channel gates $\widetilde\alpha_g$ and $\widetilde\alpha_c$ of the global stream and the context leg), so that the heads see the tree's geometry, and $\beta_{\mathsf{h}}: \mathbb{R}^{d_{\mathrm{edge}}} \to \mathbb{R}$ is a per-head scalar projector of the freshly rewritten edge, so that the heads see the renormalised physics, again akin to Wilson's renormalisation group. Both biases reappear in full in the level attention of \S\,\ref{sssec:level-attn}. The attention and FFN candidates are incorporated through the bounded updates above.

After the second block, the global takes one closing update: the factored row-and-column edge reading of Eq.~\eqref{eq:edge-rowcol}, now over the rewritten edges and under the full slot mask. Holes participate here precisely because the contextualizer has just written their states, where the post-trunk reading of \S\,\ref{ssec:trunk-leaf} saw physical sites alone. The outputs after the second block are therefore the per-slot contexts $c_t^{(0)}$ that seed the level-$0$ context array, the rewritten edge tensor, which becomes $E^{(0)}$, and the refreshed global. As such, this contextualizer manufactures the structural information for the holes that balance the binary tree, which is what makes hole placement a meaningful learned decision rather than dead padding and allows the model to generalise over different sized systems. With the frame loaded, the merge recursion has everything it needs at $\kappa = 0$. For $\kappa \ge 1$, the recursion additionally requires a coarse-grained edge field $E^{(\kappa)}$ and the attention that consumes it; these are defined in the next subsection.

\subsubsection{The edge field and level attention}
\label{sssec:level-attn}

Even with bit-reversed leaves and a quadrilinear merge, the cascade still imposes a $\log_2 \hat N$ depth penalty on long-range pair correlations. This is because information between two leaves at common-ancestor depth $K - \mathrm{lca}(a, b)$ has to travel through as many merge nodes as there are between it and its lca before they meet.  We remove this penalty by interweaving lateral self-attention between the levels' edge data during the tree contraction.

Recall from \eqref{eq:tree-cast} that level $\kappa$ carries an edge field $E^{(\kappa)} \in \mathbb{R}^{N_\kappa \times N_\kappa \times d_{\mathrm{edge}}}$, one feature vector per ordered pair of level-$\kappa$ subtrees, seeded at $\kappa = 0$ by the leaf contextualizer's rewritten edge tensor (\S\,\ref{sssec:merge-upgrade}). As the tree contracts, the size of the edge field must be coarse-grained so that it has the correct shape for a given level's attention layer.

Thus, when a level's nodes are merged pairwise, the edge field is coarse-grained alongside them. The four child entries connecting the two children of $P$ to the two children of $Q$, together with the four child contexts, are combined into the single parent entry,
\begin{align}
  \bar E^{(\kappa+1)}_{PQ}
    &:= \mathrm{Norm}\!\left(
      \frac14\sum_{a,b\in\{0,1\}}
      E^{(\kappa)}_{2P+a,\,2Q+b}
      \right), \\
  \Delta^{\mathrm{coarse}}_{PQ}
    &:= \mathrm{FFN}_{\mathrm{coarse}}\!\Bigl(
      [\,E^{(\kappa)}_{2P,2Q},
         E^{(\kappa)}_{2P,2Q+1},
         E^{(\kappa)}_{2P+1,2Q},
         E^{(\kappa)}_{2P+1,2Q+1},
         c^{(\kappa)}_{2P},c^{(\kappa)}_{2P+1},
         c^{(\kappa)}_{2Q},c^{(\kappa)}_{2Q+1}\,]
      \Bigr), \\
  E^{(\kappa+1)}_{PQ}
    &:= \mathcal N_{\mathrm{coarse}}\!\left(
       \bar E^{(\kappa+1)}_{PQ},
       \Delta^{\mathrm{coarse}}_{PQ}
       \right).
  \label{eq:edge-coarsen}
\end{align}
This update is applied at every parent pair $(P,Q)$, with one $\mathrm{FFN}_{\mathrm{coarse}}$ shared across all levels, like the merge tensor itself.
This again has a direct analogy to Wilson's real space renormalisation group at the level of a block-spin step \citep{kadanoff1966scaling, wilson1975theory}. Each level halves the system and rewrites the effective couplings between blocks. Despite needing no spectral rank truncation, the four child edge cells and their contexts are nevertheless compressed into one fixed-width parent cell, and the information retained is exactly what the learned $\mathrm{FFN}_{\mathrm{coarse}}$ and nGPT interpolation preserve. The result is a learned, rather than spectral, coarse-graining of the edge field across levels.

The diagonal is included deliberately, since the self-edge $E^{(\kappa)}_{PP}$, built from the on-diagonal $2 \times 2$ child block, carries the \emph{internal} structure of block $P$ (such as the frustration its subtree encloses) while the off-diagonal entries carry the renormalised couplings between blocks.

After this coarsening step, the edge field is refined by the same 2-FWL path-composition feature already used in the trunk's edge update, applying Eq.~\ref{eq:triangle-aggregate} to the merge-node indices of level $\kappa$. This is so triangle-closure information survives the blocking. Figure~\ref{fig:edge-coarsen} draws one coarsening step.
\begin{figure}[!htbp]
  \centering
  \input{architecture-edge-coarsen.tikz}
  \caption{One edge-coarsening step
    (\S\,\ref{sssec:level-attn}). A $2 \times 2$ block of child edge cells (teal, highlighted) and the four child contexts (magenta) produce one parent edge cell via Eq.~\eqref{eq:edge-coarsen} --- a learned block-spin step. Diagonal self-edges (shaded) summarise the internal structure of their own subtree. The coarsened field is refined by the tree-level 2-FWL block and then biases the level-attention logits of Eq.~\eqref{eq:level-edge-logit}.}
  \label{fig:edge-coarsen}
\end{figure}
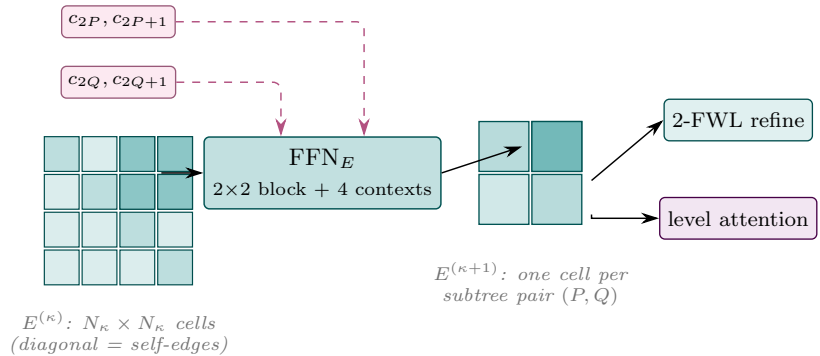

The attention layer between contraction levels is defined as follows. Let $c^{(\kappa)} \in \mathbb{R}^{N_\kappa \times d_c}$ be the context array at level $\kappa \in \{1, \dots, K\}$ with one row per merge node, and recall $N_\kappa = \hat N / 2^\kappa$. Let $m^{(\kappa)} \in \{0, 1\}^{N_\kappa}$ be the level mask that marks live merge nodes. The level-attention insert applies a pre-norm multi-head self-attention to block to $c^{(\kappa)}$ and applies two bounded updates,
\begin{align}
  \Delta^{\mathrm{attn},(\kappa)}
    &:= \mathrm{Attn}\!\bigl(
       \mathrm{RMS}(c^{(\kappa)});E^{(\kappa)}
       \bigr), \\
  c^{(\kappa)}
    &\longleftarrow
      m^{(\kappa)}\odot
      \mathcal N_{\mathrm{level-attn}}\!\left(
        c^{(\kappa)},\Delta^{\mathrm{attn},(\kappa)}
      \right)+ (1-m^{(\kappa)})\odot c^{(\kappa)},
      \\
  \Delta^{\mathrm{ffn},(\kappa)}
    &:= \mathrm{FFN}_{\mathrm{level}}\!\bigl(
       \mathrm{RMS}(c^{(\kappa)})
       \bigr), \\
  c^{(\kappa)}
    &\longleftarrow
      m^{(\kappa)}\odot
      \mathcal N_{\mathrm{level-ffn}}\!\left(
        c^{(\kappa)},\Delta^{\mathrm{ffn},(\kappa)}
      \right) + (1-m^{(\kappa)})\odot c^{(\kappa)}.
  \label{eq:level-attn-update}
\end{align}

As per the rank-4 contraction tensor and the FFN coarse graining network, the block's parameters are \emph{shared across all $K$ levels}. This is again to allow us to act on arbitrary system sizes, since the tree metric is self-similar in that siblings are at distance one at every level. Hence one attention layer can serve every level, and at evaluation time, levels deeper than any seen in training can be constructed.

This block is our functional stand-in for the disentanglers of the multiscale entanglement renormalisation ansatz (MERA)~\citep{vidal2007entanglement}. Like them, it communicates across block boundaries at every scale, but it acts on the $q$-independent conditioning stream rather than on the state itself. The carrier path receives that signal only through the context-conditioned gates that modulate its multilinear coefficients. Consequently the block can alter cross-subtree correlations represented by the wavefunction without introducing nonlinear dependence on any
quaternion.

Whilst we want the initial attention layers here to act as identity functions, we still need to encode the tree geometry with a positional encoding. Otherwise the attention layer will treat the merge nodes at each level as an unordered set. ALiBi \citep{press2021train} is the simplest construction that does this for a one-dimensional sequence. For a sequence of length $L$ and an attention head $h \in \{1, \dots, H\}$, the sequence-ALiBi logit at query position $i$ and key position $j$ is
\begin{equation}
  \mathrm{logit}^{\text{seq}}_{h, i, j}
  \;=\;
  \frac{q_{h, i} \cdot k_{h, j}}{\sqrt d}
  \;-\;
  \alpha_h \, \lvert i - j \rvert,
  \label{eq:alibi-seq}
\end{equation}
with $\alpha_h \ge 0$ a per-head slope drawn from a geometric schedule (\citet{press2021train, press2021shortformer} use $\alpha_h = 2^{-8h/H}$) and $d$ the per-head dimension. Here the content term $q_{h, i} \cdot k_{h, j} / \sqrt d$ scores how compatible the query at position $i$ is with the key at position $j$ in content space, with no positional information mixed into $Q$ or $K$. Positions enter the attention map only through the additive bias term, and the projection matrices are free to encode content purely. Notice also that the geometric slope schedule spans orders of magnitude, so the heads partition naturally by attention range; steep-slope heads see a short neighbourhood, while shallow-slope heads attend nearly globally, without any learned alignment. Finally, the bias is a function of the \emph{distance} $\lvert i - j \rvert$, so longer sequences at inference slot into the same linear penalty, and no out-of-distribution positional embeddings are needed when generalising to larger systems.

We adapt the ALiBi encoding \citep{press2021train} to a tree by replacing the sequence distance $\lvert i - j \rvert$ with the tree distance $d_{\text{tree}} = 2(K - \mathrm{lca})$ of Eq.~\eqref{eq:tree-distance} from \S\,\ref{sssec:balanced-merge}. This means the LCA-ALiBi logit absorbs the factor of $2$ into the head slope,
\begin{equation}
  \
  \mathrm{logit}^{\text{tree}}_{h, a, b}
  \;=\;
  \frac{q_{h, a} \cdot k_{h, b}}{\sqrt d}
  \;-\;
  \alpha_h \, \bigl(K - \mathrm{lca}(a, b)\bigr),
  \qquad
  \alpha_h \;=\; \tfrac{3}{2} \cdot 2^{-4h / (H - 1)},
  \label{eq:lca-alibi}
\end{equation}
where $h \in \{0, \dots, H - 1\}$. We set the slopes to fixed constants spanning the geometric ladder from $1.5$ down to $\approx 0.09$ (tuned empirically), so the heads partition by tree-distance band exactly as in the sequence case, with zero learnable positional parameters. The bias depends on the node pair only through $a \oplus b$ (via $\mathrm{lca}$, Eq.~\ref{eq:lca-bitrev}). This means we have only one XOR plus one bit-scan per pair to compute, and is the same function at every level, meaning the computation is self-similar across the tree and we can share the block across different layers. Hence the same position encoding drives the route policy of \S\,\ref{ssec:routes}, and we will frequently refer back to it in \S\,\ref{ssec:routes}.

Finally, the coarsened edge field enters the logits as a learned additive bias alongside the tree-distance term,
\begin{equation}
  \mathrm{logit}^{\text{tree, edge}}_{h, a, b}
  \;=\;
  \frac{q_{h, a} \cdot k_{h, b}}{\sqrt d}
  \;-\;
  \alpha_h\,(K - \mathrm{lca}(a, b))
  \;+\;
  \beta_h\!\bigl(E^{(\kappa)}_{a b}\bigr),
  \label{eq:level-edge-logit}
\end{equation}
with $\beta_h: \mathbb{R}^{d_{\mathrm{edge}}} \to \mathbb{R}$ a small per-head scalar projector, scaled by the same $1/\sqrt d$ as the content term so that changing the head width does not reweight geometry against content. Thus the heads see where two blocks sit in the tree \emph{and} how strongly the renormalised Hamiltonian couples them.

\subsubsection{Root readout and the structural invariant}
\label{sssec:root-readout}

After $K$ levels of merges and level attention, one node remains carrying the root state $(c_{\mathrm{root}},\, u_{\mathrm{root}},\, s_{\mathrm{root}})$ together with the top of the edge stream $e_{\mathrm{root}} := E^{(K)}_{00} \in \mathbb{R}^{d_{\mathrm{edge}}}$ as a fully renormalised state of the whole system's coupling structure. The readout maps the root carrier to a pair of real numbers through a hypernet-generated readout matrix $W \in \mathbb{R}^{2 \times d_u}$, produced from the root's full even state,
\begin{equation}
  W \;=\; W\!\bigl(\bigl[\, \mathrm{RMS}(e_{\mathrm{root}}),\;
    c_{\mathrm{root}},\; g \,\bigr]\bigr),
  \label{eq:root-hypernet-input}
\end{equation}
by the same low-rank gate\footnote{its third and final appearance} as Eqs.~\eqref{eq:leaf-builder} and \eqref{eq:merge-hypernet}: the renormalised coupling structure $e_{\mathrm{root}}$, the root context $c_{\mathrm{root}}$, and the global stream's final value, refreshed through every level of the tree and read through the root's own projection (\S\,\ref{sssec:balanced-merge}). This allows us to assemble a complex log-amplitude from that pair along with the banked log-scale,
\begin{equation}
  (\psi_1, \psi_2)
  \;=\;
  W\, u_{\mathrm{root}}
  \;\in\; \mathbb{R}^2,
  \qquad
  \log \psi_\theta(q)
  \;=\;
  \underbrace{\tfrac{1}{2}\,
    \log\bigl(\psi_1^2 + \psi_2^2\bigr)
    + s_{\mathrm{root}}}_{\log\lvert\psi_\theta\rvert}
  \;+\;
  i\, \mathrm{atan2}(\psi_2, \psi_1).
  \label{eq:root-readout}
\end{equation}
Here $\mathrm{atan2}$ is the two-argument arctangent, giving the usual phase of the complex scalar $\psi_1 + i \psi_2$ in $(-\pi, \pi]$. Hence, decomposition is exactly $\psi_\theta = e^{s_{\mathrm{root}}} (\psi_1 + i \psi_2)$, and the pair is the amplitude in Cartesian form at unit scale, with $s_{\mathrm{root}}$ re-attaching every magnitude the scale normalisations banked on the way up (\S\,\ref{sssec:merge-upgrade}). We can formalise this by the following lemma.

\begin{lemma}[Restored-scale multilinearity]
\label{lem:restored-multilinearity}
For every active site $i$, let $a_i^{(0)}:=e^{s_i^{(0)}}u_i^{(0)}$.  Under the leaf initialisation \eqref{eq:leaf-rms}, the bilinear merge \eqref{eq:reconstructed-merge}, and the root readout \eqref{eq:root-readout}, the represented amplitude $\psi_\theta(q_1,\ldots,q_N)$ is degree one in each quaternion $q_i$ separately.
\end{lemma}

\begin{proof}
At a leaf, Eq.~\eqref{eq:leaf-scale-cancellation} gives $a_i^{(0)}=B_iq_i$, with $B_i$ independent of every quaternion.  Suppose inductively that $a_A$ and $a_B$ are multilinear in the active leaves of two disjoint subtrees $A$ and $B$. When both children contain physical leaves, the map $M_PT$ is bilinear; when exactly one is empty, Eq.~\eqref{eq:masked-carrier-merge} passes the other child through linearly. Induction over the nonempty leaves therefore preserves degree one in every physical quaternion. Induction up the balanced tree gives a multilinear $a_{\mathrm{root}}$.  Finally, if $W_1,W_2$ are the two rows of the root readout, then
\begin{equation}
  \psi_\theta
  \;=\;
  e^{s_{\mathrm{root}}}(\psi_1+i\psi_2)
  \;=\;
  (W_1+iW_2)a_{\mathrm{root}},
\end{equation}
which is a linear functional of a multilinear carrier. Therefore $\psi_\theta$ is degree one in each $q_i$.
\end{proof}

One might ask why the readout does not simply emit $\log \psi_\theta$ as a linear function of $u_{\mathrm{root}}$, which looks more direct. However we emphasise that this would violate the central oddness of our phase space since \eqref{eq:per-site-oddness} demands that a single-site flip send $\psi_\theta \to -\psi_\theta$, i.e.\ $\log \psi_\theta \to \log \psi_\theta + i\pi$. Eq.~\eqref{eq:root-readout} is built so that the linearity sits exactly there, and thus,
\begin{equation}
  q_i \to -q_i
  \;\;\Longrightarrow\;\;
  u_i^{(0)} \to -u_i^{(0)}
  \;\;\Longrightarrow\;\;
  u_{\mathrm{root}} \to -u_{\mathrm{root}}
  \;\;\Longrightarrow\;\;
  (\psi_1, \psi_2) \to -(\psi_1, \psi_2)
  \;\;\Longrightarrow\;\;
  \psi_\theta \to -\psi_\theta.
  \label{eq:oddness-chain}
\end{equation}
The first arrow is the leaf oddness \eqref{eq:leaf-odd}. The second holds because each merge is linear in each child and its normalisation is parity-odd, so the sign flip propagates up the unique root-ward path through exactly one node per level. Every log-scale $\tilde s$ is even and stays put, and nothing else can change. The third arrow holds because $W$ is a function of $(e_{\mathrm{root}}, c_{\mathrm{root}}, g)$, all even, $q$-independent quantities. And the fourth is \eqref{eq:root-readout}: $\psi_1^2 + \psi_2^2$ is invariant, so $\log \lvert \psi_\theta \rvert$ does not move, while the $\mathrm{atan2}$ picks up exactly $\pi$. Per-site oddness therefore holds at \emph{every} parameter value of our model, and the symmetry is a property of the architecture itself.

Figure~\ref{fig:tree-readout} draws the full readout of this section, from the initial bit-reversed canonical frame to the root readout on top. Notice that many different shapes and tensor ranks flow through it --- $d_u$-wide carriers, scalar log-scales, the complex pair at the root --- yet every $q$-dependent intermediate, irrespective of its overall dimension, shares one structure on its two leading axes: at level $\kappa$ it splits into $\hat N / 2^\kappa$ independent chunks, one per node, and the chunk at node $P$ depends on $q$ only through the $3 \cdot 2^\kappa$ Lie coordinates (\S\,\ref{sec:ops_and_autodiff}) of $P$'s own subtree. The merge is the only operation that grows this dependence set, uniting two subtrees into one; everything else in the readout leaves it untouched. This is what will allow the energy kernel of \S\,\ref{ssec:energy-kernels} to propagate each chunk's derivatives along its own $3 \cdot 2^\kappa$ coordinates and no others, bringing the forward-mode Laplacian to $\mathcal{O}(\hat N^2)$ instead of the $\mathcal{O}(\hat N^3)$ a generic implementation would pay.

\subsection{Deep reinforcement learning of leaf permutations}
\label{ssec:routes}

In this section we detail our deep reinforcement learning strategy to learn the leaf-to-site map, which we call the \emph{router}. We will not introduce deep reinforcement learning here from first principles; see \citep{sutton1998reinforcement} for a comprehensive introduction.

The bit-reversed leaf assignment of \S\,\ref{sssec:bitrev} fixes the address geometry, but the map from physical sites to leaf positions is currently the same irrespective of the Hamiltonian interaction data. Indeed, site $i$ goes to leaf $\mathrm{brev}_K(i - 1)$, regardless of $(J, h)$. On systems whose correlation structure aligns with the bit-reversed prior (mainly the short-range and translation-invariant ones), a static assignment is fine. However, for systems whose correlation structure aligns with no fixed assignment, such as frustrated couplings, quenched disorder, topologically non-trivial geometry, or long-range tails, we want the leaf-to-site map to be \emph{learnable} and to condition on the $(J, h)$ tensors of the Hamiltonian.

We let the leaf-to-site map be a discrete latent variable, and define \emph{route} as a permutation of the whole slot block,
\begin{equation}
  p \;\in\; S_{\hat N},
  \label{eq:route-def}
\end{equation}
where $S_{\hat N}$ is the symmetric group on the $\hat N$ slots, containing physical sites and holes on an equal footing. We emphasise that the holes move, per our earlier discussion on scaling to unseen system sizes. The natural alternative is to constrain routes to be mask-preserving, $m_{p_t} = m_t$, so that padding stays frozen at its canonical slots. But \S\,\ref{sssec:balanced-merge} made hole placement informative, since a hole carries a context, and where the holes sit shapes every context merge above them to the root readout. The router therefore frees the holes and decodes their positions with the same care as the sites'. At decode step $t=0$ we restrict the candidate set to physical sites. This is a policy-support convention, it anchors the autoregressive prefix in a physical token before the hole-placement statistics are updated.  Freed holes do cost an $(\hat N - N)!$-fold exchange redundancy, as routes that differ only by which hole went where produce the identical routed system. However, our sampler detailed below removes this redundancy exactly at every decode step, by collapsing the interchangeable holes into a single candidate class.

Given a route $p$, every per-site and per-bond tensor is relabelled into the route's frame (including masks for the holes),
\begin{equation}
  q'_t := q_{p_t},
  \quad
  h'_t := h_{p_t},
  \quad
  J'_{tu} := J_{p_t, p_u},
  \quad
  \ell'_t := \ell^{(L)}_{p_t},
  \quad
  e'_{tu} := e^{(L)}_{p_t, p_u},
  \quad
  m'_t := m_{p_t}.
  \label{eq:routed-frame}
\end{equation}
Indeed, under mask-preserving routes the mask was a constant of the system, but once real and virtual positions mix, \emph{which slots are real} is itself route-dependent, and the energy kernel, the sampler, and the carrier gating of Eq.~\eqref{eq:carrier-mask} all read the routed mask $m'$. Inside the routed frame everything downstream  (the leaf contextualizer, leaf builder, merge tree, and root readout) operates on an ordinary already-routed system. Thus each kernel sees $(q', h', \ell', e', m')$ exactly as if the permutation had been applied at the input, and the route is invisible to the conditional ansatz beyond this relabelling.

Let $\pi_\theta(p \mid h, \ell^{(L)}, e^{(L)}, m)$ denote a Hamiltonian-conditioned route policy (parametrised below), and let $\psi_\theta(\cdot \mid p)$ be the routed ansatz of \S\,\ref{ssec:trunk-leaf} and \S\,\ref{ssec:tree-readout} evaluated in the $p$-frame. The policy returns probabilities, not amplitudes, and a route is drawn afresh at each evaluation, so the object the model represents is not a single wavefunction but an incoherent convex mixture of the routed states, that is, the density matrix
\begin{equation}
  \rho_\theta(q, q')
  \;=\;
  \sum_{p \,\in\, S_{\hat N}}
    \pi_\theta(p \mid h, \ell^{(L)}, e^{(L)}, m)\,
    \frac{\psi_\theta(q \mid p)\,\psi_\theta^{*}(q' \mid p)}{Z_p},
  \label{eq:routed-mixture}
\end{equation}
with $Z_p = \int |\psi_\theta(q \mid p)|^2\,\mathrm{d}q$ the squared norm of the routed state. The parameters $\theta$ are shared between the policy and the conditional ansatz. Its energy is the matching convex combination of the routed states' Rayleigh quotients, a variational principle over mixed states, but a benign one: the energy is linear in $\rho_\theta$, so its minimum over the convex set of density matrices is attained at an extreme point, a pure state, namely the ground state. The relaxation therefore costs nothing, and at the optimum $\rho_\theta$ is pure even where the policy keeps its weight spread over routes that realise one and the same physical state. Conditioning only on the trunk's outputs $(h, \ell^{(L)}, e^{(L)})$ and the slot mask $m$ keeps the learned raw scoring map permutation-equivariant under relabelling of the slot block.  The minimum-index representative retained after quotienting fixes a bookkeeping gauge inside each WL class, so the labelled representative route itself need not transform equivariantly; class probabilities and the routed physical state remain invariant when the classes are genuine symmetry orbits. Since the policy never sees $q$, each term of the mixture preserves the per-site oddness of
\S\,\ref{sssec:leaf-builder}.

The rest of this section sets up the architecture of the policy, and the machinery to train it with policy gradients via a discrete sampling step, a score-function estimator, and some variance-control analysis.

We conclude by characterising optimal symmetrised policies through the automorphism orbits of the routed problem. A policy supported uniformly on one energy-minimising orbit $\mathcal O$ has entropy $\log|\mathcal O|$, which is generally far below $\log(\hat N!)$. This orbit construction,  establishes that a non-uniform optimum exists beyond the guarantees of having a lower bound below the maximum.

The sum in Eq.~\eqref{eq:routed-mixture} contains $\hat N!$ routes and is therefore estimated by sampling. Tractable ancestral sampling and exact log-probabilities come from factorising the policy autoregressively into $\hat N$ masked slot picks. Differentiation through the discrete draw is handled separately by the score-function estimator of \S\,\ref{ssec:rl-routes}.

To that end, let $p_{<t} := (p_0, p_1, \dots, p_{t-1})$ denote the prefix of picks taken before step $t$, and let $p_{\le t} := (p_0, p_1, \dots, p_t)$. Then
\begin{equation}
  \pi_\theta(p \mid h, \ell^{(L)}, e^{(L)}, m)
  \;=\;
  \pi_\theta(p_0 \mid h, \ell^{(L)}, e^{(L)}, m)
  \,\prod_{t=1}^{\hat N - 1}
    \pi_\theta(p_t \mid p_{<t},\, h, \ell^{(L)}, e^{(L)}, m),
  \label{eq:policy-factor}
\end{equation}
where this runs over the full slot block, and holes are placed by the same factors that place sites. The first pick anchors the tree, with a physical site at slot $p_0$.  If $\mathcal Q_0^{\rm real}$ denotes the WL classes after restricting the empty-prefix candidate set to physical sites and $r(C)$ is the retained representative of class $C$, then
\begin{equation}
  \pi_\theta(p_0 = r(C) \mid h, \ell^{(L)}, e^{(L)}, m)
  \;=\;
  \frac{\exp\bigl(\tilde\eta_{0,C}\bigr)}
       {\sum_{D\in\mathcal Q_0^{\rm real}}
        \exp\bigl(\tilde\eta_{0,D}\bigr)},
  \label{eq:first-pick}
\end{equation}
where $\tilde\eta$ is the class logit of Eq.~\eqref{eq:quotient-collapse}. We emphasise that the anchor is \emph{learned}, not uniform. This is because which site the contraction is built around is itself part of the structure the policy should discover. The raw logits at all $\hat N$ positions are Hamiltonian-conditioned; the final factor has one available class and therefore contributes the identically zero log-probability and gradient.

\subsubsection{Hamiltonian-conditioned scoring row}
\label{sssec:router-scoring-row}

Before quotienting, each factor in Eq.~\eqref{eq:policy-factor} has raw pointer logits $\eta_t\in\mathbb R^{\hat N}$ over the candidate slots, with the empty prefix at $t=0$; the actual factor is the softmax over the resulting WL-class logits $\tilde\eta_t$. The decoder is a pointer network in the sense of \citet{vinyals2015pointer}, with one structural difference: a standard pointer decoder represents the placed prefix as a flat sequence, whereas ours materialises the partial merge tree implied by the picks made so far. The logits are produced by a candidate composer, a tree-prefix encoder, a tied dot-product pointer, and a conditional symmetry quotient. We describe them in turn.

\smallskip\noindent
\emph{(a) Candidate composer.}
Let $P_\ell: \mathbb{R}^{d_\ell} \to \mathbb{R}^{d_h}$ be a small bias-free linear projector on per-site features from the trunk, as in \S\,\ref{ssec:featurizer}, and let $\phi_{\mathrm{pre}},\phi_{\mathrm{suf}}: \mathbb{R}^{2d_e}\to\mathbb{R}^{d_h}$ be independent per-channel SiLU MLPs.  Both read the two directed bonds of a pair through $E_{iu}:=[e^{(L)}_{i,u},e^{(L)}_{u,i}]$. Define the projected node embedding at slot $i$,
\begin{equation}
  \mathrm{node}_i \;:=\; P_\ell(\ell^{(L)}_i).
  \label{eq:centre-node}
\end{equation}
Let $U_t:=\{0,\ldots,\hat N-1\}\setminus p_{<t}$ be the available candidate set, and write $\bar g=W_{\mathrm{global}}g+b_{\mathrm{global}}$ for the projection of the fixed forked router global.  Define
\begin{equation}
  P_{t,i}:=\frac{1}{\sqrt{\max(t,1)}}
      \sum_{u\in p_{<t}}\phi_{\mathrm{pre}}(E_{iu}),
  \qquad
  S_{t,i}:=\frac{1}{\sqrt{\max(|U_t|-1,1)}}
      \sum_{v\in U_t\setminus\{i\}}\phi_{\mathrm{suf}}(E_{iv}).
  \label{eq:candidate-summaries}
\end{equation}
The implemented composer does not first make a static vocabulary token. Let $T_g:\mathbb R^{d_g}\to\mathbb R^{256}$ denote the active bottleneck on the raw-global channel, and write the router's stabilized RMS map as
\begin{equation}
  \operatorname{RMS}_{\rm r}(z)
  :=\gamma_{\rm rms}\odot\frac{z}
  {\sqrt{\max\{\overline{z^2},10^{-2}\}}}.
  \label{eq:router-rms}
\end{equation}
Here every occurrence has its own learned scale $\gamma_{\rm rms}$. The composer RMS-normalises each learned vector stream, leaves the three hole ratios raw, and adds learned affine projections into one row--candidate accumulator,
\begin{equation}
  \begin{aligned}
  a_i^{(t)}={}&b+W_\ell\operatorname{RMS}_{\rm r}(\ell_i^{(L)})
    +W_g\operatorname{RMS}_{\rm r}(T_g g)
    +W_{\bar g}\operatorname{RMS}_{\rm r}(\bar g+\gamma(t))\\
   &+W_P\operatorname{RMS}_{\rm r}(P_{t,i})
    +W_O\operatorname{RMS}_{\rm r}(O_{t,i})
    +W_S\operatorname{RMS}_{\rm r}(S_{t,i})\\
   &+W_H\operatorname{RMS}_{\rm r}(H_{t,i})+W_\rho\rho_{t,i}.
  \end{aligned}
  \label{eq:candidate-embed}
\end{equation}
where $O_{t,i}$ is the order-aware prefix stream described below, $H_{t,i}$ is the ordered hole-prefix stream, and $\rho_{t,i}$ collects the three remaining-hole ratios.  In the reported model one residual candidate FFN acts on $a_i^{(t)}$; an output projection is then added to $\mathrm{node}_i$ to give $x_i^{(t)}\in\mathbb R^{d_h}$.  Thus $g$ and $\bar g$ are fixed across decode rows, while the root-centred clock $\gamma(t)$ tells the composer how deep into the sequence it is.  The prefix summary $P_{t,i}$ runs over already-picked slots and the candidate-specific suffix $S_{t,i}$ runs over every other available slot. Both can be maintained as cumulative scans, so each append at step $t$ costs $\mathcal{O}(\hat N d_h)$ work to update the prefix and suffix summaries at every candidate.  Across a complete route this bookkeeping costs $\mathcal{O}(\hat N^2d_h)$.  The dense candidate refinement below contributes a cubic core; sampled tree recomputation and the conditional quotient are accounted for separately.


The order-aware stream $O_{t,i}$ is a second read of the prefix: the candidate-specific edge message from position $s<t$ is weighted by a learned Gaussian filterbank over the dyadic LCA distance between decode positions $t$ and $s$, so that \emph{when} a neighbour was placed matters as well as whether. The hole streams track both how many holes remain and where they have been going.  Every one of these quantities depends on the current row $t$; there is no candidate embedding that is reused unchanged throughout a route.

\smallskip\noindent
\emph{(b) Tree-prefix encoder.}
The structural heart of the decoder is that the placed prefix is not summarised as a flat sequence; it is materialised as the thing it really is, a partial merge tree. A learned dyadic merge, whose inputs and combination rule (children, directed sibling edges, dyadic clocks, and the normalised interpolation of Eq.~\eqref{eq:context-interp} with its own gate $\alpha_r$) deliberately mirror the physical merge of \S\,\ref{sssec:merge-upgrade}, scans the placed tokens bottom-up, with causal per-level 2-FWL refinement and causal level attention, the machinery of \S\,\ref{sssec:level-attn} restricted to the prefix. The prefix $[0, t)$ is then represented by its \emph{dyadic cover}: the at most $\log_2 \hat N$ complete subtrees given by the binary decomposition of $t$. Cover segments self-attend, and the candidate tokens cross-attend to the cover through their own pooled candidate-to-segment edges. Four read-only cover-query blocks produce the decode state $y_t \in \mathbb{R}^{d_h}$; two candidate-to-cover blocks and four post-prefix blocks refine the candidate states before the pointer is applied. The decode state is, in other words, a learned surrogate of the partial wavefunction contraction that the next choice is about to extend. Positional information enters only through structure-aware signals, decode-position clocks, dyadic segment clocks, and the tree-distance attention biases of Eq.~\eqref{eq:lca-alibi}, and never through absolute slot identities, which is what keeps the raw scoring map equivariant under relabelling of the slot block before the representative gauge is fixed.

The router's global has two distinct roles.  It owns the second copy of the contextualizer stack of \S\,\ref{sssec:merge-upgrade}: this copy forks from the same post-trunk value as the physics leg, refreshes through that stack's per-block updates, and closes with the factored edge reading of Eq.~\eqref{eq:edge-rowcol} over the route-refined edges. The resulting \emph{fixed} fork $g$ supplies the two composer channels in Eq.~\eqref{eq:candidate-embed}. A projection of the same fixed $g$ enters every selected-leaf merge and every cover-query FFN; the same-level 2-FWL and attention refinements receive no additional global tap.

Only after the prefix cover has been formed does the stream acquire a row-specific value,
\begin{equation}
  g^{(t)} \;:=\;
  \begin{cases}
    g, & t=0,\\
    \mathrm{Update}\!\bigl(g,\;
      \mathrm{pool}\bigl(g,\mathcal C_t\bigr)\bigr), & t>0,
  \end{cases}
  \label{eq:route-global-step}
\end{equation}
one update of Eq.~\eqref{eq:ngpt-update} reading the complete-subtree roots in $\mathcal C_t$; at an empty cover the learned update is masked out and the global remains unchanged.  This refreshed value enters only the FFNs of the two candidate-to-cover blocks.  It is not fed back into the composer, the query seed, the tree merge, the cover-query blocks, or the four post-prefix blocks.  Each position computes Eq.~\eqref{eq:route-global-step} independently from the same fixed $g$, so teacher scoring of a stored route and sequential sampling of a fresh one evaluate the same row map (up to floating-point evaluation order).

The attention pattern is defined explicitly as follows. For a query sequence $A=(a_i)$, a key--value sequence $B=(b_j)$, write $\hat a_i=\operatorname{RMS}_{\rm r}(a_i)$ and $\hat b_j=\operatorname{RMS}_{\rm r}(b_j)$ for the block's tokenwise pre-norms. For an additive structural bias $\beta_h$, a per-head key dimension $d_k$, and a mask $M\in\{0,-\infty\}^{|A|\times|B|}$, write
\begin{equation}
  \operatorname{Attn}_h(A,B;\beta_h,M)
  :=\operatorname{softmax}\!\left(
    \frac{(\widehat A W_Q^h)(\widehat B W_K^h)^\mathsf{T}}
         {\sqrt{d_k}}
    +\beta_h+M\right)\widehat B W_V^h .
  \label{eq:router-attention}
\end{equation}
Equation~\eqref{eq:router-attention} displays one head; in the residual equations below, $\operatorname{Attn}$ denotes the implemented gated multihead collapse followed by the learned bias-free output projection $W_O$. All projections are learned separately in each block.  The important object retained at decode position $s$ is not a static embedding of the selected site, but the raw row-conditioned candidate state $\tilde v_s$.  The active tree normalisation maps it through the floored RMS sphere transform before the first merge,
\begin{equation}
  \tilde v_s:=x_{p_s}^{(s)},
  \qquad
  v_s:=\operatorname{Sphere}(\tilde v_s)
      :=\frac{\tilde v_s}
      {\sqrt{\max\{\overline{\tilde v_s^{2}},10^{-4}\}}}.
  \label{eq:router-dynamic-leaf}
\end{equation}
Thus a tree leaf already records the Hamiltonian, the prefix and suffix seen when it was selected, the hole statistics, and the decode clock.

The states $v_0,\ldots,v_{t-1}$ are merged bottom-up over aligned dyadic intervals.  After every merge level, ordered subtree-pair features undergo a row-causal 2-FWL update, and roots of the same dyadic size attend one another using Eq.~\eqref{eq:router-attention} with the refined pair features as $\beta$ and the mask $M_{ab}=0$ for $b\leq a$ and $-\infty$ otherwise.  This is a same-scale causal refinement, since information can move between complete subtrees without flattening their internal contractions.

Let the binary expansion of $t$ decompose the prefix into its unique set of maximal aligned intervals,
\begin{equation}
  [0,t)=\bigsqcup_{a=1}^{c_t}I_{t,a},
  \qquad
  \mathcal C_t:=\{z(I_{t,a})\}_{a=1}^{c_t},
  \qquad
  c_t=\operatorname{popcount}(t)\leq\lceil\log_2\hat N\rceil ,
  \label{eq:router-dyadic-cover}
\end{equation}
where $z(I)$ is the root state of the complete subtree on $I$. The fresh query seed $r_t^{(0)}=\bar g+\gamma(t)$, with $\gamma(t)$ the root-centred decode clock, is appended to this cover.  Four edge-biased cover blocks then apply Eq.~\eqref{eq:router-attention} with a deliberately asymmetric mask: every live cover root may read all live cover roots, and the appended query may read the roots and itself, but no root may read the appended query.  Padding is masked in both directions.  The cover roots generally have different dyadic sizes; their all-to-all interaction is causal because every member of $\mathcal C_t$ lies strictly inside the already selected prefix.  The resulting query
\begin{equation}
  y_t=\operatorname{CoverRead}_\theta
      \bigl(r_t^{(0)},\mathcal C_t,E_t^{\mathrm{cover}}\bigr)
  \label{eq:router-cover-query}
\end{equation}
where $E_t^{\mathrm{cover}}$ is the directed root-pair message tensor formed by pooling the physical pair features over the leaves of each ordered pair of cover segments.  Thus $y_t$ is a read-only observation of every complete subtree that exactly tiles the current prefix.  At $t=0$ the cover is empty and $y_0$ is the learned four-block transform of the global--clock seed, which attends itself.  In particular, $y_t$ is rebuilt for every $t$ from the complete-subtree roots in that row's cover.

First, two blocks let each $x_i^{(t)}$, $i\in U_t$, cross-attend to $\mathcal C_t$; their biases pool the directed candidate--subtree edges over the leaves of each cover segment, and their FFNs alone receive the refreshed $g^{(t)}$ of Eq.~\eqref{eq:route-global-step}.  Four post-prefix blocks then repeat the ordered update.  Let $X_t^{(0)}$ be the output of the second candidate-to-cover block, set $L_{\mathrm{post}}=4$, and write $\lambda_{\mathrm{post}}:=L_{\mathrm{post}}^{-1/2}$. Then
\begin{align}
  H_t^{(r)}
    &=X_t^{(r)}+\lambda_{\mathrm{post}}
       \operatorname{Attn}^{(r)}_{\mathrm{hist}}
       \bigl(X_t^{(r)},Y_{<t};
       \beta^{\mathrm{hist}}_t,M^{\mathrm{hist}}_t\bigr), \\
       M^{\mathrm{hist}}_{t,s}&=
       \begin{cases}0,&s<t,\\-\infty,&s\geq t,\end{cases}
       \label{eq:router-history-attn}\\
  S_t^{(r)}
    &=H_t^{(r)}+\lambda_{\mathrm{post}}
       \operatorname{Attn}^{(r)}_{\mathrm{self}}
       \bigl(H_t^{(r)},H_t^{(r)};
       \beta^{\mathrm{cand}}_t,M^{U_t}\bigr),
       \label{eq:router-suffix-attn}\\
  X_t^{(r+1)}
    &=S_t^{(r)}+\lambda_{\mathrm{post}}
       \operatorname{FFN}_r\bigl(\operatorname{RMS}_{\rm r}(S_t^{(r)})\bigr),
    \qquad r=0,\ldots,L_{\mathrm{post}}-1 .
  \label{eq:router-post-block}
\end{align}
We set $\operatorname{Attn}(A,\varnothing;\beta,M):=0$. Hence at $t=0$ the candidate-to-cover and history-attention residual sublayers are identities, while their FFNs and the remaining-candidate self-attention still run.  Here $Y_{<t}=(y_0,\ldots,y_{t-1})$; the strict history mask excludes the current and future queries, while $M^{U_t}$ permits all-to-all attention among the still-available candidates and masks every placed slot.  The history bias contains the direct bidirectional edge pair between candidate $i$ and the site selected at history position $s$, together with a dyadic-LCA term; $\beta_t^{\mathrm{cand}}$ contains the candidate-pair edge features. Consequently, before the pointer chooses one site, the candidates may compare themselves conditional on the same prefix and on one another. Writing $\bar x_i^{(t)}:=X_{t,i}^{(L_{\mathrm{post}})}$, the pointer of Eq.~\eqref{eq:pointer-logit} pairs $y_t$ with these refined candidate states.

For Hamilton-Zero, the active depths are four cover-query blocks, two candidate-to-cover blocks, and four post-prefix blocks, at width $d_h=512$ with $16$ logical heads and a $512$-dimensional pointer score. Each repeated stack uses its own depth gain: $\lambda_{\mathrm{cover}}=4^{-1/2}=1/2$, $\lambda_{\mathrm{cand}}=2^{-1/2}=1/\sqrt 2$, and $\lambda_{\mathrm{post}}=4^{-1/2}=1/2$. There is no tokenwise causal backbone and no key--value attention cache in this model: previous $y_s$ are retained as states and their keys and values are projected afresh in each post block.  Only the support of the prefix query is logarithmic in $\hat N$; the candidate-to-cover and remaining-candidate stages cost $\mathcal{O}(\hat N\log\hat N)$ and $\mathcal{O}(\hat N^2)$ per dense row, respectively. Across all $\hat N$ rows, the latter makes the dense attention core $\mathcal{O}(\hat N^3)$; the conditional WL quotient is accounted for separately. Figure~\ref{fig:tree-prefix-query} shows the resulting computation.

\begin{figure}[!htbp]
  \centering
  \resizebox{\textwidth}{!}{%
    \input{architecture-tree-prefix-query.tikz}%
  }
  \caption{The multiscale prefix query at
$t=7$. The selected prefix $p_{<7}$ is decomposed into its canonical dyadic cover $\mathcal C_7=\{z[0,4),z[4,6),v_6\}$, comprising complete subtrees of four, two, and one leaves. Same-level causal 2-FWL and attention refine the subtree roots during their construction. Four read-only query blocks attend from the conditioned seed $\bar g+\gamma(7)$ to the three cover roots, without updating those roots, and produce the query state $y_7$, which is retained for later decoding positions. In parallel, each remaining candidate receives a fresh prefix-dependent embedding $x_i^{(7)}$. Two candidate--cover cross-attention blocks condition these embeddings on $\mathcal C_7$, using the pooled edge bias and the routed global state $g^{(7)}$. Four post-prefix blocks then apply history cross-attention to $y_{<7}$, suffix self-attention among the remaining candidates, and a conditioned FFN, producing $\bar x_i^{(7)}$. The tied pointer scores $y_7$ against every refined candidate; the conditional symmetry quotient collapses equivalent candidates into Weisfeiler--Leman classes before the Gumbel pick. Here $|\mathcal C_7|=\operatorname{popcount}(7)=3$, so only the support of the read-only prefix query grows logarithmically with prefix length.}
  \label{fig:tree-prefix-query}
\end{figure}
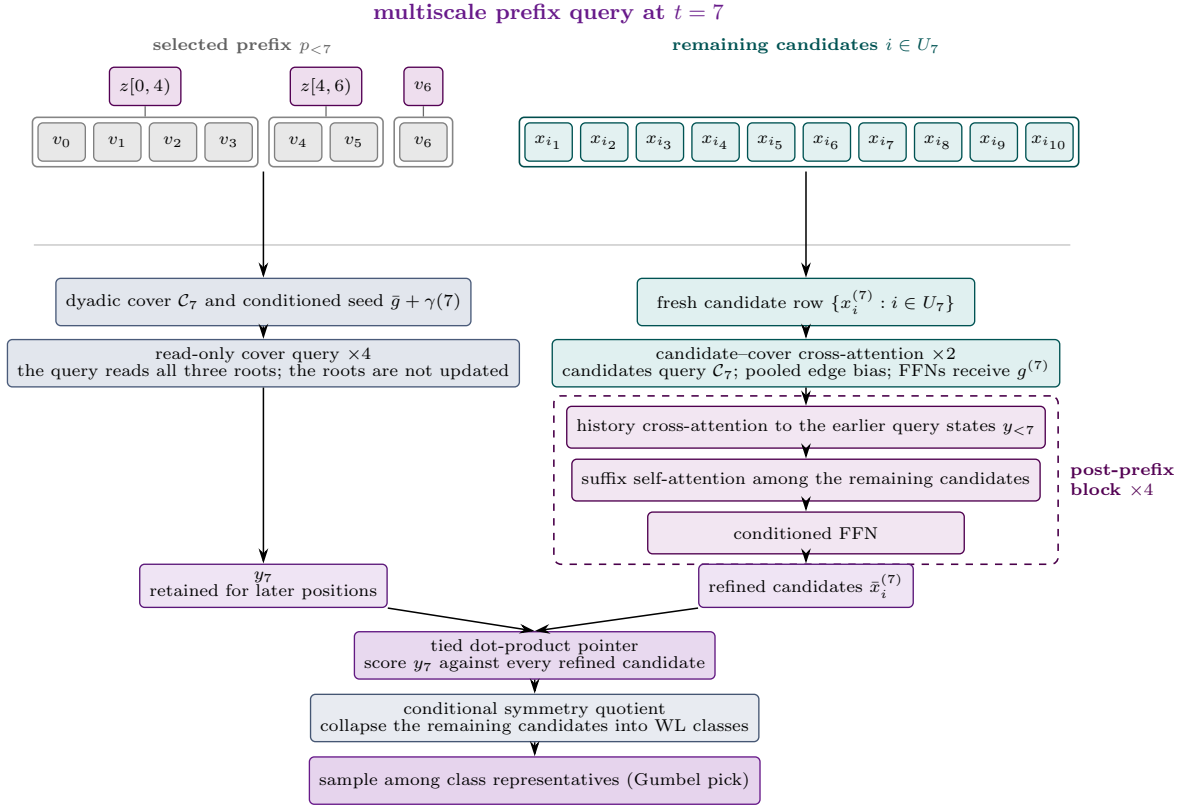

The two router refinement residuals retain their ordinary pre-norm skip connections. Every named update in the tree and global streams uses Eq.~\eqref{eq:ngpt-interpolation}, including the route-global pool update of Eq.~\eqref{eq:route-global-step}. 

\smallskip\noindent
\emph{(c) Pointer.} \citep{vinyals2015pointer}
Let $W_q, W_k \in \mathbb{R}^{d_h \times d_h}$ be a single pair of bias-free maps shared across all decode steps, and let $\tau > 0$ be a temperature parameter that sets the exploration scale of the policy. The per-step logit at candidate slot $i$ is the tied dot product of the projected decode state against the projected candidate token,
\begin{equation}
  \eta_{t, i}
  \;=\;
  \frac{1}{\tau}\,
  \frac{\bigl(W_q\, y_t\bigr) \cdot \bigl(W_k\, \bar x_i^{(t)}\bigr)}
       {\sqrt{d_h}}
  \;+\;
  \mathrm{mask}_{t, i},
  \label{eq:pointer-logit}
\end{equation}
with the availability mask
\begin{equation}
  \mathrm{mask}_{t, i}
  \;=\;
  \begin{cases}
    -\infty & \text{if } i \in p_{<t},
              \text{ or if } t = 0 \text{ and } m_i = 0, \\
    0       & \text{otherwise.}
  \end{cases}
\end{equation}
The availability mask first removes placed slots, holes remain eligible, except at the physical anchor step, and Eq.~\eqref{eq:quotient-collapse} then converts the surviving raw pointer logits into the class logits whose softmax is the actual pick distribution. The query projection $W_q$ is initialised at zero, so the raw pointer logits start equal.  After the mean quotient of Eq.~\eqref{eq:quotient-collapse}, routing therefore starts uniform over valid WL classes (not necessarily over individual slots), and every early preference is learned rather than inherited from the initialisation.

Figure~\ref{fig:route-decoder} shows the complete pipeline.

\smallskip\noindent
\emph{(d) Symmetry quotient.}
Before any sampling, each step's logits are collapsed over WL colour classes of the remaining candidates. The pair refinement is the second-order update the trunk's edge block already realises (Eqs.~\eqref{eq:edge-update}--\eqref{eq:triangle-aggregate}): colours live on ordered pairs of vertices of the coloured coupling graph and update from the multiset of two-leg paths through every third vertex, run here with the placed prefix individualised, and vertex colours are read off the diagonal.  Prefix-preserving automorphisms leave these colours invariant, so candidates in the same exact orbit cannot be separated; the converse is not a theorem for 2-WL on arbitrary graphs.  At the empty prefix, the converged colour classes matched the exact orbits on all $247$ verifiable systems in the checked panel; prefix-conditioned checks additionally matched on a $20$-system chain-and-ladder spot sample.  The per-system dispatch uses the cheaper first-order refinement when its empty-prefix partition agrees with the pair refinement and the system metadata permits it, and uses the pair version elsewhere. Let $\mathcal Q_t$ be the resulting partition of $U_t$, and let $C_t(i)\in\mathcal Q_t$ contain candidate $i$. Each class keeps a single representative whose logit averages its members,
\begin{equation}
  \tilde\eta_{t, C}
  \;=\;
  \log \frac{1}{\lvert C \rvert} \sum_{i \in C} e^{\eta_{t, i}},
  \label{eq:quotient-collapse}
\end{equation}
and the step samples among representatives only (a Gumbel-max draw \citep{gumbel1954statistical}, which is exact ancestral sampling). The mean, rather than a sum, is what stops a class's probability from scaling with its class size. All hole slots share one class at every step, so the $(\hat N - N)!$ hole-exchange redundancy of Eq.~\eqref{eq:route-def} is removed here. The same collapse is applied when scoring, so the factors of Eq.~\eqref{eq:route-logprob} are understood over the collapsed logits.

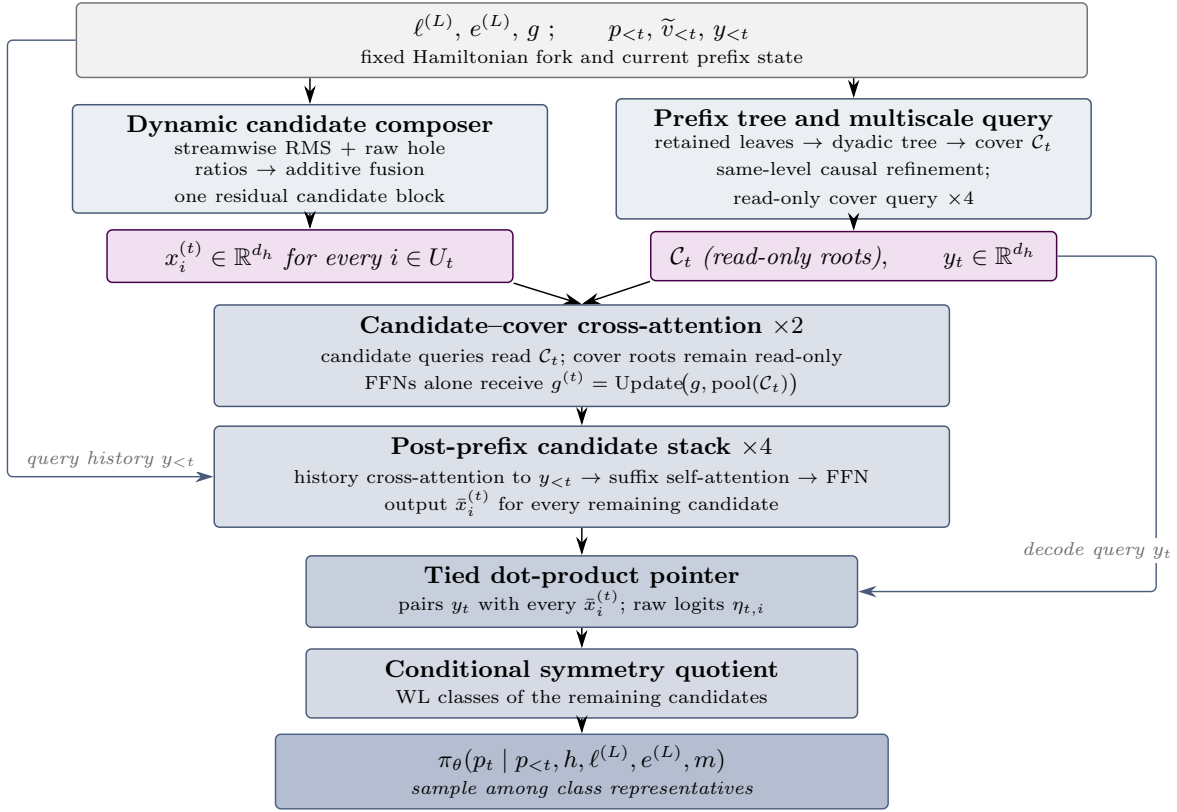
\begin{figure}[!htbp]
  \centering
  \resizebox{\textwidth}{!}{%
    \input{architecture-route-decoder.tikz}%
  }
  \caption{The route decoder
    (\S\,\ref{ssec:routes}), read from top to bottom. The candidate composer RMS-normalises and adds the learned vector streams with the raw hole ratios, then applies one residual candidate block to obtain the row-dependent tokens $x_i^{(t)}$ (Eqs.~\ref{eq:centre-node}--\ref{eq:candidate-embed}). In parallel, the prefix branch materialises the selected leaves as a dyadic tree, extracts its cover, and forms the read-only multiscale query $y_t$. Candidates cross-attend to the cover, then each post-prefix block applies history cross-attention, suffix self-attention, and an FFN in that order. The fixed router global conditions the composer, tree, and cover query; the cover-refreshed $g^{(t)}$ (Eq.~\eqref{eq:route-global-step}) enters only the two candidate-to-cover FFNs. Finally, the tied pointer scores $y_t$ against every refined candidate (Eq.~\ref{eq:pointer-logit}), and the conditional symmetry quotient collapses WL-indistinguishable choices before sampling (Eq.~\ref{eq:quotient-collapse}).}
  \label{fig:route-decoder}
\end{figure}

We refresh the route at \emph{every} training step. The refresh is cheap because nothing downstream is redrawn: sampling the new permutation costs one decoder pass, and the walker population is simply relabelled into the new frame, which is a measure-preserving coordinate change. After this, the route re-equilibrates in a few sampler sweeps (SM~\S~\ref{sec:part4}, \S\,\ref{ssec:mcmc}) with no dedicated burn-in.

\FloatBarrier
\subsubsection{Sequential sampling and vectorised teacher scoring}
\label{sssec:router-schedules}

The policy has two execution schedules for one mathematical scoring row.  Denote that row map by
\begin{equation}
  \mathcal D_\theta(p_{<t};h,\ell^{(L)},e^{(L)},m)
  =\bigl(y_t,\,\bar X_t,\,\eta_t\bigr),
  \label{eq:router-row-map}
\end{equation}
where $\bar X_t$ contains the fully refined states of the available candidates and $\eta_t$ is obtained by the tied pointer; the symmetry quotient is applied immediately afterwards.

In \emph{sequential sampling}, the Hamiltonian-dependent projections, directed message tables, and structural attention biases that do not depend on the sampled prefix are computed once.  At step $t$, the prefix, suffix, order and hole scans form every fresh $x_i^{(t)}$. The already selected states $\tilde v_s=x_{p_s}^{(s)}$ are retained and their spherical images $v_s=\operatorname{Sphere}(\tilde v_s)$ form the partial tree, its cover gives $y_t$, and the candidate-to-cover and post-prefix blocks produce $\bar X_t$ before the tied pointer produces $\eta_t$. A quotient-aware Gumbel-max draw supplies $p_t$, after which the summary scans, raw selected-state history, and query-state history are updated. In Hamilton-Zero the selected tree is recomputed from the retained $\tilde v_s$ at the next row. The retained history consists of states, not projected attention keys and values, and thus there is (counter-intuitively) no tokenwise key--value recurrence behind Eq.~\eqref{eq:router-row-map}.

In \emph{teacher-forced scoring}, a complete stored route makes every prefix known.  All step--candidate states are therefore constructed together,
\begin{equation}
  \mathbf X:=\bigl(x_i^{(t)}\bigr)_{t,i}
  \in\mathbb R^{\hat N\times\hat N\times d_h},
  \qquad
  \tilde v_t=\mathbf X_{t,p_t},\qquad
  v_t=\operatorname{Sphere}(\tilde v_t).
  \label{eq:router-teacher-gather}
\end{equation}
The prefix and suffix terms are vector cumulative sums; the order-aware terms are triangular masked contractions.  A single bottom-up tree scan builds every complete-subtree state at every level.  For each row $t$, dynamic indexing gathers precisely the roots in $\mathcal C_t$ from Eq.~\eqref{eq:router-dyadic-cover}; the cover-query, candidate-to-cover, strict-history, remaining-candidate, pointer and quotient operations then run over all rows in parallel with their row-specific masks.  In particular, the raw leaf gathered at position $t$ is $\mathbf X_{t,p_t}$ and is then sphered for the tree; it is not a site embedding shared across decode positions.

The two schedules consequently satisfy, in exact arithmetic,
\begin{equation}
  \eta^{\mathrm{sample}}_{t,i}
  =\bigl[\mathcal D_\theta(p_{<t};h,\ell^{(L)},e^{(L)},m)\bigr]_{\eta,i}
  =\eta^{\mathrm{teacher}}_{t,i},
  \qquad i\in U_t.
  \label{eq:router-schedule-equivalence}
\end{equation}
The distinction is computational, sampling
evaluates one changing row after another, whereas teacher forcing materialises every changing row and every dyadic cover in one masked batch. Figure~\ref{fig:route-schedules} displays the two schedules.

\begin{figure}[!htbp]
  \centering
  \resizebox{\textwidth}{!}{%
    \input{architecture-route-schedules.tikz}%
  }
  \caption{Two schedules for the same route
    scoring row.  Left: sequential sampling constructs the dynamic candidates and prefix tree, samples one quotient class, and stores the selected row-conditioned leaf and query state before the next row.  Right: teacher forcing constructs all $\hat N^2$ dynamic candidate states, gathers and spheres the selected leaf in every row, performs one multilevel tree scan, dynamically gathers every row's dyadic cover, and evaluates all masked scoring rows together.}
  \label{fig:route-schedules}
\end{figure}
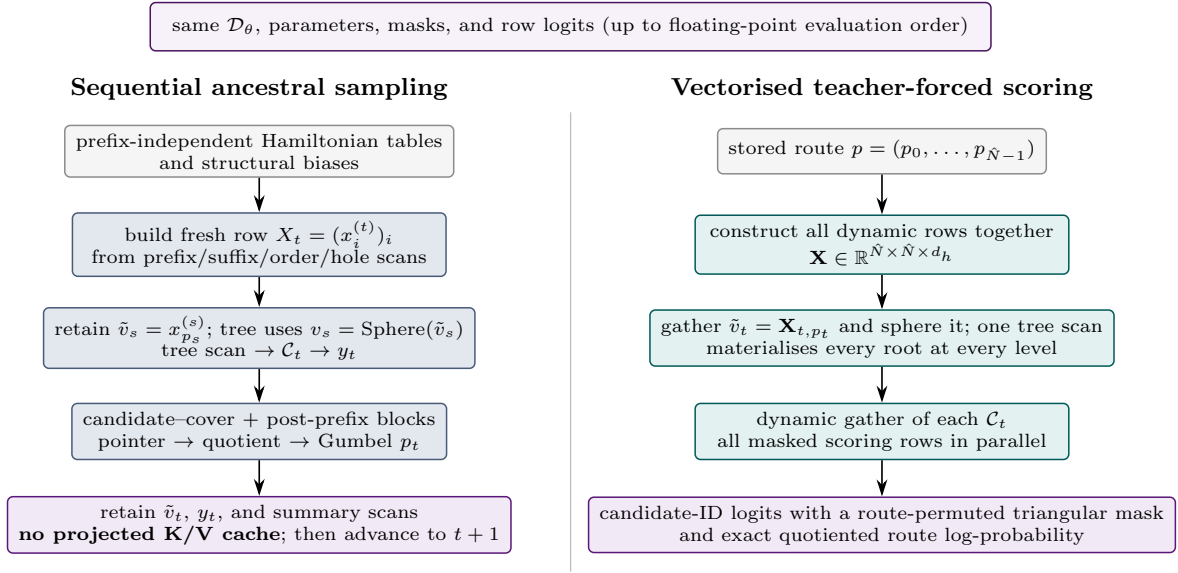

For teacher-forced scoring, the route log-probability therefore reads
\begin{equation}
  \log \pi_\theta(p \mid h, \ell^{(L)}, e^{(L)}, m)
  \;=\;
  \sum_{t = 0}^{\hat N - 1}
    \Bigl[\, \tilde\eta_{t,\, C_t(p_t)}
            \;-\;
            \log \sum_{C\in\mathcal Q_t} e^{\tilde\eta_{t,C}} \,\Bigr].
  \label{eq:route-logprob}
\end{equation}
Before detailing the training algorithm, we emphasise here that each route row contributes separately to the summed log-probability and to the Fisher statistics.  The second-order tooling of \S\,\ref{ssec:kfac} aggregates those contributions over the route-position repeat axis into the shared policy-parameter blocks. \emph{This lets the model learn to route by the same second-order natural gradient that trains the wavefunction itself, without any extra added optimization machinery to calculate such gradients.}

\subsubsection{Training the route policy: the score-function gradient}
\label{ssec:rl-routes}

The route is discrete, so its derivative is estimated with the score-function identity \citep{williams1992simple}, whereas the conditional wavefunction retains the standard VMC derivative. The sampling rule reads hierarchically as,
\begin{equation}
  p\sim\pi_\theta,
  \qquad
  q\mid p\sim\rho_{\theta,p},
  \qquad
  \rho_{\theta,p}(q)
    :=\frac{|\psi_\theta(q\mid p)|^2}{Z_{\theta,p}}.
  \label{eq:rl-hierarchical-law}
\end{equation}
Here the fixed Hamiltonian-conditioning arguments of $\pi_\theta$ are suppressed, and $Z_{\theta,p}:=\int|\psi_\theta(q\mid p)|^2\,\mathrm{d}q$ is the squared norm that makes $\rho_{\theta,p}$ a probability density. Thus the route and walker are not statistically independent. They are sampled by separate procedures: the policy first draws $p$, after which the MCMC kernel targets the routed density $\rho_{\theta,p}$. The procedures do not share transition steps, but the walker's target does however depend on the route.

For a fixed route, we can define
\begin{equation}
  E_{\mathrm{loc}}^{(p)}(q)
    :=\frac{(\hat H\psi_\theta(\cdot\mid p))(q)}
            {\psi_\theta(q\mid p)},
  \qquad
  \bar E^{(p)}
    :=\mathbb E_{q\sim\rho_{\theta,p}}
      [E_{\mathrm{loc}}^{(p)}(q)],
  \qquad
  E:=\mathbb E_{p\sim\pi_\theta}[\bar E^{(p)}].
  \label{eq:rl-conditional-energy}
\end{equation}
Differentiating the two factors gives
\begin{equation}
\begin{aligned}
  \nabla_\theta E
  ={}&
  \mathbb E_{p\sim\pi_\theta}
  \mathbb E_{q\sim\rho_{\theta,p}}
  \left[
    \bigl(E_{\mathrm{loc}}^{(p)}(q)-\bar E^{(p)}\bigr)
    2\,\mathrm{Re}\,\nabla_\theta\log\psi_\theta(q\mid p)
  \right] \\
  &+
  \mathbb E_{p\sim\pi_\theta}
  \left[
    \bigl(\bar E^{(p)}-E\bigr)
    \nabla_\theta\log\pi_\theta(p)
  \right].
\end{aligned}
  \label{eq:rl-two-channel}
\end{equation}
The first line is the ordinary VMC gradient for each routed state and must be centred at that route's own mean $\bar E^{(p)}$. Replacing it by the global mean $E$ changes the finite-sample estimator and introduces a systematic cross-route term. The second line is the REINFORCE gradient from the policy-gradient theorem \cite{sutton1998reinforcement}. The following construction supplies robust energy scales and self-normalised importance sampling for this estimator.

Let $e_s$ collect the local energies used by one policy update, over its route and walker indices, and let $B$ be their total count. We first form the batch mean and mean (absolute) total variation (TV),
\begin{equation}
  \mu_1:=\frac1B\sum_s e_s,
  \qquad
  \mathrm{TV}_1:=\frac1B\sum_s|e_s-\mu_1|.
\end{equation}
We then clip before estimating the standard deviation,
\begin{equation}
  \begin{aligned}
  e_s^{\mathrm{clip}}
    &:=\operatorname{clip}
      (e_s,\mu_1-5\,\mathrm{TV}_1,\mu_1+5\,\mathrm{TV}_1),\\
  \bar e^{\mathrm{clip}}
    &:=\frac1B\sum_s e_s^{\mathrm{clip}},\qquad
  \sigma_1
    :=\sqrt{
      \frac1B\sum_s
      \bigl(e_s^{\mathrm{clip}}-\bar e^{\mathrm{clip}}\bigr)^2
      +\varepsilon_\sigma}.
  \end{aligned}
  \label{eq:rl-mad-clip}
\end{equation}
Here $\operatorname{clip}(x,a,b)$ clamps $x$ to $[a,b]$, and $\varepsilon_\sigma>0$ is the numerical floor under the variance. We use the shared normalised reward $R_s:=-(e_s^{\mathrm{clip}}-\bar e^{\mathrm{clip}})/\sigma_1$. We write $R(p,q)$ for this same normalised reward evaluated at walker $q$ under route $p$. The clipping therefore protects the scale estimate itself (since $\sigma_1$ is not computed from the unclipped tail).

Beam search is used to approximate the policy mode,
\begin{equation}
  p_0:=\arg\max_{p\in\mathcal B_{\mathrm{beam}}}
       \log\pi_\theta(p).
  \label{eq:rl-beam-mode}
\end{equation}
Here $\mathcal B_{\mathrm{beam}}$ is the set of complete routes retained by the beam search. For a sampled action $p_i$ with $B_w$ walkers $q_{ib}\sim\rho_{\theta,p_i}$, indexed by $b=1,\ldots,B_w$, the mode expectation is estimated on those walkers with self-normalised weights,
\begin{equation}
  w_{ib}
    :=\frac{|\psi_\theta(q_{ib}\mid p_0)|^2}
            {|\psi_\theta(q_{ib}\mid p_i)|^2},
  \qquad
  \bar w_{ib}:=\frac{w_{ib}}{\sum_{b'=1}^{B_w}w_{ib'}},
  \qquad
  \widehat b_i^{\mathrm{mode}}
    :=\sum_{b=1}^{B_w}\bar w_{ib}
      R(p_0,q_{ib}).
  \label{eq:rl-mode-snis}
\end{equation}
The unknown normalising constants cancel under self-normalisation. The estimator depends on the sampled action through the proposal density in the denominator. It is biased at finite $B_w$, but it is consistent and asymptotically unbiased as $B_w\to\infty$ under the usual support and finite-weight conditions.

With $K$ the number of sampled routes in the policy update, $\bar R_i:=B_w^{-1}\sum_b R(p_i,q_{ib})$, and $J_{\mathrm{route}} :=\mathbb E_{p\sim\pi_\theta} \mathbb E_{q\sim\rho_{\theta,p}}[R(p,q)]$ the expected routed reward, the policy channel is estimated by
\begin{equation}
  \widehat{\nabla_\theta J}_{\mathrm{route}}
  =\frac1K\sum_{i=1}^{K}
    \Bigl(\bar R_i-
      \widehat b_i^{\mathrm{mode}}\Bigr)
    \nabla_\theta\log\pi_\theta(p_i).
  \label{eq:rl-final-estimator}
\end{equation}
All $\log\pi_\theta(p_i)$ values are evaluated by the teacher-forced autoregressive scoring pass of Eq.~\eqref{eq:route-logprob}.

\emph{Cost.}\;
Table~\ref{tab:rl-cost} gives the compute distribution per system per training step.
\begin{table}[H]
\centering
\small
\begin{tabular}{p{0.29\textwidth}p{0.14\textwidth}p{0.43\textwidth}}
\hline
\textbf{Operation} & \textbf{When} & \textbf{Cost} \\
\hline
ancestral dense candidate core      & $K$ routes & $\mathcal{O}(K\hat N^3d_h)$ arithmetic \\
sampled prefix-tree recomputation   & $K$ routes & $\mathcal{O}(K\hat N^4d_e)$ in slot count \\
beam baseline, width $K_{\mathrm{beam}}$ & every step & $K_{\mathrm{beam}}$ times the sampled-route costs \\
teacher scoring, dense core + one tree & $K$ routes & $\mathcal{O}(K\hat N^3(d_h+d_e))$ arithmetic, batched \\
conditional quotient, one route row & every row & $\mathcal{O}(R\hat N^3)$ (1-WL); $\mathcal{O}(R\hat N^4)$ (2-FWL) \\
advantage + policy KFAC factor      & every step                          & $\mathcal{O}(\lvert \text{policy params} \rvert)$ \\
\hline
\end{tabular}
\caption{Per-step cost of the route policy machinery. $\hat N$ is the routed slot count; $d_h$ and $d_e$ are fixed feature widths; $K = 8$ is the per-physical-system route count, and $K_{\mathrm{beam}} = 16$ the beam width in the reported checkpoint.  The repeated sampled-tree bound comes from rerunning a cubic same-level 2-FWL tree scan at each route row.  In the quotient row, $R$ is the refinement-round limit ($R=\hat N$ by default); its per-row cost must additionally be multiplied by the route rows and routes actually evaluated.  The compiled checkpoint maps the eight sampled routes to separate lanes, which changes elapsed critical path but not total arithmetic.}
\label{tab:rl-cost}
\end{table}
The teacher schedule does not reduce the arithmetic order: it exposes the $K\hat N$ scoring rows to masked batched kernels.  Sampling and beam search remain sequential in route position, whereas the teacher-forced scoring that carries the policy gradient is vectorised over rows and routes.

\subsubsection{An entropy floor for the optimal route policy}
\label{sssec:entropy-floor}

We close the routing story with an aside. Nothing in this subsubsection is needed to build or train the architecture. We include it because it is a sharp theoretical observation about what routing can achieve, and because it doubles as a powerful sanity check on the implementation.

We characterise the invariant policies through automorphism orbits and identify the minimum entropy among invariant policies supported on energy-minimising routes. Here we will show that the result depends on both route-orbit sizes, their stabilisers, and the group size.

Define the Hamiltonian's automorphism group as the permutations of the slot block that preserve $J$, $h$, and $m$ simultaneously:
\begin{equation}
  \operatorname{Aut}(J,h,m)
  \;:=\;
  \Bigl\{\, g \in S_{\hat N} \;:\;
    J_{g \cdot t,\, g \cdot u} = J_{t, u},\;\;
    h_{g \cdot t} = h_t,\;\;
    m_{g \cdot t} = m_t
    \;\;\forall\, t, u \,\Bigr\}.
  \label{eq:rl-aut-def}
\end{equation}
$\operatorname{Aut}(J,h,m)$ is finite and always contains the identity. Notice that it also always contains the $(\hat N - N)!$ permutations of the holes among themselves: hole rows of $J$ and $h$ are zero and the mask is preserved, so exchanging holes changes nothing the network can see. On top of this padding symmetry sit the physical ones. For a Heisenberg ring, $\operatorname{Aut}(J,h,m)$ contains the lattice translations and reflections; for an irregular disordered system, it contains nothing more than the hole exchanges.

\emph{Gauge equivalence under automorphism.}\;
For every $g \in \operatorname{Aut}(J,h,m)$ the routed wavefunction satisfies
\begin{equation}
  \psi_\theta(q \mid p)
  \;=\;
  \psi_\theta(g \cdot q \mid g \cdot p) \cdot e^{i \varphi(g)},
  \label{eq:rl-gauge}
\end{equation}
up to a global phase $\varphi(g)$ that depends only on $g$. The proof is by induction up the readout of \S\,\ref{ssec:tree-readout}: every Hamiltonian-side input, contexts, edges, clocks, masks, is built equivariantly from $(J, h, m)$, which $g$ preserves, so the carrier path evaluates the same contraction on relabelled inputs. Consequently the routes $p$ and $g \cdot p$ yield the same observable energy at every walker: $E_{\mathrm{loc}}^{(p)}(q) = E_{\mathrm{loc}}^{(g \cdot p)}(g \cdot q)$.

\emph{Optimal policy invariance.}\;
Among unbiased policies, the variance-minimising one is invariant under the $\operatorname{Aut}(J,h,m)$ action:
\begin{equation}
  \pi_\theta^*(p) \;=\; \pi_\theta^*(g \cdot p)
  \quad
  \forall\, g \in \operatorname{Aut}(J,h,m).
  \label{eq:rl-policy-invariance}
\end{equation}
The argument is symmetrisation. Take any non-invariant optimum $\nu$ and average it over the group, $\bar\nu(p) := \lvert \operatorname{Aut}(J,h,m)\rvert^{-1} \sum_g \nu(g^{-1} \cdot p)$. By \eqref{eq:rl-gauge} the average has the same expected reward, and by Jensen's inequality on the variance functional it has weakly smaller gradient variance.

The policy entropy decomposes over route orbits as follows. Let $G:=\operatorname{Aut}(J,h,m)$ act on the route set, and let $\{\mathcal O_k\}$ be its route orbits. Write $o_k:=|\mathcal O_k|$ and let $\omega_k$ be the total policy mass assigned to orbit $k$. A $G$-invariant policy is uniform within each occupied orbit, so, with $H[\pi]:=-\sum_p\pi(p)\log\pi(p)$ and $H(\omega):=-\sum_k\omega_k\log\omega_k$ denoting the corresponding Shannon entropies,
\begin{equation}
  H[\pi]
  =-\sum_k\omega_k\log\frac{\omega_k}{o_k}
  =H(\omega)+\sum_k\omega_k\log o_k.
  \label{eq:rl-orbit-entropy}
\end{equation}
By orbit--stabiliser,
\begin{equation}
  o_k
  =\frac{|G|}{|\operatorname{Stab}_G(p_k)|},
  \qquad p_k\in\mathcal O_k.
  \label{eq:rl-orbit-stabiliser}
\end{equation}
where $\operatorname{Stab}_G(p_k):=\{g\in G:g\!\cdot\!p_k=p_k\}$ is the stabiliser subgroup of the representative route $p_k$. Let $\mathfrak O_\star$ be the set of orbits containing energy-minimising routes. Because the energy objective is linear in the policy, there exists an optimal $G$-invariant policy supported uniformly on one orbit $\mathcal O_\star\in\mathfrak O_\star$. Among such optimal invariant policies, the minimum entropy is
\begin{equation}
  H_\star(J,h,m)
  =\min_{\mathcal O\in\mathfrak O_\star}\log|\mathcal O|,
  \label{eq:rl-floor}
\end{equation}
attained on a smallest energy-minimising orbit. Spreading mass across several optimal orbits adds the nonnegative mixing entropy $H(\omega)$.

Equation~\eqref{eq:rl-floor} is an orbit-specific optimum, not the universal quantity $\log|G|$ unless the action on the relevant route is free, meaning that its stabiliser is the identity subgroup. Whenever $|\mathcal O_\star|<\hat N!$, the explicit single-orbit construction yields a non-uniform optimum; other optima may be uniform.

\section{Datasets and further detail on training}
\label{sec:datasets}

We pretrain the foundation ansatz on a corpus of quadratic spin Hamiltonians, spanning a century of many body literature \cite{anderson1958absence,lieb1961two,hubbard1963electron,majumdar1969next,ising1925beitrag,su1979solitons,rice1982elementary,kitaev2001unpaired,chen2011classification,sachdev1999quantum,sandvik2007evidence,senthil2004deconfined,monroe2021programmable,browaeys2020many,hensgens2017quantum,heras2014digital,kitaev2006anyons,anderson1973rvb,balents2010spinliquids,savary2016qsl,basko2006mbl,nandkishore2015mbl,abanin2019mbl,page1993average,lipkin1965validity,greenberger1989ghz,wilson1974confinement,kogut1975hamiltonian,kogut1979lattice,jordan1928pauli,bravyi2002fermionic,tranter2018comparison,goldstone1961field,nambu1960quasiparticles,auerbach1994interacting,turner2018scars,bernier2017rydberg,aklt1988vbs,yang1962odlro,sachdev1993gapless,kitaevsuh2018soft,maldacena2016remarks,dzyaloshinsky1958weak,moriya1960anisotropic,kugel1982jahn,nussinov2015compass}. In this section, we detail the core composition of this dataset in Sec.~\ref{ssec:panel-construction}, and then walk through in detail the augmented training and per-round perturbation and held-out evaluation datasets that were mentioned in the main text.

The pretraining data uses a quality-balanced mixture rather than empirical family frequencies. Rare regimes are over-represented, and each system's neighbourhood is densified by the perturbative cover as training proceeds. During pretraining, curvature damping is set as low as the natural-gradient norm allows (SM~\S~\ref{sec:part4}). Model selection uses a held-out out-of-distribution dataset, followed by evaluation on a separate unseen dataset.

\subsection{Pretraining data core composition}
\label{ssec:panel-construction}

Our pretraining data's core is a finite set of quadratic spin Hamiltonians $\{(J^{(s)}, h^{(s)})\}_{s=1}^{S}$ in the Pauli-pair representation of SM~\S~\ref{sec:part2}. There $J$ is the $N \times N \times 3 \times 3$ bond tensor carrying every Pauli-pair channel, and $h$ is the per-site field. We group the systems first by where they commonly arise in the  physics literature and refer to each group as a cell.  Within a cell the Hamiltonians range over size, interaction graph, coupling channel, field and anisotropy.

Construction is deterministic. A parametrised builder emits every
system from a seed committed to the repository, and the exact systems can be rebuilt from source.

The systems are grouped into physics cells, with three further
coverage tiers\footnote{the dataset employed is publically evailable at \url{ https://github.com/simulacra-research/HamiltonZero/blob/main/datasets/train/foundation_5000.jsonl}}. The defining property of each cell's Hamiltonians are briefly given below. 
\begin{itemize}
  \item \emph{Canonical:} Heisenberg, XXZ, transverse-field Ising and
        Majumdar--Ghosh on regular lattices, through their exactly
        solvable points.
  \item \emph{Topological / SPT:} SSH, Rice--Mele, the Kitaev chain and
        Haldane, cluster and CZX proxies across topological transitions.
  \item \emph{Critical-point fans:} non-Ising quantum critical points,
        the $2$D XXZ, $J$--$Q$ and Ising--Dzyaloshinskii--Moriya families.
  \item \emph{Hardware-native:} Hamiltonians native to superconducting,
        trapped-ion, Rydberg, spin-qubit and molecular platforms.
  \item \emph{Frustration / spin liquids:} kagome, triangular,
        Kitaev-honeycomb, $J_1$--$J_2$ and Shastry--Sutherland magnets.
  \item \emph{Disorder / MBL:} random-field and random-bond chains
        across the many-body-localisation transition.
  \item \emph{Extremal entanglement:} volume-law ground states, from
        random Hermitians to the Lipkin--Meshkov--Glick, Page and GHZ
        parents.
  \item \emph{Lattice gauge:} $\mathbb{Z}_2$ and U($1$) gauge theories
        in two-body form.
  \item \emph{Quantum chemistry / Hubbard:} Fermi--Hubbard and molecular
        Hamiltonians through the Jordan--Wigner map.
  \item \emph{Goldstone:} gapless symmetry-broken magnets with
        long-wavelength modes.
  \item \emph{Many-body scars:} athermal eigenstate towers, the PXP,
        AKLT and $\eta$-pairing families.
  \item \emph{SYK / random graphs:} Hamiltonians with no lattice
        structure, dense and sparse SYK and random-graph couplings.
  \item \emph{Cross-channel $J$:} off-diagonal Pauli couplings, the
        Dzyaloshinskii--Moriya, compass and per-bond-random families.
  \item \emph{Coverage tiers:} an easy low-correlation band, a
        space-filling random-Hamiltonian cover of the small-$N$ regime,
        and graphs that defeat the routing quotient of
        SM~\S~\ref{sec:part2}.
\end{itemize}
Fig.~\ref{fig:panel-v9-size} reports how many systems each cell contributes to the dataset.

\definecolor{cellCanon}{HTML}{F16F7E}
\definecolor{cellSYK}{HTML}{5742AA}
\definecolor{cellFrust}{HTML}{AE9905}
\definecolor{cellSPT}{HTML}{9B204A}
\definecolor{cellCross}{HTML}{DE96FD}
\definecolor{cellHW}{HTML}{268B2D}
\definecolor{cellMBL}{HTML}{16C1FE}
\definecolor{cellScars}{HTML}{BD4D00}
\definecolor{cellQCP}{HTML}{7F84F1}
\definecolor{cellExtremal}{HTML}{FCAA2B}
\definecolor{cellChem}{HTML}{0275B3}
\definecolor{cellGauge}{HTML}{3ED691}
\definecolor{cellGoldstone}{HTML}{A34CA4}
\definecolor{cellCover}{HTML}{DCDCDC}
\definecolor{cellEasy}{HTML}{888888}
\definecolor{cellWL}{HTML}{3C3C3C}
\newcommand{\swatch}[1]{%
  \begingroup\setlength{\fboxsep}{0pt}\setlength{\fboxrule}{0.2pt}%
  \fcolorbox{black!45}{#1}{\rule{0pt}{1.15ex}\rule{1.15ex}{0pt}}\endgroup}

\newcommand{\cellswatch}[1]{%
  \begingroup
  \setlength{\fboxsep}{0pt}%
  \setlength{\fboxrule}{0.2pt}%
  \fcolorbox{black!45}{#1}{%
    \rule{0pt}{1.15ex}\rule{1.15ex}{0pt}}%
  \endgroup
}

\begin{figure*}[!t]
\centering

\begin{minipage}[c]{0.48\textwidth}
  \centering
  \small
  \begin{tabular}{@{}c l r@{}}
    \hline
      & Cell & Systems \\
    \hline
    \cellswatch{cellCanon}
      & Canonical & 1051 \\
    \cellswatch{cellSYK}
      & SYK \& random graphs & 396 \\
    \cellswatch{cellFrust}
      & Frustration \& spin liquids & 394 \\
    \cellswatch{cellSPT}
      & Topological \& SPT & 322 \\
    \cellswatch{cellCross}
      & Cross-channel $J$ & 297 \\
    \cellswatch{cellHW}
      & Hardware-platform-native & 391 \\
    \cellswatch{cellMBL}
      & Disorder \& MBL & 215 \\
    \cellswatch{cellScars}
      & Quantum many-body scars & 509 \\
    \cellswatch{cellQCP}
      & Critical-point fans & 153 \\
    \cellswatch{cellExtremal}
      & Extremal entanglement & 86 \\
    \cellswatch{cellChem}
      & Quantum chemistry \& Hubbard & 79 \\
    \cellswatch{cellGauge}
      & Lattice gauge & 42 \\
    \cellswatch{cellGoldstone}
      & Continuous symmetry \& Goldstone & 17 \\
    \hline
    \multicolumn{3}{@{}l}{\emph{Coverage tiers}} \\
    \cellswatch{cellCover}
      & Small-$N$ random cover & 664 \\
    \cellswatch{cellEasy}
      & Simple systems & 287 \\
    \cellswatch{cellWL}
      & WL-$1$-breaking & 97 \\
    \hline
      & \textbf{Total} & \textbf{5000} \\
    \hline
  \end{tabular}
\end{minipage}
\hfill
\begin{minipage}[c]{0.48\textwidth}
  \centering
  \includegraphics[width=\linewidth]{%
    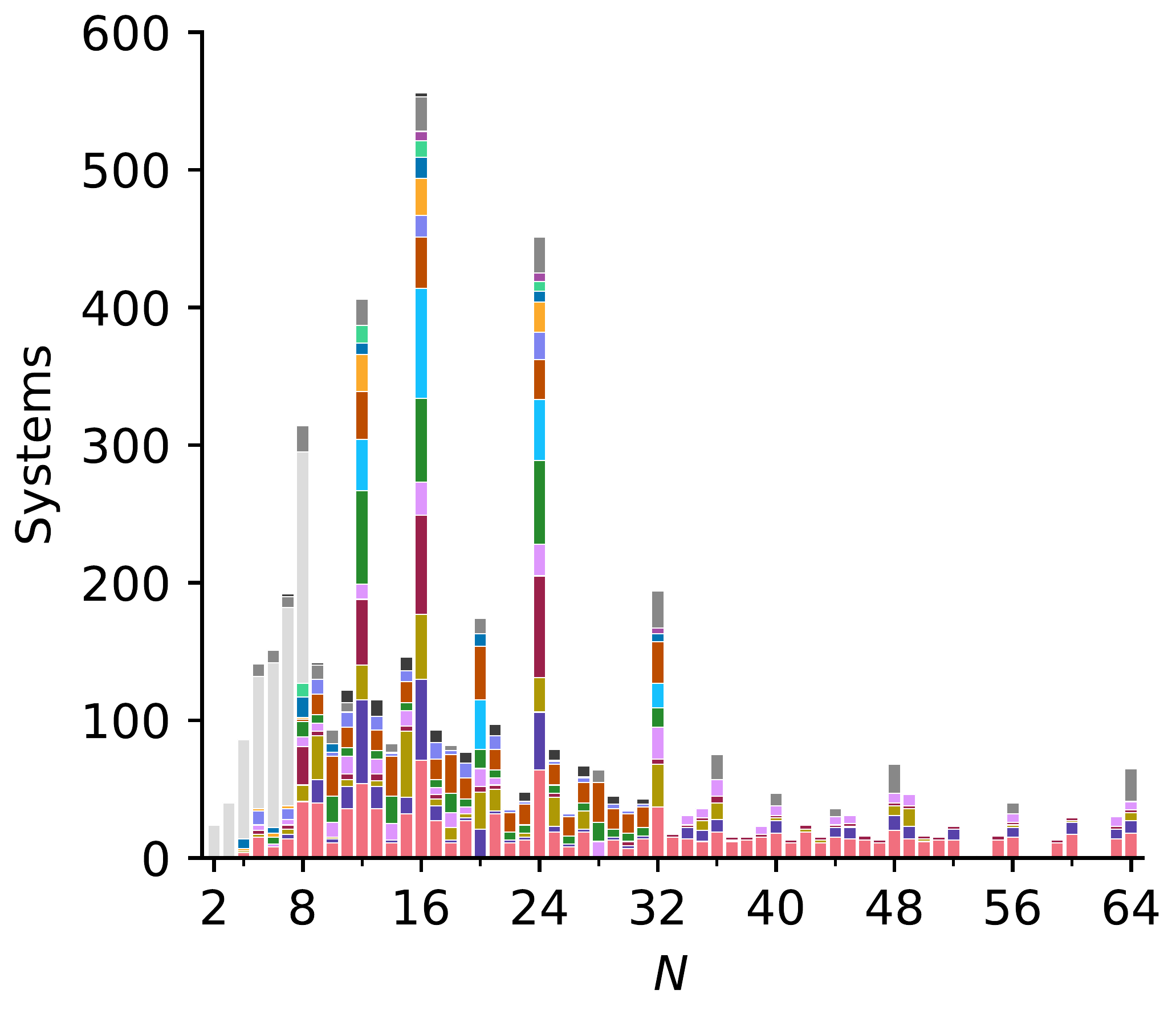}
\end{minipage}

\caption{
Composition and system-size distribution of the 5,000-system pretraining dataset
Left: the number of systems contributed by each physics cell and coverage tier. Right: the number of systems at each spin count $N$, stacked by cell. Colours match the table colours and follow the table
order from the bottom of each bar upwards.}
\label{fig:panel-v9-size}
\end{figure*}


\subsection{Augmented training by per-round perturbation}
\label{ssec:augmented}

Every datapoint in the pretrain set is a single set of edge values on one interaction graph $(J,h)$. Trained on that point alone, the model can overfit to the exact edge values rather than the physics they carry. Therefore in our pretraining, we augment the training data with a perturbation routine that we detail in this subsection.

Each round, we perturb the edge values of every system independently, so the optimiser sees a \emph{neighbourhood} of values around each datapoint and averages over it instead of fitting one point. A round is one such perturbed version of the whole training dataset, swapped in at an epoch end. Round 0 is an unperturbed base, and the run settles back on it after having trained for a few epochs on the perturbed data.

For each structured system and round, a subroutine draws one of five mutations with equal probability: (i) leave the system unchanged,
(ii) add one bond, 
(iii) perturb one existing bond, 
(iv) remove one bond,
or (v) add one orbit-symmetric random field. Sites, bonds, and orbits are sampled stratified over the Weisfeiler--Leman orbit partition of $(J,h)$ \cite{babel2010algebraic, faradzev2013investigations}. In case (v), one random three-vector of magnitude $\eta s$, with $\eta\sim\mathcal U(0.25,1)$ and $s$ the system's coupling scale, is added to every site in the selected orbit. For small random systems with $N\le\dsDensifyN$, the worker additionally densifies towards all-to-all coupling; a field-free system receives a substantial random dense field with probability $1/2$, together with occasional bond drops.

Speculative re-gauging is evaluated after the MCMC update, when the routed walker population is already available. It applies
\begin{equation}
  J_{ij}\mapsto R_iJ_{ij}R_j^\top,
  \qquad h_i\mapsto R_ih_i,
  \qquad q_i\mapsto q_i\bar u_i,
  \label{eq:su2-gauge}
\end{equation}
with one $u_i$ per Weisfeiler--Leman orbit, $\bar u_i:=u_i^{-1}$, and $R_i:=R_{\mathrm{ker}}(u_i)\in\mathrm{SO}(3)$ the adjoint rotation expressed in the kernel frame. The transformed Hamiltonian is unitarily equivalent to the original and the transformed walkers target its routed density exactly.

We emphasise that a round differs from the base in the values of $(J,h)$ and nothing else. As such, the perturbation is cheap, because it changes only the values in $(J,h)$, meaning shapes remain static, and thus we avoid any recompiling. 

To make sure none of this stalls training, we
make the routine CPU-based such that it executes the perturbation run asynchronously. The routine stays \dsLookahead\ rounds in front of the inner training loop. For each upcoming round it perturbs the pretrain dataset and precomputes the reference data, the exact-diagonalisation energies to $N \le \dsEDmax$ and the per-system routing fields, so the next round is ready before the loop reaches it. The loop consumes one round per epoch and, at each epoch boundary, swaps in the next round if the routine has produced it. The look-ahead is the margin that absorbs the time the perturbation and the diagonalisations take, so the training loop has zero latency with respect to this routine. Figure~\ref{fig:augmented} shows how the two processes run asynchronously with zero latency.

\begin{figure}[!htbp]
  \centering
  \input{augmented-training-hotswap.tikz}
  \caption{Augmented training and asynchronous hot-swap. The CPU samples one of five system perturbations, builds its references and routing fields $\dsLookahead$ rounds ahead, and caches the result. At each epoch end the GPU swaps in the next cached round. Post-MCMC speculative re-gauging reuses the current walkers and adds no measured overhead to the inner loop.}
  \label{fig:augmented}
\end{figure}
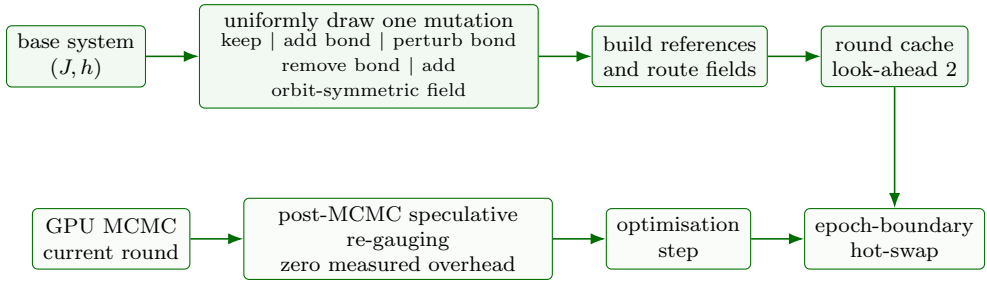

The cover is local. Each round perturbs a system within a neighbourhood of its base, so augmentation fills in the surroundings of the curated physics rather than reaching for new phases. The critical points the model must reproduce are placed in the pretraining data itself, in the fans of \S\,\ref{ssec:panel-construction}, since they are not reached by the perturbation in general.

The exact diagonalisations the worker computes are a second payoff. Training is self-supervised and needs no reference energies, so they serve only as a convergence diagnostic, and because each is valid only within its own round, the gap to exact is read per round. Yet each round contributes \dsEDcount\ of them at $N \le \dsEDmax$, and across a run they accumulate into a corpus of tens of thousands of exactly solved systems, a byproduct of the cover rather than a cost of it.

\subsection{Evaluation and fine-tuning datasets}
\label{ssec:v4-eval}

In this section, we describe the three datasets we used to investigate how well Hamilton-Zero generalises to unseen interaction topologies and interaction types. To do this, we kept several types Hamiltonian families out of training, using them only here for evaluation and fine-tuning.

First in the set of three is Hamiltonian's corresponding to combinatorial optimisation problems. These are Hamiltonians whose bonding structure stays inside the training distribution (graph Ising, $H = \sum_{(i,j) \in E} J_{ij} \sigma^z_i \sigma^z_j + \sum_i h_i \sigma^z_i$) but whose \emph{problem semantics} and interaction topologies are absent from training. Every Hamiltonian in this set encodes an NP-hard combinatorial optimisation problem whose ground state is the problem's optimal solution. We emphasise that unlike many systems in training and the other evaluation sets, such solutins correspond to \textit{separable} ground states, so the difficulty here not in the entanglement of ground states, but rather their single-direction support in an otherwise exponentially large Hilbert space. The sub-families are listed in Table~\ref{tab:eval-axisA}.

\begin{table}[H]
\centering
\small
\begin{tabular}{@{}p{5.0cm} r p{8.6cm}@{}}
\hline
Sub-family & \#cells & Graph / structure \\
\hline
MaxCut on random regular graphs & 20 & $k$-regular ($k \in \{3,5\}$),
BA scale-free, planted-cut \\
MaxCut on Erd\H{o}s--R\'enyi & 8  & Density $p \in \{0.4, 0.6\}$ \\
QUBO (dense + sparse) & 16 & $J_{ij} \in \{\pm 1\}$,
$h_i \sim \mathcal{N}(0,1)$;
dense $K_N$ or sparse-ER \\
Number partitioning  & 10 & Rank-$1$ $J = v\,v^\top$,
$v_i \sim \text{Uniform}$ or power-law \\
Max-2-SAT            & 10 & Random clauses at clause-to-variable
ratio $\alpha \in \{0.8, 2\}$ \\
Max independent set / min vertex cover & 14 & Penalty form
$H = \sum_v -\sigma^z_v + \lambda
\sum_{(i,j) \in E}
(\sigma^z_i + 1)(\sigma^z_j + 1)/4$ \\
Graph 2-colouring (adversarial) & 8  & $5$-regular graphs with many
odd cycles \\
Hamiltonian-path proxy & 6 & $2$-body approximation of the
$4$-body Lucas encoding \\
\hline
\end{tabular}
\caption{Table of systems for the first evaluation set}
\label{tab:eval-axisA}
\end{table}

The second evaluation set covers unseen topologies and interaction types, including fractal, quasi-periodic, and higher-dimensional interaction geometries such as hypercubes. Its subfamilies are listed in Table~\ref{tab:eval-axisB}.

\begin{table}[H]
\centering
\small
\begin{tabular}{@{}p{5.0cm} r p{8.6cm}@{}}
\hline
Sub-family & \#cells & Graph \\
\hline
Fractal: Sierpinski + T-fractal & 14 & Hausdorff dim
$\log 3 / \log 2 \approx 1.585$
(Sierpinski); branching tree
(T-fractal) \\
Penrose / Fibonacci quasi-periodic & 10 & Fibonacci chain $\pm$
  pentagon ring \\
Decorated lattices  & 16 & Lieb (flat band), dice ($T_3$, AB caging),
square-octagon ($4.8.8$) \\
3D high-coordination & 10 & FCC (frustrated), BCC (bipartite),
diamond \\
4D hypercube $N{=}16$  & 8  & $K_2^4$; AFM / FM / XXZ / TFIM \\
\textbf{5D hypercube $N{=}32$} & 4 & $K_2^5$, all out of scope for ED;
the headline extrapolation cell \\
Hyperbolic $\{7,3\}$ & 10 & Heptagon $r=1$ + $r=2$ patches
(2 systems out of scope for ED) \\
M\"obius / complete-bipartite & 7 & M\"obius-twisted chain,
$K_{a,b}$ \\
\hline
\end{tabular}
\caption{The second evaluation set containing unseen interaction topology and unseen interaction channels. The training dataset has no fractal, no quasi-periodic-2D, no decorated, no 4D/5D hypercube, and no hyperbolic lattices.}
\label{tab:eval-axisB}
\end{table}

The third evaluation set contains Hamiltonians with physical structures absent from training: multi-spin-$1/2$-per-physical-site encodings, complex hoppings, and a gauge theory. Its subfamilies are listed in Table~\ref{tab:eval-axisC}.

\begin{table}[H]
\centering
\small
\begin{tabular}{@{}p{4.0cm} r p{9.6cm}@{}}
\hline
Sub-family & \#cells & Encoding \\
\hline
SU($3$) ULS chain (PROXY) & 14 & $2$ spin-$1/2$ per physical site;
 FM intra-bond $\rightarrow S{=}1$;
 Heisenberg + ZZ-boost approximation
 of the bilinear-biquadratic
 spin-$1$ chain. \\
Kugel--Khomskii spin-orbital (PROXY) & 14 & $2$ spin-$1/2$ per site
(spin + orbital);
factorised as $2$-body
$\mathbf{S}\!\cdot\!\mathbf{S}
+ \boldsymbol{\tau}\!\cdot\!\boldsymbol{\tau}
+ $ cross-Ising (PROXY
for the native $4$-body
$(\mathbf{S}\!\cdot\!\mathbf{S})
 (\boldsymbol{\tau}\!\cdot\!\boldsymbol{\tau})$). \\
Hofstadter strip & 24 & $2$/$3$/$4$-leg ladders with complex hop
$t \to t\,e^{i\phi}$ $\Rightarrow$ DM$_z$
terms in real spin language; flux $\alpha$
per plaquette (incl.\ golden-mean
quasi-periodic). Native $2$-body.\\
Spin-$1$ XXZ chain (encoded) & 12 & Bilinear-only spin-$1$ chain via
$2$ spin-$1/2$; $\Delta$ sweep
covers Haldane / AKLT-like /
easy-axis. Bethe-ansatz
integrable reference points. \\
$1$D Hubbard JW & 10 & Lieb--Wu integrable line $U=4$,
Mott $U=8$, weakly correlated $U=2$
(truncated-JW PROXY). \\
Schwinger $\theta$-vacuum & 11 & $1$D U($1$) gauge after staggered
 + JW + Gauss-law; $\theta \in
 \{0, \pi/2, \pi\}$ sweep; Coleman
 transition at $\theta=\pi$.
 Native $2$-body. \\
\hline
\end{tabular}
\caption{The third and hardest evaluation dataset.}
\label{tab:eval-axisC}
\end{table}

\subsubsection{Evaluation and fine-tuning configuration}
\label{ssec:eval_finetune_config}

For evaluation, we let the pretrained router select the tree ordering (eight candidate trees
decoded at temperature $\tau=1$ from a width-64 beam, each raced over a short burn-in and measurement window before collapsing to the winner), compile the wavefunction once on that ordering, and sample the fixed model for 512 measurement steps with 256 walkers across eight tempered replicas and 24 adaptive Langevin moves per step, retaining
$512\times256\times8\simeq1.05\times10^{6}$ samples per system. Energies are reported with the pooled standard deviation of per-walker means as the uncertainty.

fine-tuning restarts from the same checkpoint but commits the route once, at initialisation, with a width-16 beam search: the wavefunction is compiled on that fixed ordering and the eager model and router are discarded, leaving a compiled ansatz of one shared kernel plus per-spin tree weights ($\sim$4.7M parameters, against 547M for the full model) as the object being trained. Each system is optimised independently on a single A100 for $10^{4}$ KFAC steps with 256 walkers, four MCMC moves per step and a 128-step burn-in, with learning rate $0.01/(1+t/10^{4})$ (ending at $\approx0.005$), damping $10^{-3}$, curvature EMA $0.9$ and curvature/inverse updates every 2/4 steps. After training, every system is re-evaluated with frozen weights: the final checkpoint is resumed at exactly zero learning rate for a further 512 measurement steps with the sampler state carried over, the freeze is certified per system by a byte-identical model checkpoint written inside the evaluation window, and the standard error carries an AR(1) autocorrelation correction. ED references are never supplied to evaluation or fine-tuning and enter only
in the analysis of the main text.

\section{Optimization and Engineering}
\label{sec:part4}

In this section, we detail our state-of-the-art sampling scheme for MCMC
on the configuration manifold $SU(2)^N$, as well as our energy kernels, our extension of the KFAC optimizer to the trunk's stacked weights and the rank-4 merge tensor, all of which we made shard-mappable and GEMM-optimised.  We take them in the order the data flows; the sampler produces walker configurations, then energy kernels evaluate the energy and its gradient on them, and finally the KFAC extensions consume both.

\subsection{MCMC on the SU(2) manifold}
\label{ssec:mcmc}

Improving the energy and updating the weights both require samples drawn from the model's own density $|\psi_{\theta}(q)|^{2}$ on $(\mathrm{SU}(2))^{N}$. The variational energy is the loss, and its gradient drives the optimiser, so the variance of these Monte Carlo estimates sets how noisy each training step is. That variance grows with the chain's integrated autocorrelation time and shrinks with the number of samples, so a sampler that decorrelates faster reaches a given accuracy from fewer samples. The sampler is therefore worth as much engineering as the energy kernel. Ours goes well beyond the Metropolis--Hastings (MH) algorithm \citep{metropolis1953equation, hastings1970monte} that variational Monte Carlo conventionally uses, whose limitations for electronic-structure sampling are analysed in \citep{falcioni2025benchmarking}.

MH algorithm proposes a move $q \to q'$ from a kernel $g(q' \mid q)$ and accepts it with probability $\min\bigl(1,\, |\psi_{\theta}(q')|^{2}\, g(q \mid q') / (|\psi_{\theta}(q)|^{2}\, g(q' \mid q))\bigr)$. The usual random-walk proposal forces a bad trade-off in high dimensions. Large steps land in low-density regions and are almost all rejected; small steps are accepted but barely move. Either way, the chain decorrelates slowly, the autocorrelation time grows, and the effective sample size collapses \citep{roberts1997weak}.

Langevin dynamics breaks the trade-off by drifting the proposal up the density with gradient ascent. The Langevin diffusion
\begin{equation}
  \mathrm{d} q_{t}
  \;=\; \nabla \log |\psi_{\theta}(q_{t})|^{2}\, \mathrm{d}t
        \;+\; \sqrt{2}\, \mathrm{d} W_{t}
  \label{eq:langevin-sde}
\end{equation}
has $|\psi_{\theta}|^{2}$ as its stationary density, but a finite discretisation of it does not. A single Euler--Maruyama step of size $\sigma$ gives the unadjusted Langevin proposal, and because it only approximates the diffusion, it samples a perturbed density whose bias away from $|\psi_{\theta}|^{2}$ vanishes only as $\sigma \to 0$. The Metropolis-adjusted Langevin algorithm \citep{roberts1996exponential} removes this bias exactly by passing the Langevin proposal through the MH accept step, and so keeps the drift's fast mixing while restoring $|\psi_{\theta}|^{2}$ as the stationary density. Its optimal step size scales as $\sigma \sim D^{-1/3}$ rather than the random walk's $D^{-1/2}$ \citep{roberts1998optimal}, a large gain at our dimension.

MALA alone, however, is still not enough. Each rejected move discards a gradient evaluation, which is the most expensive quantity in the step, so reducing $\sigma$ to cut rejections also shortens the steps and slows exploration. And a single local chain with a single scalar step size mixes slowly across the rugged, multimodal landscape of a frustrated or disordered Hamiltonian, where it can remain in one basin for many steps. The subsections below close these gaps in turn. First, a manifold-aware MALA proposal that respects the group geometry, then a ladder of tempered replicas \citep{swendsen1986replica, hukushima1996exchange} that exchange to carry walkers between basins, followed by a global Haar-random refresh that erases stuck sites outright, and finally online adaptations that hold each piece at its optimal operating point \citep{andrieu2008tutorial}. Assembled, the scheme cuts the chain's integrated autocorrelation time, and with it the variance of every energy and gradient estimate, so a target accuracy needs far fewer samples. And it returns batches differentiable in both the manifold coordinates and the model weights, the two derivatives the energy estimator and the natural-gradient step consume.

Concretely, the sampler runs once per training step, in three stages. A setup stage runs first and once: it hoists the per-system quantities that do not depend on the configuration $q$, the featurizer and trunk outputs, and refreshes the cached log-density $\log |\psi_{\theta}|^{2}$ and its gradient at the current $q$ against the freshly updated model. A single global Haar refresh of $m$ sites  then fires, before any local move. Finally, a loop of inter-steps runs, each one a manifold MALA move followed by a cross-chain swap whose edge set alternates with the step parity. Three adaptations run on a slower cadence and condition the next pass: a multiplicative bang-bang rule drives the MALA step size to its optimal acceptance, a proportional rule tunes the number of Haar-refreshed sites to its own optimal acceptance, and a monotone-cubic interpolation re-spaces the temperature ladder to equalise swap rejection. Figure~\ref{fig:mcmc-engine} shows this pipeline and where each subsection below fits.

\begin{figure}[!htbp]
  \centering
  \resizebox{\textwidth}{!}{\input{architecture-mcmc.tikz}}
  \caption{The per-chain MCMC pipeline at inverse temperature $\beta_{r}$ (\S\,\ref{ssec:mcmc}), run once per training step. A setup stage runs first and once: it hoists the per-system, $q$-independent featurizer and trunk outputs, then refreshes the cached log-density and its gradient at the current $q$ (\S\,\ref{sssec:mcmc-cache}). A single global Haar refresh of $m$ sites (\S\,\ref{sssec:mcmc-haar}, dashed border for its stochastic accept) fires before any local move. The inner loop then runs a fixed number of inter-steps (the loop-back on the left), each a manifold MALA move (\S\,\ref{sssec:mcmc-mala}) followed by a deterministic even/odd cross-chain swap (\S\,\ref{sssec:mcmc-deo}) whose edge set alternates with the step parity, producing $q^{(r)}_{t+1}$. Three adaptations (magenta, dashed) run on a slower cadence and condition the next pass: the refresh count $m$ is tuned to the Haar acceptance rate (\S\,\ref{sssec:mcmc-haar}), the multiplicative rule of \S\,\ref{sssec:mcmc-rm} drives the MALA step size $\sigma$ to its acceptance, and the Hyman monotone cubic (\S\,\ref{sssec:mcmc-beta}) re-spaces the inverse-temperature grid $\{\beta_{r}\}$ to the pair-rejection rates. All boxes are shades of the leaf-green hue used for the spin-configuration input in Figure~\ref{fig:architecture}, darkening from the setup stage along the loop; the colour match marks this figure as the subroutine that produces $q$.}
  \label{fig:mcmc-engine}
\end{figure}
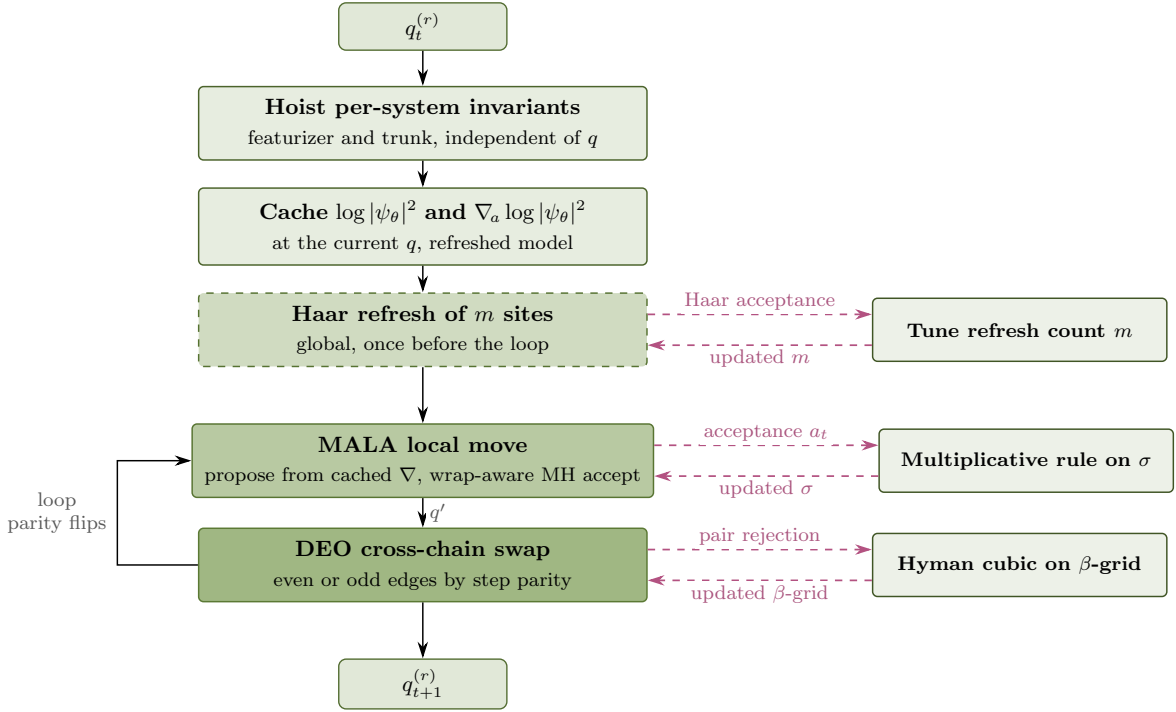

\subsubsection{Proposal density on the Lie group}
\label{sssec:mcmc-proposal}

A proposal $q \to q'$ on $\mathrm{SU}(2)$ is parameterised by a Lie algebra increment $\xi \in \mathbb{R}^3$ via
\begin{equation}
  q' = q \cdot \exp\bigl(\sigma\, \xi^a\, e_a(q)\bigr),
\end{equation}
where $\{e_a(q)\}$ is the left-invariant frame. The Gaussian density of $\xi$ on the algebra is closed-form, but the induced density on the group is not: it must be wrapped around the periodicity of the exponential map. We use the wrap sum
\begin{equation}
  p_{\text{wrap}}(\xi; \sigma)
  \;=\; \sum_{n \in \mathbb{Z}^3,\, \|n\| \le N_{\text{wrap}}}
        \mathcal{N}\bigl(\xi + 2\pi n;\, 0,\, \sigma^2 I_3\bigr),
  \label{eq:wrap-sum-prop}
\end{equation}
truncated to $N_{\text{wrap}} = 8$ shells in the experiments reported here. For $\sigma < 7\pi$, vastly larger than any practical step, the truncation error falls below double-precision noise. The wrap sum costs about $3\times$ a single Gaussian and takes under $0.3\%$ of MCMC wall time.

\subsubsection{Manifold MALA on $\mathrm{SU}(2)^{N}$}
\label{sssec:mcmc-mala}

Langevin and Hamiltonian proposals on manifolds have a developed literature \citep{girolami2011riemann, byrne2013geodesic}; the construction below specialises to the product group $\mathrm{SU}(2)^N$, where the exponential map and its wrapped proposal density are available in closed form. The MALA proposal of the introduction needs two changes on the Lie group. Vanilla addition $q + \xi$ is not a group operation, so we exponentiate a Lie-algebra increment back onto the group, $q' = q \cdot \exp(\sigma\, \xi^{a}\, e_{a}(q))$, as in the proposal density above. And the Euclidean gradient becomes the left-invariant gradient $\nabla_{\!a} \log |\psi_{\theta}|^{2} = e_{a}(q) \cdot \nabla \log |\psi_{\theta}|^{2}$ along the frame $\{e_{a}(q)\}$. The manifold-MALA proposal is then
\begin{equation}
  \xi
  \;=\; \tfrac{\sigma^{2}}{2}\, \nabla_{\!a} \log |\psi_{\theta}(q)|^{2}
        \;+\; \sigma\, \eta,
  \qquad \eta \sim \mathcal{N}(0, I_{3}).
\end{equation}
The MH correction now needs both the wrap-sum proposal density \eqref{eq:wrap-sum-prop} and the volume Jacobian of the exponential map. Either correction alone breaks detailed balance, because the wrapped Gaussian is the proposal density on the group while the Jacobian accounts for the change of measure between the algebra and the group.

\subsubsection{Replica exchange with deterministic even-odd swaps}
\label{sssec:mcmc-deo}

A single chain, however well tuned, still struggles to cross between well-separated modes, so we run a ladder of them. $R$ chains run in parallel at decreasing inverse temperatures $\beta_1 > \beta_2 > \dots > \beta_R$; the coldest replica, at $\beta_1$, sits near the modes of $|\psi_{\theta}|^{2}$, while the hottest replica, at $\beta_R=0$, explores broadly.

After each local MALA move the chain proposes swaps between adjacent replicas along the deterministic even/odd (DEO) schedule of \citet{syed2022non}, whose eligible edges alternate with the step parity. At even steps the even edges swap, $(1, 2), (3, 4), \dots$; at odd steps the odd edges swap, $(2, 3), (4, 5), \dots$; so every adjacent pair is offered a swap once in every two steps. Each pair accepts with the standard MH ratio,
\begin{equation}
  \alpha_{r,r+1}
  \;=\;
  \min\!\Bigl(1,\;
    \exp\bigl((\beta_r-\beta_{r+1})\,
      (\log|\psi_\theta(x_{r+1})|^2
       -\log|\psi_\theta(x_r)|^2)\bigr)
  \Bigr).
\end{equation}
The alternation is non-reversible. A pair swapped at an even step is not offered again until the next even step, two steps later; in between, the odd-step swap carries the walker on to its new neighbour rather than back. This directed transport along the ladder is what maximises information flow between the hot and cold chains \citep{syed2022non}. Figure~\ref{fig:replica-deo} shows the swap pattern over a few steps.

\begin{figure}[!htbp]
  \centering
  \input{architecture-replica.tikz}
  \caption{Deterministic even/odd (DEO) replica-exchange schedule (\S\,\ref{sssec:mcmc-deo}). $R$ chains run in parallel at decreasing inverse temperatures $\beta_{1} > \beta_{2} > \dots > \beta_{R}$ (rows, cold $\to$ hot colour gradient). At each MCMC time step (columns), pairs of \emph{adjacent} chains are proposed for Metropolis swaps: even-edge pairs $(1,2)$ and $(3,4)$ at even $t$, odd-edge pair $(2, 3)$ at odd $t$. The alternation is deterministic and exhaustive, every pair firing once in every two steps, and non-reversible in the sense that consecutive swaps cannot trivially undo each other, which maximises information flow along the temperature ladder.}
  \label{fig:replica-deo}
\end{figure}
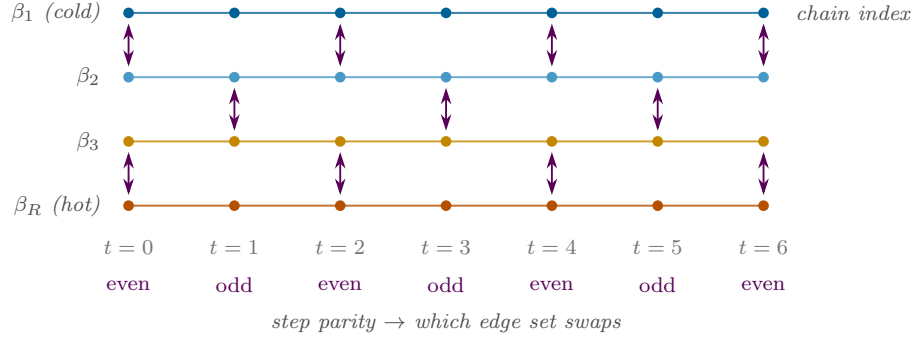

\subsubsection{Online $\beta$-adaptation via monotone cubic interpolation}
\label{sssec:mcmc-beta}

The ladder only carries walkers if adjacent rungs overlap enough to swap: too coarse a spacing rejects every swap, too fine a spacing wastes replicas. At stationarity the optimal schedule equalises the adjacent-pair rejection rates across the ladder \citep{atchade2011towards}: for a target rate $\bar\lambda$, every adjacent pair $(r, r+1)$ should reject with probability $\bar\lambda$. This is the schedule-optimisation problem that Syed et al. \cite{syed2022non} solve through the cumulative rejection (their communication barrier) $\Lambda(\beta) = \int_{\beta_1}^{\beta} \lambda(\beta')\, \lvert\mathrm{d}\beta'\rvert$: the optimal grid divides $\Lambda$ into equal increments, $\beta_r = \Lambda^{-1}\bigl((r-1)\,\bar\lambda\bigr)$ with $\bar\lambda := \Lambda(\beta_R)/(R-1)$.  What we add is the online version: the running chain gives only the $R$ samples $\{(\beta_r, \hat\Lambda(\beta_r))\}_r$, so recovering the grid means inverting a monotone interpolant of $\Lambda$, refreshed as the model, and with it the rejection profile, drifts during training.

Piecewise cubic Hermite interpolation fixes both the value $f(\beta_{r})$ and the derivative $f'(\beta_{r})$ at each node and fits each segment with a cubic. The choice of node derivative is what separates the methods. PCHIP, the piecewise cubic Hermite interpolating polynomial of \citet{fritsch1980monotone}, sets the derivative to zero wherever the local secant slope $\delta_{r} = (f(\beta_{r+1}) - f(\beta_{r})) / (\beta_{r+1} - \beta_{r})$ changes sign and to the harmonic mean of $\delta_{r-1}, \delta_{r}$ otherwise; it is monotone by construction, but the harmonic mean loses accuracy through a smooth inflection. Hyman's \citet{hyman1983accurate} derivative is instead the centred secant $f'(\beta_{r}) = (f(\beta_{r+1}) - f(\beta_{r-1})) / (\beta_{r+1} - \beta_{r-1})$, which is exact for quadratic data, and applies the Fritsch--Carlson \citep{fritsch1980monotone} monotonicity clamp only as a safety net at a detected sign change. On the quadratic test $f(\beta) = \beta^{2}$ Hyman recovers $f'$ to double precision ($\sim 10^{-12}$ relative error) where PCHIP errs at $\sim 10^{-3}$ near an inflection. Spin-manifold rejection curves carry smooth inflection regions where that harmonic-mean error, not monotonicity, is the binding constraint, so we use Hyman.

Each adaptation interval then runs four steps. It updates $\hat\Lambda(\beta_{r})$ by an exponential moving average over the last $T$ swap attempts at each pair, fits the Hyman cubic to $\{(\beta_{r}, \hat\Lambda(\beta_{r}))\}_{r}$, bisects it to recover the equal-rejection grid $\{\tilde\beta_{r}\}$, and blends that grid into the old one, $\beta_{r} \leftarrow (1 - \mu)\, \beta_{r} + \mu\, \tilde\beta_{r}$, for stability.

\subsubsection{Global $m$-site Haar refresh and the $0.234$ optimum}
\label{sssec:mcmc-haar}

Local moves and swaps still change the configuration only gradually, so a basin that no replica has entered stays unvisited. A global move sidesteps that. Once per training step, before the local-move loop, the chain refreshes $m$ randomly chosen sites to Haar-uniform values on $\mathrm{SU}(2)$. This non-local move breaks the within-chain correlation that MALA builds up and reaches disconnected modes that the local Langevin proposal cannot bridge.

The refresh acceptance rate falls as $m$ grows, since a larger refresh lands further from the current state and is accepted less often. We tune $m$ on the same slow cadence as the step size, by a proportional rule: increment it by one when the running Haar acceptance exceeds the target, decrement it by one when it falls short, clipped to $[1, n_{\text{active}}]$. A unit step is the right granularity because acceptance is discrete in $m$ and the search space has only $n_{\text{active}}$ candidates, so the optimum is reached quickly. The hottest replica, at $\beta = 0$, is held at full erasure, $m = n_{\text{active}}$, where every refresh is accepted and the mixing is free.

The target $\alpha^{\star}$ is the acceptance that maximises the expected squared-jump distance per move. For a $k$-site refresh on a compact group the expected squared-jump distance (ESJD) grows linearly in $m$, and so does the variance $V$ of the log-density increment, which is a sum of $m$ per-site contributions; the two are therefore proportional, and $\mathbb{E}[\mathrm{ESJD}] \propto \alpha\, V$. Because the increment is a $m$-term sum, the central limit theorem makes it approximately Gaussian whenever the site contributions are weakly correlated, and exactly Gaussian when the target factorises over sites, as for the rank-one ferromagnetic ground state. We therefore take
\begin{equation}
  \Delta \log |\psi_{\theta}|^{2} \;=\; M + Z,
  \qquad Z \sim \mathcal{N}(0, V),
\end{equation}
with mean $M < 0$. The refresh draws its sites independently of the current state, so the proposal is symmetric; detailed balance at stationarity then gives $\mathbb{E}[e^{\Delta}] = 1$, which for a Gaussian increment implies
\begin{equation}
  M \;=\; -V/2,
  \qquad \text{equivalently} \qquad
  V \;=\; -2M.
  \label{eq:fdt}
\end{equation}.

The expected MH acceptance probability is then the
truncated Gaussian-tail integral
\begin{equation}
  \alpha(V)
  \;=\; \mathbb{E}\bigl[\min(1, e^{\Delta \log |\psi_{\theta}|^{2}})\bigr]
  \;=\; 2\, \Phi\!\bigl(-\sqrt{V}/2\bigr),
  \label{eq:alpha-of-V}
\end{equation}
where $\Phi$ is the standard-normal CDF. The expected squared-jump distance (ESJD) is proportional to $\alpha(V) \cdot V$. Setting $\mathrm{d}/\mathrm{d}V[\alpha(V)\, V] = 0$ and solving numerically:
\begin{equation}
  \sqrt{V^{\star}}/2 \;\approx\; 1.19,
  \qquad
  V^{\star} \;\approx\; 5.66,
  \qquad
  \alpha^{\star} \;\approx\; 0.234.
\end{equation}
The resulting optimum, $\alpha^\star\approx0.234$, matches the random-walk Metropolis limit of \newcitet{roberts1997weak}. Both derivations reduce to maximising the expected squared-jump distance when the log-density increment is asymptotically Gaussian and stationarity fixes its mean to $-V/2$. The present derivation concerns an $m$-site Haar refresh on a compact group and uses neither a random walk, Euclidean geometry, nor a small-step limit.

\subsubsection{Multiplicative step-size adaptation}
\label{sssec:mcmc-rm}

Every piece above assumes its proposal sits at the right acceptance rate. The step size $\sigma$ of the manifold MALA proposal (\S\,\ref{sssec:mcmc-mala}) is the component that holds the local-move acceptance at its optimum. For MALA that optimum is $0.574$ in the high-dimensional limit \citep{roberts1998optimal}; a plain random-walk proposal would instead target $0.234$. The Haar refresh has no step size and reaches its own $0.234$ optimum through the site count $m$ of \S\,\ref{sssec:mcmc-haar}, not through $\sigma$.

Each replica drives its own $\sigma$ to the target with a multiplicative rule. After every adaptation window it measures the local acceptance $A$ and moves one fixed factor toward the target,
\begin{equation}
  \sigma \;\longleftarrow\;
  \mathrm{clip}\Bigl(\sigma\, f^{\,\mathrm{sign}(A - A^{\star})},\;
    10^{-4},\; \pi\Bigr),
  \qquad f = 1.1,
  \quad A^{\star} = 0.574.
  \label{eq:sigma-adapt}
\end{equation}
The rule is bang-bang: $\sigma$ rises by the factor $f$ when acceptance runs hot and falls by it when acceptance runs cold, so at equilibrium it oscillates within a factor $f$ of the value that holds $A = A^{\star}$. That is all the precision the sampler needs, because the MALA efficiency curve is flat near its optimum, and the rule carries no schedule and no state. The clip keeps $\sigma$ inside the range where the proposal is meaningful on the group: steps below $10^{-4}$ are numerically idle, and $\pi$ is the diameter-scale move on $\mathrm{SU}(2)$.

The natural alternative is a Robbins--Monro stochastic-approximation schedule on $\log \sigma$ \citep{robbins1951stochastic}, whose diminishing gains converge almost surely to the exact target. But a diminishing gain is built for a stationary target, and ours moves: training reshapes $|\psi_\theta|^2$ continuously, so an adapter that never stops moving is the safer default.

\subsubsection{Gradient cache between MCMC inter-steps}
\label{sssec:mcmc-cache}

A second saving falls out of the chain's own structure. The gradient the chain computes when it accepts a move at $q$ is the same gradient the next inter-step needs before it proposes from $q$. Threading this cached gradient through the scan saves one \texttt{value\_and\_grad} per inter-step and halves the per-step gradient cost. At the start of each training step the cache is refreshed against the freshly updated model, which is the caching step the figure marks before the Haar refresh.

\subsection{Energy kernels}
\label{ssec:energy-kernels}

The bottleneck of variational Monte Carlo on a Hamiltonian quadratic in the Pauli operators is the energy estimator. For a sampled configuration $q \in (\mathrm{SU}(2))^{N}$ the local energy splits into a two-body exchange term and a linear field term. The field is the first order in the derivatives of $\log \psi$ and is computed as by-product of computing second order terms. The whole cost sits in the exchange term, which contracts the coupling tensor against the second derivatives of $\log \psi$,
\begin{equation}
  E_{\text{exch}}(q) \;=\;
  -\tfrac{1}{4} \sum_{i, j} \sum_{a, b}
    J^{ab}_{ij}\, \Bigl[
      L_i^a L_j^b \log \psi(q)
      + \bigl(L_i^a \log \psi(q)\bigr)\,
        \bigl(L_j^b \log \psi(q)\bigr)
    \Bigr],
  \label{eq:e-exchange}
\end{equation}
where $L_i^a$ is the left-invariant derivative along Pauli direction $a$ at site $i$ (\S\ref{ssec:peter-weyl}) and the prefactor $\tfrac14$ comes from $S = \tfrac12 \sigma$. The first bracket is a Hessian of $\log \psi$ contracted against $J$; the second is an outer product of its gradient. The Hessian term is the expensive one. Our custom JAXPR interpreter evaluates it, together with every other term the estimator needs, in a single forward pass: a custom forward-mode second-order operator written for the tree, described next.

\subsubsection{Custom forward-mode Second-order operator}
\label{sssec:custom-lap}

One JAXPR interpreter serves the whole local energy. Its function signature reads,
\begin{equation}
  \texttt{custom\_lap}: \;
  q \;\longmapsto\;
  \Bigl(\, \log \psi_\theta(q),\;
    \nabla \log \psi_\theta(q),\;
    \mathrm{Tr}\bigl(W\, H_{\log\psi_\theta}(q)\bigr) \,\Bigr),
  \label{eq:custom-lap-contract}
\end{equation}
one forward pass per walker, from which the score, the linear field term and the exchange term of Eq.~\eqref{eq:e-exchange} are all read off: no Hessian is ever materialised and no Hessian-vector product is ever formed on the energy path. The kernel propagates a jet through the network, so each intermediate carries its value, its Jacobian along the Lie-frame input directions, and a running weighted trace, and each primitive advances the three together. The trace advances as a quadratic form of the Jacobian it already carries,
\begin{equation}
    \mathrm{tr} \;\longleftarrow\; \mathrm{tr} \;+\;
  \sum_{\mu, \nu, c} \mathcal{J}_{\mu c}\;
    \Lambda_{c \mu \nu}\; \mathcal{J}_{\nu c},
  \label{eq:jet-quadratic},
\end{equation}
where $\mathcal{J}_{\mu c}$ is the intermediate's Jacobian ($\mu, \nu$ over the Lie-frame input directions, $c$ over channels),not to be confused with the coupling tensor $J$,  and $\Lambda_{c\mu\nu}$ is the second-order weight the chain rule assigns to the primitive at hand; for every linear primitive $\Lambda$ vanishes, so the sum is paid only where genuine curvature enters. \eqref{eq:jet-quadratic} is the update for element-wise primitives, whose curvature is diagonal in the channel index. The merge tensor is the one primitive with cross-channel curvature: there the update couples the two child direction blocks through $T$ itself, and it is the only point in the network where a cross-block entry enters the running trace.

The natural alternative is a generic forward-mode Laplacian\citep{li2024computational}, as implemented in \texttt{folx} \citep{folx}. Its tracer treats the second-derivative tangent as densely coupled across sites, because the Pauli tangents hide the structure of $J$: it sees a dense $4N \times 3N$ Jacobian fan-out and pays for all of it. The tree of \ref{ssec:tree-readout} makes that coupling sparse. Every $q$-dependent intermediate at merge level $\kappa$ carries a Jacobian of the chunked shape $[\,3k,\, N/k,\, *S\,]$, where $k = 2^{\kappa}$ is the number of sites per chunk, $N/k$ the number of chunks, the factor $3$ the Lie-frame directions per site, and $*S$ the trailing feature axes; the only new cross-block coupling in the Jacobian is the one the merge itself creates between its two children. \texttt{custom\_lap} is a forward-mode Second order tracer written for exactly this layout.  Its tracer treats the second-derivative tangent as densely coupled across sites, because the Pauli tangents hide the structure of $\mathcal{J}$: it sees a dense $4N \times 3N$ Jacobian fan-out and pays for all of it. \texttt{custom\_lap} is our extension of this technique: a \emph{weighted-trace} generalisation ($\mathrm{Tr}(WH)$ rather than $\mathrm{Tr}(H)$, which is what lets one pass serve an arbitrary coupling tensor $J^{ab}$), specialised to the tree's layout, where chunk structure sets the cost. At level $\kappa$ the $\hat N/2^{\kappa}$ chunks each carry $3 \cdot 2^{\kappa}$ Lie-frame directions, so the quadratic-form updates at that level cost $\mathcal{O}\bigl((\hat N/2^{\kappa}) \cdot (3 \cdot 2^{\kappa})^{2} \cdot S\bigr) = \mathcal{O}\bigl(\hat N\, 2^{\kappa} S\bigr)$ with $S$ the feature width. Summing the geometric series over $\kappa \le K$ gives $\mathcal{O}(\hat N^{2} S)$ for the full pass, against the $\mathcal{O}(\hat N^{3})$ of a dense forward-mode tracer. This claim is anticipated in \S\,\ref{sssec:root-readout}.

We emphasize that our extension here cannot be used as a universal differentiator since it assumes the balanced binary tree, and a symmetric two-body weight $W : \mathbb{R}^{3N} \to \mathbb{R}^{3N}$. Thus, the custom second order operator that we built for our model, can produce correct second order differential operator values only on the tree-like ansatz. This is the core reason we aren't bound by the generic algorithmic complexity of evaluating the Hessian of a black box function.

\subsection{Kronecker-factored curvature for the foundation ansatz}
\label{ssec:kfac}

KFAC (Kronecker-factored approximate curvature \citep{martens2015optimizing}) is the natural-gradient method we use for pretraining. Natural gradient multiplies the raw gradient by the inverse of the Fisher information matrix, the curvature of the model's distribution; the resulting step is invariant to how the model is parameterised and reaches a given loss in far fewer iterations than plain gradient descent. The Fisher information is a $P \times P$ matrix for $P$ parameters, far too large to form or invert at our scale, and KFAC's contribution is to approximate it cheaply.

For a single dense layer with weight$W \in \mathbb{R}^{d_{\text{out}} \times d_{\text{in}}}$, input activation $a \in \mathbb{R}^{d_{\text{in}}}$ and back-propagated output gradient $\delta \in \mathbb{R}^{d_{\text{out}}}$, the block of the Fisher information belonging to $W$ is the covariance of the per-example gradient $\mathrm{vec}(\delta a^{\top})$,
\begin{equation}
  F_{W}
  \;=\;
  \mathbb{E}\bigl[\mathrm{vec}(\delta a^{\top})\,
                  \mathrm{vec}(\delta a^{\top})^{\top}\bigr]
  \;=\;
  \mathbb{E}\bigl[a a^{\top} \otimes \delta \delta^{\top}\bigr].
\end{equation}
KFAC replaces the expectation of this Kronecker product by the product of expectations, which is exact when $a$ and $\delta$ are independent,
\begin{equation}
  F_{W} \;\approx\; A \otimes G,
  \qquad
  A := \mathbb{E}[a a^{\top}],\ G := \mathbb{E}[\delta \delta^{\top}].
  \label{eq:kfac-2factor}
\end{equation}
The factorisation is what makes the method affordable. Inverting the two factors separately, $(A \otimes G)^{-1} = A^{-1} \otimes G^{-1}$, costs $\mathcal{O}(d^{3})$ per layer in place of the $\mathcal{O}(d^{6})$ of the full block. Before inversion, the KFAC implementation adds damping $\gamma$, $\widetilde{F}^{-1} \approx (\widehat{F} + \gamma I)^{-1}$, which trades fidelity for stability: a smaller $\gamma$ keeps the preconditioner closer to the true Fisher information and each step closer to the natural-gradient direction, so fewer steps reach a given loss.

Our model breaks three assumptions that KFAC's standard blocks build in.  The trunk's $L$ repeated blocks stack their weights through a loop, so the Fisher information of those weights couples across the loop iterations. The attention projections couple across the query, key, value and head channels. And the merge tensor of \S\,\ref{ssec:tree-readout} is cubic in its width, so its full Fisher information is far too large to form. The patches below extend the \texttt{kfac\_jax} library \citep{botev2017practical} to each case. Together they remove from the model every dense Fisher-information block, and leave a curvature estimate faithful enough to run near the edge of stability. The Fisher information itself is registered on both output channels, the log-amplitude and the phase, as two normal-predictive registrations: the Fubini--Study metric of the complex output.

In the variational Monte Carlo literature this natural gradient is stochastic reconfiguration \citep{sorella1998green}, with the Fubini-Study metric playing the role of the Fisher information \cite{heightman2025deep}. KFAC-preconditioned wavefunction optimisation was introduced for deep ansatze by \cite{pfau2020ferminet}, whose registration strategy ours extends to a complex output.

\subsubsection{Fusing the query, key and value projections}

Each trunk block's attention sublayer from \S\,\ref{ssec:trunk-leaf} projects the per-site stream $\ell \in \mathbb{R}^{d_\ell}$ into queries $q_{\mathsf{h}}$, keys $k_{\mathsf{h}}$ and values $v_{\mathsf{h}}$ for each of its $n_{\mathsf{h}}$ heads. We compute all three projections with one weight matrix, $W_{\mathsf{qkv}} \in \mathbb{R}^{(3 n_{\mathsf{h}} d_{\mathsf{h}}) \times d_\ell}$, rather than three separate ones. The fusion needs no custom curvature class: the fused weight registers with KFAC as one ordinary dense layer. The standard factorisation hands KFAC a single gradient-covariance factor $G$ of size $(3 n_{\mathsf{h}} d_{\mathsf{h}})^{2}$, which carries the covariance between the query, key, value and head channels; three separate registrations would set that cross-covariance to zero by construction.

\subsubsection{The stacked trunk and a third Kronecker factor}

A standard KFAC block factors a dense layer's Fisher information as $A \otimes G$. The trunk's dense layers, however, are stacked: the same layer appears in all $L$ trunk blocks, carried on an extra leading axis of length $L$. Sharing one $A \otimes G$ across the $L$ copies treats them as independent and throws away any correlation between the blocks. Our block keeps that correlation, factoring the Fisher information as a three-way Kronecker product
\begin{equation}
  F_{\text{scan}} \;\approx\;
    M \otimes A \otimes G,
  \label{eq:kfac-3factor}
\end{equation}
where $M$ is the $L \times L$ covariance across the $L$ blocks and $A$, $G$ are the activation and gradient covariances as before. Here and below $\mathbb{S}_{++}^{n}$ denotes the $n \times n$ symmetric positive-definite matrices, so $M \in \mathbb{S}_{++}^{L}$, $A \in \mathbb{S}_{++}^{d_{\text{in}}}$ and $G \in \mathbb{S}_{++}^{d_{\text{out}}}$. Two-factor KFAC is the special case $M \propto I_{L}$, which discards the cross-block correlation entirely. The third factor is cheap: $L$, the number of trunk blocks, is far smaller than the layer widths, so $M$ adds only an $L^{2}$ term to factors that already cost $d^{2}$. Kronecker structure across a loop axis has precedent in the recurrent setting \newcitep{martens2018kronecker}, where the shared weight's Fisher couples across time steps. Carrying an explicit learned $L \times L$ factor across \emph{stacked, distinct} block weights, capturing the cross-block correlation that layer-wise block-diagonal KFAC discards, is, as far as we are aware, novel.

The block estimates $M$ online, as a running exponentially-weighted average of the gradient correlations between the $L$ blocks. This is the only statistic it keeps beyond what two-factor KFAC already tracks, and at $\mathcal{O}(L^{2})$ per block its cost is negligible against the $\mathcal{O}(d^{2})$ of the activation and gradient factors.

\subsubsection{Stacked RMS scale parameters}

The RMS normalisation layers of SM~\S~\ref{sec:part2} carry a per-channel scale and no shift, and the scales are stacked across the $L$ blocks in the same way. Their Fisher information takes the same form, with the gradient factor replaced by a diagonal: a block-coupling factor across the $L$ copies, times a diagonal per-feature term $D$. With a single block it reduces, exactly, to the standard diagonal treatment.

\subsubsection{The merge tensor}

The rank-4 merge tensor $T$ of \S\,2.4 is the hardest block. With $B$ sub-blocks of contraction width $d_r$ it carries $B\, d_r^3$ parameters, of order $10^{6}$ at the widths used here, so its full Fisher information would have $(B\, d_r^3)^2$ entries, of order $10^{13}$, roughly ten terabytes, which cannot be formed. The custom block instead matricises the four axes $(i, j, k, l)$ into a balanced pair and factorises the Fisher information as the two-factor Kronecker product $\Sigma_L \otimes \Sigma_R$ over that split, storing and inverting two Gram factors of a few megabytes together; the factored state is about six orders of magnitude smaller than the dense block. The sub-block axis rides inside one of the two factors, a genuine Kronecker leg rather than a batch dimension. Without it the merge tensor would fall back to a dense block and, at this width, would not train; with it the tensor receives a genuine natural-gradient preconditioner.

\subsubsection{Damping}

The faithful three-way Fisher information changes what limits the damping. With the approximation tightened, the floor on how small $\gamma$ can go is the noise in the curvature estimate itself, not the approximation error. We hold $\gamma$ at that noise floor: a constant identity-weight damping of $10^{-3}$, with a hard floor of $10^{-4}$ and a norm constraint of $10^{-3}$ on the preconditioned step.

Large-model pretraining has shown \citep{cohen2021gradient} that deep networks converge in an \emph{edge-of-stability} regime, where the largest curvature eigenvalue sits just above $2 / \eta$ for learning rate $\eta$ and the loss oscillates without diverging. Our schedule rides the analogous edge. A cumulative-sum change detector (CUSUM) on the natural-gradient norm classifies the run as stable, marginal or divergent; $\gamma$ is raised only when the run crosses into the divergent regime and drifted back down afterwards, so the optimiser sits at the smallest stable damping rather than behind a conservative margin.

\subsubsection{\texorpdfstring{Benchmarking the third factor versus per-layer curvature}{Benchmarking the third factor versus per-layer curvature}}
\label{ssec:kfac-validation}

The natural alternative to the third factor of Eq. \eqref{eq:kfac-3factor} is per-layer curvature: give each of the $L$ stacked copies its own two-factor block and precondition with the block diagonal $\widetilde{F}_{\text{bd}} := \bigoplus_{k=1}^{L} \widehat{A}_{k} \otimes \widehat{G}_{k}$. Per diagonal block this is the more expressive family --- every diagonal block of $\widetilde{F}_{3} = \widehat{M} \otimes \widehat{A} \otimes \widehat{G}$ is the same $\widehat{A} \otimes \widehat{G}$ up to the scalar $\widehat{M}_{kk}$, whereas $\widetilde{F}_{\text{bd}}$ fits all $L$ blocks independently --- and $L$ times more expensive in factor state and inverse work; for the fused $W_{\mathsf{qkv}}$ projection alone the factor state is $8\times$ larger. The extra freedom buys accuracy only at initialisation, and only on the fused QKV block, where per-layer curvature wins decisively, $0.956$ against $0.343$ in the damped-action cosine \eqref{eq:kfac-action-cos}; the three-factor block leads on the other three trunk modules already at initialisation, and at the trained point it wins on all four.

We evaluate both families against the exact sampled Fisher of Hamilton-Zero. Curvature capture runs through the estimator's own one-hot seeded vector--Jacobian products, so the per-layer activations $a_{k}$ and cotangents $\delta_{k}$ are precisely the tensors the curvature blocks consume; the exact block $F_{\text{scan}}$ is assembled from the per-walker gradients and inverted exactly through the Woodbury identity at damping $\gamma = 10^{-3}$. A candidate $\widetilde{F}$ is scored by its damped-action alignment on the model's true gradient $g$,
\begin{equation}
\cos\angle\bigl(\,(\widetilde{F} + \gamma I)^{-1} g,\ (F_{\text{scan}} + \gamma I)^{-1} g\,\bigr),
\label{eq:kfac-action-cos}
\end{equation}
evaluated at two points: random initialisation, on a $64$-site open-boundary Heisenberg chain, and a trained checkpoint (step $30\,000$) with equilibrated MCMC walkers, aggregated across the $13$ training Hamiltonians categories of \S\ref{sec:datasets} as the estimator aggregates, i.e. factors averaged across systems as the exponential moving average averages them, gradients and damped actions per system. 

\begin{table}[!htbp]
\centering
\small
\begin{tabular}{l rr rr r}
\hline
 & \multicolumn{2}{c}{Initialisation} & \multicolumn{2}{c}{Trained (step $30\,000$)} & \\
Trunk module & $\widetilde{F}_{\text{bd}}$ & $\widetilde{F}_{3}$ & $\widetilde{F}_{\text{bd}}$ & $\widetilde{F}_{3}$ & $\mathrm{cond}(\widehat{M})$ \\
\hline
Fused QKV ($W_{\mathsf{qkv}}$) & $0.956$ & $0.343$ & $0.112$ & $0.153$ & $4.08$ \\
FFN in-projection & $0.0717$ & $0.699$ & $0.0657$ & $0.224$ & $3.23$ \\
FFN out-projection & $0.0309$ & $0.422$ & $0.0468$ & $0.160$ & $9.21$ \\
Attention output & $0.676$ & $0.833$ & $0.0761$ & $0.215$ & $4.81$ \\
\hline
\end{tabular}
\caption{Damped-action alignment of the two stacked-trunk curvature families: the cosine \eqref{eq:kfac-action-cos} between the approximate damped natural gradient $(\widetilde{F} + \gamma I)^{-1} g$ and the exact $(F_{\text{scan}} + \gamma I)^{-1} g$ on the model's true gradient, at damping $\gamma = 10^{-3}$}; larger is better. $\widetilde{F}_{\text{bd}}$ is the per-layer block diagonal $\bigoplus_{k} \widehat{A}_{k} \otimes \widehat{G}_{k}$; $\widetilde{F}_{3}$ is the deployed three-factor block $\widehat{M} \otimes \widehat{A} \otimes \widehat{G}$ of \eqref{eq:kfac-3factor}. Initialisation: random parameters on a $64$-site open-boundary Heisenberg chain. Trained: checkpoint step $30\,000$ with equilibrated MCMC walkers; entries are medians of per-system cosines across the $11$ training Hamiltonians, with factors aggregated across systems using the same exponential moving average as in training. Last column: median $\mathrm{cond}(\widehat{M})$ at the trained point over the three systems with retained per-walker captures.
\label{tab:kfac-perlayer}
\end{table}

Table~\ref{tab:kfac-perlayer} fixes why the trained point belongs to the third factor. First, $F_{\text{scan}}$ is dominated by cross-layer coupling: $0.86$--$0.87$ of its squared Frobenius mass lies off the layer-diagonal, at initialisation and at the trained point alike. $\widetilde{F}_{\text{bd}}$ is layer-diagonal by construction, so no choice of per-layer factors touches this mass, while the off-diagonal blocks $\widehat{M}_{kl}\, \widehat{A} \otimes \widehat{G}$ of $\widetilde{F}_{3}$ model it directly. Second, the trained $\widehat{M}$ is far from $c\, I_{L}$: across the four trunk modules $\mathrm{cond}(\widehat{M})$ has medians $3.23$--$9.21$.

The trained regime is three-factor shaped because the residual stream makes both ends of every block gradient shared across the stack. Each stacked copy reads the same normalised stream state, so the activations $a_{k}$ are nearly independent of $k$; the stream update is $r \mapsto r + f_{k}(r)$, so the cotangents propagate through near-identity Jacobians, $\delta_{k} = J_{k}^{\top} \delta_{k+1}$ with $J_{k} = I + \partial f_{k}/\partial r$, and stay near-collinear in turn. The per-sample gradient of the $k$-th copy, $g_{k} = \delta_{k} a_{k}^{\top}$, therefore shares a rank-one direction with layer-dependent amplitude, and every block of $F_{\text{scan}}$ factorises accordingly,
\begin{equation}
g_{k} \;\approx\; \alpha_{k}\, \delta\, a^{\top} + \varepsilon_{k}
\qquad \Longrightarrow \qquad
(F_{\text{scan}})_{kl} \;\approx\; \mathbb{E}\bigl[\alpha_{k} \alpha_{l}\bigr]\, A \otimes G,
\label{eq:kfac-rank1}
\end{equation}
with $\delta$, $a$ shared within each sample and $\varepsilon_{k}$ the idiosyncratic remainder --- exactly \eqref{eq:kfac-3factor}, with $M = \mathbb{E}\bigl[\alpha \alpha^{\top}\bigr]$ far from a multiple of the identity. Trained attention concentrates the tangents onto a few specialised low-rank directions, strengthening the shared component; the per-layer family spends its $L$-fold extra freedom fitting the remainders $\varepsilon_{k}$, i.e.\ batch noise.

Initialisation favours per-layer curvature on the fused QKV block for an equally structural reason. At random parameters the blocks are weakly input-sensitive, since attention is near-uniform and $\partial f_{k}/\partial r$ is small, so per-block gradient statistics are near-Gaussian and second moments suffice: on precisely this block, $\widehat{A}_{k} \otimes \widehat{G}_{k}$ reproduces its diagonal Fisher block to a relative error of $1.19 \times 10^{-3}$, and winning the layer-diagonal is everything a layer-diagonal preconditioner can win. Training grows the input sensitivity and installs the shared low-rank stream of \eqref{eq:kfac-rank1}, which simultaneously breaks per-block Gaussianity, so second moments stop being sufficient, and organises the cross-layer mass into the coherent correlation that only the $M$ factor models. The trained regime is where the optimiser spends the run; there the deployed block is both the better preconditioner and the cheaper one, by a factor of $L$.

\hypersetup{linkcolor=red,citecolor=red,urlcolor=red}
\subsection{Sequence-sharded router evaluation}
\label{ssec:router-sharding}

For the $N=8100$ square-lattice evaluation we work in the padded $\hat N=8192$ slot frame.  We cannot carry the ordinary width-64 beam of \S\,\ref{ssec:eval_finetune_config} to this size: its order-aware prefix buffer has shape $[K_{\mathrm{beam}},\hat N,\hat N,d_h]$.  At the reported width $d_h=512$, that buffer occupies $128\,\mathrm{GiB}$ per beam member in fp32, or $8\,\mathrm{TiB}$ for the full beam.  The unsharded fp32 edge stream $e^{(L)}\in\mathbb R^{\hat N\times\hat N\times96}$ already occupies $24\,\mathrm{GiB}$.

We therefore decode one route greedily.  At row $t$ we take the largest remaining logit $\eta_{t,i}$ and append its index to $p$.  We run the trained scoring blocks one row at a time and omit both the beam axis and the teacher tensor $\mathbf X$ of Eq.~\eqref{eq:router-teacher-gather}.  We keep the prefix and suffix accumulators in $[\hat N,d_h]$ arrays and the order-dependent history in dyadic frontiers of shape $[\log_2\hat N,\hat N,d_h]$.  We also evaluate the router summary and bias MLPs in 128-column tiles and shard the dense attention logits over their query rows.  This reduces the order-history storage from $\mathcal O(K_{\mathrm{beam}}\hat N^2d_h)$ to $\mathcal O(\hat N\log\hat N\,d_h)$; the edge, bias and incremental-tree states remain quadratic, and dense candidate attention retains $\mathcal O(\hat N^3)$ arithmetic.

We opted for sequence sharding over the first axis of every full-width node and pair tensor.  With $n_{\mathrm{loc}}=\hat N/n_{\mathrm{dev}}$, each device owns $n_{\mathrm{loc}}$ node rows and $n_{\mathrm{loc}}\times\hat N$ pair rows of $e^{(L)}$ and $J''\in\mathbb R^{\hat N\times\hat N\times10}$; we replicate parameters, masks and global vectors.  Full-width pair kernels consequently operate on local $n_{\mathrm{loc}}\times\hat N$ rectangles without replicating the quadratic state.

Internally we use \texttt{all\_gather} for linear-sized descriptors and attention keys and values, keeping queries local and streaming the trunk, contextualiser and reducer softmax in blocks.  For 2-FWL and the routed row permutation, we circulate one lane-local panel at a time with ring \texttt{ppermute}.  We obtain the column view needed by row--column pooling with an \texttt{all\_to\_all} transpose, and combine the two scalar exchange statistics used by the featurizer's global scale features with \texttt{psum}.  We keep the physical tree row-sharded until its width reaches $512$; the resulting pair state is at most $96\,\mathrm{MiB}$ and can enter the replicated tail compiler.  We place readiness barriers between route selection, pair permutation and the physical compilation substages, then drop each phase-specific array after its final consumer.

\subsection{Fermionic compilation}\label{app:fermionic-systems}

\subsubsection{The Jordan--Farhi gadget in Degenerate Perturbation Theory}
For a Pauli word $P=\prod_{j=1}^{k}\sigma^{a_j}_{q_j}$ of weight $k>2$ and coefficient $\alpha$, the gadget \cite{jordanfarhi} introduces ancillas $r_1,\dots,r_k$ and the two-local Hamiltonian $H^{\rm gad}=H^{\rm anc}+V$ with
\begin{equation}\label{eq:app-gadget-def} H^{\rm anc}=\Delta\sum_{1\le i<j\le k}\frac{I-Z_{r_i}Z_{r_j}}{2},\qquad V=\sum_{j=1}^{k}g_j\,\sigma^{a_j}_{q_j}X_{r_j},\qquad \prod_{j=1}^{k}g_j=\frac{(-1)^{k-1}(k-1)!}{k}\,\Delta^{k-1}\alpha . \end{equation}
The ground space of $H^{\rm anc}$ is spanned by the two aligned ancilla states, and $V$ flips one ancilla at a time. Projected onto the $+1$ eigenspace of $X_{r_1}\cdots X_{r_k}$, with projector $P_+$, degenerate perturbation theory gives
\begin{equation}\label{eq:app-gadget-eff} H_{\rm eff}-fI=\alpha\,P\otimes P_+ + O\!\left(\lVert V\rVert^{k+1}/\Delta^{k}\right), \end{equation}
where $f$ is a scalar series. The proof is short. Every process that starts and ends in the ancilla ground space must flip all $k$ ancillas, so no nonscalar term appears below order $k$. After $j$ distinct flips the ferromagnet carries excitation energy $E_j=\Delta\,j(k-j)$, so each ordering of the $k$ flips contributes the denominator $\prod_{j=1}^{k-1}(-E_j)^{-1}=(-1)^{k-1}/[\Delta^{k-1}((k-1)!)^2]$, and summing the $k!$ orderings under the coupling condition of Eq.~\ref{eq:app-gadget-def} leaves exactly $\alpha P$~\cite{jordanfarhi,kkr}. Gadgets for different words use disjoint ancilla registers, so they add without cross terms at this order.

The remainder in Eq.~\ref{eq:app-gadget-eff} does not vanish quickly. A symmetric choice of couplings gives $|g_j|\sim\Delta^{(k-1)/k}|\alpha|^{1/k}$, so $\lVert V\rVert\sim k\,\Delta^{(k-1)/k}|\alpha|^{1/k}$ and the remainder scales as $k^{k+1}|\alpha|^{(k+1)/k}\,\Delta^{-1/k}$. This is the $\Delta^{-1/k}$ law behind the main text: each factor of ten in accuracy for a weight-$k$ word costs a factor of $10^{k}$ in penalty.

The penalty above is largely absent in gate-based dynamical quantum simulation. A circuit composes unitaries in time, and the exponential of a weight-$k$ Pauli word factorises exactly into a ladder of $O(k)$ two-qubit gates around a single-qubit rotation, so in that setting weight is a linear gate cost, with approximation entering only through the Trotter splitting of non-commuting terms~\cite{lloyd,vqe}. A static two-local Hamiltonian has no time axis to sequence over, so high-weight terms must be reproduced by sums of low-weight ones, and sums of non-commuting terms compose only perturbatively. This is strong reason for our outlook on dynamical models, which would make this perturbative approach obsolete.

\subsubsection{Minimising the compiled register}
The deployed register after gadgetisation is
\begin{equation}\label{eq:app-cfinal} C_{\rm final}=n+\sum_{\ell\,:\,w_\ell>2}w_\ell , \end{equation}
the $n$ mapping qubits plus one ancilla per factor of every word above weight two. The weights $w_\ell$ depend on the fermion--qubit mapping, so $C_{\rm final}$ is the natural objective when choosing one: fewer and lighter heavy words mean fewer ancillas and faster compute. For the JW-order-CF frame the search space is the mode ordering $\pi$, and we minimise $C_{\rm final}$ over $\pi$ by simulated annealing with pair-transposition proposals~\cite{chiew}; the mapping is exact for every ordering, so the anneal is a small compiling cost only.

\subsubsection{Exact sector compilation}
Given $(h_{pq},v_{pqrs})$, the compilation then proceeds in three steps.

First, we fix the sector using quantum numbers. We take the balanced split $N_\alpha=N_\beta=N_e/2$ at spin projection $m$, and when $S^2$ commutes with $H_{\rm f}$ we also fix the total spin, which is $S=0$ since the systems here are closed shell. For BeH$_2$ this splits $36=20_{S=0}\oplus15_{S=1}\oplus1_{S=2}$ and we keep the singlet block. We then enumerate the occupation states of the sector and build the restricted matrix $\widetilde M=W^\dagger H_{\rm f}W-\bar c\,I$, where $W$ is the isometry onto the sector and $\bar c$ centres the trace.

Second, we encode $\widetilde M$ as a two-local system. For $d\le4$ the block embeds densely on two qubits, since ${\rm Herm}(\mathbb C^4)$ is exactly the real span of the two-qubit Pauli words, and unused register qubits are pinned by weak $Z$ fields. For $d>4$ we use a one-hot register of $d$ qubits with
\begin{equation}\label{eq:app-onehot} H_2^{\mathcal V}=\sum_{k<l}\tfrac{\widetilde M_{kl}}{2}\left(X_kX_l+Y_kY_l\right)+\sum_k\widetilde M_{kk}\,\hat n_k+\lambda\big(\hat N_1-1\big)^2,\qquad \hat n_k=\tfrac{I-Z_k}{2},\quad \hat N_1=\textstyle\sum_k\hat n_k .
\end{equation}
The penalty commutes with the rest of $H_2^{\mathcal V}$ and vanishes on the $\hat N_1=1$ subspace, where $H_2^{\mathcal V}$ restricts exactly to $\widetilde M$, so the compilation is bias-free at every finite $\lambda$. The only job of $\lambda$ is to keep the global ground state in the logical subspace. The required size follows from bounds that need no diagonalisation. The minimal diagonal entry of $\widetilde M$ upper-bounds the sector ground energy, and Gershgorin bounds lower-bound the $\hat N_1\ne1$ blocks.

Third, we rescale so that the largest coupling is one, and we record the affine map $E_{\rm molecular}={\rm scale}\times E_{\rm model}+{\rm shift}$.

Since the register is small, we can also check every map directly: the sector choice, the zero compilation bias, and the penalty margins, by full diagonalisation when $2^{d}$ is tractable and by exact block bounds otherwise. Exact diagonalisation only enters through these checks and through the CASCI reference we judge the model against. The model itself optimises the compiled $(J,h)$ and has access to neither.

\subsubsection{The seven compiled systems}
Figure~\ref{fig:app-fermionic-graphs} shows the compiled interaction graphs and Table~\ref{tab:app-fermionic} the compilation parameters. The two CAS(2e,2o) systems use the dense two-qubit encoding with two pinned qubits; the CAS(4e,3o) systems are one-hot registers on $d=9$ qubits; BeH$_2$ is the $S^2$-resolved singlet block on $d=20$. Every one-hot graph is complete, since the commuting penalty couples all pairs, and the edge weights on top of that background carry the restricted Hamiltonian.

An end-to-end example is included in the public release as \texttt{examples/finetune\_h2o\_sector.ipynb}, which compiles the H$_2$O system of Table~\ref{tab:app-fermionic} from its packaged integrals, fine-tunes the released checkpoint, and reproduces the chemical-accuracy evaluation.

\begin{table}[H]
\centering
\begin{tabular}{@{}llrlrrl@{}}
\toprule
System & active space & $d$ & encoding & $\lambda$ (Eh) & gap (Eh) & notes\\
\midrule
H$_2$O & CAS(2e,2o) & 4 & dense & -- & 0.441 & two pinned qubits\\
LiH & CAS(2e,2o) & 4 & dense & -- & 0.140 & two pinned qubits\\
CO$_2$ & CAS(4e,3o) & 9 & one-hot & 1.27 & 0.272 & \\
ethylene & CAS(4e,3o) & 9 & one-hot & 1.36 & 0.176 & \\
benzene & CAS(4e,3o) & 9 & one-hot & 1.22 & 0.215 & \\
Cr$_2$ & CAS(4e,3o) & 9 & one-hot & 1.03 & 0.100 & rank-one lifts\\
BeH$_2$ & CAS(4e,4o) & 20 & one-hot & 2.76 & 0.310 & $S=0$ block\\
\bottomrule
\end{tabular}
\caption{Compilation parameters of the seven systems. $d$ is the sector dimension and equals the logical register; $\lambda$ is the confinement penalty of Eq.~\ref{eq:app-onehot}; the gap is the first excitation above the compiled ground state, from exact diagonalisation of the compiled register (block-exact for $d=20$).}
\label{tab:app-fermionic}
\end{table}

\begin{figure}[t]
\centering
\includegraphics[width=\linewidth]{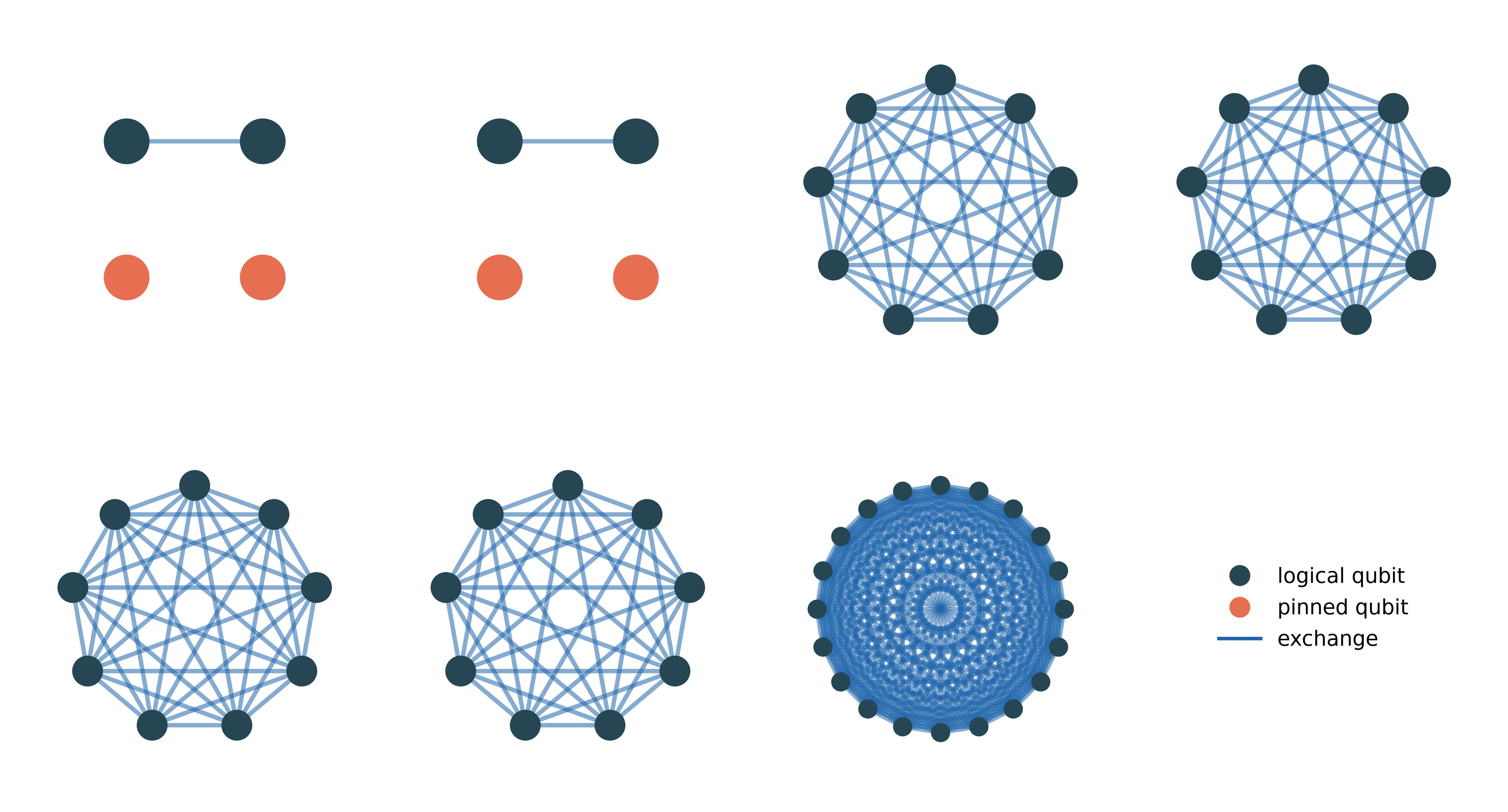}
\caption{Compiled interaction graphs of the seven systems (top row: H$_2$O, LiH, CO$_2$, ethylene; bottom row: benzene, Cr$_2$, BeH$_2$). Dark nodes are logical qubits, orange nodes pinned qubits, and edge width scales with the Frobenius norm of the exchange block $J_{ij}$, normalised within each panel. The two CAS(2e,2o) registers carry a single logical edge and two pinned spectators; the one-hot registers are complete graphs, with the penalty contributing a uniform $ZZ$ background and the restricted Hamiltonian the structure on top.}
\label{fig:app-fermionic-graphs}
\end{figure}

\appendix

\section{\texorpdfstring{The Casimir regulariser for nonlinear ans\"atze}{The Casimir regulariser for nonlinear ansatze}}
\label{app:casimir-reg}

This appendix constructs a variational principle for ans\"atze that satisfy the per-site oddness \eqref{eq:oddness-constraint} but are \emph{nonlinear} in the quaternions, so their support spreads across the half-integer sectors. This is the setting of the phase-space programme of~\cite{heightman2025foundations}, where maintaining a variational principle was left open; the closed-form couplings below resolve it. What oddness does not do is single out the lowest half-integer. The network may still place weight on $j_i = \tfrac{3}{2}$ and above, or on superpositions across them. Pinning it to a $\tfrac{1}{2}$ variational principle requires a second ingredient.

The second ingredient leans on the shape of the Casimir spectrum. By \eqref{eq:casimir-eigen} the per-site Casimir $\hat S_l^2 = -\Delta_l$ reads off the sector: it acts as $j_l(j_l+1)$, which oddness restricts to the values $\tfrac{3}{4}, \tfrac{15}{4}, \tfrac{35}{4}, \dots$ at $j_l = \tfrac{1}{2}, \tfrac{3}{2}, \tfrac{5}{2}, \dots$. The target value $\tfrac{3}{4}$ is the floor, and every step away from it grows quadratically in $j_l$. Oddness is also what lets us charge that gap \emph{linearly}: on the odd subspace every site sits at $\hat S_l^2 \ge \tfrac{3}{4}$, so the bare per-site deviation is non-negative, and so is its sum
\begin{equation}
  \hat C \;=\; \sum_{l = 1}^{N} \bigl(\hat S_l^2 - \tfrac{3}{4}\bigr),
  \qquad
  C[\psi]
  \;=\; \frac{\langle \psi |\, \hat C\, | \psi \rangle}{\langle \psi | \psi \rangle}
  \;=\; \sum_{l = 1}^{N}
      \mathbb{E}_{x \sim |\psi|^{2}}\!\Bigl[\,
        \tfrac{-\Delta_l\, \psi}{\psi}(x) - \tfrac{3}{4}
      \,\Bigr],
  \label{eq:casimir-regulariser}
\end{equation}
which vanishes exactly on the target sector (Proposition~\ref{prop:halfsector}). Were the integer sectors still present, the $j = 0$ site at $\hat S^2 = 0$ would sit below the floor and this penalty would \emph{reward} it; removing them is precisely what makes the bare deviation a sound penalty, and is why we did not need to square it.

The reason this rescues the variational principle is a race between two rates. An excursion into a higher sector lowers the energy only through the spin magnitudes $\sqrt{j_l(j_l+1)}$, which grow linearly in $j_l$, whereas the penalty charges the Casimir gap $j_l(j_l+1) - \tfrac{3}{4}$, which grows quadratically. The gap outpaces the gain, so there is a finite multiplier $\lambda$ above which every unit of energy bought by drifting off-target is overpaid by the penalty, and the regularised objective
\begin{equation}
  E_\lambda[\psi]
  \;=\; \frac{\langle \psi | \hat H | \psi \rangle}{\langle \psi | \psi \rangle}
        \;+\; \lambda\, C[\psi]
  \label{eq:regularised-objective}
\end{equation}
stays at or above the physical ground-state energy $E_0$. Because oddness has already removed every integer sector, the nearest a state can sit to the target without lying in it is $j_l = \tfrac{3}{2}$, the tightest case and the one that fixes how large $\lambda$ must be.

Evaluating $C[\psi]$ costs nothing beyond the energy autodiff scheme above, since the log-derivative identity \eqref{eq:casimir-diagnostic} of \S\,\ref{sec:ops_and_autodiff} turns each per-site Casimir ratio $-\Delta_l\, \psi / \psi$ into $\log \psi$ and its first two automatic derivatives. Indeed these are the same quantities the local energy already uses. What remains is to compute, from the couplings $(J, h)$, a $\lambda$ large enough that $E_\lambda[\psi] \ge E_0$ for every per-site-odd $\psi$.

In this appendix, we prove there is a $\lambda$ such that  $E_\lambda[\psi] \ge E_0$ for every per-site-odd $\psi$ for \emph{any} given Hamiltonian. We then find a tighter bound in which the scalar $\lambda$ is promoted to a per-edge tensor $\mu_{\star}$ that saturates the variational principle, i.e. our variational space contains the ground state energy on the tigher bound only. Both are a function of a given Hamiltonian's interaction data, $J,h$.

We use two quantities of interest that can be efficiently computed. First, at each edge $(i, k)$ the relevant size of the $3 \times 3$ interaction block $J_{ik} \subset J$ is its spectral norm, i.e. the top singular value,
\begin{equation}
  A_{ik} \;:=\; \bigl\| J_{ik} \bigr\|_{2}
              \;=\; \sigma_{1}\bigl(J_{ik}\bigr),
  \label{eq:Aij}
\end{equation}
From it we form the per-site \emph{operator-norm weighted degree} and the
\emph{field magnitude},
\begin{equation}
  \mathrm{wd}_i \;:=\; \sum_{k \neq i} A_{ik},
  \qquad
  b_i \;:=\; \bigl\| h_i \bigr\|_{2}.
  \label{eq:wd-and-b}
\end{equation}
The weighted degree $\mathrm{wd}_i$ collects how strongly site $i$ is coupled to its neighbours through any quadratic Pauli channel, and $b_i$ is the size of its transverse field. Both are closed-form functions of $(J, h)$, and can be efficiently computed once.

\begin{theorem}[Sufficient coupling]
\label{thm:lambda-min}
Let $\hat C = \sum_{l} (\hat S_l^{2} - \tfrac{3}{4})$ and $H(\lambda) = \hat H + \lambda\, \hat C$, and write $\varphi = (1 + \sqrt{5})/2$ for the golden ratio. For every $\lambda \ge \lambda_{\min}(J, h)$, where
\begin{equation}
  \lambda_{\min}(J, h)
      \;=\; (1 + \epsilon)\,
            \max_{i \in [N]}
              \Bigl[\,
                \tfrac{\varphi}{2}\, \mathrm{wd}_i
                \;+\; \tfrac{2}{3}\, b_i
              \,\Bigr],
  \label{eq:lambda-min}
\end{equation}
and every per-site-odd $\psi \in L^{2}((\mathrm{SU}(2))^{N})$,
\begin{equation*}
  \langle \psi |\, H(\lambda)\, | \psi \rangle \;\ge\; E_{0}\, \|\psi\|^{2},
\end{equation*}
with $E_{0}$ the physical spin-$\tfrac{1}{2}$ ground-state energy.
\end{theorem}

The constants $\tfrac{\varphi}{2}$ and $\tfrac{2}{3}$ are the maxima of three gain-to-penalty ratios over the allowed sectors. The slack $\epsilon > 0$ absorbs numerical safety holes (normalisation conventions, fp32 drift in $\mathrm{wd}_i$, the isotropy test at the $10^{-9}$ threshold); we set $\epsilon = 0.1$ in all reported runs.

The proof is given in Appendix~\ref{app:suff}.

For orientation, the 1D Heisenberg chain at unit coupling has $\mathrm{wd}_i = 2$ in the bulk and no field, so \eqref{eq:lambda-min} reads $\lambda_{\min} = (1+\epsilon)\, \varphi \approx 1.618\,(1+\epsilon)$. The remaining gap to the true physical threshold is over-conservatism, not a safety hole, and the next refinement closes most of it.

The coefficient $\tfrac{\varphi}{2}$ in \eqref{eq:lambda-min} is the universal one: it rests on the spectral norm \eqref{eq:Aij} alone and holds for every $J_{ik}$. It is loose for structured couplings, and the implementation tightens it edge by edge. The only place the proof touched $J_{ik}$ was its Cauchy--Schwarz step (Appendix~\ref{app:suff}, 3b), through the top singular value, so using more of the matrix sharpens the per-edge coefficient. Call $\alpha \ge 0$ a \emph{certificate} for a pair tensor $J$ if
\begin{equation}
  \bigl\| \vec S_i^{\top} J \vec S_k \bigr\|_{\text{op},(s_i, s_j)}
  + \bigl\| \vec S_i^{\top} J \vec S_k \bigr\|_{\text{op},(\tfrac12, \tfrac12)}
  \;\le\; \alpha\, \bigl[\, q(s_i) + q(s_j) \,\bigr]
  \label{eq:cert-master}
\end{equation}
for every non-trivial half-integer pair $(s_i, s_j) \neq (\tfrac12, \tfrac12)$. This is exactly the per-edge condition the proof of Theorem~\ref{thm:lambda-min} verified for $\alpha = \tfrac{\varphi}{2}\|J\|_{\text{op}}$; any certificate may take its place, and three are worth recording.

\begin{lemma}[Per-edge certificates]
\label{lem:certs}
Write $\|J\|_{\text{op}} = \sigma_1(J)$ and $\|J\|_* = \sum_a \sigma_a(J)$
for the spectral and nuclear norms. Each of
\begin{equation*}
  \alpha_{\text{op}} = \tfrac{\varphi}{2}\, \|J\|_{\text{op}},
  \qquad
  \alpha_{\text{nuc}} = \tfrac{1}{2}\, \|J\|_*,
  \qquad
  \alpha_{\text{orth}} = \tfrac{3}{4}\, |\lambda| \ \ (\text{for } J = \lambda Q,\ Q \in \mathcal{O}(3)),
\end{equation*}
is a certificate in the sense of \eqref{eq:cert-master}, and certificates add: $\alpha(J^{(1)} + J^{(2)}) \le \alpha(J^{(1)}) + \alpha(J^{(2)})$.
\end{lemma}

The three certificates and their additivity are proved in Appendix~\ref{app:certs}.

Additivity turns the three primitives into a family. Peeling off an isotropic part $J = \tau I + K$ with $\tau = \tfrac13 \mathrm{tr}(J)$ (so $\tau I$ takes the orthogonal certificate with $Q = I$) and bounding the remainder $K$ by the operator or nuclear certificate, or peeling off the nearest orthogonal part $J = \lambda Q + K_{\text{pol}}$ with $Q = UV^{\top}$ from the SVD and $\lambda = \tfrac13 \mathrm{tr}(Q^{\top} J)$, gives six certificates whose minimum is used for each edge,
\begin{equation}
  \alpha_{\text{full}}(J) = \min\Bigl\{
    \tfrac{\varphi}{2}\|J\|_{\text{op}},\;
    \tfrac12\|J\|_*,\;
    \tfrac34|\tau| + \tfrac{\varphi}{2}\|K\|_{\text{op}},\;
    \tfrac34|\tau| + \tfrac12\|K\|_*,\;
    \tfrac34|\lambda| + \tfrac{\varphi}{2}\|K_{\text{pol}}\|_{\text{op}},\;
    \tfrac34|\lambda| + \tfrac12\|K_{\text{pol}}\|_*
  \Bigr\}
  \label{eq:alpha-full}
\end{equation}

\begin{theorem}[Tight per-edge threshold]
\label{thm:tight}
With $\alpha_{\text{full}}$ from \eqref{eq:alpha-full}, the threshold
\begin{equation}
  \mu_{\star}(J, h) \;=\; (1 + \epsilon)\, \max_{i \in [N]}
    \Bigl[\, \sum_{k \neq i} \alpha_{\text{full}}(J_{ik}) + \tfrac{2}{3}\, b_i \,\Bigr]
  \label{eq:mu-star}
\end{equation}
obeys $\mu_{\star} \le \lambda_{\min}$ and the conclusion of Theorem~\ref{thm:lambda-min}: for every $\lambda \ge \mu_{\star}$ and per-site-odd $\psi$, $\langle \psi | H(\lambda) | \psi \rangle \ge E_0\, \|\psi\|^2$.
\end{theorem}

The proof is given in Appendix~\ref{app:tight}.

On the canonical couplings the minimum lands on the physically right primitive. Heisenberg $J = J_0 I$ gives $\alpha_{\text{full}} = \tfrac34 |J_0|$ from the isotropic entry, recovering the textbook $\mu_{\star} = \tfrac34 \mathrm{wd}_i$, that is $1.5$ on the 1D unit chain rather than the universal $\varphi \approx 1.618$; pure Ising or XY ($J$ rank one) drops to $\tfrac12 |J|$ through the nuclear entry; pure Dzyaloshinskii--Moriya stays at $\tfrac{\varphi}{2}|D|$, the operator entry, since its singular values are dense.

\section{Proofs from \texorpdfstring{Supp.\ Mat.~S1}{Supp. Mat. S1}}
\label{app:part1-proofs}

The five results stated in SM~\S~\ref{sec:part1} and Appendix~\ref{app:casimir-reg} are proved here: Proposition~\ref{prop:halfsector}, Theorem~\ref{thm:expressivity-ladder}, Theorems~\ref{thm:lambda-min} and \ref{thm:tight}, and Lemma~\ref{lem:certs}. Proposition~\ref{prop:halfsector} is elementary and can be checked by hand. 

\subsection{Proof of Proposition~\ref{prop:halfsector}}
\label{app:prop}
\begin{proof}
By the Peter--Weyl expansion above, $\psi = \sum_{\boldsymbol{J}, \boldsymbol{m}, \boldsymbol{n}} c_{\boldsymbol{J}, \boldsymbol{m}, \boldsymbol{n}}\, \Phi_{\boldsymbol{J}, \boldsymbol{m}, \boldsymbol{n}}$, with $-\Delta_k$ diagonal of eigenvalue $j_k(j_k + 1)$ by
\eqref{eq:casimir-eigen}. For the forward direction,
\eqref{eq:target-sector} subtracts to
\begin{equation}
  0 \;=\; \sum_{\boldsymbol{J}, \boldsymbol{m}, \boldsymbol{n}}
    c_{\boldsymbol{J}, \boldsymbol{m}, \boldsymbol{n}}\,
    \bigl(j_k(j_k+1) - \tfrac{3}{4}\bigr)\,
    \Phi_{\boldsymbol{J}, \boldsymbol{m}, \boldsymbol{n}},
\end{equation}
and orthogonality of the basis forces $c_{\boldsymbol{J}, \boldsymbol{m}, \boldsymbol{n}} = 0$ whenever $j_k \neq \tfrac{1}{2}$. Holding at every site, only $\boldsymbol{J} = (\tfrac{1}{2}, \dots, \tfrac{1}{2})$ survives, and $T_{\boldsymbol{m}, \boldsymbol{n}} := c_{(\tfrac{1}{2}, \dots, \tfrac{1}{2}), \boldsymbol{m}, \boldsymbol{n}}$ gives the stated form. Conversely, applying $-\Delta_k$ to the tensor form, the factors at sites $r \neq k$ pass through and the site-$k$ factor is a spin-$\tfrac{1}{2}$ Wigner matrix on which $-\Delta_k$ acts as $\tfrac{3}{4}$, so $-\Delta_k\, \psi = \tfrac{3}{4}\, \psi$ at every site.
\end{proof}

\subsection{Proof of Theorem~\ref{thm:expressivity-ladder}}
\label{app:expressivity}

\begin{proof}
Schur orthogonality gives, on one copy of $G$,
\begin{equation}
  \int_G \overline{D^{1/2}_{mn}(q)}D^{1/2}_{m'n'}(q)\,\mathrm d q
  =\frac12\delta_{mm'}\delta_{nn'}.
\end{equation}
Taking the product over $N$ sites and using $\|\chi\|=1$ shows $\|\iota_\chi(\Psi)\|_{L^2(G^N)}=\|\Psi\|_{\mathcal H_N}$; hence $\iota_\chi$ is an isometry and in particular is injective.

Peter--Weyl identifies the spin-$\tfrac12$ sector with $\mathcal K_N\otimes\mathcal H_N$, which is exactly $\mathcal F_{\mathrm{lin}}$. Since
\begin{equation}
  \dim \iota_\chi(\mathcal H_N)=2^N
  \quad\text{and}\quad
  \dim \mathcal F_{\mathrm{lin}}=4^N,
\end{equation}
the first inclusion in Eq.~\eqref{eq:expressivity-ladder} is strict for $N\geq1$.

The central element $-I\in G$ acts on $V_j$ as $(-1)^{2j}I$. Consequently a Peter--Weyl coefficient is odd under $q_i\mapsto-q_i$ if and only if its site-$i$ spin $j_i$ is half-integer.  This proves the direct sum in Eq.~\eqref{eq:F-odd} and gives $\mathcal F_{\mathrm{lin}}\subseteq\mathcal F_{\mathrm{odd}}$.  The inclusion is strict because, for example,
\begin{equation}
  D^{3/2}_{3/2,3/2}(q_1)
  \prod_{i=2}^{N}D^{1/2}_{1/2,1/2}(q_i)
  \in\mathcal F_{\mathrm{odd}}\setminus\mathcal F_{\mathrm{lin}},
\end{equation}
with the product over $i\geq2$ omitted when $N=1$.

On $V_{1/2}^{*}\otimes V_{1/2}$, the right-regular differential action generated by the left-invariant fields acts on the second, column/physical factor and leaves the first, row/multiplicity factor invariant. Tensoring over sites therefore gives
\begin{equation}
  \widetilde H\big|_{\mathcal F_{\mathrm{lin}}}
  \cong I_{\mathcal K_N}\otimes H_{\mathcal H_N}.
\end{equation}
Its lowest eigenvalue is $E_0$, which proves Eq.~\eqref{eq:linear-variational-safety}; equality is attained by $\iota_\chi(\Psi_0)$ for any physical ground state $\Psi_0$.

Finally, on a per-site-odd Peter--Weyl block labelled by $\boldsymbol j=(j_1,\ldots,j_N)$,
\begin{equation}
  \widehat C
  =\sum_{i=1}^{N}\left(j_i(j_i+1)-\frac34\right)I.
\end{equation}
Every summand is nonnegative for half-integer $j_i\geq\tfrac12$, and the sum vanishes if and only if every $j_i=\tfrac12$.  Thus $\widehat C\succeq0$ on $\mathcal F_{\mathrm{odd}}$ and $\ker\widehat C=\mathcal F_{\mathrm{lin}}$. Theorem~\ref{thm:tight} gives Eq.~\eqref{eq:nonlinear-variational-safety} for $\lambda\geq\mu_\star(J,h)$.  Conversely the lifted ground state $\iota_\chi(\Psi_0)$ lies in $\ker\widehat C$ and has regularised quotient $E_0$, so the infimum over $\mathcal F_{\mathrm{odd}}$ is exactly $E_0$.
\end{proof}

\subsection{Proof of Theorem~\ref{thm:lambda-min} (sufficient coupling)}
\label{app:suff}
\begin{proof}
Write $q(j) := j(j+1) - \tfrac{3}{4}$ for the Casimir gap and $r(j) := \sqrt{j(j+1)}$ for the Casimir root, and recall the Peter--Weyl decomposition of the per-site-odd subspace
\begin{equation*}
  \mathcal H_{\mathrm{odd}}
  =\bigoplus_{s:j_i\in\{\tfrac12,\tfrac32,\tfrac52,\dots\}}
    \mathcal H_s,
  \qquad
  \mathcal H_s
  :=\bigotimes_{i=1}^{N}\bigl(V_{j_i}^{*}\otimes V_{j_i}\bigr).
\end{equation*}

\begin{itemize}
  \item[\textbf{(1)}] \emph{$\hat H$ and $\hat C$ are block-diagonal in $s$.}
        The spin operators $\hat S_i^a$ act through the right-regular differential representation generated by the left-invariant fields. This action preserves every Peter--Weyl isotypic component $V_{j_i}^{*}\otimes V_{j_i}$ and acts on its second factor. Hence each $\hat S_i^a$ acts within $\mathcal H_s$, and $\hat H=\sum J_{ij}^{ab}\hat S_i^a\hat S_j^b +\sum h_i^a\hat S_i^a$ preserves $\mathcal H_s$. The penalty $\hat C = \sum_l (\hat S_l^2 - \tfrac{3}{4})$ is a function of the per-site Casimirs, so it acts on $\mathcal{H}_{s}$ as the scalar $\hat C \equiv Q(s)\, I$ with $Q(s) := \sum_{l} q(j_{l})$.

  \item[\textbf{(2)}] \emph{Sector-gap upper bound (variational + purification, derived inline).}
        Fix a sector $s$ with raised set $R := R(s) = \{i : j_{i} \neq \tfrac{1}{2}\}$ non-empty, and $R^{c} := [N] \setminus R$. Split the Hamiltonian by which edges/fields touch $R$:
        \begin{equation*}
          \hat H \;=\; \hat H_{\text{off}}
                       \;+\; \hat H_{\text{inc}}
                       \;+\; \hat H^{R}_{\text{field}},
        \end{equation*}
        with $\hat H_{\text{off}}$ collecting pair terms with both endpoints in $R^{c}$ and the field on $R^{c}$; $\hat H_{\text{inc}}$ collecting pair terms with at least one endpoint in $R$; and $\hat H^{R}_{\text{field}}$ the field on $R$. By construction $\hat H_{\text{off}}$ acts trivially on the $R$ factor and lives entirely on $R^{c}$.

        \emph{Upper bound on $E_{0}(\mathcal{H}_{\tfrac{1}{2}})$ (variational).} Let $|\phi_{0}\rangle_{R^{c}}$ be the spin-$\tfrac{1}{2}$ ground state of $\hat H_{\text{off}}$ on
        \begin{equation*}
          \bigotimes_{i \in R^{c}} (V_{1/2}^* \otimes V_{1/2}),
        \end{equation*}
        with energy $E_{0}^{\text{off}, 1/2}$. Let $|\xi\rangle_{R}$ be any unit vector in
        \begin{equation*}
          \bigotimes_{i \in R} (V_{1/2}^* \otimes V_{1/2}).
        \end{equation*}
        The product $|\psi^{\text{var}}_{\tfrac{1}{2}}\rangle := |\phi_{0}\rangle_{R^{c}} \otimes |\xi\rangle_{R}$ lies in $\mathcal{H}_{\tfrac{1}{2}}$, so
        \begin{align*}
          E_{0}(\mathcal{H}_{\tfrac{1}{2}})
            &\;\le\; \langle \psi^{\text{var}}_{\tfrac{1}{2}} | \hat H | \psi^{\text{var}}_{\tfrac{1}{2}} \rangle \\
            &\;=\; \langle \phi_{0} | \hat H_{\text{off}} | \phi_{0} \rangle
                   + \langle \psi^{\text{var}}_{\tfrac{1}{2}} | (\hat H_{\text{inc}} + \hat H^{R}_{\text{field}}) | \psi^{\text{var}}_{\tfrac{1}{2}} \rangle \\
            &\;\le\; E_{0}^{\text{off}, 1/2}
                   \;+\; \bigl\| \hat H_{\text{inc}} + \hat H^{R}_{\text{field}} \bigr\|_{\text{op}, \tfrac{1}{2}}.
        \end{align*}
        The choice of $|\xi\rangle_{R}$ is free; the inequality holds for any choice and hence for the supremum, the operator-norm
        upper bound is independent of $|\xi\rangle$.

        \emph{Lower bound on $E_{0}(\mathcal{H}_{s})$
        (purification / partial trace).} Let $|\psi_{s}\rangle$ be the ground state of $\hat H$ on $\mathcal{H}_{s}$, energy $E_{0}(\mathcal{H}_{s})$. Define the reduced density matrix $\rho_{R^{c}} := \mathrm{tr}_{R} |\psi_{s}\rangle\langle\psi_{s}|$ on $\bigotimes_{i \in R^{c}}(V_{1/2} \otimes V_{1/2}^{*})$ (since $j_{i} = \tfrac{1}{2}$ for $i \in R^{c}$ by definition of $R$). Then
        \begin{align*}
          E_{0}(\mathcal{H}_{s})
            &\;=\; \langle \psi_{s} | \hat H_{\text{off}} | \psi_{s} \rangle
                   \;+\; \langle \psi_{s} | (\hat H_{\text{inc}} + \hat H^{R}_{\text{field}}) | \psi_{s} \rangle \\
            &\;=\; \mathrm{tr}\!\bigl(\hat H_{\text{off}}\, \rho_{R^{c}}\bigr)
                   \;+\; \langle \psi_{s} | (\hat H_{\text{inc}} + \hat H^{R}_{\text{field}}) | \psi_{s} \rangle \\
            &\;\ge\; E_{0}^{\text{off}, 1/2}
                   \;-\; \bigl\| \hat H_{\text{inc}} + \hat H^{R}_{\text{field}} \bigr\|_{\text{op}, s}.
        \end{align*}
        The first equality uses that $\hat H_{\text{off}}$ acts on the $R^{c}$ factor only, so $\langle \hat H_{\text{off}} \rangle = \mathrm{tr}(\hat H_{\text{off}}\, \rho_{R^{c}})$. The inequality uses (i) $\mathrm{tr}(\hat H_{\text{off}} \rho_{R^{c}}) \ge E_{0}^{\text{off}, 1/2} \cdot \mathrm{tr}(\rho_{R^{c}}) = E_{0}^{\text{off}, 1/2}$ (spectral lower bound on $\hat H_{\text{off}}$, which has spectrum $\ge E_{0}^{\text{off}, 1/2}$ on the spin-$\tfrac{1}{2}$ sector where $\rho_{R^{c}}$ is supported); and (ii) $|\langle \psi_{s} | M | \psi_{s} \rangle| \le \|M\|_{\text{op}, s}$ for any operator $M$ acting on $\mathcal{H}_{s}$.

        \emph{Subtract:}
        \begin{equation*}
          E_{0}(\mathcal{H}_{\tfrac{1}{2}}) - E_{0}(\mathcal{H}_{s})
          \;\le\; \underbrace{\bigl\| \hat H_{\text{inc}} + \hat H^{R}_{\text{field}} \bigr\|_{\text{op}, \tfrac{1}{2}}
                  + \bigl\| \hat H_{\text{inc}} + \hat H^{R}_{\text{field}} \bigr\|_{\text{op}, s}}_{=: \; B(s)}.
        \end{equation*}
        The $E_{0}^{\text{off}, 1/2}$ pieces cancel exactly.

  \item[\textbf{(3)}] \emph{Per-edge bounds, derived inline.}
        By triangle inequality on the sum defining $\hat H_{\text{inc}} + \hat H^{R}_{\text{field}}$, each side $t \in \{\tfrac{1}{2}, s\}$ satisfies
        \begin{equation*}
          \bigl\| \hat H_{\text{inc}} + \hat H^{R}_{\text{field}} \bigr\|_{\text{op}, t}
          \;\le\; \sum_{(i, k) \text{ inc. to } R} \bigl\| \vec S_{i}^{\top} J_{ik}\, \vec S_{k} \bigr\|_{\text{op}, t}
                  \;+\; \sum_{i \in R} \bigl\| \vec h_{i}^{\top} \vec S_{i} \bigr\|_{\text{op}, t}.
        \end{equation*}
        Two factor bounds are needed; both follow from one-line spin-algebra calculations.

        \emph{(3a) Field op-norm:
        $\|\vec h_{i}^{\top} \vec S_{i}\|_{\text{op}, V_{j}} = b_{i} \cdot j$.} Pick the orthogonal frame in which $\vec h_{i}$ points along $\hat z$, so $\vec h_{i}^{\top} \vec S_{i} = b_{i}\, S_{i}^{z}$. On $V_{j}$, $S_{i}^{z}$ is diagonal with eigenvalues $\{-j, -j+1, \dots, j\}$; its operator norm equals the maximum absolute eigenvalue, $j$. Rotational invariance of the operator norm on the spin representation gives the result for any direction of $\vec h_{i}$.

        \emph{(3b) Cauchy--Schwarz pair op-norm:
        $\|\vec S_{i}^{\top} J_{ik} \vec S_{k}\|_{\text{op}, V_{j_{i}} \otimes V_{j_{k}}}
        \le A_{ik}\, r(j_{i})\, r(j_{k})$.}
        Take the real SVD $J_{ik} = U \Sigma V^{\top}$ with $\Sigma = \mathrm{diag}(\sigma_{1}, \sigma_{2}, \sigma_{3})$ and $\sigma_{1} = A_{ik}$; rewrite
        \begin{equation*}
          \vec S_{i}^{\top} J_{ik}\, \vec S_{k}
          \;=\; \sum_{a = 1}^{3} \sigma_{a}\,
                (\vec u_{a} \cdot \vec S_{i})\,
                (\vec v_{a} \cdot \vec S_{k})
          \;=\; \sum_{a} \sigma_{a}\, A_{a}\, B_{a},
        \end{equation*}
        with $A_{a} := \vec u_{a} \cdot \vec S_{i}$, $B_{a} := \vec v_{a} \cdot \vec S_{k}$. For any unit $|\psi\rangle \in V_{j_{i}} \otimes V_{j_{k}}$, two successive Cauchy--Schwarz steps give
        \begin{equation*}
          \bigl| \langle\psi| \textstyle\sum_{a} \sigma_{a} A_{a} B_{a} |\psi\rangle \bigr|
          \;\le\; \sum_{a} \sigma_{a}\, \|A_{a} |\psi\rangle\|\, \|B_{a} |\psi\rangle\|
          \;\le\; \sqrt{\textstyle\sum_{a} \sigma_{a}^{2} \|A_{a}|\psi\rangle\|^{2}}\,
                  \sqrt{\textstyle\sum_{a} \|B_{a}|\psi\rangle\|^{2}}.
        \end{equation*}
        For the first factor:
        $\sum_{a} \sigma_{a}^{2} \|A_{a}|\psi\rangle\|^{2}
        \le \sigma_{1}^{2} \sum_{a} \langle\psi| A_{a}^{\dagger} A_{a} |\psi\rangle
        = A_{ik}^{2}\, \langle\psi| \vec S_{i}^{\top} (U U^{\top}) \vec S_{i} |\psi\rangle
        = A_{ik}^{2}\, j_{i}(j_{i}+1)$,
        where $U U^{\top} = I_{3}$ (orthogonal) and
        $\vec S_{i} \cdot \vec S_{i} = j_{i}(j_{i}+1) I$ on $V_{j_{i}}$
        (Casimir). 
        For the second factor:
        $\sum_{a} \|B_{a}|\psi\rangle\|^{2} = \langle\psi| \vec S_{k}
        \cdot \vec S_{k} |\psi\rangle = j_{k}(j_{k}+1)$ by the same
        argument with $V V^{\top} = I_{3}$. Product:
        $|\langle\psi| \vec S_{i}^{\top} J_{ik} \vec S_{k} |\psi\rangle|
        \le A_{ik}\, r(j_{i})\, r(j_{k})$, which is the operator-norm
        bound by taking the supremum over unit $|\psi\rangle$.

        \emph{Packaging.} Using $j_{i}^{e} := j_{i}$ for $i \in R$ and $\tfrac{1}{2}$ for $i \in R^{c}$, and adding the $t = \tfrac{1}{2}$
        and $t = s$ sides:
        \begin{equation*}
          B(s)
          \;\le\; \sum_{(i, k) \text{ inc.}} A_{ik}\!
                  \bigl[\, r(j_{i}^{e})\, r(j_{k}^{e}) + r(\tfrac{1}{2})^{2} \,\bigr]
                  \;+\; \sum_{i \in R} b_{i}\!
                  \bigl[\, j_{i}^{e} + \tfrac{1}{2} \,\bigr].
        \end{equation*}

  \item[\textbf{(4)}] \emph{Per-site accounting.}
        Assign each contribution in the packaged $B(s)$ to a raised endpoint, measured against the gap $q(j)$ that the linear penalty collects there. The reader can read each ratio
        straight off the packaging.

        \emph{Boundary edges.} One endpoint $i$ raised, the other $k$ at $\tfrac{1}{2}$, so the contribution is $A_{ik}\, [r(j_{i}) r(\tfrac{1}{2}) + r(\tfrac{1}{2})^{2}]$. Assigning it to $i$ and dividing by $q(j_{i})$ defines
        \begin{equation*}
          \gamma_{\text{bdy}}(j)
          \;:=\; \frac{r(\tfrac{1}{2})\, \bigl(r(j) + r(\tfrac{1}{2})\bigr)}{q(j)}.
        \end{equation*}

        \emph{Internal edges.} Both endpoints raised, contribution $A_{ik}\, [r(j_{i}) r(j_{k}) + r(\tfrac{1}{2})^{2}]$. Apply AM--GM, $r(j_{i})\, r(j_{k}) \le \tfrac{1}{2}(j_{i}(j_{i}+1) + j_{k}(j_{k}+1))$, split the constant $r(\tfrac{1}{2})^{2} = \tfrac{1}{2} r(\tfrac{1}{2})^{2} + \tfrac{1}{2} r(\tfrac{1}{2})^{2}$ symmetrically, and split half to each endpoint; the $i$-side share, over $q(j_{i})$, defines
        \begin{equation*}
          \gamma_{\text{int}}(j)
          \;:=\; \frac{j(j+1) + \tfrac{3}{4}}{2\, q(j)}.
        \end{equation*}

        \emph{Field at $i \in R$.} Assigning $b_{i}(j_{i} + \tfrac{1}{2})$ to $i$ and dividing by $q(j_i)$ defines
        \begin{equation*}
          \gamma_{h}(j) \;:=\; \frac{j + \tfrac{1}{2}}{q(j)}.
        \end{equation*}

        With $\gamma_{J}(j) := \max\bigl(\gamma_{\text{bdy}}(j), \gamma_{\text{int}}(j)\bigr)$, every incident edge contributes at most $A_{ik}\, \gamma_{J}(j_{i})\, q(j_{i})$ to the $i$-side share, and summing the incident edges into $\mathrm{wd}_{i}$,
        \begin{equation*}
          B(s) \;\le\; \sum_{i \in R} \bigl[\, \gamma_{J}(j_{i})\, \mathrm{wd}_{i}
                                            + \gamma_{h}(j_{i})\, b_{i} \,\bigr]\, q(j_{i}).
        \end{equation*}

  \item[\textbf{(5)}] \emph{Monotonicity: the worst harmonic component sits at $j = \tfrac{3}{2}$.} Treat $j$ as a continuous variable on $(\tfrac{1}{2}, \infty)$ and write $u := j(j+1)$, so $q(j) = u - \tfrac{3}{4}$ and $u' = 2j+1 > 0$. Each ratio is strictly decreasing. For $\gamma_{h}(j) = (j + \tfrac{1}{2}) / q(j)$, the quotient rule with $q'(j) = 2j+1$ and $2(j + \tfrac{1}{2}) = 2j+1$ gives
        \begin{equation*}
          \partial_{j} \gamma_{h}
          \;=\; \frac{q(j) - (j + \tfrac{1}{2})(2j+1)}{q(j)^{2}}
          \;=\; \frac{q(j) - \tfrac{1}{2}(2j+1)^{2}}{q(j)^{2}}
          \;=\; \frac{-\,j^{2} - j - \tfrac{5}{4}}{q(j)^{2}} \;<\; 0.
        \end{equation*}
    For $\gamma_{\text{int}}(j) = (u + \tfrac{3}{4}) / \bigl(2(u - \tfrac{3}{4})\bigr)$, differentiating in $u$ (sign-preserving since $u' > 0$),
        \begin{equation*}
          \partial_{u} \gamma_{\text{int}}
          \;=\; \frac{(u - \tfrac{3}{4}) - (u + \tfrac{3}{4})}{2 (u - \tfrac{3}{4})^{2}}
          \;=\; \frac{-\tfrac{3}{2}}{2 (u - \tfrac{3}{4})^{2}} \;<\; 0.
        \end{equation*}
    For $\gamma_{\text{bdy}}(j) = \tfrac{\sqrt 3}{2}\bigl(\sqrt u + \tfrac{\sqrt 3}{2}\bigr) / (u - \tfrac{3}{4})$, the quotient rule on $u$ (with $\partial_{u} \sqrt u = 1/(2\sqrt u)$) carries the numerator $\tfrac{u - 3/4}{2\sqrt u} - \sqrt u - \tfrac{\sqrt 3}{2} = -\tfrac{u + 3/4}{2\sqrt u} - \tfrac{\sqrt 3}{2} < 0$, so $\partial_{u} \gamma_{\text{bdy}} < 0$. All three decrease on $(\tfrac{1}{2}, \infty)$, so $\gamma_{J} = \max(\gamma_{\text{bdy}}, \gamma_{\text{int}})$ does too. Per-site oddness (\S\,\ref{ssec:casimir}) leaves $j = \tfrac{3}{2}$ as the smallest allowed value above $\tfrac{1}{2}$, where the ratios reach their worst case,
        \begin{equation*}
          \gamma_{\text{bdy}}\!\bigl(\tfrac{3}{2}\bigr) = \frac{1 + \sqrt 5}{4} = \frac{\varphi}{2},
          \qquad
          \gamma_{\text{int}}\!\bigl(\tfrac{3}{2}\bigr) = \frac{3}{4},
          \qquad
          \gamma_{h}\!\bigl(\tfrac{3}{2}\bigr) = \frac{2}{3},
        \end{equation*}
    so $\gamma_{J}(j_{i}) \le \tfrac{\varphi}{2}$ (the boundary ratio wins, $\tfrac{\varphi}{2} > \tfrac{3}{4}$) and $\gamma_{h}(j_{i}) \le \tfrac{2}{3}$ at every raised site. Substituting into the per-site bound of step (4),
        \begin{align*}
          B(s)
            &\;\le\; \sum_{i \in R} \Bigl[\, \tfrac{\varphi}{2}\, \mathrm{wd}_{i}
                                              + \tfrac{2}{3}\, b_{i} \,\Bigr]\, q(j_{i}) \\
            &\;\le\; \Bigl(\max_{i} \bigl[\tfrac{\varphi}{2}\, \mathrm{wd}_{i}
                                    + \tfrac{2}{3}\, b_{i}\bigr]\Bigr)
                       \sum_{i \in R} q(j_{i}) \\
            &\;=\; \frac{\lambda_{\min}(J, h)}{1 + \epsilon}\, Q(s)
             \;\le\; \lambda_{\min}(J, h)\, Q(s).
        \end{align*}

  \item[\textbf{(6)}] \emph{Assemble.}
        Decompose any per-site-odd $\psi = \sum_{s} \psi_{s}$ over the orthogonal sectors. Using block-diagonality (1) on $\hat H$ and the scalar action of $\hat C$ from (1),
        \begin{align*}
          \langle \psi | H(\lambda) | \psi \rangle
            &\;=\; \sum_{s} \Bigl[\, \langle \psi_{s} | \hat H | \psi_{s} \rangle
                                      + \lambda\, Q(s)\, \|\psi_{s}\|^{2} \,\Bigr]
              &&\text{(block-diagonal in $s$)} \\
            &\;\ge\; \sum_{s} \Bigl[\, E_{0}(\mathcal{H}_{s})
                                      + \lambda\, Q(s) \,\Bigr]\, \|\psi_{s}\|^{2}
              &&\text{(within-sector variational principle)} \\
            &\;\ge\; \sum_{s} \Bigl[\, E_{0} - B(s) + \lambda\, Q(s) \,\Bigr]\, \|\psi_{s}\|^{2}
              &&\text{(step 2: $E_{0}(\mathcal{H}_{s}) \ge E_{0} - B(s)$)} \\
            &\;\ge\; \sum_{s} E_{0}\, \|\psi_{s}\|^{2}
              &&\text{(step 5 with $\lambda \ge \lambda_{\min}$)} \\
            &\;=\; E_{0}\, \|\psi\|^{2}.
        \end{align*}
        For the all-$\tfrac{1}{2}$ sector $s = (\tfrac{1}{2}, \dots, \tfrac{1}{2})$ the raised set is empty, $B(s) = 0$ and $Q(s) = 0$, and the chain collapses to the physical variational bound $\langle \psi_{\tfrac{1}{2}} | \hat H | \psi_{\tfrac{1}{2}} \rangle \ge E_{0}\, \|\psi_{\tfrac{1}{2}}\|^{2}$.
\end{itemize}
\end{proof}

\subsection{Proof of Lemma~\ref{lem:certs} (per-edge certificates)}
\label{app:certs}
\begin{proof}
\emph{Operator.} Step (3b) of Appendix~\ref{app:suff} gives $\|\vec S_i^{\top} J \vec S_k\|_{(s_i, s_j)} \le \|J\|_{\text{op}}\, r(s_i) r(s_j)$, and $\le \|J\|_{\text{op}}\, \tfrac34$ on $(\tfrac12, \tfrac12)$, so the left side of \eqref{eq:cert-master} is at most $\|J\|_{\text{op}}[r(s_i) r(s_j) + \tfrac34]$. The ratio $[r(s_i) r(s_j) + \tfrac34]/[q(s_i) + q(s_j)]$ over non-trivial pairs is largest at $(\tfrac32, \tfrac12)$, where it equals $(\tfrac{\sqrt{15}}{2}\tfrac{\sqrt3}{2} + \tfrac34)/3 = (\sqrt5 + 1)/4 = \tfrac{\varphi}{2}$.

\emph{Nuclear.} For rank one $J = \sigma\, u v^{\top}$, $\vec S_i^{\top} J \vec S_k = \sigma\, (u\!\cdot\!\vec S_i)(v\!\cdot\!\vec S_k)$; each factor has operator norm $s$ (eigenvalues $-s, \dots, s$), so the pair norm is $\le |\sigma|\, s_i s_j$, while on $(\tfrac12, \tfrac12)$ the two $\pm\tfrac12$ factors give $|\sigma|/4$. Then $[s_i s_j + \tfrac14]/[q(s_i) + q(s_j)] \le \tfrac12$, because $q(s_i) + q(s_j) - 2(s_i s_j + \tfrac14) = (s_i - s_j)^2 + (s_i + s_j) - 2 \ge 0$ for every non-trivial half-integer pair ($s_i + s_j \ge 2$). Summing the rank-one pieces of the SVD gives $\tfrac12 \|J\|_*$.

\emph{Orthogonal.} For $J = \lambda Q$, $Q \in \mathcal{O}(3)$, the rotated operators $\tilde S_k = Q \vec S_k$ keep the Casimir $\sum_a (\tilde S_k^a)^2 = s_k(s_k+1)$, so $\vec S_i \!\cdot\! \tilde S_k$ is unitarily ($\det Q = +1$) or anti-unitarily ($\det Q = -1$) equivalent to $\vec S_i \!\cdot\! \vec S_k$, whose pair operator norm is $s_i s_j + \min(s_i, s_j)$ (the extreme total-spin multiplets). The ratio $[s_i s_j + \min + \tfrac34]/[q(s_i) + q(s_j)]$ peaks at $(\tfrac32, \tfrac32)$, value $(\tfrac94 + \tfrac32 + \tfrac34)/6 = \tfrac34$.

Additivity is the triangle inequality applied to each operator-norm term in \eqref{eq:cert-master}.
\end{proof}

\subsection{Proof of Theorem~\ref{thm:tight} (tight per-edge threshold)}
\label{app:tight}
\begin{proof}
By Lemma~\ref{lem:certs} and its additivity each entry of \eqref{eq:alpha-full} is a certificate (the isotropic and polar entries split $J$ and sum a primitive on each part), and the minimum of certificates is a certificate, so $\alpha_{\text{full}}(J_{ik})$ satisfies \eqref{eq:cert-master}. Assigning it to the endpoints (a boundary edge has $q(\tfrac12) = 0$, so its whole weight lands on the raised site) and keeping the field bound $j + \tfrac12 \le \tfrac23 q(j)$ of Appendix~\ref{app:suff} (its step 5), steps (2) and (6) there go through with $\sum_k \alpha_{\text{full}}(J_{ik})$ in place of $\tfrac{\varphi}{2}\mathrm{wd}_i$. Since every entry of the minimum is $\le \tfrac{\varphi}{2}\|J_{ik}\|_{\text{op}}$, we have $\mu_{\star} \le \lambda_{\min}$.
\end{proof}

\bibliography{main}

\addtocontents{toc}{\protect\endgroup}
\end{document}

%% file: dataset-instance.tex
\newcommand{\dsEDmax}{22}
\newcommand{\dsEDcount}{2882}

\newcommand{\dsDensifyN}{8}

\newcommand{\dsLookahead}{2}



%% file: generalisation-surface.tikz
\begin{tikzpicture}[
  font=\footnotesize,
  herobar/.style   ={fill=surfhero!85, draw=surfhero!55!black, line width=0.7pt, rounded corners=1.5pt},
  priorbar/.style  ={fill=black!7, draw=black!45, line width=0.4pt, rounded corners=1.5pt},
  partialbar/.style={pattern=north east lines, pattern color=black!40, draw=black!40, line width=0.4pt, rounded corners=1.5pt},
  collbl/.style    ={anchor=base, font=\footnotesize},
  regionlbl/.style ={anchor=base, font=\footnotesize, color=black!75},
  glyphnode/.style ={circle, fill=black!55, inner sep=0pt, minimum size=2.6pt},
  glyphedge/.style ={line width=0.5pt, color=black!45, line cap=round},
  leglbl/.style    ={anchor=west, font=\scriptsize, color=black!75},
]
\def\Xo{4.0}\def\colw{2.7}\def\barh{0.5}
\pgfmathsetmacro{\cA}{\Xo + 0.5*\colw}      
\pgfmathsetmacro{\cB}{\Xo + 1.5*\colw}      
\pgfmathsetmacro{\cC}{\Xo + 2.5*\colw}      
\pgfmathsetmacro{\cD}{\Xo + 3.5*\colw}      
\pgfmathsetmacro{\bnd}{\Xo + 1*\colw}       
\pgfmathsetmacro{\gC}{\Xo + 2*\colw}        
\pgfmathsetmacro{\gD}{\Xo + 3*\colw}        
\pgfmathsetmacro{\Xmax}{\Xo + 4*\colw}      
\def\rA{0.0}\def\rB{-0.85}\def\rC{-1.7}\def\rD{-2.55}

\draw[black!15, line width=0.4pt] (\gC, 0.30) -- (\gC, -2.85);
\draw[black!15, line width=0.4pt] (\gD, 0.30) -- (\gD, -2.85);

\draw[herobar]    (\Xo, \rA-0.5*\barh) rectangle (\Xmax,     \rA+0.5*\barh);
\draw[priorbar]   (\Xo, \rB-0.5*\barh) rectangle (\gD,       \rB+0.5*\barh);
\draw[priorbar]   (\Xo, \rC-0.5*\barh) rectangle (\gC,       \rC+0.5*\barh);
\draw[partialbar] (\gC, \rC-0.5*\barh) rectangle (\gC+0.55,  \rC+0.5*\barh);
\draw[priorbar]   (\Xo, \rD-0.5*\barh) rectangle (\bnd,      \rD+0.5*\barh);
\draw[partialbar] (\bnd,\rD-0.5*\barh) rectangle (\bnd+0.55, \rD+0.5*\barh);

\draw[black!55, dashed, line width=0.6pt] (\bnd, 2.10) -- (\bnd, -2.95);

\node[anchor=east, align=right] at (\Xo-0.22, \rA)
  {{\color{surfhero!72!black}Hamilton-Zero (this work)}\\[-2pt]{\scriptsize\color{black!55}spin-$\tfrac12$ interactions}};
\node[anchor=east, align=right] at (\Xo-0.22, \rB)
  {molecular VMC\\[-2pt]{\scriptsize\color{black!55}Orbformer, Gao}};
\node[anchor=east, align=right] at (\Xo-0.22, \rC)
  {electron gas\\[-2pt]{\scriptsize\color{black!55}LEM, QERNEL}};
\node[anchor=east, align=right] at (\Xo-0.22, \rD)
  {spin NQS\\[-2pt]{\scriptsize\color{black!55}FNQS, Attn-FM}};

\node[collbl] at (\cA, 0.92) {edge weights};
\node[collbl] at (\cB, 0.92) {system size};
\node[collbl] at (\cC, 0.92) {topology};
\node[collbl] at (\cD, 0.92) {interaction type};


\draw[pauliz!35, line width=0.7pt, line cap=round] (\cA+0.30,1.50) -- (\cA,1.95);
\draw[pauliz!65, line width=1.3pt, line cap=round] (\cA,1.95) -- (\cA-0.30,1.50);
\draw[pauliz,    line width=2.2pt, line cap=round] (\cA-0.30,1.50) -- (\cA+0.30,1.50);
\node[glyphnode] at (\cA,1.95) {};
\node[glyphnode] at (\cA-0.30,1.50) {};
\node[glyphnode] at (\cA+0.30,1.50) {};

\draw[glyphedge] (\cB-0.15,1.92) -- (\cB+0.15,1.92);
\node[glyphnode] at (\cB-0.15,1.92){};\node[glyphnode] at (\cB+0.15,1.92){};
\draw[glyphedge] (\cB-0.25,1.69) -- (\cB+0.25,1.69);
\node[glyphnode] at (\cB-0.25,1.69){};\node[glyphnode] at (\cB,1.69){};\node[glyphnode] at (\cB+0.25,1.69){};
\draw[glyphedge] (\cB-0.36,1.46) -- (\cB+0.36,1.46);
\node[glyphnode] at (\cB-0.36,1.46){};\node[glyphnode] at (\cB-0.12,1.46){};\node[glyphnode] at (\cB+0.12,1.46){};\node[glyphnode] at (\cB+0.36,1.46){};

\draw[glyphedge] (\cC-0.34,1.92) -- (\cC-0.52,1.55) -- (\cC-0.16,1.55) -- (\cC-0.34,1.92);
\node[glyphnode] at (\cC-0.34,1.92){};\node[glyphnode] at (\cC-0.52,1.55){};\node[glyphnode] at (\cC-0.16,1.55){};
\draw[glyphedge] (\cC+0.14,1.90) -- (\cC+0.50,1.90) -- (\cC+0.50,1.54) -- (\cC+0.14,1.54) -- (\cC+0.14,1.90);
\node[glyphnode] at (\cC+0.14,1.90){};\node[glyphnode] at (\cC+0.50,1.90){};\node[glyphnode] at (\cC+0.50,1.54){};\node[glyphnode] at (\cC+0.14,1.54){};

\fill[paulix]       (\cD-0.29,1.99) rectangle ++(0.18,-0.18);
\fill[pauliy]       (\cD-0.09,1.79) rectangle ++(0.18,-0.18);
\fill[pauliz]       (\cD+0.11,1.59) rectangle ++(0.18,-0.18);
\fill[surfhero!40]  (\cD+0.11,1.99) rectangle ++(0.18,-0.18);
\fill[surfhero!40]  (\cD-0.29,1.59) rectangle ++(0.18,-0.18);
\foreach \r in {0,1,2}{\foreach \c in {0,1,2}{
  \pgfmathsetmacro{\cx}{\cD-0.29+\c*0.20}\pgfmathsetmacro{\cy}{1.99-\r*0.20}
  \draw[black!30, line width=0.3pt] (\cx,\cy) rectangle ++(0.18,-0.18);
}}

\node[herobar,    minimum width=0.5cm, minimum height=0.28cm, anchor=west] (lh)  at (1.00,-3.40) {};
\node[leglbl] at (lh.east)  {\,this work};
\node[priorbar,   minimum width=0.5cm, minimum height=0.28cm, anchor=west] (lp)  at (3.10,-3.40) {};
\node[leglbl] at (lp.east)  {\,prior foundation models};
\node[partialbar, minimum width=0.5cm, minimum height=0.28cm, anchor=west] (lpa) at (7.30,-3.40) {};
\node[leglbl] at (lpa.east) {\,partial / model-dependent};
\end{tikzpicture}

%% file: graph-heisenberg.tikz
\begin{tikzpicture}[
  font=\footnotesize,
  site/.style ={circle, fill=black, inner sep=0pt, minimum size=5pt},
  edge/.style ={line width=1.2pt, line cap=round},
  slabel/.style={font=\footnotesize, below=5pt},
]
\node[site] (s1) at (0.0,0) {};
\node[site] (s2) at (1.3,0) {};
\node[site] (s3) at (2.6,0) {};
\node[site] (s4) at (3.9,0) {};

\foreach \i/\j in {s1/s2, s2/s3, s3/s4} {
  \draw[paulix, edge, transform canvas={yshift= 3pt}] (\i) -- (\j);
  \draw[pauliy, edge]                                  (\i) -- (\j);
  \draw[pauliz, edge, transform canvas={yshift=-3pt}] (\i) -- (\j);
}

\node[slabel] at (s1) {1};
\node[slabel] at (s2) {2};
\node[slabel] at (s3) {3};
\node[slabel] at (s4) {4};
\end{tikzpicture}

%% file: graph-kitaev.tikz
\begin{tikzpicture}[
  font=\footnotesize,
  site/.style ={circle, fill=black, inner sep=0pt, minimum size=5pt},
  edge/.style ={line width=1.2pt, line cap=round},
  slabel/.style={font=\footnotesize, below=5pt},
]
\node[site] (s1) at (0.0,0) {};
\node[site] (s2) at (1.3,0) {};
\node[site] (s3) at (2.6,0) {};
\node[site] (s4) at (3.9,0) {};

\draw[pauliy, edge] (s1) -- ($(s1)!0.5!(s2)$);
\draw[pauliz, edge] ($(s1)!0.5!(s2)$) -- (s2);
\draw[pauliz, edge] (s2) -- ($(s2)!0.5!(s3)$);
\draw[paulix, edge] ($(s2)!0.5!(s3)$) -- (s3);
\draw[paulix, edge] (s3) -- ($(s3)!0.5!(s4)$);
\draw[pauliy, edge] ($(s3)!0.5!(s4)$) -- (s4);

\node[slabel] at (s1) {1};
\node[slabel] at (s2) {2};
\node[slabel] at (s3) {3};
\node[slabel] at (s4) {4};
\end{tikzpicture}

%% file: graph-compass.tikz
\begin{tikzpicture}[
  font=\footnotesize,
  site/.style ={circle, fill=black, inner sep=0pt, minimum size=5pt},
  edge/.style ={line width=1.2pt, line cap=round},
  slabel/.style={font=\footnotesize, inner sep=1.5pt, fill=white},
]
\node[site] (s11) at (0.0, 0.0) {};
\node[site] (s21) at (1.3, 0.0) {};
\node[site] (s31) at (2.6, 0.0) {};
\node[site] (s12) at (0.0, 1.1) {};
\node[site] (s22) at (1.3, 1.1) {};
\node[site] (s32) at (2.6, 1.1) {};

\foreach \i/\j in {s11/s21, s21/s31, s12/s22, s22/s32} {
  \draw[paulix, edge] (\i) -- (\j);
}
\foreach \i/\j in {s11/s12, s21/s22, s31/s32} {
  \draw[pauliy, edge] (\i) -- (\j);
}
\end{tikzpicture}

%% file: graph-longrange.tikz
\begin{tikzpicture}[
  font=\footnotesize,
  site/.style ={circle, fill=black, inner sep=0pt, minimum size=5pt},
  edge/.style ={line width=1.2pt, line cap=round},
  slabel/.style={font=\footnotesize, below=5pt},
]
\node[site] (s1) at (0.0, 0) {};
\node[site] (s2) at (1.1, 0) {};
\node[site] (s3) at (2.2, 0) {};
\node[site] (s4) at (3.3, 0) {};
\node[site] (s5) at (4.4, 0) {};

\foreach \i/\j in {s1/s2, s2/s3, s3/s4, s4/s5} {
  \draw[pauliz, edge] (\i) -- (\j);
}

\foreach \i/\j in {s1/s3, s2/s4, s3/s5} {
  \draw[pauliz, edge, bend left=45] (\i) to (\j);
}

\node[slabel] at (s1) {1};
\node[slabel] at (s2) {2};
\node[slabel] at (s3) {3};
\node[slabel] at (s4) {4};
\node[slabel] at (s5) {5};
\end{tikzpicture}

%% file: architecture-short.tikz
\begin{tikzpicture}[
  x=1cm, y=1cm,
  font=\scriptsize,
  block/.style={draw, rounded corners=2pt, align=center,
    inner xsep=3pt, inner ysep=2.5pt, minimum height=0.74cm, line width=0.45pt},
  input/.style ={block, fill=black!5,   draw=black!55,        text width=1.05cm, font=\scriptsize\itshape},
  prep/.style  ={block, fill=blue!9,    draw=blue!55!black,   text width=1.45cm},
  feat/.style  ={block, fill=cyan!16,   draw=cyan!65!black,   text width=1.45cm},
  trunk/.style ={block, fill=teal!24,   draw=teal!60!black,   text width=1.95cm},
  ctxer/.style ={block, fill=teal!12,   draw=teal!55!black,   text width=2.05cm},
  leaf/.style  ={block, fill=violet!14, draw=violet!55!black, text width=1.30cm},
  tree/.style  ={block, fill=violet!30, draw=violet!65!black, text width=1.70cm},
  output/.style={block, fill=black!4,   draw=black!50,        text width=1.50cm, font=\scriptsize\itshape},
  route/.style ={block, fill=routecol!22,   draw=routecol!75!black,   text width=2.70cm},
  qInput/.style={block, fill=mcmcsample!28, draw=mcmcsample!85!black, text width=1.85cm},
  wire/.style={-{Stealth[length=2mm, width=1.3mm]}, line width=0.5pt,
    shorten <=1pt, shorten >=1pt},
  routewire/.style={wire, draw=routecol!80!black},
  gwire/.style={-, line width=0.9pt, draw=surfhero!75!black, opacity=0.55},
]

\node[input]  (H)    at ( 0.00, 0) {$(J,h)$};
\node[prep]   (prep) at ( 1.70, 0) {Spectral\\normalise};
\node[feat]   (feat) at ( 3.55, 0) {Featuriser};
\node[trunk]  (tr)   at ( 5.65, 0) {Pre-norm trunk\\[-0.25em]{\tiny $\times L$}};
\node[ctxer]  (ctx)  at ( 8.10, 0) {Leaf\\contextualiser};
\node[leaf]   (lb)   at (10.25, 0) {Odd leaf\\builder};
\node[tree]   (tree) at (12.15, 0) {Merge-tree\\readout};
\node[output] (out)  at (14.15, 0) {$\log\psi_\theta(q\,|\,p)$};

\node[route]  (rp) at (6.90, 1.62) {Route policy\\[-0.20em]{ deep RL}};
\node[qInput] (q)  at (10.25, 1.62) {Spin config $q$\\[-0.20em]{\tiny MCMC on $\mathrm{SU}(2)^N$}};

\draw[gwire, rounded corners=3pt]
  (feat.south) -- (3.55,-1.05) -- (12.15,-1.05) -- (tree.south);
\node[font=\tiny, text=surfhero!75!black, anchor=south]
  at (7.85,-1.05) {global stream $g$};

\draw[wire] (H)    -- (prep);
\draw[wire] (prep) -- (feat);
\draw[wire] (feat) -- (tr);
\draw[wire] (tr)   -- (ctx);
\draw[wire] (ctx)  -- (lb);
\draw[wire] (lb)   -- (tree);
\draw[wire] (tree) -- (out);

\draw[routewire] (tr.north)  -- (tr.north |- rp.south);
\draw[routewire] (rp.south -| ctx.north) -- (ctx.north);
\draw[wire]      (q.south)   -- (lb.north);

\end{tikzpicture}

%% file: architecture.tikz
\begin{tikzpicture}[
  >={Stealth[length=2.5mm, width=1.6mm]},
  font=\small,
  block/.style={
    draw, rounded corners=2pt, align=center,
    minimum height=0.95cm, minimum width=5.2cm,
    inner xsep=6pt, inner ysep=4pt,
    line width=0.55pt,
  },
  input/.style    ={block, fill=black!5,           draw=black!55,
                    minimum width=3.6cm, minimum height=0.65cm,
                    font=\small\itshape},
  qInput/.style   ={block, fill=mcmcsample!28,     draw=mcmcsample!85!black,
                    minimum width=3.2cm, minimum height=0.85cm,
                    font=\small\itshape},
  prep/.style     ={block, fill=blue!9,            draw=blue!55!black},
  feat/.style     ={block, fill=cyan!16,           draw=cyan!65!black},
  trunk/.style    ={block, fill=teal!24,           draw=teal!60!black,
                    minimum height=1.05cm},
  ctxer/.style    ={block, fill=teal!12,           draw=teal!55!black,
                    minimum height=0.95cm},
  route/.style    ={block, fill=routecol!22,       draw=routecol!75!black,
                    minimum width=3.6cm, minimum height=1.05cm},
  leaf/.style     ={block, fill=violet!14,         draw=violet!55!black,
                    minimum height=1.0cm},
  tree/.style     ={block, fill=violet!30,         draw=violet!65!black,
                    minimum height=1.1cm},
  output/.style   ={font=\small\itshape, align=center,
                    draw=black!50, fill=black!4,
                    rounded corners=2pt,
                    minimum height=0.7cm, minimum width=5.2cm,
                    inner xsep=6pt, inner ysep=4pt,
                    line width=0.55pt},
  arr/.style      ={->, thick, line width=0.55pt},
  shp/.style      ={font=\scriptsize\itshape, color=black!60,
                    fill=white, inner sep=1.5pt, align=center},
  routelbl/.style ={font=\scriptsize\itshape, color=routecol!75!black,
                    fill=white, inner sep=1.5pt, align=center},
]

\node[input]   at (0,  0)     (H)    {Hamiltonian\\$(J,\,h)$};
\node[prep]    at (0, -1.50)  (Jpp)  {Build $J'',\,h''$};
\node[feat]    at (0, -3.00)  (feat) {SystemFeaturizer};
\node[trunk]   at (0, -4.55)  (tr)   {Pure-even trunk, $\times L$ blocks};
\node[ctxer]   at (0, -6.15)  (ctx)
  {Leaf contextualizer $\times 2$\\
   \scriptsize refines $(\ell^{(L)}, e^{(L)})$ over the slot frame};
\node[leaf]    at (0, -7.80)  (lb)
  {Per-site leaf builder\\
   \scriptsize takes $(\hat\ell_i,\, g,\, q_i)$, emits the odd carrier $u_i^{(0)}$};
\node[tree]    at (0, -9.40)  (tree)
  {Merge-tree readout\\
   \scriptsize bit-rev leaves; rank-4 quadrilinear merge; level attention};
\node[output]  at (0, -10.90) (out)
  {$\log\psi_\theta(q \mid p)\;\in\;\mathbb{C}$};

\node[route]   at ( 6.35, -4.55) (rp)
  {Route policy\\
   \scriptsize $p \sim \pi_\theta(p \,|\, h, \ell^{(L)}, e^{(L)}, m)$};
\node[qInput]  at ( 5.6, -7.80) (q)
  {Spin config $q$\\$\in(\mathrm{SU}(2))^N$\\from MCMC};

\draw[arr] (H)   -- node[shp, right=2pt] {$J,h$} (Jpp);
\draw[arr] (Jpp) -- node[shp, right=2pt] {$J''\in\mathbb{R}^{N\!\times\! N\!\times\! 10},\;\, h''\in\mathbb{R}^{N\!\times\! 3}$} (feat);
\draw[arr] (feat) -- node[shp, right=2pt] {$e^{(0)},\,\ell^{(0)},\,g^{(0)}$} (tr);

\draw[arr] (tr.south) --
  node[shp, right=2pt]
  {$\ell^{(L)},\, e^{(L)},\, g$}
  (ctx.north);
\draw[arr] (ctx.south) --
  node[shp, right=2pt]
  {$\hat\ell \in \mathbb{R}^{\hat N\!\times\! d_\ell},\;\, g$}
  (lb.north);

\draw[arr] (lb) --
  node[shp, right=2pt]
  {$u^{(0)}\in\mathbb{R}^{\hat N \!\times\! d_u}$}
  (tree);

\draw[arr] (tree) -- node[shp, right=2pt] {$u_{\mathrm{root}}$} (out);

\draw[arr] (tr.east)
  -- node[shp, above=2pt] {$\ell^{(L)},\, e^{(L)},\, g$}
  (rp.west);

\draw[arr] (rp.south west)
  -- node[routelbl, pos=0.55, above=2pt, sloped]
       {route $p$}
     (ctx.north east);

\draw[arr] (q.west) -- node[shp, above=2pt] {$q$} (lb.east);

\end{tikzpicture}

%% file: architecture-featurizer.tikz
\begin{tikzpicture}[
  >={Stealth[length=2.4mm, width=1.5mm]},
  font=\small,
  block/.style={
    draw, rounded corners=2pt, align=center,
    minimum height=0.9cm, minimum width=4.5cm,
    inner xsep=4pt, inner ysep=3pt,
    line width=0.55pt,
  },
  iotxt/.style    ={font=\small\itshape, align=center,
                    draw=black!55, fill=black!5,
                    rounded corners=2pt,
                    minimum height=0.65cm, minimum width=3.6cm,
                    inner xsep=4pt, inner ysep=3pt,
                    line width=0.55pt},
  polar/.style    ={block, fill=cyan!4,    draw=cyan!50!black,
                    minimum height=0.7cm},
  embed/.style    ={block, fill=cyan!7,    draw=cyan!60!black},
  attn/.style     ={block, fill=cyan!20,   draw=cyan!70!black},
  cross/.style    ={block, fill=cyan!34,   draw=cyan!80!black},
  arr/.style      ={->, thick, line width=0.55pt},
  skip/.style     ={->, thick, dashed, color=violet!70!black, line width=0.6pt},
  ann/.style      ={font=\footnotesize, fill=white, inner sep=1.5pt, align=center},
  skiplbl/.style  ={font=\footnotesize, color=violet!70!black, fill=white, inner sep=1.5pt, align=center},
]

\node[iotxt] (inJ)  at (-3.0, 0)     {$J'' \in \mathbb{R}^{N{\times}N{\times}10}$};
\node[polar] (polJ) at (-3.0,-1.45)  {Polar split};
\node[embed] (emb)  at (-3.0,-3.00)  {Per-bond embedding\\{\footnotesize grouped RMS}};
\node[attn]  (col)  at (-3.0,-4.55)  {Column attention};
\node[attn]  (rraw) at (-3.0,-6.10)  {Row attention (bonds)};

\node[iotxt] (inh)  at ( 3.0, 0)     {$h'' \in \mathbb{R}^{N{\times}3}$};
\node[polar] (polh) at ( 3.0,-1.45)  {Polar split};
\node[attn]  (zee)  at ( 3.0,-3.00)  {Zeeman base\\{\footnotesize grouped RMS}};
\node[attn]  (rz)   at ( 3.0,-6.10)  {Row attention (field)};

\node[cross] (fuse) at ( 0,-7.70)    {Global fusion};
\node[cross] (loc)  at ( 0,-9.10)    {Local cross-conditioning};
\node[cross] (edge) at ( 0,-10.50)   {Edge cross-conditioning\\{\footnotesize global FiLM}};
\node[iotxt] (out)  at ( 0,-11.85)
  {$(e,\;\ell,\;g)$};

\draw[arr] (inJ)  -- (polJ);
\draw[arr] (polJ) -- node[ann, right=1pt] {$\hat{x}^{J}\!\in\!\mathbb{R}^{N{\times}N{\times}24}$} (emb);
\draw[arr] (emb)  -- node[ann, right=1pt] {$B$}              (col);
\draw[arr] (col)  -- node[ann, right=1pt] {$\ell^{\text{raw}}$} (rraw);

\draw[arr] (inh)  -- (polh);
\draw[arr] (polh) -- node[ann, right=1pt] {$\hat{x}^{h}\!\in\!\mathbb{R}^{N{\times}7}$} (zee);
\draw[arr] (zee)  -- node[ann, right=1pt] {$\ell^{\text{z}}$} (rz);

\draw[arr] (rraw.south) |- node[ann, pos=0.78, above] {$g^{\text{raw}}$} (fuse.west);
\draw[arr] (rz.south)   |- node[ann, pos=0.78, above] {$g^{\text{z}}$}   ($(fuse.east)+(0,0.14)$);

\draw[arr] (fuse) -- node[ann, right=1pt] {$g$}      (loc);
\draw[arr] (loc)  -- node[ann, right=1pt] {$\ell$}   (edge);
\draw[arr] (edge) -- node[ann, right=1pt] {$e$}      (out);

\draw[skip, rounded corners=3pt] (col.west) -- ++(-0.5,0)
  |- node[skiplbl, pos=0.85, above=1pt] {$\ell^{\text{raw}}$} (loc.west);
\draw[skip, rounded corners=3pt] (emb.west) -- ++(-0.9,0)
  |- node[skiplbl, pos=0.88, above=1pt] {$B$} (edge.west);
\draw[skip, rounded corners=3pt] (zee.east) -- ++(0.5,0)
  |- node[skiplbl, pos=0.85, above=1pt] {$\ell^{\text{z}}$} (loc.east);
\draw[skip, rounded corners=3pt, preaction={draw=white, line width=3pt}]
  ($(fuse.east)+(0,-0.14)$) -- ++(0.35,0) |- (edge.east);
\node[skiplbl, anchor=west] at (2.72,-9.75) {$g$ (FiLM)};
\draw[skip, -] (inJ.east) -- (0,0);
\draw[skip, -] (inh.west) -- (0,0);
\fill[violet!70!black] (0,0) circle (1.2pt);
\draw[skip] (0,0)
  -- node[skiplbl, right=1pt, pos=0.55] {$\phi_{\mathrm{spec}},\, \phi_{Jh}$}
  (fuse.north);

\end{tikzpicture}

%% file: architecture-gladder.tikz
\begin{tikzpicture}[
  >={Stealth[length=2.4mm, width=1.5mm]},
  font=\small,
  block/.style={
    draw, rounded corners=2pt, align=center,
    inner xsep=4pt, inner ysep=3.5pt,
    line width=0.55pt,
  },
  gIn/.style  ={block, fill=magenta!9,  draw=magenta!55!black,
                minimum height=0.55cm, minimum width=1.0cm,
                font=\small\itshape},
  gSt/.style  ={block, fill=magenta!14, draw=magenta!60!black,
                minimum height=0.62cm, minimum width=1.85cm,
                font=\scriptsize},
  sink/.style ={block, fill=teal!10,    draw=teal!60!black,
                minimum height=0.62cm, minimum width=1.85cm,
                font=\scriptsize},
  garr/.style ={->, thick, color=magenta!70!black, line width=0.55pt},
  readarr/.style={->, thick, dashed, color=teal!65!black, line width=0.5pt},
  ann/.style  ={font=\scriptsize\itshape, color=black!55,
                fill=white, inner sep=1.5pt, align=center},
]

\node[gIn] (seed) at (-6.6, 0) {$g$};
\node[gSt] (st1)  at (-4.5, 0) {trunk block $\times L$};
\node[gSt] (st2)  at (-2.2, 0) {post-trunk};

\draw[garr] (seed.east)
  -- node[ann, above=0.5pt, pos=0.42] {$g^{(0)}$}
  (st1.west);
\draw[garr] (st1.east) -- (st2.west);

\node[gSt] (p1) at (0.6,  1.15) {ctx block $\times 2$};
\node[gSt] (r1) at (0.6, -1.15) {ctx block $\times 2$};

\draw[garr, rounded corners=3pt]
  (st2.east) -- ++(0.35, 0) |- (p1.west);
\draw[garr, rounded corners=3pt]
  (st2.east) -- ++(0.35, 0) |- (r1.west);

\node[gSt]  (p2) at (2.9,  1.15) {post-ctx};
\node[gSt]  (p3) at (5.2,  1.15) {merge level $\times K$};
\node[sink] (p4) at (7.5,  1.15) {root readout};

\draw[garr] (p1.east) -- (p2.west);
\draw[garr] (p2.east) -- (p3.west);
\draw[garr] (p3.east) -- (p4.west);

\node[gSt]  (r2) at (2.9, -1.15) {post-ctx};
\node[gSt]  (r3) at (5.2, -1.15) {position $t$};
\node[sink] (r4) at (7.5, -1.15) {pointer};

\draw[garr] (r1.east) -- (r2.west);
\draw[garr] (r2.east)
  -- node[ann, above=0.5pt, pos=0.5] {$g$}
  (r3.west);
\draw[garr] (r3.east)
  -- node[ann, above=0.5pt, pos=0.42] {$g^{(t)}$}
  (r4.west);

\end{tikzpicture}

%% file: architecture-trunk.tikz
\begin{tikzpicture}[
  >={Stealth[length=2.4mm, width=1.5mm]},
  font=\small,
  block/.style={
    draw, rounded corners=2pt, align=center,
    inner xsep=6pt, inner ysep=4pt,
    line width=0.55pt,
  },
  ioL/.style   ={block, fill=teal!10,  draw=teal!60!black,
                 minimum height=0.55cm, minimum width=2.4cm,
                 font=\small\itshape},
  ioE/.style   ={block, fill=teal!22,  draw=teal!65!black,
                 minimum height=0.55cm, minimum width=2.4cm,
                 font=\small\itshape},
  gIn/.style   ={block, fill=magenta!9, draw=magenta!55!black,
                 minimum height=0.55cm, minimum width=1.2cm,
                 font=\small\itshape},
  edgeUp/.style={block, fill=teal!26,  draw=teal!70!black,
                 minimum height=0.7cm, minimum width=3.6cm},
  attn/.style  ={block, fill=teal!14,  draw=teal!60!black,
                 minimum height=0.7cm, minimum width=3.6cm},
  ffn/.style   ={block, fill=teal!8,   draw=teal!55!black,
                 minimum height=0.7cm, minimum width=3.6cm},
  add/.style   ={draw, circle, inner sep=0pt, minimum size=4.5mm,
                 line width=0.55pt, fill=white, font=\footnotesize},
  gUp/.style   ={block, fill=magenta!14, draw=magenta!60!black,
                 minimum height=0.7cm, minimum width=2.6cm},
  arr/.style   ={->, thick, line width=0.55pt},
  biasarr/.style={->, thick, dashed, color=teal!65!black, line width=0.55pt},
  garr/.style  ={->, thick, color=magenta!70!black, line width=0.55pt},
  ann/.style   ={font=\scriptsize\itshape, color=black!55,
                 fill=white, inner sep=1.5pt, align=center},
]

\node[ioE] (ein) at (-3.0,  0.0) {$e^{(\beta)}$};
\node[ioL] (lin) at ( 2.2,  0.0) {$\ell^{(\beta)}$};
\node[gIn] (g)   at ( 6.1,  0.0) {$g^{(\beta)}$};

\node[edgeUp] (eu)   at (-3.0, -1.7) {pre-RMS edge update};
\node[add]    (eadd) at (-3.0, -2.9) {$+$};
\node[ioE]    (eout) at (-3.0, -3.8) {$e^{(\beta+1)}$};

\draw[arr] (ein.south) -- (eu.north);
\draw[arr] (eu.south)  -- (eadd.north);
\draw[arr] (eadd.south) -- (eout.north);
\draw[arr, rounded corners=3pt]
  (ein.west) -- ++(-0.9, 0) |- (eadd.west);
\node[ann] at (-5.5, -1.45) {$\tfrac{1}{\sqrt L}$};

\draw[biasarr] (lin.south west)
  -- node[ann, pos=0.45, above=1pt] {$\ell_i,\, \ell_j$}
  (eu.north east);

\node[attn] (at)    at (2.2, -2.3) {pre-RMS edge-biased MHA};
\node[add]  (ladd1) at (2.2, -3.5) {$+$};
\node[ffn]  (ff)    at (2.2, -4.7) {pre-RMS FFN};
\node[add]  (ladd2) at (2.2, -5.7) {$+$};
\node[ioL]  (lout)  at (2.2, -6.5) {$\ell^{(\beta+1)}$};

\draw[arr] (lin.south)  -- (at.north);
\draw[arr] (at.south)   -- (ladd1.north);
\draw[arr] (ladd1.south) -- (ff.north);
\draw[arr] (ff.south)   -- (ladd2.north);
\draw[arr] (ladd2.south) -- (lout.north);

\draw[arr, rounded corners=3pt]
  (lin.east) -- ++(0.95, 0) |- (ladd1.east);
\node[ann] at (4.75, -1.45) {$\tfrac{1}{\sqrt L}$};

\draw[arr, rounded corners=3pt]
  ($(ladd1.south) + (0, -0.12)$) -- ++(2.15, 0) |- (ladd2.east);
\node[ann] at (3.30, -5.44) {$\tfrac{1}{\sqrt L}$};

\draw[biasarr, rounded corners=3pt]
  (eout.east)
  -| node[ann, pos=0.25] {fresh bonds}
  ($(at.south) + (-1.0, 0)$);

\node[gUp] (gup)  at (6.1, -5.8) {pool $+$ update};
\node[gIn] (gout) at (6.1, -6.9) {$g^{(\beta+1)}$};

\draw[garr] (g.south) -- (gup.north);
\draw[garr] (gup.south) -- (gout.north);

\draw[garr, preaction={draw=white, line width=2.6pt}] (6.1, -4.7)
  -- node[ann, pos=0.35, above=1pt] {$W_g$}
  (ff.east);

\draw[biasarr, rounded corners=3pt]
  (lout.east)
  -- node[ann, above=1pt, pos=0.5] {$\ell^{(\beta+1)}$}
  (5.5, -6.5) -- (5.5, -6.15);

\end{tikzpicture}

%% file: architecture-leaf-builder.tikz
\begin{tikzpicture}[
  >={Stealth[length=2.4mm, width=1.5mm]},
  font=\small,
  block/.style={
    draw, rounded corners=2pt, align=center,
    minimum height=0.74cm, minimum width=2.6cm,
    inner xsep=4pt, inner ysep=3pt,
    line width=0.55pt,
  },
  ctxIn/.style    ={block, fill=teal!16,       draw=teal!60!black,
                    minimum width=2.4cm, minimum height=0.6cm,
                    font=\small\itshape},
  gIn/.style      ={block, fill=teal!28,       draw=teal!65!black,
                    minimum width=1.5cm, minimum height=0.6cm,
                    font=\small\itshape},
  qIn/.style      ={block, fill=mcmcsample!28, draw=mcmcsample!85!black,
                    minimum width=2.0cm, minimum height=0.6cm,
                    font=\small\itshape},
  evenGate/.style ={block, fill=teal!36,       draw=teal!70!black,
                    minimum height=1.0cm, minimum width=3.6cm},
  qLin/.style     ={block, fill=orange!18,     draw=orange!75!black,
                    minimum width=3.4cm},
  qProj/.style    ={block, fill=orange!30,     draw=orange!80!black,
                    minimum width=3.4cm},
  mBox/.style     ={block, fill=violet!14,     draw=violet!55!black,
                    minimum width=3.4cm,
                    font=\small\itshape},
  uProj/.style    ={block, fill=violet!26,     draw=violet!65!black,
                    minimum width=3.4cm},
  rmsBox/.style   ={block, fill=violet!18,     draw=violet!60!black,
                    minimum width=3.4cm},
  carrierOut/.style ={block, fill=violet!40,   draw=violet!75!black,
                    minimum width=5.4cm,
                    font=\small\itshape},
  had/.style      ={draw, circle, inner sep=0pt, minimum size=7mm,
                    line width=0.6pt, fill=white,
                    font=\small},
  arr/.style      ={->, thick, line width=0.55pt},
]

\node[ctxIn]   (lin) at (-2.5,  0.0) {$\hat\ell_i\!\in\!\mathbb{R}^{d_\ell}$};
\node[gIn]     (gin) at ( 0.0,  0.0) {$g\!\in\!\mathbb{R}^{d_g}$};
\node[qIn]     (qin) at ( 2.5,  0.0) {$q_i\!\in\!\mathbb{R}^4$};

\node[evenGate] (gate) at (-1.5, -1.55)
  {even gate\\$m_i = W_h\,[\hat\ell_i,\, g]$};
\node[mBox]     (mout) at (-1.5, -3.30)
  {$m_i\in\mathbb{R}^{R}$ \;(no $q_i$)};

\draw[arr] (lin.south) -- (lin.south |- gate.north);
\draw[arr] (gin.south) -- (gin.south |- gate.north);
\draw[arr] (gate.south) -- (mout.north);

\node[qLin]  (wqz)  at ( 2.5, -1.55) {$z_i = W_{q\to z}\,q_i$\; };
\node[qProj] (vproj) at ( 2.5, -3.30) {$V z_i\in\mathbb{R}^{R}$\;};

\draw[arr] (qin.south) -- (wqz.north);
\draw[arr] (wqz.south) -- (vproj.north);

\node[had]                               at (0.5, -4.45) (had) {$\odot$};
\node[uProj,  below=0.50cm of had]  (uproj)
  {$U\,(m_i \odot V z_i)$};
\node[rmsBox, below=0.40cm of uproj] (rms)
  {scale-norm\; (parity-odd)};
\node[carrierOut, below=0.40cm of rms] (uout)
  {$u_i^{(0)} \in\mathbb{R}^{d_u}$,\;\; odd in $q_i$};

\draw[arr] (mout.south)  |- (had.west);
\draw[arr] (vproj.south) |- (had.east);
\draw[arr] (had.south)   -- (uproj.north);
\draw[arr] (uproj.south) -- (rms.north);
\draw[arr] (rms.south)   -- (uout.north);

\end{tikzpicture}

%% file: architecture-tree.tikz
\begin{tikzpicture}[
  >={Stealth[length=2.4mm, width=1.5mm]},
  font=\small,
  leaf/.style={
    draw, rounded corners=2pt, align=center,
    fill=orange!28, draw=orange!75!black,
    minimum height=0.55cm, minimum width=0.72cm,
    inner sep=2pt, line width=0.5pt,
    font=\scriptsize,
  },
  merge/.style={
    draw, rounded corners=2pt, align=center,
    fill=violet!18, draw=violet!60!black,
    minimum height=0.50cm, minimum width=0.95cm,
    inner sep=2pt, line width=0.55pt,
    font=\scriptsize,
  },
  mergeDeep/.style={
    draw, rounded corners=2pt, align=center,
    fill=violet!32, draw=violet!65!black,
    minimum height=0.50cm, minimum width=0.95cm,
    inner sep=2pt, line width=0.55pt,
    font=\scriptsize,
  },
  mergeRoot/.style={
    draw, rounded corners=2pt, align=center,
    fill=violet!46, draw=violet!75!black,
    minimum height=0.55cm, minimum width=1.05cm,
    inner sep=2pt, line width=0.6pt,
    font=\scriptsize,
  },
  lvlband/.style={
    draw, fill=violet!7, draw=violet!35!black,
    rounded corners=2pt,
    minimum height=0.40cm,
    line width=0.4pt,
    font=\scriptsize\itshape,
  },
  edgeBox/.style={
    draw, rounded corners=2pt, align=center,
    fill=teal!14, draw=teal!60!black,
    minimum height=0.55cm, minimum width=1.85cm,
    inner sep=2pt, line width=0.5pt,
    font=\scriptsize\itshape,
  },
  readout/.style={
    draw, rounded corners=2pt, align=center,
    fill=black!7, draw=black!55,
    minimum height=0.55cm, minimum width=3.1cm,
    inner sep=3pt, line width=0.5pt,
    font=\scriptsize\itshape,
  },
  rootOut/.style={
    draw, rounded corners=2pt, align=center,
    fill=black!4, draw=black!50,
    minimum height=0.55cm, minimum width=2.6cm,
    inner sep=3pt, line width=0.5pt,
    font=\small\itshape,
  },
  gSt/.style={
    draw, rounded corners=2pt, align=center,
    fill=magenta!14, draw=magenta!60!black,
    minimum height=0.50cm, minimum width=1.55cm,
    inner sep=2pt, line width=0.5pt,
    font=\scriptsize,
  },
  gIn/.style={
    draw, rounded corners=2pt, align=center,
    fill=magenta!9, draw=magenta!55!black,
    minimum height=0.50cm, minimum width=1.0cm,
    inner sep=2pt, line width=0.5pt,
    font=\scriptsize\itshape,
  },
  edge/.style       ={->, thick, line width=0.45pt, color=black!75},
  edgearr/.style    ={->, thick, color=teal!65!black, line width=0.55pt},
  biasarr/.style    ={->, thick, dashed, color=teal!65!black, line width=0.5pt},
  garr/.style       ={->, thick, color=magenta!70!black, line width=0.5pt},
  greadarr/.style   ={->, thick, dashed, color=magenta!60!black, line width=0.45pt},
  lvlabel/.style    ={font=\footnotesize\itshape, color=black!55, anchor=east},
]

\def\leafX#1{#1*0.85}

\foreach \t/\site in {0/1, 1/5, 2/3, 3/7, 4/2, 5/6, 6/4, 7/8} {
  \node[leaf] (L\t) at (\leafX{\t}, 0) {\site};
}

\foreach \i in {0,...,3} {
  \pgfmathsetmacro\xpos{\i*1.70 + 0.425}
  \node[merge] (M1\i) at (\xpos, 1.4) {merge};
}
\foreach \i in {0,...,3} {
  \pgfmathtruncatemacro\la{\i*2}
  \pgfmathtruncatemacro\lb{\i*2 + 1}
  \draw[edge] (L\la) -- (M1\i);
  \draw[edge] (L\lb) -- (M1\i);
}

\node[lvlband, minimum width=7.0cm] at (2.975, 2.05) (B1)
  {level attention (edge bias $+$ LCA-ALiBi)};

\foreach \i in {0,1} {
  \pgfmathsetmacro\xpos{\i*3.4 + 1.275}
  \node[mergeDeep] (M2\i) at (\xpos, 2.7) {merge};
}
\draw[edge] (M10.north) -- (M10.north |- B1.south);
\draw[edge] (M11.north) -- (M11.north |- B1.south);
\draw[edge] (M12.north) -- (M12.north |- B1.south);
\draw[edge] (M13.north) -- (M13.north |- B1.south);
\draw[edge] (M10.north |- B1.north) -- (M20.south west);
\draw[edge] (M11.north |- B1.north) -- (M20.south east);
\draw[edge] (M12.north |- B1.north) -- (M21.south west);
\draw[edge] (M13.north |- B1.north) -- (M21.south east);

\node[lvlband, minimum width=7.0cm] at (2.975, 3.35) (B2)
  {level attention (edge bias $+$ LCA-ALiBi)};

\node[mergeRoot] (Mr) at (2.975, 4.0) {merge};
\draw[edge] (M20.north) -- (M20.north |- B2.south);
\draw[edge] (M21.north) -- (M21.north |- B2.south);
\draw[edge] (M20.north |- B2.north) -- (Mr.south west);
\draw[edge] (M21.north |- B2.north) -- (Mr.south east);

\node[readout] (rr) at (2.975, 4.85)
  {root readout\;
   $W\bigl([\,\mathrm{RMS}(e_{\mathrm{root}}),\, c_{\mathrm{root}},\, g\,]\bigr)$};
\node[rootOut] (out) at (2.975, 5.65)
  {$\log\psi_\theta(q)\in\mathbb{C}$};
\draw[edge] (Mr.north) -- (rr.south);
\draw[edge] (rr.north) -- (out.south);

\node[edgeBox] (E0) at (7.9, 0.0)
  {$E^{(0)}$: $8 \times 8$};
\node[edgeBox] (E1) at (7.9, 1.4)
  {$E^{(1)}$: $4 \times 4$};
\node[edgeBox] (E2) at (7.9, 2.7)
  {$E^{(2)}$: $2 \times 2$};
\node[edgeBox] (E3) at (7.9, 4.0)
  {$e_{\mathrm{root}} = E^{(3)}_{00}$};
\draw[edgearr] (E0.north) -- (E1.south);
\draw[edgearr] (E1.north) -- (E2.south);
\draw[edgearr] (E2.north) -- (E3.south);
\node[font=\scriptsize\itshape, color=teal!50!black, anchor=west]
  at (8.95, 0.7)  {coarsen $+$ 2-FWL};
\node[font=\scriptsize\itshape, color=teal!50!black, anchor=west]
  at (8.95, 2.05) {coarsen $+$ 2-FWL};
\node[font=\scriptsize\itshape, color=teal!50!black, anchor=west]
  at (8.95, 3.35) {coarsen $+$ 2-FWL};

\draw[biasarr] (E1.west) -- (E1.west -| B1.east);
\draw[biasarr] (E2.west) -- (E2.west -| B2.east);
\draw[biasarr] (E3.west) |- (rr.east);

\node[gIn] (G0)  at (-1.6, 0)   {$g$};
\node[gSt] (GS1) at (-1.6, 1.4) {pool $+$ update};
\node[gSt] (GS2) at (-1.6, 2.7) {pool $+$ update};
\node[gSt] (GS3) at (-1.6, 4.0) {pool $+$ update};

\draw[garr] (G0.north)  -- (GS1.south);
\draw[garr] (GS1.north) -- (GS2.south);
\draw[garr] (GS2.north) -- (GS3.south);
\draw[garr, rounded corners=3pt] (GS3.north) |- (rr.west);

\draw[greadarr] (M10.west) -- (GS1.east);
\draw[greadarr] (M20.west) -- (GS2.east);
\draw[greadarr] (Mr.west)  -- (GS3.east);

\node[lvlabel] at (-2.75, 0)   {$\kappa = 0$ (leaves)};
\node[lvlabel] at (-2.75, 1.4) {$\kappa = 1$};
\node[lvlabel] at (-2.75, 2.7) {$\kappa = 2$};
\node[lvlabel] at (-2.75, 4.0) {$\kappa = K$ (root)};

\end{tikzpicture}

%% file: architecture-bitrev.tikz
\begin{tikzpicture}[
  >={Stealth[length=2.0mm, width=1.3mm]},
  font=\footnotesize,
  leaf/.style={
    draw, rounded corners=2pt, align=center,
    fill=orange!28, draw=orange!75!black,
    minimum height=0.55cm, minimum width=0.62cm,
    inner sep=2pt, line width=0.5pt,
    font=\scriptsize,
  },
  leafHi/.style={
    draw, rounded corners=2pt, align=center,
    fill=orange!50, draw=pauliz!85!black,
    minimum height=0.55cm, minimum width=0.62cm,
    inner sep=2pt, line width=1.0pt,
    font=\scriptsize\bfseries,
  },
  merge/.style={
    draw, rounded corners=2pt,
    fill=violet!15, draw=violet!60!black,
    minimum height=0.35cm, minimum width=0.45cm,
    inner sep=0pt, line width=0.5pt,
  },
  mergeHi/.style={
    draw, rounded corners=2pt,
    fill=violet!28, draw=pauliz!85!black,
    minimum height=0.35cm, minimum width=0.45cm,
    inner sep=0pt, line width=1.0pt,
  },
  edge/.style    ={-, thick, line width=0.45pt, color=black!70},
  hi/.style      ={-, ultra thick, line width=1.4pt, color=pauliz!85!black},
  ttl/.style     ={font=\small\itshape, color=black!75, align=center},
  capt/.style    ={font=\scriptsize\itshape, color=pauliz!80!black, align=center},
]


\begin{scope}[xshift=0cm]
  \node[ttl] at (2.625, 3.6) {Identity placement};
  \node[ttl, font=\scriptsize] at (2.625, 3.25)
    {site $i$ at slot $i-1$};

  \foreach \pos/\sitelabel in {0/1, 1/2, 2/3, 3/4, 4/5, 5/6, 6/7, 7/8} {
    \node[leaf] (Lc\pos) at (\pos*0.75, 0) {\sitelabel};
  }
  \node[leafHi] at (0.0, 0)  {1};
  \node[leafHi] at (3.0, 0)  {5};

  \foreach \i in {0,...,3} {
    \node[merge] (Mc1\i) at (\i*1.5 + 0.375, 0.85) {};
  }
  \foreach \i in {0,1} {
    \node[merge] (Mc2\i) at (\i*3.0 + 1.125, 1.7) {};
  }
  \node[merge] (Mcr) at (2.625, 2.55) {};

  \foreach \i in {0,1,2,3} {
    \pgfmathtruncatemacro\la{\i*2}
    \pgfmathtruncatemacro\lb{\i*2 + 1}
    \draw[edge] (Lc\la) -- (Mc1\i);
    \draw[edge] (Lc\lb) -- (Mc1\i);
  }
  \draw[edge] (Mc10) -- (Mc20); \draw[edge] (Mc11) -- (Mc20);
  \draw[edge] (Mc12) -- (Mc21); \draw[edge] (Mc13) -- (Mc21);
  \draw[edge] (Mc20) -- (Mcr); \draw[edge] (Mc21) -- (Mcr);

  \draw[hi] (Lc0) -- (Mc10) -- (Mc20) -- (Mcr) -- (Mc21) -- (Mc12) -- (Lc4);

  \node[mergeHi] at (Mc10) {};
  \node[mergeHi] at (Mc20) {};
  \node[mergeHi] at (Mcr)  {};
  \node[mergeHi] at (Mc21) {};
  \node[mergeHi] at (Mc12) {};

  \node[capt] at (2.625, -0.85)
    {$\mathrm{lca}(\text{site 1, site 5}) = 0$ (root)};
\end{scope}

\begin{scope}[xshift=6.9cm]
  \node[ttl] at (2.625, 3.6) {Bit-reversed (canonical frame)};
  \node[ttl, font=\scriptsize] at (2.625, 3.25)
    {site $i$ at slot $\mathrm{brev}_3(i-1)$};

  \foreach \pos/\sitelabel in {0/1, 1/5, 2/3, 3/7, 4/2, 5/6, 6/4, 7/8} {
    \node[leaf] (Lb\pos) at (\pos*0.75, 0) {\sitelabel};
  }
  \node[leafHi] at (0.0, 0)   {1};
  \node[leafHi] at (0.75, 0)  {5};

  \foreach \i in {0,...,3} {
    \node[merge] (Mb1\i) at (\i*1.5 + 0.375, 0.85) {};
  }
  \foreach \i in {0,1} {
    \node[merge] (Mb2\i) at (\i*3.0 + 1.125, 1.7) {};
  }
  \node[merge] (Mbr) at (2.625, 2.55) {};

  \foreach \i in {0,1,2,3} {
    \pgfmathtruncatemacro\la{\i*2}
    \pgfmathtruncatemacro\lb{\i*2 + 1}
    \draw[edge] (Lb\la) -- (Mb1\i);
    \draw[edge] (Lb\lb) -- (Mb1\i);
  }
  \draw[edge] (Mb10) -- (Mb20); \draw[edge] (Mb11) -- (Mb20);
  \draw[edge] (Mb12) -- (Mb21); \draw[edge] (Mb13) -- (Mb21);
  \draw[edge] (Mb20) -- (Mbr); \draw[edge] (Mb21) -- (Mbr);

  \draw[hi] (Lb0) -- (Mb10) -- (Lb1);
  \node[mergeHi] at (Mb10) {};

  \node[capt] at (2.625, -0.85)
    {$\mathrm{lca}(\text{site 1, site 5}) = K - 1 = 2$ (siblings)};
\end{scope}

\end{tikzpicture}

%% file: architecture-merge-node.tikz
\begin{tikzpicture}[
  >={Stealth[length=2.4mm, width=1.5mm]},
  font=\small,
  block/.style={
    draw, rounded corners=2pt, align=center,
    minimum height=0.62cm, minimum width=2.2cm,
    inner xsep=4pt, inner ysep=3pt,
    line width=0.55pt,
  },
  carrierBar/.style ={block, fill=violet!16, draw=violet!60!black,
                      minimum width=2.5cm, font=\small\itshape},
  carrierMid/.style ={block, fill=violet!28, draw=violet!65!black,
                      minimum width=3.4cm, font=\small\itshape},
  carrierOut/.style ={block, fill=violet!42, draw=violet!75!black,
                      minimum width=3.4cm, font=\small\itshape},
  mergeT/.style    ={block, fill=violet!50, draw=violet!75!black,
                     minimum height=1.1cm, minimum width=4.6cm,
                     text=white, font=\small},
  hyperBox/.style  ={block, fill=magenta!11, draw=magenta!60!black,
                     minimum height=0.95cm, minimum width=3.6cm},
  ctxIn/.style     ={block, fill=magenta!9,  draw=magenta!55!black,
                     minimum width=1.35cm, font=\small\itshape},
  ctxFFN/.style    ={block, fill=magenta!16, draw=magenta!60!black,
                     minimum height=1.0cm, minimum width=3.4cm},
  ctxOut/.style    ={block, fill=magenta!30, draw=magenta!65!black,
                     minimum width=2.6cm, font=\small\itshape},
  edgeIn/.style    ={block, fill=teal!16,    draw=teal!60!black,
                     minimum width=1.6cm, font=\small\itshape},
  clkIn/.style     ={block, fill=black!6,    draw=black!50,
                     minimum width=1.6cm, font=\footnotesize\itshape},
  sBar/.style      ={block, fill=black!6,    draw=black!50,
                     minimum width=1.25cm, font=\small\itshape},
  sOp/.style       ={block, fill=black!10,   draw=black!55,
                     minimum height=0.85cm, minimum width=3.0cm,
                     font=\footnotesize, align=center},
  sclNorm/.style   ={block, fill=black!7,    draw=black!55,
                     minimum width=3.4cm, font=\footnotesize, align=center},
  arr/.style       ={->, thick, line width=0.55pt},
  hyparr/.style    ={->, thick, color=magenta!70!black, line width=0.6pt},
  garr/.style      ={->, thick, color=black!45, line width=0.5pt},
  sarr/.style      ={->, thick, color=black!55, line width=0.55pt},
  ann/.style       ={font=\scriptsize\itshape, color=black!55,
                     fill=white, inner sep=1.5pt, align=center},
]

\node[ctxIn]      at (-5.6,  0.0) (cA) {$c_A$};
\node[ctxIn]      at (-3.6,  0.0) (cB) {$c_B$};
\node[carrierBar] at (-0.6,  0.0) (uA) {$u_A \in \mathbb{R}^{d_u}$};
\node[carrierBar] at ( 2.4,  0.0) (uB) {$u_B \in \mathbb{R}^{d_u}$};
\node[sBar]       at ( 4.9,  0.0) (sA) {$s_A$};
\node[sBar]       at ( 6.4,  0.0) (sB) {$s_B$};

\node[edgeIn] at (-8.4, -1.0)  (eAB) {$E^{(\kappa)}_{AB},\, E^{(\kappa)}_{BA}$};
\node[clkIn]  at (-8.4, -1.85) (clk)
  {clock $\tau_P$\\counts $\Phi_P$};
\node[clkIn]  at (-8.4, -2.7)  (gin) {$g$ (per level)};
\node[ctxFFN] at (-4.6, -1.6) (cffn)
  {context FFN};
\node[ctxFFN, minimum height=0.85cm] at (-4.6, -3.3) (cinterp)
  {normalised interpolation\\ \scriptsize gate $\alpha_c$};
\node[ctxOut] at (-4.6, -4.5) (cP) {$c_P \in \mathbb{R}^{d_c}$};

\draw[arr] (cA.south) -- (cA.south |- cffn.north);
\draw[arr] (cB.south) -- (cB.south |- cffn.north);
\draw[arr]  (eAB.east) -- ++(0.5, 0) |- ($(cffn.west) + (0, 0.28)$);
\draw[arr]  (clk.east) -- ++(0.5, 0) |- (cffn.west);
\draw[arr] (gin.east) -- ++(0.5, 0) |- ($(cffn.west) + (0, -0.28)$);
\draw[arr] (cffn.south)
  -- node[ann, right=1pt, pos=0.4] {$\Delta_P$}
  (cinterp.north);
\draw[arr] (cinterp.south) -- (cP.north);
\draw[arr, rounded corners=3pt]
  (cA.west) -- ++(-3.8, 0)
  |- node[ann, pos=0.78, below=1.5pt] {$\tfrac{1}{2}(c_A + c_B)$}
  (cinterp.west);

\node[mergeT] at (0.9, -1.6) (T)
  {blockwise rank-4 merge $T$\\
   \scriptsize literal, shared};
\node[carrierMid] at (0.9, -3.0) (raw) {$\tilde u \in \mathbb{R}^{d_u}$};
\node[hyperBox]   at (0.9, -4.15) (H)
  {$+$ residual hypernet\\
   \scriptsize $H(\tilde u, c_P)$};
\node[sclNorm]    at (0.9, -5.35) (sn)
  {scale-norm $\div\, \tilde s$\\
   \scriptsize unit RMS};
\node[carrierOut] at (0.9, -6.45) (uP) {$u_P \in \mathbb{R}^{d_u}$};

\draw[arr] (uA.south) -- (uA.south |- T.north);
\draw[arr] (uB.south) -- (uB.south |- T.north);
\node[ann] at ($(uA.south) + (0,-0.32cm)$) {reshape $B\!\times\!d_r$};
\node[ann] at ($(uB.south) + (0,-0.32cm)$) {reshape $B\!\times\!d_r$};
\draw[arr] (T.south)   -- (raw.north);
\draw[arr] (raw.south) -- (H.north);
\draw[arr] (H.south)   -- (sn.north);
\draw[arr] (sn.south)  -- (uP.north);
\draw[hyparr] (cP.east) -- ++(0.5, 0) |- (H.west);

\node[sOp]  at (5.65, -5.35) (sop)
  {log-scale update\\
   \scriptsize banks $\log \tilde s$};
\node[sBar, minimum width=2.2cm] at (5.65, -6.45) (sP) {$s_P$};

\draw[arr] (sA.south) -- (sA.south |- sop.north);
\draw[arr] (sB.south) -- (sB.south |- sop.north);
\draw[sarr] (sn.east) -- node[ann, above=2pt] {$\log \tilde s$} (sn.east -| sop.west);
\draw[arr] (sop.south) -- (sP.north);

\end{tikzpicture}

%% file: architecture-edge-coarsen.tikz
\begin{tikzpicture}[
  >={Stealth[length=2.4mm, width=1.5mm]},
  font=\small,
  cell/.style={draw=teal!60!black, fill=teal!14, line width=0.45pt,
               minimum width=0.46cm, minimum height=0.46cm, inner sep=0pt},
  cellHot/.style={cell, fill=teal!45},
  cellSelf/.style={cell, fill=teal!26},
  pcell/.style={draw=teal!65!black, fill=teal!18, line width=0.5pt,
                minimum width=0.66cm, minimum height=0.66cm, inner sep=0pt},
  pcellHot/.style={pcell, fill=teal!55},
  ctxChip/.style={draw=magenta!55!black, fill=magenta!10, rounded corners=2pt,
                  line width=0.5pt, font=\scriptsize\itshape,
                  minimum height=0.42cm, inner sep=2.5pt},
  opBox/.style={draw, rounded corners=2pt, align=center, line width=0.55pt,
                fill=teal!24, draw=teal!65!black,
                minimum height=0.95cm, minimum width=2.9cm},
  chip/.style={draw, rounded corners=2pt, align=center, line width=0.5pt,
               font=\footnotesize, minimum height=0.55cm, inner sep=3pt},
  arr/.style={->, thick, line width=0.55pt},
  ctxarr/.style={->, thick, dashed, color=magenta!70!black, line width=0.55pt},
  ann/.style={font=\scriptsize\itshape, color=black!55, align=center},
]

\foreach \i in {0,...,3} {
  \foreach \j in {0,...,3} {
    \ifnum\i=\j
      \node[cellSelf] (c\i\j) at (0.5*\j, -0.5*\i) {};
    \else
      \node[cell]     (c\i\j) at (0.5*\j, -0.5*\i) {};
    \fi
  }
}
\foreach \i in {0,1} { \foreach \j in {2,3} {
  \node[cellHot] at (0.5*\j, -0.5*\i) {};
}}
\node[ann, anchor=north] at (0.75, -1.95)
  {$E^{(\kappa)}$: $N_\kappa \times N_\kappa$ cells\\
   (diagonal = self-edges)};

\node[ctxChip] (cK) at ( 0.75,  1.75) {$c_{2P}, c_{2P+1}$};
\node[ctxChip] (cS) at ( 0.75,  0.95) {$c_{2Q}, c_{2Q+1}$};

\node[opBox] (op) at (3.45, -0.25)
  {$\mathrm{FFN}_E$\\ \scriptsize $2{\times}2$ block $+$ 4 contexts};

\draw[arr]    (c12.east) ++(0.08,0.25) -- (op.west);
\draw[ctxarr, rounded corners=3pt] (cK.east)  -| ($(op.north) + ( 0.55, 0)$);
\draw[ctxarr, rounded corners=3pt] (cS.east)  -| ($(op.north) + (-0.55, 0)$);

\foreach \i in {0,1} {
  \foreach \j in {0,1} {
    \ifnum\i=\j
      \node[pcell, fill=teal!28] (p\i\j) at (5.85 + 0.7*\j, 0.1 - 0.7*\i) {};
    \else
      \node[pcell] (p\i\j) at (5.85 + 0.7*\j, 0.1 - 0.7*\i) {};
    \fi
  }
}
\node[pcellHot] at (6.55, 0.1) {};
\node[ann, anchor=north] at (6.2, -1.3)
  {$E^{(\kappa+1)}$: one cell per\\ subtree pair $(P, Q)$};

\draw[arr] (op.east) -- ($(p01.west) + (-0.12, 0)$);

\node[chip, fill=teal!12,   draw=teal!60!black]   (fwl) at (8.95,  0.45)
  {2-FWL refine};
\node[chip, fill=violet!10, draw=violet!55!black] (lat) at (8.95, -0.85)
  {level attention};
\draw[arr] ($(p11.east) + (0.12, 0.25)$) -- (fwl.west);
\draw[arr] ($(p11.east) + (0.12,-0.20)$) |- (lat.west);

\end{tikzpicture}

%% file: architecture-tree-prefix-query.tikz
\begin{tikzpicture}[
  >={Stealth[length=2.2mm,width=1.4mm]},
  font=\scriptsize,
  box/.style={draw,rounded corners=2pt,align=center,inner xsep=4pt,
              inner ysep=3pt,line width=0.55pt},
  prefixslot/.style={box,fill=black!9,draw=black!50,
                    minimum width=0.68cm,minimum height=0.55cm,
                    inner xsep=1pt},
  candidateslot/.style={box,fill=teal!11,draw=teal!65!black,
                       minimum width=0.68cm,minimum height=0.55cm,
                       inner xsep=1pt},
  rootchip/.style={box,fill=violet!13,draw=violet!65!black,
                  minimum height=0.55cm},
  groupframe/.style={draw=black!48,rounded corners=2pt,
                    line width=0.65pt,inner sep=2pt},
  candidateframe/.style={draw=teal!60!black,rounded corners=2pt,
                        line width=0.65pt,inner sep=2pt},
  route/.style={box,fill=routecol!18,draw=routecol!80!black,
                minimum width=4.15cm,minimum height=0.68cm},
  cand/.style={box,fill=teal!12,draw=teal!65!black,
               minimum width=4.85cm,minimum height=0.68cm},
  post/.style={box,fill=violet!10,draw=violet!62!black,
               minimum width=4.55cm,minimum height=0.60cm},
  postframe/.style={draw=violet!58!black,rounded corners=3pt,
                   dashed,line width=0.65pt,inner xsep=5pt,
                   inner ysep=4pt},
  state/.style={box,fill=routeraddition!12,draw=routeraddition!75!black,
                minimum width=2.15cm,minimum height=0.62cm},
  result/.style={box,fill=routeraddition!17,draw=routeraddition!80!black,
                 minimum width=4.5cm,minimum height=0.68cm},
  quotient/.style={box,fill=routecol!14,draw=routecol!75!black,
                   minimum width=4.5cm,minimum height=0.68cm},
  pick/.style={box,fill=routeraddition!19,draw=routeraddition!85!black,
               minimum width=4.5cm,minimum height=0.68cm},
  arr/.style={->,line width=0.62pt},
  edge/.style={line width=0.52pt,draw=black!58},
  note/.style={font=\scriptsize\itshape,text=black!62,align=center},
]

\node[font=\small\bfseries,text=routeraddition!85!black]
  at (7.45,7.05) {multiscale prefix query at $t=7$};
\node[font=\scriptsize\bfseries,text=black!62]
  at (3.05,6.58) {selected prefix $p_{<7}$};
\node[font=\scriptsize\bfseries,text=teal!65!black]
  at (11.15,6.58) {remaining candidates $i\in U_7$};

\foreach \x/\s in {0.45/0,1.25/1,2.05/2,2.85/3,3.85/4,4.65/5,5.65/6}{
  \node[prefixslot] (v\s) at (\x,5.20) {$v_{\s}$};
}
\node[groupframe,fit=(v0)(v1)(v2)(v3)] (g04) {};
\node[groupframe,fit=(v4)(v5)] (g46) {};
\node[groupframe,fit=(v6)] (g67) {};

\node[rootchip] (z04) at (1.65,6.00) {$z[0,4)$};
\node[rootchip] (z46) at (4.25,6.00) {$z[4,6)$};
\node[rootchip] (z67) at (5.65,6.00) {$v_6$};
\draw[edge] (z04.south)--(g04.north);
\draw[edge] (z46.south)--(g46.north);
\draw[edge] (z67.south)--(g67.north);

\foreach \x/\s in {7.45/1,8.25/2,9.05/3,9.85/4,10.65/5,11.45/6,12.25/7,13.05/8,13.85/9,14.65/10}{
  \node[candidateslot] (c\s) at (\x,5.20) {$x_{i_{\s}}$};
}
\node[candidateframe,fit=(c1)(c2)(c3)(c4)(c5)(c6)(c7)(c8)(c9)(c10)]
  (candidateRow) {};

\node[note] at (3.05,4.55)
  {};
\node[note,text=routeraddition!82!black] (coverSet) at (3.35,4.90)
  {};
\node[note,text=teal!65!black] (candidateSet) at (11.15,4.90)
  {};

\draw[black!25,line width=0.5pt] (0.05,3.72)--(14.95,3.72);

\node[route] (queryInput) at (3.35,2.90)
  {dyadic cover $\mathcal C_7$ and conditioned seed $\bar g+\gamma(7)$};
\node[route] (coverQuery) at (3.35,2.02)
  {read-only cover query $\times4$\\[-1pt]
   \scriptsize the query reads all three roots; the roots are not updated};
\node[state] (y7) at (3.35,-1.18)
  {$y_7$\\[-1pt]\scriptsize retained for later positions};
\draw[arr] (queryInput.south)--(coverQuery.north);
\draw[arr] (coverQuery.south)--(y7.north);

\node[cand] (candidateInput) at (11.15,2.90)
  {fresh candidate row $\{x_i^{(7)}:i\in U_7\}$};
\node[cand] (candidateCover) at (11.15,2.02)
  {candidate--cover cross-attention $\times2$\\[-1pt]
   \scriptsize candidates query $\mathcal C_7$; pooled edge bias; FFNs receive $g^{(7)}$};
\draw[arr] (candidateInput.south)--(candidateCover.north);
\draw[arr] (coverSet.south)--(queryInput.north);
\draw[arr] (candidateSet.south)--(candidateInput.north);

\node[post] (history) at (11.15,1.10)
  {history cross-attention to the earlier query states $y_{<7}$};
\node[post] (suffix) at (11.15,0.34)
  {suffix self-attention among the remaining candidates};
\node[post] (ffn) at (11.15,-0.42)
  {conditioned FFN};
\node[postframe,fit=(history)(suffix)(ffn)] (postStack) {};
\node[font=\scriptsize\bfseries,text=violet!65!black,anchor=west,
      align=left,inner sep=1pt]
  at ($(postStack.east)+(0.12,0)$) {post-prefix\\block $\times4$};
\node[state,minimum width=2.75cm] (refined) at (11.15,-1.18)
  {refined candidates $\bar x_i^{(7)}$};
\draw[arr] (candidateCover.south)--(history.north);
\draw[arr] (history.south)--(suffix.north);
\draw[arr] (suffix.south)--(ffn.north);
\draw[arr] (ffn.south)--(refined.north);

\node[result] (pointer) at (7.25,-2.18)
  {tied dot-product pointer\\[-1pt]
   \scriptsize score $y_7$ against every refined candidate};
\draw[arr] (y7.south east)--(pointer.north);
\draw[arr] (refined.south west)--(pointer.north);

\node[quotient] (quotient) at (7.25,-3.08)
  {conditional symmetry quotient\\[-1pt]
   \scriptsize collapse the remaining candidates into WL classes};
\node[pick] (pick) at (7.25,-3.98)
  {sample among class representatives (Gumbel pick)};
\draw[arr] (pointer.south)--(quotient.north);
\draw[arr] (quotient.south)--(pick.north);

\node[note] at (1.05,-3.48)
  {};

\end{tikzpicture}

%% file: architecture-route-decoder.tikz
\begin{tikzpicture}[
  >={Stealth[length=2.4mm, width=1.5mm]},
  font=\small,
  block/.style={draw,rounded corners=2pt,align=center,
                inner xsep=6pt,inner ysep=4pt,line width=0.55pt},
  context/.style={block,draw=black!55,fill=black!5,
                  minimum width=13.2cm,minimum height=0.88cm},
  branch/.style={block,fill=routecol!12,draw=routecol!75!black,
                 minimum width=6.2cm,minimum height=1.15cm,
                 text width=5.78cm},
  state/.style={block,fill=violet!12,draw=violet!55!black,
                minimum width=5.3cm,minimum height=0.62cm,
                font=\small\itshape},
  stage/.style={block,fill=routecol!22,draw=routecol!78!black,
                minimum width=9.6cm,minimum height=1.15cm},
  ptr/.style={block,fill=routecol!34,draw=routecol!80!black,
              minimum width=7.2cm,minimum height=0.82cm},
  quot/.style={block,fill=routecol!26,draw=routecol!80!black,
               minimum width=7.2cm,minimum height=0.82cm},
  outBox/.style={block,fill=routecol!46,draw=routecol!85!black,
                 minimum width=8.0cm,minimum height=0.72cm,
                 font=\small\itshape},
  arr/.style={->,thick,line width=0.55pt},
  rail/.style={->,thick,color=routecol!80!black,line width=0.52pt},
  ann/.style={font=\scriptsize\itshape,color=black!58,
              fill=white,inner sep=1.5pt,align=center},
]

\node[context] (input) at (0,0)
  {$\ell^{(L)},\,e^{(L)},\,g\ ;\qquad
    p_{<t},\,\widetilde v_{<t},\,y_{<t}$\\[-1pt]
   \scriptsize fixed Hamiltonian fork and current prefix state};

\node[branch] (composer) at (-3.55,-1.55)
  {\textbf{Dynamic candidate composer}\\
   \scriptsize streamwise RMS + raw hole ratios $\to$ additive fusion\\[-1pt]
   \scriptsize one residual candidate block};
\node[branch] (prefix) at (3.55,-1.55)
  {\textbf{Prefix tree and multiscale query}\\
   \scriptsize retained leaves $\to$ dyadic tree $\to$ cover $\mathcal C_t$\\[-1pt]
   \scriptsize same-level causal refinement; read-only cover query $\times4$};

\draw[arr] ($(input.south)+(-3.55cm,0)$) -- (composer.north);
\draw[arr] ($(input.south)+( 3.55cm,0)$) -- (prefix.north);

\node[state] (fresh) at (-3.55,-2.82)
  {$x_i^{(t)}\in\mathbb R^{d_h}$ for every $i\in U_t$};
\node[state] (cover) at (3.55,-2.82)
  {$\mathcal C_t\ \text{(read-only roots)},\qquad y_t\in\mathbb R^{d_h}$};
\draw[arr] (composer.south) -- (fresh.north);
\draw[arr] (prefix.south) -- (cover.north);

\node[stage] (candcover) at (0,-4.12)
  {\textbf{Candidate--cover cross-attention $\times2$}\\
   \scriptsize candidate queries read $\mathcal C_t$; cover roots remain read-only\\[-1pt]
   \scriptsize FFNs alone receive
     $g^{(t)}=\operatorname{Update}\!\bigl(g,\operatorname{pool}(\mathcal C_t)\bigr)$};
\draw[arr,rounded corners=2pt] (fresh.south east) -- (candcover.north);
\draw[arr,rounded corners=2pt] (cover.south west) -- (candcover.north);

\node[stage] (post) at (0,-5.70)
  {\textbf{Post-prefix candidate stack $\times4$}\\
   \scriptsize history cross-attention to $y_{<t}$ $\to$ suffix self-attention $\to$ FFN\\[-1pt]
   \scriptsize output $\bar x_i^{(t)}$ for every remaining candidate};
\draw[arr] (candcover.south) -- (post.north);

\draw[rail,rounded corners=2pt]
  (input.west) -- (-7.50,0) -- (-7.50,-5.70) -- (post.west);
\node[ann] at (-6.15,-5.47) {query history $y_{<t}$};

\node[ptr] (pointer) at (0,-7.20)
  {\textbf{Tied dot-product pointer}\\[-1pt]
   \scriptsize pairs $y_t$ with every $\bar x_i^{(t)}$; raw logits $\eta_{t,i}$};
\draw[arr] (post.south) -- (pointer.north);
\draw[rail,rounded corners=2pt]
  (cover.east) -- (7.50,-2.82) -- (7.50,-7.20) -- (pointer.east);
\node[ann] at (6.73,-6.66) {decode query $y_t$};

\node[quot] (quotient) at (0,-8.40)
  {\textbf{Conditional symmetry quotient}\\[-1pt]
   \scriptsize WL classes of the remaining candidates};
\draw[arr] (pointer.south) -- (quotient.north);

\node[outBox] (policy) at (0,-9.55)
  {$\pi_\theta(p_t\mid p_{<t},h,\ell^{(L)},e^{(L)},m)$\\[-1pt]
   {\scriptsize sample among class representatives}};
\draw[arr] (quotient.south) -- (policy.north);

\end{tikzpicture}

%% file: architecture-route-schedules.tikz
\begin{tikzpicture}[
  >={Stealth[length=2.2mm,width=1.4mm]},
  font=\scriptsize,
  box/.style={draw,rounded corners=2pt,align=center,inner xsep=4pt,
              inner ysep=3pt,line width=0.55pt},
  sample/.style={box,fill=routecol!17,draw=routecol!80!black,
                 minimum width=4.8cm,minimum height=0.72cm},
  teacher/.style={box,fill=teal!11,draw=teal!70!black,
                  minimum width=4.8cm,minimum height=0.72cm},
  state/.style={box,fill=routeraddition!10,draw=routeraddition!75!black,
                minimum width=4.8cm,minimum height=0.65cm},
  io/.style={box,fill=black!4,draw=black!45,
             minimum width=3.5cm,minimum height=0.58cm},
  arr/.style={->,line width=0.6pt},
]

\node[box,fill=routeraddition!6,draw=routeraddition!60!black,
      minimum width=10.8cm,minimum height=0.62cm]
  (same) at (0,4.25)
  {same $\mathcal D_\theta$, parameters, masks, and row logits
   (up to floating-point evaluation order)};

\node[font=\small\bfseries,text=black]
  at (-4.0,3.45) {Sequential ancestral sampling};
\node[font=\small\bfseries,text=black]
  at (4.0,3.45) {Vectorised teacher-forced scoring};
\draw[black!25,line width=0.5pt] (0,3.15)--(0,-2.75);

\node[io] (staticL) at (-4.0,2.65)
  {prefix-independent Hamiltonian tables\\and structural biases};
\node[sample] (rowL) at (-4.0,1.45)
  {build fresh row $X_t=(x_i^{(t)})_i$\\
   from prefix/suffix/order/hole scans};
\node[sample] (treeL) at (-4.0,0.25)
  {retain $\tilde v_s=x_{p_s}^{(s)}$; tree uses
   $v_s=\operatorname{Sphere}(\tilde v_s)$\\
   tree scan $\to\mathcal C_t\to y_t$};
\node[sample] (scoreL) at (-4.0,-0.95)
  {candidate--cover + post-prefix blocks\\
   pointer $\to$ quotient $\to$ Gumbel $p_t$};
\node[state] (cacheL) at (-4.0,-2.15)
  {retain $\tilde v_t$, $y_t$, and summary scans\\
   \textbf{no projected K/V cache}; then advance to $t+1$};
\draw[arr] (staticL.south)--(rowL.north);
\draw[arr] (rowL.south)--(treeL.north);
\draw[arr] (treeL.south)--(scoreL.north);
\draw[arr] (scoreL.south)--(cacheL.north);

\node[io] (routeR) at (4.0,2.65)
  {stored route $p=(p_0,\ldots,p_{\hat N-1})$};
\node[teacher] (allR) at (4.0,1.45)
  {construct all dynamic rows together\\
   $\mathbf X\in\mathbb R^{\hat N\times\hat N\times d_h}$};
\node[teacher] (gatherR) at (4.0,0.25)
  {gather $\tilde v_t=\mathbf X_{t,p_t}$ and sphere it; one tree scan\\
   materialises every root at every level};
\node[teacher] (coverR) at (4.0,-0.95)
  {dynamic gather of each $\mathcal C_t$\\
   all masked scoring rows in parallel};
\node[state] (logitR) at (4.0,-2.15)
  {candidate-ID logits with a route-permuted triangular mask\\
   and exact quotiented route log-probability};
\draw[arr] (routeR.south)--(allR.north);
\draw[arr] (allR.south)--(gatherR.north);
\draw[arr] (gatherR.south)--(coverR.north);
\draw[arr] (coverR.south)--(logitR.north);

\end{tikzpicture}

%% file: augmented-training-hotswap.tikz
\begin{tikzpicture}[
  >=Latex,
  font=\footnotesize,
  box/.style={
    draw=changegreen!85!black,
    rounded corners=2pt,
    align=center,
    minimum height=8mm,
    inner xsep=4pt,
    inner ysep=3pt,
    fill=changegreen!5,
    text=black
  },
  cpu/.style={box},
  gpu/.style={box, fill=changegreen!2},
  arr/.style={
    ->,
    semithick,
    draw=changegreen!85!black
  }
]

\node[cpu] (base) {
  base system\\
  $(J,h)$
};

\node[cpu, right=7mm of base, text width=42mm] (mutate) {
  uniformly draw one mutation\\[-1pt]
  \scriptsize
  keep $\mid$ add bond $\mid$ perturb bond\\
  remove bond $\mid$ add orbit-symmetric field
};

\node[cpu, right=7mm of mutate] (refs) {
  build references\\
  and route fields
};

\node[cpu, right=7mm of refs] (cache) {
  round cache\\
  look-ahead $\dsLookahead$
};

\draw[arr] (base) -- (mutate);
\draw[arr] (mutate) -- (refs);
\draw[arr] (refs) -- (cache);

\node[gpu, below=16mm of cache] (swap) {
  epoch-boundary\\
  hot-swap
};

\node[gpu, left=7mm of swap] (train) {
  optimisation\\
  step
};

\node[gpu, left=7mm of train, text width=38mm] (gauge) {
  post-MCMC speculative\\
  re-gauging\\
  zero measured overhead
};

\node[gpu, left=7mm of gauge] (mcmc) {
  GPU MCMC\\
  current round
};

\draw[arr] (mcmc) -- (gauge);
\draw[arr] (gauge) -- (train);
\draw[arr] (train) -- (swap);
\draw[arr] (cache.south) -- (swap.north);

\end{tikzpicture}

%% file: architecture-mcmc.tikz
\begin{tikzpicture}[
  >={Stealth[length=2.4mm, width=1.5mm]},
  font=\small,
  box/.style={draw, rounded corners=2pt, align=center, inner sep=5pt,
              minimum height=0.95cm, minimum width=6.4cm, line width=0.6pt},
  state/.style={draw, rounded corners=3pt, fill=mcmcsample!20,
                draw=mcmcsample!75!black, minimum height=0.6cm,
                minimum width=2.4cm, align=center, inner sep=4pt,
                font=\small\itshape, line width=0.55pt},
  setup/.style={box, fill=mcmcsample!16, draw=mcmcsample!70!black},
  haar/.style ={box, fill=mcmcsample!32, draw=mcmcsample!80!black, dashed},
  move/.style ={box, fill=mcmcsample!48, draw=mcmcsample!82!black},
  swapb/.style={box, fill=mcmcsample!64, draw=mcmcsample!85!black},
  adapt/.style={box, fill=mcmcsample!12, draw=mcmcsample!60!black,
                minimum height=0.9cm, minimum width=4.2cm, font=\footnotesize},
  arr/.style  ={->, line width=0.6pt},
  feed/.style ={->, dashed, color=magenta!70!black, line width=0.7pt},
  lbl/.style  ={font=\footnotesize, color=black!65, fill=white,
                inner sep=1.5pt, align=center},
  flbl/.style ={font=\footnotesize, color=magenta!65!black, fill=white,
                inner sep=1.5pt, align=center},
]

\node[state] (qt) {$q^{(r)}_{t}$};
\node[setup, below=0.5cm of qt] (hoist) {%
  \textbf{Hoist per-system invariants}\\[1pt]
  {\footnotesize featurizer and trunk, independent of $q$}};
\node[setup, below=0.38cm of hoist] (cache) {%
  \textbf{Cache} $\log |\psi_{\theta}|^{2}$ \textbf{and} $\nabla_{\!a} \log |\psi_{\theta}|^{2}$\\[1pt]
  {\footnotesize at the current $q$, refreshed model}};
\node[haar, below=0.38cm of cache] (haar) {%
  \textbf{Haar refresh of $m$ sites}\\[1pt]
  {\footnotesize global, once before the loop}};
\node[move, below=0.8cm of haar] (mala) {%
  \textbf{MALA local move}\\[1pt]
  {\footnotesize propose from cached $\nabla$, wrap-aware MH accept}};
\node[swapb, below=0.42cm of mala] (swap) {%
  \textbf{DEO cross-chain swap}\\[1pt]
  {\footnotesize even or odd edges by step parity}};
\node[state, below=0.8cm of swap] (qt1) {$q^{(r)}_{t+1}$};

\draw[arr] (qt)    -- (hoist);
\draw[arr] (hoist) -- (cache);
\draw[arr] (cache) -- (haar);
\draw[arr] (haar)  -- (mala);
\draw[arr] (mala)  -- node[lbl, right=1pt]{$q'$} (swap);
\draw[arr] (swap)  -- (qt1);

\draw[arr] (swap.west) -- ++(-1.15,0) |- (mala.west);
\node[lbl, anchor=east]
  at ($(mala.west)!0.5!(swap.west) + (-1.2,0)$)
  {loop \\ parity flips};

\node[adapt, right=3.2cm of haar] (am) {%
  \textbf{Tune refresh count $m$}};
\node[adapt, right=3.2cm of mala] (rm) {%
  \textbf{Multiplicative rule on $\sigma$}};
\node[adapt, right=3.2cm of swap] (hyman) {%
  \textbf{Hyman cubic on $\beta$-grid}};

\draw[feed] ([yshift=2.2mm]haar.east) -- node[flbl, above]{Haar acceptance}
            ([yshift=2.2mm]am.west);
\draw[feed] ([yshift=-2.2mm]am.west) -- node[flbl, below]{updated $m$}
            ([yshift=-2.2mm]haar.east);
\draw[feed] ([yshift=2.2mm]mala.east) -- node[flbl, above]{acceptance $a_{t}$}
            ([yshift=2.2mm]rm.west);
\draw[feed] ([yshift=-2.2mm]rm.west) -- node[flbl, below]{updated $\sigma$}
            ([yshift=-2.2mm]mala.east);
\draw[feed] ([yshift=2.2mm]swap.east) -- node[flbl, above]{pair rejection}
            ([yshift=2.2mm]hyman.west);
\draw[feed] ([yshift=-2.2mm]hyman.west) -- node[flbl, below]{updated $\beta$-grid}
            ([yshift=-2.2mm]swap.east);

\end{tikzpicture}

%% file: architecture-replica.tikz
\begin{tikzpicture}[
  font=\small,
  >={Stealth[length=2.2mm, width=1.4mm]},
  chainline/.style={line width=0.75pt},
  state/.style={circle, inner sep=0pt, minimum size=4pt},
  swap/.style={<->, thick, line width=0.7pt, color=violet!70!black},
  ylabel/.style={font=\footnotesize\itshape, anchor=east, color=black!70},
  xlabel/.style={font=\footnotesize\itshape, color=black!70},
  ttick/.style={font=\footnotesize, color=black!55},
  edgetype/.style={font=\footnotesize, color=violet!70!black, fill=white, inner sep=1.5pt},
]

\def\dt{1.4}     
\def\dy{0.85}    
\def\Nt{6}       

\definecolor{cHot}{HTML}{D55E00}
\definecolor{cWarm}{HTML}{E69F00}
\definecolor{cCool}{HTML}{56B4E9}
\definecolor{cCold}{HTML}{0072B2}

\foreach \i/\col/\lab in {1/cCold/{$\beta_{1}$ (cold)},
                          2/cCool/{$\beta_{2}$},
                          3/cWarm/{$\beta_{3}$},
                          4/cHot/{$\beta_{R}$ (hot)}} {
  \draw[chainline, color=\col!85!black, opacity=0.8]
        (0, -\i*\dy) -- (\Nt*\dt, -\i*\dy);
  \foreach \t in {0,...,\Nt} {
    \node[state, fill=\col!85!black] at (\t*\dt, -\i*\dy) {};
  }
  \node[ylabel] at (-0.25, -\i*\dy) {\lab};
}

\foreach \t in {0, 2, 4, 6} {
  \draw[swap] (\t*\dt, -1*\dy - 0.14) -- (\t*\dt, -2*\dy + 0.14);
  \draw[swap] (\t*\dt, -3*\dy - 0.14) -- (\t*\dt, -4*\dy + 0.14);
}
\foreach \t in {1, 3, 5} {
  \draw[swap] (\t*\dt, -2*\dy - 0.14) -- (\t*\dt, -3*\dy + 0.14);
}

\foreach \t/\parity in {0/even, 1/odd, 2/even, 3/odd, 4/even, 5/odd, 6/even} {
  \node[ttick] at (\t*\dt, -4*\dy - 0.55) {$t=\t$};
  \node[edgetype] at (\t*\dt, -4*\dy - 1.05) {\parity};
}

\node[xlabel, anchor=west] at (\Nt*\dt + 0.3, -1*\dy) {chain index};
\node[xlabel] at (\Nt*\dt/2, -4*\dy - 1.55) {step parity $\to$ which edge set swaps};


\end{tikzpicture}